\documentclass[11pt]{article}

\usepackage[a4paper,margin=1in]{geometry}

\usepackage{booktabs}
\usepackage{array}
\usepackage{multirow}
\usepackage{float}
\usepackage{tabularx}
\usepackage{ragged2e}

\usepackage{microtype}
\usepackage[utf8]{inputenc}
\usepackage{mathpazo} 

\usepackage{mathtools,amssymb,amsthm}
\usepackage{dsfont}

\usepackage{enumitem}

\usepackage{graphicx}
\usepackage{tikz}

\usepackage[dvipsnames,svgnames,x11names]{xcolor}

\usepackage{cite}
\usepackage{doi}

\usepackage{hyperref}
\usepackage[capitalize]{cleveref}

\DeclareMathOperator{\Erf}{erf}
\DeclareMathOperator{\Erfc}{erfc}
\DeclarePairedDelimiterXPP\erf[1]{\Erf\mkern1mu}(){}{#1}
\DeclarePairedDelimiterXPP\erfc[1]{\Erfc\mkern1mu}(){}{#1}

\newcommand{\deq}{\coloneqq}
\newcommand{\eps}{\varepsilon}
\newcommand{\ot}{\otimes}

\newcommand{\C}{\mathbb{C}}
\newcommand{\Z}{\mathbb{Z}}

\newcommand{\Nat}{\mathbb{N}}

\newcommand{\zo}{\{0,1\}}
\newcommand{\bits}{\{0,1\}}

\newcommand{\norm}[1]{\left \| #1 \right \|}
\newcommand{\paren}[1]{ \left ( #1 \right )}

\newcommand{\ket}[1]{{|#1\rangle}}
\newcommand{\bra}[1]{{\langle#1|}}

\newcommand{\ketbra}[2]{|#1\rangle\! \langle #2|}

\newcommand{\braket}[2]{ \langle #1 | #2 \rangle}

\newcommand{\id}{\mathds{1}} 

\makeatletter
\DeclareMathOperator{\Tr@noarg}{Tr}
\newcommand{\Tr}{
   \@ifnextchar\bgroup
      {\Tr@witharg}
      {\Tr@noarg}
}
\newcommand{\Tr@witharg}[1]{
   \ensuremath{\Tr\left[#1\right]}
}
\makeatother
\DeclareMathOperator{\Pos}{Pos}
\DeclareMathOperator{\Obs}{Obs}
\DeclareMathOperator{\Cliff}{Cliff}
\DeclareMathOperator{\Herm}{Herm}

\newcommand{\textprotocol}[1]{\textnormal{\textsc{#1}}}

\renewcommand{\H}{\mathcal{H}} 

\newcommand{\Qa}{\mathcal{Q}_A}
\newcommand{\Qb}{\mathcal{Q}_B}

\DeclareMathOperator*{\E}{\mathbb{E}}

\DeclareMathOperator{\poly}{poly}

\newcommand{\BQP}{\mathsf{BQP}}
\newcommand{\QPT}{\mathsf{QPT}}
\newcommand{\nuQPT}{\mathsf{nuQPT}}
\newcommand{\PPT}{\mathsf{PPT}}

\newcommand{\Gen}{\mathsf{Gen}}
\newcommand{\Enc}{\mathsf{Enc}}
\newcommand{\Dec}{\mathsf{Dec}}
\newcommand{\Eval}{\mathsf{Eval}}

\newcommand{\sk}{\mathsf{sk}}
\newcommand{\evk}{\mathsf{evk}}

\newcommand{\ct}{\mathsf{ct}}

\newcommand{\QFHE}{\mathsf{QFHE}}

\makeatletter
\newcommand{\draw@noarg}{\,{\overset{\raisebox{1px}{$\scriptscriptstyle\; R$}}{\leftarrow}}\,}
\newcommand{\draw}{
   \@ifnextchar\bgroup
      {\draw@witharg}
      {\draw@noarg}
}
\newcommand{\draw@witharg}[1]{\,{\overset{\raisebox{1px}{$\scriptscriptstyle\;\;#1$}}{\leftarrow}}\,}
\makeatother

\newcommand{\Encsk}{\mathsf{Enc}_\mathsf{sk}}
\newcommand{\Decsk}{\mathsf{Dec}_\mathsf{sk}}

\newcommand{\cA}{\mathcal{A}}  \newcommand{\cC}{\mathcal{C}} 
\newcommand{\cE}{\mathcal{E}}   \newcommand{\cH}{\mathcal{H}}

\newcommand{\cQ}{\mathcal{Q}}  \newcommand{\cS}{\mathcal{S}} 
   \newcommand{\cX}{\mathcal{X}}

\newcommand{\ea}{\eta(\lambda)}
\newcommand{\pca}{{\psi^c_\alpha}}
\newcommand{\pea}[1]{{\psi^{\Enc(#1)}_\alpha}}
\newcommand{\pe}[1]{{\psi^{\Enc(#1)}}}
\newcommand{\ind}[1]{\mathbf{1}\{#1\}}
\newcommand{\peq}{{\psi^{\Enc(q)}}}

\newcommand{\kca}{\ket{\pca}}

\newcommand{\Eazz}{\E_{\substack{a \in \zo^n\\ |a| \equiv 0\pmod 2 }}}
\newcommand{\Eazzm}{\E_{\substack{a \sim\mu\\ |a| \equiv 0\pmod 2 }}}

\newcommand{\numberthis}{\addtocounter{equation}{1}\tag{\theequation}}
\newcommand{\restr}[2]{\left.#1\right|_{#2}}
\newcommand{\exten}[2]{#1^{(#2)}}

\newcommand{\floor}[1]{\left\lfloor #1 \right\rfloor}
\newcommand{\ceil}[1]{\left\lceil #1 \right\rceil}

\makeatletter
\newtheorem*{rep@theorem}{\rep@title}
\newcommand{\newreptheorem}[2]{%
\newenvironment{rep#1}[1]{%
 \def\rep@title{#2 \ref{##1} (restated)}%
 \begin{rep@theorem}}%
 {\end{rep@theorem}}}
\makeatother

\newtheorem{theorem}{Theorem}[section]
\newreptheorem{theorem}{Theorem}

\newtheorem{informaltheorem}{Informal Theorem}[section]
\newreptheorem{informaltheorem}{Informal Theorem}
\newtheorem*{informaltheorem*}{Informal Theorem}

\newtheorem{lemma}{Lemma}[section]
\newreptheorem{lemma}{Lemma}

\newreptheorem{proposition}{Proposition}

\newtheorem{corollary}{Corollary}[section]
\newreptheorem{corollary}{Corollary}

\newtheorem{definition}{Definition}[section]

\theoremstyle{remark}
\newtheorem{remark}{\textit{Remark}}[section]

\numberwithin{equation}{section}

\usepackage{tcolorbox}
\tcbuselibrary{breakable}
\usetikzlibrary{arrows.meta,positioning,fit,backgrounds}
\newcounter{protocol}
\newtcolorbox{protocolbox}[1]{
    breakable,
    colback=white,
    label type = protocol,
    colframe=black,
    fonttitle=\bfseries,
    title={%
        Protocol~\theprotocol
        \ifx#1\empty\else:~#1\fi
    },
    code={\refstepcounter{protocol}},
    left=6pt, right=6pt, top=6pt, bottom=6pt
}
\crefname{protocol}{protocol}{protocols}
\Crefname{protocol}{Protocol}{Protocols}

\hypersetup{%
  colorlinks=true,
  citecolor=DarkGreen,
  linkcolor=DarkRed,
  pdftitle  = {Classical Verification of Quantum Computation with Quasilinear Resources, from Compiled Nonlocal Games},
  pdfauthor = {Finn Holler and Anand Natarajan},
  pdfkeywords = {quantum computation, classical verification, compiled nonlocal games, self-testing, LWE},
}

\title{\bfseries Classical Verification of Quantum Computation with\\
       Quasilinear Resources, from Compiled Nonlocal Games}

\author{Finn Holler\footnote{fholler@ethz.ch}\\
\footnotesize ETH Zurich
\and
Anand Natarajan\footnote{anand@mit.edu}\\
\footnotesize MIT}

\date{September 29, 2026}

\begin{document}

\tikzset{every picture/.style={line width=0.75pt}}
\setlength{\parindent}{0cm}

\maketitle

\begin{abstract}

Computational self-testing gives a classical verifier command over the quantum register of a single computationally bounded prover. We use this framework to construct the first argument system for $\mathsf{BQP}$ with quasilinear total resource requirements in the circuit model. Our argument system is based on the learning with errors (LWE) assumption and requires total resources of $O(\poly(\lambda, \log g)\cdot g)$ for delegating a circuit with $g$ gates, where $\lambda$ is the LWE security parameter. This is achieved by constructing a new computational self-test for certifying the prover's quantum state and using it to dequantize the efficient verification protocol of Broadbent (ToC 2018). Specifically, this self-test enables the verifiable, random remote state preparation of tensor product states of the single-qubit Clifford observables $\sigma_X, \sigma_Y, \sigma_Z, (\sigma_Y-\sigma_X)/\sqrt{2}$ and $(\sigma_Y+\sigma_X)/\sqrt{2}$,
with constant robustness: the verification error is independent of the number of prepared qubits. This approach was first proposed by Coladangelo et al. (ToC 2024) in the multi-prover setting. We replicate their result in the single-prover setting by applying the compiler proposed by Kalai et al. (STOC 2023)---which turns any nonlocal game into a single-prover argument system---to a modified version of their self-test.

\end{abstract}
\newpage
{
\setcounter{tocdepth}{3}
\tableofcontents
}

\newpage

\section{Introduction}\label{chapter:introduction}

Quantum computers are becoming more capable, with a centralized `quantum cloud' being the dominant deployment model. In the near term, users are thus unlikely to own quantum hardware and they will have to access it remotely, on machines they can't control. This makes trustless classical delegation of quantum computation an important primitive.
Classically, this topic has been extensively studied and its practical relevance is clear from numerous applications in blockchain technology and cloud computing. Extending the idea to the quantum realm, we would like that a fully classical verifier could delegate an arbitrary quantum computation to a single untrusted prover, having a strong correctness guarantee for the purported outcome. How to achieve this functionality is far from obvious, and doing so efficiently with information-theoretic security is still a major open problem. 
\par
\medskip
In a breakthrough result Mahadev showed that \emph{classical verification of quantum computations} (CVQC) is indeed possible under cryptographic assumptions, constructing the first single-prover, classical-verifier argument system for $\mathsf{BQP}$ \cite{Mahadev2018}. Her protocol has overall complexity of $O(g^3)$ to verify a circuit of $g$ gates, placing high resource demands on the quantum prover. The overhead stems from a circuit-to-Hamiltonian reduction, which allows the verifier to check the computation by testing that there exists a Feynman-Kitaev history state with sufficiently low energy. Since the best-known circuit-to-Hamiltonian reductions \cite{Bausch2018} have an inverse-polynomial promise gap---distinguishing good from bad computations---a polynomial overhead is inherent to all works taking a Hamiltonian approach\footnote{A positive resolution of the quantum PCP conjecture as formulated in \cite{Aharonov2013} would improve this situation.}. This situation is unsatisfactory, since CVQC acts as a primitive in other works, which inherit the cubic overhead.
\par 
\medskip
Zhang partially addressed this issue by proposing the first linear-time CVQC protocol for measurement-based computations in the quantum random-oracle model \cite{Zhang2022}\footnote{The introduction of his paper provides a thorough overview of the protocol landscape.}. Although the measurement-based and circuit models are equivalent, converting from the latter to the former for a fixed cluster state incurs a polynomial factor in the circuit size, due to locality restrictions stemming from the geometry of the cluster state. Constructing an efficient protocol in the circuit model was explicitly left open in \cite{Zhang2022}; we resolve this by constructing the first CVQC protocol in the circuit model with quasilinear resource requirements and proving its soundness.
\par 

\paragraph{Delegation through dequantisation.} As inspiration for our construction, we turned to an older body of results in the nonlocal, information-theoretic model. Here a classical verifier interacts with \emph{multiple} provers sharing entanglement, who are assumed to not communicate. These interactive protocols are called \emph{nonlocal games}. In this model more efficient protocols have been constructed \cite{Coladangelo2024} using the powerful phenomenon of \emph{self-testing} \cite{Mayers2004}, in which nonlocal correlations can be used to certify specific quantum states and measurements.
\par 
\medskip
A direct application of self-testing can yield a test for an accepting history state, in the sense of the Mahadev protocol: this is done by the elegant nonlocal protocol of Grilo \cite{Grilo2019}. But self-testing enables richer control over the prover's system: Coladangelo et al. realized that it could be used to simulate protocols with quantum communication between the prover and verifier, which are known to have very low overhead. In particular, they were able to implement an efficient protocol due to Broadbent \cite{Broadbent2018} in the nonlocal setting, achieving quasi-linear total resources with a fully classical verifier. In the language of a more recent line of work \cite{Gheorghiu2019, Zhang2022, Gheorghiu2022, Zhang2025}, the self-testing results they relied on can be viewed as a nonlocal form of verifiable, random \emph{remote state preparation} (RSP) --- a primitive that can be used to construct CVQC by `dequantising' verification protocols requiring quantum communication. Informally, verifiable RSP is a single-prover interactive protocol with classical messages, which realises the functionality of sending a quantum state, even in the presence of cheating provers.
\par

\paragraph{Single-prover protocols from nonlocal games.}  We obtain an efficient CVQC protocol by constructing a verifiable random RSP, informed by the self-test of \cite{Coladangelo2024}. Specifically, we employ the compiler of Kalai et al. \cite{Kalai2023}, which serves as a bridge between the nonlocal and single-prover settings. Compilation here means taking any nonlocal game  and converting it into a single-prover argument system using a fixed transformation. The approach of building argument systems through compilation of proof systems has a long and fruitful history in the classical cryptographic literature\cite{Kalai2009,Kalai2021,Bitansky2012}. The KLVY transform achieves this by enforcing the structural assumption of non-communication by cryptographic means, specifically by using quantum fully homomorphic encryption (QFHE) with classical ciphertexts; a strengthening of classical FHE to allow for homomorphic evaluation of quantum circuits.
\par 
\medskip
Earlier results have shown that these \emph{compiled nonlocal games} are a powerful resource, enabling the combination of techniques from both the cryptography and nonlocal-games literature, while retaining the self-testing powers of their nonlocal counterparts \cite{Natarajan2023, Metger2024}. Unlike earlier works on computational self-testing, which had to exploit the specific structure of the underlying computational assumption and build bespoke cryptographic machinery \cite{Brakerski2021, Mahadev2018, Metger2021}, an approach through compiled nonlocal games yields a more modular construction that can be instantiated from different cryptographic assumptions \cite{Gupte2024,Bartusek2025,Bacho2025}, with proofs that closely follow their nonlocal models, where the cryptography can largely be abstracted away. For the first time we use this paradigm to construct verifiable random remote state preparation \`a la  \cite{Gheorghiu2019,Gheorghiu2022} from compiled nonlocal games.
\par 
\medskip
\paragraph{New rigidity results.} Our main contribution is a novel computational rigidity result\footnote{We believe that a straight-forward translation to the nonlocal setting should reproduce the rigidity guarantee of \cite{Coladangelo2024}, with slightly better resource requirements (no logarithmic factor).}, which generalizes the Pauli braiding test \cite{Natarajan2017} from the Heisenberg--Weyl group to what we call the \textit{extended Pauli group}. This allows us to certify that an efficient prover holds a state that is computationally indistinguishable from a uniformly random element of $\mathcal{S}^{\ot n}$, with
\[
\mathcal{S} = \left\lbrace\ket{0},\ket{1}\right\rbrace\cup\left\lbrace\ket{+_\theta} = \tfrac{1}{\sqrt{2}}\left(\ket{0}+e^{i\theta}\ket{1}\right): \theta\in\{0,\tfrac{\pi}{4},\dots,\tfrac{7\pi}{4}\}\right\rbrace,
\]
or its complex conjugate ($\overline{\mathcal{S}}^{\ot n}$), without the prover knowing which state he holds. This complex conjugate ambiguity (also called ``complex conjugation attack'') is unavoidable, since the correlations tested by interactive protocols are invariant under complex conjugation. Still, it is important that the final state held by the prover is indistinguishable from a classical mixture and not a coherent superposition of the two conjugation branches. This strength of characterization wasn't achieved by earlier works, such as \cite{Gheorghiu2019}, where a coherent mixture of the canonical and complex conjugate state can't be excluded. The RSP guarantee above can be used to dequantise Broadbent's efficient protocol for verifying $\mathsf{BQP}$ computations with constant completeness--soundness gap in the circuit model \cite{Broadbent2018}, yielding our headline result:
\begin{informaltheorem}[Formally, \cref{cor:bqp-argument}]
Under the assumption that LWE is sub-exponentially hard for non-uniform quantum adversaries, every language $L\in\mathsf{BQP}$ admits a single-prover, classical-verifier argument system with total resource requirement $\tilde O(|C_x|)$, where $C_x$ is the circuit deciding $L$ on input $x$.
\end{informaltheorem}
The sub-exponential hardness assumption is used only to set $\lambda=\mathrm{polylog}(|C_x|)$ to obtain quasilinear resources; polynomial hardness of LWE still yields a correct argument system, with almost-linear resource requirements, i.e.\ $O(|C_x|^{1+\eps})$ for all $\eps>0$. Here non-uniform adversaries do not include (potentially inefficiently computable) quantum advice. 
\par
\paragraph{A representation-theoretic perspective.} Our protocol design and discussion can serve as an independent discussion of \cite{Coladangelo2024}, who were the first to introduce the richer set of self-tests for operators beyond the Pauli group; we show that these tests can be extended to the single-prover setting. Along the way, we reformulate the analysis to use the framework of approximate group representation theory for self-testing, which was still nascent when \cite{Coladangelo2024} developed their protocol. Hopefully this new perspective can clarify some of the technical issues that were solved in a more ad-hoc way in their work.
\par    
\medskip
To the best of our knowledge, this is the first rigidity proof via
approximate group representation theory for a group whose
representation theory is richer than that of the Heisenberg--Weyl
group; we hope that this will inform future constant-robustness
rigidity results for larger classes of observables. Our construction
also yields a slight performance improvement over the self-test from
\cite{Coladangelo2024} (total resources of $O(n)$ instead of $O(n\log
n)$ for certifying $n$ qubits) for circuits compiled in the gateset
assumed by \cite{Broadbent2018}, because we no longer require
communicating a permutation on the qubits, as is done in their
`ParBell' subtest.

\section{Technical overview}\label{sec:technical-overview}
This section provides a high-level technical overview of the paper. We begin by recalling important prior results; we then introduce the idea behind our rigidity test and outline its soundness analysis, carried out in \cref{chapter:rigidity}. Finally, we sketch how the rigidity result is used to construct an efficient CVQC protocol with quasilinear resources, referring to \cref{chapter:verification} for the full soundness argument.

\subsection{Background}

In this work we use compiled nonlocal games (schematically illustrated in \cref{fig:compiled-mip}) toward an efficient argument system for $\mathsf{BQP}$.

\begin{figure}[htbp]
    \centering
    \tikzset{every picture/.style={line width=0.75pt}}
    \begin{tikzpicture}[x=0.75pt,y=0.75pt,yscale=-1,xscale=1]
    \draw   (234.73,136.1) .. controls (234.73,128.26) and (241.09,121.9) .. (248.93,121.9) -- (291.53,121.9) .. controls (299.38,121.9) and (305.73,128.26) .. (305.73,136.1) -- (305.73,222.7) .. controls (305.73,230.54) and (299.38,236.9) .. (291.53,236.9) -- (248.93,236.9) .. controls (241.09,236.9) and (234.73,230.54) .. (234.73,222.7) -- cycle ;
    \draw   (364.73,133.35) .. controls (364.73,127.76) and (369.26,123.23) .. (374.85,123.23) -- (425.61,123.23) .. controls (431.2,123.23) and (435.73,127.76) .. (435.73,133.35) -- (435.73,163.71) .. controls (435.73,169.3) and (431.2,173.83) .. (425.61,173.83) -- (374.85,173.83) .. controls (369.26,173.83) and (364.73,169.3) .. (364.73,163.71) -- cycle ;
    \draw    (306.53,139.53) -- (362.2,139.44) ;
    \draw [shift={(364.2,139.43)}, rotate = 179.9] [color={rgb, 255:red, 0; green, 0; blue, 0 }  ][line width=0.75]    (10.93,-3.29) .. controls (6.95,-1.4) and (3.31,-0.3) .. (0,0) .. controls (3.31,0.3) and (6.95,1.4) .. (10.93,3.29)   ;
    \draw    (364.87,157.6) -- (307.53,157.6) ;
    \draw [shift={(305.53,157.6)}, rotate = 360] [color={rgb, 255:red, 0; green, 0; blue, 0 }  ][line width=0.75]    (10.93,-3.29) .. controls (6.95,-1.4) and (3.31,-0.3) .. (0,0) .. controls (3.31,0.3) and (6.95,1.4) .. (10.93,3.29)   ;
    \draw  [fill={rgb, 255:red, 255; green, 255; blue, 255 }  ,fill opacity=1 ] (420.4,127.57) .. controls (420.4,124.87) and (422.59,122.68) .. (425.29,122.68) -- (434.11,122.68) .. controls (436.81,122.68) and (439,124.87) .. (439,127.57) -- (439,138.97) .. controls (439,138.97) and (439,138.97) .. (439,138.97) -- (420.4,138.97) .. controls (420.4,138.97) and (420.4,138.97) .. (420.4,138.97) -- cycle ;
    \draw  [draw opacity=0] (423.3,123.04) .. controls (423.29,122.9) and (423.29,122.76) .. (423.29,122.62) .. controls (423.29,117.25) and (426.08,112.91) .. (429.51,112.91) .. controls (432.94,112.91) and (435.72,117.25) .. (435.73,122.6) -- (429.51,122.62) -- cycle ; \draw   (423.3,123.04) .. controls (423.29,122.9) and (423.29,122.76) .. (423.29,122.62) .. controls (423.29,117.25) and (426.08,112.91) .. (429.51,112.91) .. controls (432.94,112.91) and (435.72,117.25) .. (435.73,122.6) ;  
    \draw  [fill={rgb, 255:red, 255; green, 255; blue, 255 }  ,fill opacity=1 ] (427.06,134.62) -- (428.69,127.38) -- (430.86,127.38) -- (432.49,134.62) -- cycle ;
    \draw  [draw opacity=0][fill={rgb, 255:red, 255; green, 255; blue, 255 }  ,fill opacity=1 ] (428.01,130.28) .. controls (427.55,129.85) and (427.27,129.24) .. (427.27,128.57) .. controls (427.27,127.27) and (428.33,126.22) .. (429.63,126.22) .. controls (430.93,126.22) and (431.98,127.27) .. (431.98,128.57) .. controls (431.98,129.16) and (431.76,129.7) .. (431.41,130.11) -- (429.63,128.57) -- cycle ; \draw   (428.01,130.28) .. controls (427.55,129.85) and (427.27,129.24) .. (427.27,128.57) .. controls (427.27,127.27) and (428.33,126.22) .. (429.63,126.22) .. controls (430.93,126.22) and (431.98,127.27) .. (431.98,128.57) .. controls (431.98,129.16) and (431.76,129.7) .. (431.41,130.11) ;

    \draw   (364.93,195.75) .. controls (364.93,190.16) and (369.46,185.63) .. (375.05,185.63) -- (425.81,185.63) .. controls (431.4,185.63) and (435.93,190.16) .. (435.93,195.75) -- (435.93,226.11) .. controls (435.93,231.7) and (431.4,236.23) .. (425.81,236.23) -- (375.05,236.23) .. controls (369.46,236.23) and (364.93,231.7) .. (364.93,226.11) -- cycle ;
    \draw   (350.2,123.29) .. controls (350.2,115.34) and (356.65,108.89) .. (364.6,108.89) -- (435.93,108.89) .. controls (443.89,108.89) and (450.33,115.34) .. (450.33,123.29) -- (450.33,236.23) .. controls (450.33,244.19) and (443.89,250.63) .. (435.93,250.63) -- (364.6,250.63) .. controls (356.65,250.63) and (350.2,244.19) .. (350.2,236.23) -- cycle ;
    \draw    (306.53,200.53) -- (362.2,200.44) ;
    \draw [shift={(364.2,200.43)}, rotate = 179.9] [color={rgb, 255:red, 0; green, 0; blue, 0 }  ][line width=0.75]    (10.93,-3.29) .. controls (6.95,-1.4) and (3.31,-0.3) .. (0,0) .. controls (3.31,0.3) and (6.95,1.4) .. (10.93,3.29)   ;
    \draw    (364.87,218.6) -- (307.53,218.6) ;
    \draw [shift={(305.53,218.6)}, rotate = 360] [color={rgb, 255:red, 0; green, 0; blue, 0 }  ][line width=0.75]    (10.93,-3.29) .. controls (6.95,-1.4) and (3.31,-0.3) .. (0,0) .. controls (3.31,0.3) and (6.95,1.4) .. (10.93,3.29)   ;
    
    \draw (242,168) node [anchor=north west][inner sep=0.75pt]   [align=left] {Verifier};
    \draw (381.87,139) node [anchor=north west][inner sep=0.75pt]   [align=left] {Alice};
    \draw (385,202) node [anchor=north west][inner sep=0.75pt]   [align=left] {Bob};

    \end{tikzpicture}
    
    \caption{Schematic of a compiled nonlocal game. Lock denotes quantum fully homomorphic encryption. Both `Alice' and `Bob' are invocations of the same single prover and are only distinguished for illustrative purposes. Communication is sequential and flows from top to bottom.}
    \label{fig:compiled-mip}
    \end{figure}
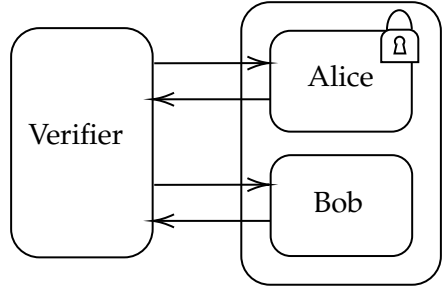
    
\par
\medskip
\begin{table}[t]
    \centering
    \renewcommand{\arraystretch}{1.25}
    \setlength{\tabcolsep}{6pt}
    \begin{tabularx}{0.96\linewidth}{@{}>{\RaggedRight\arraybackslash}m{1.9cm} c c c >{\RaggedRight\arraybackslash}X >{\Centering\arraybackslash}m{1.75cm} c@{}}
    \toprule
    Protocol & Type & Provers & Rounds & Total resources & Assumption & Blind \\
    \midrule
    \cite{Reichardt2013}             & MB & $2$ & $\mathrm{poly}(n)$  & $\ge g^{8192}$                 & IT            & Yes \\
    \cite{Gheorghiu2015}             & MB & $2$ & $\mathrm{poly}(n)$  & $\ge g^{2048}$                 & IT               & Yes \\
    \cite{Grilo2019}         & LH & $2$ & $1$                 & $\Omega(n\,g^{2})$             & IT               & No  \\
    \cite{Coladangelo2024} & C & $2$ & $O(\mathrm{depth})$/2 & $\Theta(g\log g)$ & IT & Yes/No \\
    \midrule
    \cite{Mahadev2018}     & LH & $1$ & $2$                 & $O(\poly(\lambda)\,g^{3})$ & LWE & No  \\
    \cite{Zhang2022}       & MB & $1$ & $O(\mathrm{depth})$ & $O(\poly(\lambda)\,g)$     & LWE +\newline\mbox{QROM}     & No  \\
    \hspace{4pt}\textbf{Here}      & C & $1$ & $O(\mathrm{depth})$ & $O(\poly(\lambda,\log g)\,g)$     & LWE       & No  \\
    \bottomrule
    \end{tabularx}
    \caption{Resource overheads of classical-verifier $\mathsf{BQP}$ verification protocols. The \cite{Coladangelo2024} result presents two protocols, which are compactly described as a single entry in the table (first entry is the `Leash' and second is the `Dog walker' protocol).
    $g$ is the number of gates in the delegated circuit, $n$ the qubit count, and $\lambda$ the security parameter. ``Total resources'' bundles prover time and gate complexity,
    EPR pairs, and classical communication; ``depth'' refers to the $T$-depth of the delegated circuit. Abbreviations: LH = local Hamiltonian, MB = MBQC, C = circuit, IT = information-theoretic, LWE = learning with errors, QROM = quantum random oracle model.}
    \label{tab:cvqc}
\end{table}
    
Proof and argument systems for $\mathsf{BQP}$ already exist in both the nonlocal and classical verification settings, and \cref{tab:cvqc} summarizes the resource requirements of prior works\footnote{These results are stated in different settings; for a fair comparison, we include any overhead from repetitions needed to achieve a constant completeness--soundness gap.}. We note that the polylogarithmic overhead in the circuit size in our construction is often unavoidable: most protocols assume a fixed universal gate set, and converting between gate sets typically introduces a polylogarithmic factor in the circuit size, via the Solovay--Kitaev theorem. In our case this overhead becomes explicit, since both QFHE and the Broadbent protocol require distinct gate sets, so one conversion is already baked into the protocol.
\par
\medskip

\subsubsection*{The Broadbent protocol}
Since the Broadbent protocol plays an important role in our construction, we briefly discuss it here. At a high level, its soundness rests on the indistinguishability of two kinds of rounds, called test rounds and computation rounds. The test rounds function as trapdoors that constrain a cheating prover: any strategy that is accepted with high probability in test rounds must yield the correct circuit outcome in computation rounds.
\par
\medskip
Round indistinguishability is achieved by encrypting the verifier's input states with the quantum one-time pad (QOTP). For random classical bits $a,b\in\zo$, the QOTP of a single-qubit state $\ket{\psi}\in\C^2$ is
\[
\sigma_X^a\sigma_Z^b\ket{\psi}.
\]
By the Pauli twirl, when $(a,b)$ is uniform the QOTP provides the same information-theoretic security as its classical analog:
\[
\frac{1}{4}\sum_{a,b\in\zo} \sigma_X^a\sigma_Z^b\ket{\psi}\bra{\psi}\sigma_Z^b\sigma_X^a= \frac{\id}{2},
\]
which generalizes straightforwardly to multi-qubit states.
\par
\medskip
A candidate approach, to make the different rounds indistinguishable, is for the verifier to encrypt the circuit input (assumed to be encoded in the computational basis) with the QOTP. Knowing the circuit specification, the verifier could ideally keep track of the updated QOTP keys by commuting each gate in the circuit past the tensor product of Pauli $\sigma_X$ and $\sigma_Z$ operators that implement the QOTP encryption. Unfortunately, for circuits over the universal gate set $\{\sigma_X, \sigma_Z, H, \mathrm{CNOT}, T\}$, this commutation trick succeeds only for the Clifford gates $\sigma_X$, $\sigma_Z$, $H$, and $\mathrm{CNOT}$---insufficient on its own, since any circuit using only Clifford gates can be efficiently simulated classically \cite{Aaronson2004}. By incorporating $T$-gate gadgets, however, the approach can be made to work: in addition to the QOTP-encrypted input qubits, the verifier sends encrypted magic state qubits (one per $T$-gate in the circuit) that enable the application of a $T$-gate on QOTP encrypted qubits through re-keying and hence allow the verifier to classically track all key updates through the full circuit. For every gate in the circuit, the QOTP keys thus evolve as $(a,b)\mapsto(a',b')$ in a way that is known to the verifier, so at the end of the computation the verifier knows exactly how he can decrypt the encrypted outcome returned by the prover.
\par
\medskip
We refer the reader to \cite{Broadbent2018} for a full description and security analysis. For our purposes, the crucial point is that the verifier needs only prepare states from the following set of ten:
\[
\cS=\{\ket{0}, \ket{1}\}\cup\left\lbrace \ket{+_\theta}=\tfrac{1}{\sqrt{2}}\left(\ket{0}+e^{i\theta}\ket{1}\right): \theta\in\{0,\tfrac{\pi}{4},\dots,\tfrac{7\pi}{4}\}\right\rbrace.
\]
If we can design a remote state preparation protocol which certifies that the prover holds a random tensor product of these states---without knowing which one---through purely classical interaction, we can use it to replace the quantum communication at the start of the Broadbent protocol. The set of states $\cS$ coincides with the states certified in a work by Gheorghiu and Vidick \cite{Gheorghiu2019}; unfortunately, their protocol does not have constant rigidity and thus cannot directly be used for our purposes. They obtain a per-qubit rigidity error of $\eps$ (which quantifies the deviation of the obtained state from the target state) with a resource cost of $O(1/\eps^3)$. To delegate a circuit of size $g$, the Broadbent protocol requires $O(g)$ prepared qubits, so naively invoking this RSP $g$ times already costs $O(g/\eps^3)$. Moreover, to obtain constant overall rigidity error $\eps_{\mathrm{tot}}$, the per-qubit error must scale as $\eps_{\mathrm{tot}}/g$, driving the total cost up to $O(g^4)$---worse than the Mahadev protocol. Zhang \cite{Zhang2022} gives an efficient parallel RSP protocol for states $\ket{+_\theta}$ with $\theta\in\{0,1,\ldots,7\}$, but this only covers part of $\cS$: the computational basis states are missing, which is why his construction applies only to verification in the MBQC model, through dequantizing the protocol of \cite{Ferracin2018}.

\subsection{Rigidity result}\label{section:rigidity}
So-called \textit{rigidity results}, which allow a verifier to characterize the internal quantum operations of an untrusted prover, are central to (computational) self-testing. A rigidity result for the simple fact that the prover is applying two anti-commuting observables---which is provided by the CHSH or Magic Square nonlocal games---is remarkably powerful and is sufficient for many applications \cite{Reichardt2013,Brakerski2021,Mahadev2018,Grilo2019}. How this observation about the internal workings of an untrusted prover is obtained differs strongly between nonlocal and early cryptographic works. The KLVY compiler allows us to reconcile the different settings, by enabling the use of nonlocal techniques in the cryptographic setting, uncovering their fundamental relation and reusing existing techniques, such as the framework of approximate group representations. It has been even noted that the analysis can be completely performed in the nonlocal setting, by restricting to \textit{computationally nonlocal strategies} as defined in \cite{Bartusek2025}, eliminating the need to deal with the specifics of the cryptography, but instead having to make sure that the quantum operations satisfy the definition of a computationally nonlocal strategy.
\par
\medskip
Rigidity results such as those in \cite{Mayers2004, Reichardt2013, Wu2016,Natarajan2017} are usually stated as guarantees on the observables which the provers implement and on the entangled state which they share. From such a guarantee it is typically straightforward to deduce that a specific state has been prepared, provided it is an eigenstate of the certified observables, since applying the certified measurements projects into one of its eigenstates.
\par
\medskip
The states in $\mathcal{S}$ are, up to a global phase, eigenstates of the single-qubit Clifford observables
\[
\sigma_X,\quad \sigma_Y,\quad \sigma_Z,\quad \sigma_F=\tfrac{1}{\sqrt{2}}(\sigma_Y-\sigma_X),\quad \sigma_G=\tfrac{1}{\sqrt{2}}(\sigma_Y+\sigma_X).
\]
Certifying preparation of states in $\mathcal{S}$ therefore reduces to certifying these observables. This is a stronger guarantee than that the prover is applying two anti-commuting observables, which has traditionally been extensively studied. Fortunately, the analysis of these types of rigidity results can be unified through the framework of approximate group representation theory\footnote{For a comprehensive and thorough introduction to the subject, we refer the reader to the excellent course notes of \cite{Vidick2021}.}, which has become a standard technique in self-testing after the development of the Pauli braiding test \cite{Natarajan2017}, which can also be applied in the cryptographic setting thanks to the KLVY compiler. We develop the main ideas of this framework first in the nonlocal setting, through the Heisenberg--Weyl group and the CHSH game (the scenario of certifying two anti-commuting observables), and then explain how the argument extends to the group of interest, which contains the five observables from above.
\par
\medskip
The CHSH game, named after Clauser, Horne, Shimony and Holt \cite{Clauser1969}, is a two-prover nonlocal game that distills the original Bell experiment. Beyond its role in ruling out local hidden-variable theories \cite{Hensen2015}, it has become a central tool in self-testing, including the breakthrough of \cite{Reichardt2013}. The game proceeds as follows:
\begin{itemize}
    \item The verifier samples two random bits $x,y\in\zo$ and sends them to the provers.
    \item The provers respond with one bit each, denoted $a,b\in\zo$.
    \item The verifier accepts if and only if $a\oplus b=x\cdot y$.
\end{itemize}
The no-communication assumption is what makes this game nontrivial, since both provers don't know which question the other received. Classical, unentangled provers can win with probability at most $3/4$, for instance by always outputting the same pre-agreed bit ($a=b$). Sharing entanglement, they can exploit the nonlocal correlations it enables and achieve a winning probability of $\cos^2(\pi/8)\approx 0.853$, which Tsirelson showed to be optimal among all quantum strategies \cite{Tsirelson1987}. A winning probability above $3/4$ therefore witnesses genuinely quantum behavior---the Bell-test, or proof-of-quantumness, aspect of CHSH.
\par
\medskip
The specific CHSH nonlocal game has even more power: its optimal strategy is unique up to isometry. This implies that if the provers win with a probability equal to the quantum optimum $\omega^*=\cos^2(\pi/8)$, the verifier can conclude that they are playing the optimal strategy. The optimal strategy for CHSH consists of the provers sharing an EPR pair and measuring anti-commuting observables on it. Fortunately, this also holds robustly, i.e. if the provers succeed in the CHSH game with a probability close to the quantum optimum $\omega^*-\eps$, then the provers' measurements approximately anti-commute in a suitable distance measure (defined in \Cref{chapter:preliminaries}) and they share a state which is close to an EPR pair. A classical verifier can therefore infer that successful, non-communicating provers implement approximately anti-commuting observables; in many applications, however, one needs control over \emph{which} observables are realized, not merely which algebraic relation they satisfy (anti-commutation in this case). Approximate group representation theory supplies a way to achieve this, by determining the group defined by the algebraic relations and using the irreducible representations of that group. We will illustrate this approach using the single-qubit Heisenberg--Weyl group as an example. This group is generated by elements $\omega$, $x$, and $z$, satisfying relations
\[
\langle \omega^2=x^2=z^2=\id,\, xz=\omega z x,\, z\omega=\omega z,\, x\omega=\omega x\rangle.
\]
This single-qubit group exhibits several one-dimensional irreducible representations (with $\omega=1$ and various sign choices for $x$ and $z$) and a unique two-dimensional irreducible representation,
\[
\omega = -\id,\quad x=\begin{pmatrix}0 & 1 \\ 1 & 0 \end{pmatrix}=\sigma_X,\quad z=\begin{pmatrix}1 & 0 \\ 0 & -1 \end{pmatrix}=\sigma_Z,
\]
corresponding to the Pauli matrices $\sigma_X$ and $\sigma_Z$. Imposing $\omega=-\id$ thus leaves the anti-commutation relation $\{z,x\}=0$ as the algebraic relation defining the group---precisely the relation certified by the CHSH game. Let $A_0$ and $A_1$ be the measurements implemented by one of the provers on questions 0 and 1, respectively. Because $A_0$ and $A_1$ are unitary (as binary observables) and approximately satisfy anti-commutation for a sufficiently successful prover, they form an approximate unitary representation of the single-qubit Heisenberg--Weyl group.
\par
\medskip
The Gowers--Hatami (GH) theorem \cite{Gowers2017} guarantees that any approximate unitary representation can be rounded to an exact one on a larger Hilbert space. Informally, this means that there exists an isometry $V$ such that for any $\eps > 0$,
\[
VA_0 \approx_{\eps} (\sigma_Z\ot\id) V,\quad\text{and}\quad VA_1 \approx_{\eps} (\sigma_X\ot\id) V,
\]
in the same state-dependent norm as defined in \cref{section:distance-measures}. We used the single-qubit case for illustration, but GH is most useful in the multi-qubit setting, where it certifies many qubits with \emph{constant robustness}: the approximation error $\eps$ does not grow with the number of certified qubits. Single-qubit self-testing was known before GH, via Jordan's lemma, but that approach does not easily scale. For $n$ qubits the Heisenberg--Weyl group has relations
\begin{align*}
\langle & \omega^2=x(a)^2=z(a)^2=\id,\, x(a)z(b)=\omega^{a\cdot b} z(b) x(a),\, z(a)\omega=\omega z(a),\\
&\quad  x(a)\omega=\omega x(a),\, x(a\oplus b)=x(a)x(b),\, z(a\oplus b)=z(a)z(b)\rangle,
\end{align*}
where $a,b\in\zo^n$. This group again admits many one-dimensional irreducible representations, corresponding to all possible sign assignments, but only one $2^n$-dimensional irreducible representation:
\[
\omega=-\id,\quad x(a)=\bigotimes_{i\in[n]}\sigma_X^{a_i}=\sigma_X(a),\quad z(a)=\bigotimes_{i\in[n]}\sigma_Z^{a_i}=\sigma_Z(a).
\] 
Assuming one can construct prover observables $A_0(a)$ and $A_1(a)$ that exactly satisfy $A_i(a)A_i(b)=A_i(a\oplus b)$ for $i\in\zo$ and approximately satisfy the characteristic relation $A_0(a)A_1(b)=(-1)^{a\cdot b} A_1(b) A_0(a)$, GH provides an isometry $V$ such that for any $\eps > 0$,
\[
VA_0(a) \approx_{\eps} (\sigma_Z(a)\ot\id) V,\quad\text{and}\quad VA_1(a) \approx_{\eps} (\sigma_X(a)\ot\id) V,
\]
with error $\eps$ independent of $n$. To our knowledge, the multi-qubit Heisenberg--Weyl group is the most complex group previously self-tested in this framework. The observables needed for $\mathcal{S}$ lie in a larger group, which we call the \emph{extended Pauli group}. It is generated by $\omega$, $x(a)$, $z(a)$ and $g(a)$ with relations
\begin{align*}
\langle &\omega^4=x(a)^2=z(a)^2=g(a)^2=\id,\, x(a)z(b)=\omega^{2a\cdot b} z(b) x(a),\, z(a)\omega=\omega z(a),\\
& x(a)\omega=\omega x(a),\, g(a)\omega=\omega g(a),\, g(a)z(b)=\omega^{2a\cdot b} z(b) g(a),\, g(a)x(a)g(a)=\omega^{|a|} x(a)z(a)\\
& x(a\oplus b)=x(a)x(b),\, z(a\oplus b)=z(a)z(b),\, g(a\oplus b)=g(a)g(b)\rangle.
\end{align*}
Because $\omega$ is a fourth root of unity, the representation theory is substantially richer. First, when $\omega^2=\id$, there are many $2^k$-dimensional ``classical'' irreducible representations ($0\leq k\leq n$); intermediate dimensions arise because $x$ and $g$ can anti-commute even when $\omega^2=\id$. Second, when $\omega^2=-\id$, there are $2^{n+1}$ distinct $2^n$-dimensional ``quantum'' irreducible representations which, up to sign choices and phase ambiguity, correspond to the matrices $\sigma_X$, $\sigma_Z$, and $\sigma_G=1/\sqrt{2}(\sigma_X+\sigma_Y)$. These are far more irreps than the unique quantum one in the Heisenberg--Weyl case, owing to sign freedom on $g$ induced by its conjugation relation to $x$ and $z$. Finally, both $\omega=i\id$ and $\omega=-i\id$ are admissible; this is the algebraic origin of the complex-conjugation attack discussed in the introduction and in \cite{Coladangelo2024,Zhang2022}.
\par
\medskip
To apply GH in this setting, we must self-test the full relation set, including a conjugation relation. The only prior self-test for such a relation that we are aware of is due to \cite{Coladangelo2024}, which heavily influenced our construction; it is specified in \cref{protocol:conjugation-test-mod} and its soundness is analyzed in \cref{subsec:conj}.
\par
\medskip
Even after the relations are certified and the approximate representation is rounded, a per-qubit sign ambiguity on $G(a)$ remains (here $G(a)$ denotes the prover's implementation of the abstract group element $g(a)$ in the $2^n$-dimensional branch). In the Heisenberg--Weyl case one excludes one-dimensional representations by requiring $-\id\mapsto -\id$ under the approximate representation, equivalently discarding the $\omega=\id$ branch; approximate anti-commutation enforces this automatically, as scalars cannot anti-commute. An analogous argument removes the $\omega^2=\id$ branch of ``classical'' irreps. for the extended Pauli group, but aligning the sign of $G(a)$ across qubits still requires an additional test, whose guarantee goes beyond the algebraic relations that defined the group, as we now want to pick out specific irreducible representations among the many admissible ones.
\par
\medskip
Our self-test again follows \cite{Coladangelo2024}, who confronted the same issue. They certify the same group with a protocol structurally similar to ours; the main difference is methodological, as their analysis predates the systematic GH framework and relies on more ad hoc arguments. We adapt their ideas into a self-test with a simpler structure and a more direct group-theoretic analysis.
\par
\medskip
The idea of forcing a consistent sign choice across all $G$ observables, used in \cite{Bowles2018,Coladangelo2024}, is a beautiful one. The test proceeds as follows:
\begin{itemize}
    \item The verifier asks Alice to measure pairs of her qubits in the Bell basis.
    \item The verifier asks Bob to measure all of his qubits separately in the $G$ eigenbasis.
    \item The verifier accepts if and only if the correlations between the different $G$ outcomes reported by Bob match the ones expected based on the Bell basis outcome reported by Alice.
\end{itemize}
Intuitively, the test performs entanglement swapping. Provers who pass the earlier subtests must share something that has the properties of perfect entanglement (up to efficient distinguishers); Alice's Bell measurement then converts inter-prover entanglement into entanglement internal to Bob's register. Bob does not know the post-measurement state on his side, so inconsistent per-qubit signs for $G$ alter his measurement statistics in a way the verifier can detect. We refer to \cref{protocol:cliff-test} subtests 2 and 3 for the protocol specification and to \cref{subsec:corr-analysis} for the full soundness argument; here we only sketch the main idea, using the Pauli $Y$ operator as an example.
\par
\medskip
Suppose Alice's Bell measurement projects two of her qubits onto an EPR pair $\ket{\Phi^+}$. By entanglement swapping, the two paired qubits on Bob's side are then $\ket{\Phi^+}$ as well. A direct calculation gives
\[
\bra{\Phi^+}\sigma_Y\otimes\sigma_Y\ket{\Phi^+}= -1,
\]
so if Bob applies $\sigma_Y$ with the same sign on both qubits, his outcomes are perfectly anti-correlated; a single sign flip on one of the observables converts this to perfect correlation, which the verifier would detect. The argument extends to the $\sigma_G$ observable and the other three Bell basis states. Since this test can only detect inconsistent pairwise signs, a global sign on $G(a)$ can survive; we remove it with a CHSH subtest, noting that $\sigma_F$ and $\sigma_G$ are optimal CHSH observables. 
\par
\medskip
Together, these components yield a self-test for specific $2^n$-dimensional irreducible representations of the $n$-qubit extended Pauli group: the prover's measurements are certified, up to isometry and complex conjugation, with constant robustness. This is captured by the following informal theorem:
\begin{informaltheorem}[Formally, \cref{theorem:clifford-mixed-basis}]
    For any efficient prover that passes the $n$-qubit compiled Clifford test from \cref{protocol:cliff-test} with probability $\omega^*-\eps$, there exists an isometry such that under uniform expectation over $\tilde W\in\{X,Y,Z,F,G\}^n$
    \[
        V\tilde W(a)\approx_\eps ((\sigma_{\tilde W}(a)\oplus \overline{\sigma}_{\tilde W}(a))\ot\id)V,
    \]
    where $\omega^*$ is the optimal winning probability of the test, $\tilde W(a)$ is the prover's observable on question $\tilde W$ and
    \[
        \sigma_{\tilde W}(a)\deq\bigotimes_{i\in[n]} \sigma_{\tilde W_i}^{a_i}.
    \]
\end{informaltheorem}
Here the direct sum structure is a consequence of the phase ambiguity, which was mentioned earlier. Using this characterization and a consistency test---the verifier asks both provers the same question and checks that the answers agree---we can also characterize the prover's state. This is possible because the notation $\approx_\eps$ relates to the state-dependent norm, which depends on the prover's state. Explicitly using this state dependence, we arrive at the following result:
\begin{informaltheorem}[Formally, \cref{theorem:rsp-guarantee}]
    For any efficient prover that passes the $n$-qubit compiled Clifford test from \cref{protocol:cliff-test} with probability $\omega^*-\eps$, there exists an isometry such that
    \[
        \E_{\tilde W}\sum_{v\in\zo^n} \ketbra{v,\tilde W}{v, \tilde W}_W\ot V\phi_v^{\Enc(\tilde W)}V^\dag \overset{c}{\approx}_{\sqrt{\eps}} \E_{v,\tilde W}\ketbra{v, \tilde W}{v, \tilde W}_W\ot\left(\tau_{\tilde W}^v\otimes \rho_{0} + \overline{\tau}_{\tilde W}^v\otimes \rho_{1}\right),
    \]
    where $\omega^*$ is the optimal winning probability of the test, the verifier holds the $W$ register, $\phi^{\Enc(\tilde W)}_v$ is the prover's state after receiving the encrypted question $\tilde W$ and answering with an encryption of $v$ and 
    \[
    \tau_{\tilde W}^v\deq\bigotimes_{i\in[n]} \frac{1}{2}\left(\id+(-1)^{v_i}\sigma_{\tilde W_i}\right),
    \]
\end{informaltheorem}
We have not yet introduced the notation, but $\overset{c}{\approx}_\delta$ should be understood as computational indistinguishability with an advantage of $\delta$.
The state guarantee for imaginary observables, such as $\sigma_Y$, $\sigma_F$, and $\sigma_G$, factors into two distinct blocks, because of the phase ambiguity. In general, there are also coherences between these two cases; it could be imaginable that the prover applies a coherent superposition of the canonical and complex conjugate observables. However, because the QFHE makes it impossible to consistently apply this superposition during both parts of the interaction, we can show that the coherences are cryptographically small (see \cref{cor:flag-dephasing} for the detailed statement), which achieves a stronger state characterization than prior results in our setting, such as \cite{Gheorghiu2019}, where these coherences are not generally excluded.
\par
\medskip
In what follows, we will show how this state characterization can be used to run the Broadbent protocol with purely classical communication and thus achieve $\mathsf{BQP}$ verification between a classical verifier and a single computationally bounded prover.
\subsection{BQP verification}
As mentioned before, the idea is to run the Broadbent protocol with a classical verifier, by replacing the quantum communication with verifiable remote state preparation, obtained through our rigidity result.
\par 
\medskip
To start our protocol, the verifier chooses the bases to prepare his qubits in, as in a normal execution of the Broadbent protocol; this yields the basis assignment string $\tilde W\in\{X,Y,Z,F,G\}^n$. Next, instead of preparing the qubits in the eigenstates of these bases dictated by the QOTP keys, he encrypts the basis assignment string and sends it to the prover. The prover replies with an encrypted answer $\alpha$, whose decryption ($v=\Dec(\alpha)\in\zo^n$), in an honest execution, encodes the measurement outcomes of measuring $n$ qubits in the basis specified by $\tilde W$. Recalling the previous section, we have a rigidity test, where a high passing probability forces this to be the case. Thus, by executing the rigidity test often enough to estimate this passing probability we can be sure that the prover holds the eigenstate of the tensor product of the $n$ observables specified by $\tilde W$, corresponding to the per-observable eigenvalues $v$ in his quantum register, after answering the first question.
\par 
\medskip
Importantly, the prover only receives the encryption of $\tilde W$ and homomorphically evaluated the measurement; thus, by the semantic security of the QFHE scheme, he doesn't know which bases his qubits are in. To actually perform the delegation, we need to sprinkle in executions of the Broadbent protocol. This is why there are two subtests in \cref{protocol:verification}, which are executed with different probabilities: the first is more likely to run and performs the rigidity protocol, while the second acts identical in the encrypted part of the interaction and then performs the classical part (all interactions after the verifier sent his quantum states to the prover) of the Broadbent verification protocol.
\par
\medskip
By appropriately choosing the probability with which the verification game is played, we can ensure that the rigidity guarantee (specifically the state characterization) applies in this second case as well, which makes the soundness guarantees of the Broadbent protocol kick in and ensures that the overall verification protocol is sound.

\subsection{Technical contributions}

In the compiled setting, many of the guarantees rest on the IND-CPA security of the underlying cryptography, and as a consequence some routine operations from the nonlocal setting become more delicate. For example, when `prover switching' observables corresponding to distinct questions, the analysis relies on the computational indistinguishability of the two state the prover holds after the QFHE encrypted part of the interaction. Because this indistinguishability holds only in computational distance, every operator that enters the argument---which often includes the Gowers--Hatami (GH) isometry---must be computationally efficient.
\par
\paragraph{An explicit GH isometry.} Building on the previous point, we had to show that the GH isometry is computationally efficient. To do so, we slightly modified the dilation-based proof of GH in \cite{Metger2024} to obtain an explicit expression for the isometry (see \cref{theorem:GH}). To block-diagonalize the exact representation to which GH rounds, we computed an explicit circuit for the Fourier transform over our extended Pauli group and verified that this operation is efficient. This construction may be of independent interest.
\par
\medskip
At a technical level, our modification of \cite[Theorem 3.1 (Gowers--Hatami)]{Metger2024} yields a tighter characterization of the error between the pre- and post-isometry rounded operators. In other self-testing works (including \cite{Metger2024}), the quantity that is bounded is
    \[ \|  X  - V^\dagger (\sigma_X \ot \id)V \|_\psi,\]
which led to what they called the ``$VV^\dagger$ problem.'' In our analysis, we instead bound quantities of the form
    \[ \| VX - (\sigma_X \ot \id) V \|_\psi.\]
This seemingly small change makes it substantially easier to compose error bounds when studying how the isometry acts on products of operators. Interestingly, deriving these improved bounds required the explicit expression for $V$ that our modified dilation-based proof of GH provides.
\par
\paragraph{Robust state preparation.} We reproduce the result of \cite{Metger2021}, extended to a parallel self-test of multiple EPR pairs with constant robustness (\cref{protocol:cliff-test}, subtest 2 and \cref{lemma:epr-bell}). This improves on \cite{Fu2023}, where a similar result is shown without constant robustness. Natarajan and Zhang \cite{Natarajan2023} asserted that compiled nonlocal games---specifically, their compiled Pauli braiding test---would suffice to certify EPR pairs, and our calculations confirm this intuition. Moreover, \cref{theorem:rsp-guarantee} extends the results of \cite{Gheorghiu2019,Gheorghiu2022} by yielding the first random verifiable RSP beyond BB84 states, with constant robustness. We are not the first to  achieve constant robustness computational self-testing \cite{Natarajan2023,Metger2024}, but we extend previous results from operator to state self testing and from the Heisenberg--Weyl group to the extended Pauli group. We achieve random remote state preparation, because in obtaining constant robustness, we need to retain the expectation over the test distribution of the basis assignments. In \cite{Gheorghiu2022} closeness is shown for all basis choices separately, which is achieved by a hybrid computational indistinguishability argument. We could repeat this argument, but it will always introduce a dependence on the number of prepared qubits. Our weaker notion of random remote state preparation is sufficient for any downstream application where the winning probability is taken over random basis choices anyway, which is the case in the Broadbent protocol.
\par
\paragraph{The complex conjugation ambiguity.} One of the main technical obstructions when working with imaginary observables (such as the Pauli $Y$) is the ``complex-conjugate'' attack (also called phase ambiguity). The nonlocal correlations tested during protocol execution are invariant under complex conjugation, so a prover may implement either the canonical strategy or its complex conjugate, or, in the worst case, even a coherent superposition of the two. This observation led McKague and Mosca \cite{McKague2011} to propose a generalized self-testing framework called \textit{complex self-testing}. Previous works \cite{Zhang2022,Coladangelo2024} have confronted the same issue; they are either in a different setting (nonlocal), or the states they certify are unitarily equivalent to their complex conjugates, which makes the complex conjugate attack harmless. The work closest to our setting \cite{Gheorghiu2019} contains a bug in the proof of Lemma 3.5, which ignores the phase ambiguity. It is impossible to completely remove the phase ambiguity, but in this work we show that an efficient prover can't apply a coherent superposition of the two strategies, which implies that the prepared state is computationally indistinguishable from a classical mixture of the canonical state and its complex conjugate.
\par
\paragraph{Entanglement between the simulated provers.} The results mentioned above are strong state self-testing guarantees for the prover's unencrypted part. Unfortunately, we have very little control over the state shared between the encrypted and unencrypted parts of the prover. To obtain a full analog of nonlocal games, we would like to show that these parts carry entanglement between them. We make partial progress toward this goal: we show that the state held by the prover after the QFHE encrypted part of the interaction, when marginalized over the encrypted responses, is computationally indistinguishable from the maximally mixed state (\cref{lemma:state-characterization}). Our resolution for the phase ambiguity coherences used consistency between the encrypted and unencrypted parts of the prover; in the nonlocal setting, \cite{Coladangelo2024} get rid of these coherences by certifying that the provers share EPR pairs.

\subsection{Open questions}
\begin{enumerate}
    \item At present we have only limited `inter-prover' state self-testing guarantees (e.g.\ \cref{lemma:state-characterization}), which hold up to computational indistinguishability. Is there a cheating strategy that passes our compiled $n$-qubit test while carrying far less entropy than the canonical honest strategy of preparing $n$ EPR pairs---e.g.\ a form of computational ``pseudo-entanglement''? Such strategies would bear directly on the quantum soundness of the KLVY compiler, since they would show that optimal compiled strategies can be far from any nonlocal strategy and would thereby constrain rounding-based approaches to quantum soundness. Merkulov and Arnon \cite{Merkulov2025} prove entropy lower bounds for strategies with anti-commuting observables in the cryptographic setting; stronger results of this kind would clarify how far state self-testing can be pushed in the compiled model.
    \item In the nonlocal setting, \cite{Coladangelo2024} construct a constant round protocol for $\mathsf{QMA}$ verification using the entanglement-based version of \cite{Broadbent2018}. We believe a similar extension should be possible here. We have not yet carried it out as we faced obstructions in state certification, attempting a route through an entanglement-based formulation in the compiled setting.
    \item Within the context of $\mathsf{QMA}$ verification, witness preservation is also of considerable interest. This question has been studied intensively in recent work of Kalai, Khurana, and Raizes \cite{Kalai2026}. We hope that, because our protocol delegates a verification circuit gate-by-gate, it may be easier to obtain witness-preserving protocols by delegating the Marriott--Watrous amplified verification circuit.
    \item Ultimately, the ``pie in the sky'' goal, as proposed by Justin Raizes~\cite{Raizes2026}, is an interactive argument for $\mathsf{QMA}$ that is as good as, or better than, the direct proof system: it should use as few rounds as possible (ideally a constant number), the prover should need only one copy of the witness state, and completeness and soundness should be negligibly close to $1$ and $0$, respectively. One would also hope for succinctness. It remains far from clear how to achieve all of these properties simultaneously, but we hope the tools introduced in this work will prove useful.
    \item Among more immediate protocol improvements, one straightforward route to near-perfect completeness would be to extend our current CHSH test (\cref{protocol:cliff-test}, subtest 5) to a parallel version with multi-qubit observables, along the lines of \cite{Coladangelo2024}; we have not yet completed the corresponding calculation. We also believe that blindness can be added to our protocol by delegating a universal circuit, following the standard approach of \cite{Coladangelo2024}, at the cost of an additional logarithmic factor in resource overhead and a potential sequential execution of the rigidity test.
    \item For succinctness, one promising direction is the recently proposed black-box compiler of Bartusek, Liu, and Malavolta \cite{Bartusek2026}. Another is to apply the techniques of \cite{Metger2024}, based on the ideas of de la Salle \cite{DeLaSalle2025}, to sparsify the group self-test. An important first step to achieving succinctness is to reduce the round complexity of our verification protocol.
\end{enumerate}

\subsection{Paper outline}

The remainder of this paper is organized as follows.
\par
\medskip
\Cref{chapter:preliminaries} introduces the background, concepts, and helpful tools needed for the analysis. We fix notation, introduce the extended Pauli group (\cref{def:extended-pauli-group}), and define the state-dependent distance measures used throughout (\cref{def:state-dependent-norm}). We also state the Gowers--Hatami theorem (\cref{theorem:GH}) that underlies our self-testing arguments.
\par
\medskip
\Cref{chapter:rigidity} develops the core technical contribution. We specify the Clifford rigidity test (\cref{protocol:cliff-test}) and the constituent subtests, introduce compiled prover switching as a tool for analyzing protocols in the cryptographic setting, and prove soundness of each subtest. The section concludes with a state characterization theorem (\cref{theorem:rsp-guarantee}) that certifies remote preparation of the states required by the Broadbent protocol.
\par
\medskip
\Cref{chapter:verification} applies this rigidity result to obtain the headline result: an efficient CVQC protocol. The protocol specification is given in \cref{protocol:verification} and its completeness and soundness for the $\mathsf{BQP}$-complete promise problem $\textnormal{Q-CIRCUIT}$ (\cref{def:q-circuit}) are proven in the subsequent lemmas, culminating in \cref{theorem:bqp-verification}. A constant number of sequential repetitions then yields the standard completeness and soundness bounds of $2/3$ and $1/3$ (\cref{lemma:sequential-repetition}), and thus an argument system for every language in $\mathsf{BQP}$ (\cref{cor:bqp-argument}).
\par
\medskip
Finally, the appendix (\cref{appendix:supplementary-material}) provides supplementary material, including an explicit efficient circuit for the Fourier transform over the extended Pauli group, which is needed to implement the Gowers--Hatami isometry efficiently and may be of independent interest.

\subsection*{Acknowledgments and AI statement}
AN was supported by NSF CAREER grant 2339948. We thank Thomas Vidick and Andru Gheorghiu for helpful conversations, and Tina Zhang and Tony Metger for sharing an unpublished note on self-testing the Pauli $Y$ operator.
\par
\medskip
The authors acknowledge the use of ChatGPT 5.5 for assistance in the initial analysis of the irreducible representations of the extended Pauli group and Claude 4.8 for suggestions regarding the explicit circuit for the quantum Fourier transform over that group. 
\par
\medskip
Anthropic's Claude models provided the proof ideas for Lemmas \ref{lemma:povm-indistinguishability}, \ref{lemma:QPT-measurable}, and \ref{lemma:sequential-repetition}, and the statement and proof of Corollary~\ref{cor:efficient-isometry-state-switching}. Claude was also used for generating \Cref{fig:cliff-structure}, for reviewing the paper, finding mistakes in the proofs and for minor writing improvements. The authors take responsibility for all content.

\section{Preliminaries}\label{chapter:preliminaries}

\subsection{Notation}\label{section:notation}

Throughout this work we assume basic knowledge of quantum information theory and linear algebra, including familiarity with Dirac bra--ket notation. \cref{table:notation-conventions} collects the most frequently used symbols for quick reference; each entry is introduced and discussed in more detail in the following subsections.

\subsubsection*{Sets and bitstrings}
For a positive integer $n$, we let
\[
    [n]\coloneq \{1,2,\ldots, n\}.
\]
We write $\zo^n$ for the set of $n$-bit binary strings, whose coordinates we index by $[n]$. For $x\in\zo^n$, its \emph{Hamming weight} $|x|$ is the number of ones in $x$, i.e.\ $|x|=\sum_{i\in[n]} x_i$, and its \emph{element-wise complement} is $\overline{x}\in\zo^n$ with $\overline{x}_i = 1\oplus x_i$ for all $i\in[n]$. For $i\in[n]$ we write $e_i\in\zo^n$ for the string with a $1$ in position $i$ and zeros elsewhere, where the length $n$ should be clear from the context.
\par 
\medskip
From an index set $I\subseteq [n]$ and a bit-string $x\in\zo^n$, we can define the following two objects:
\begin{itemize}
    \item the \emph{restriction} $\restr{x}{I}\in\zo^{|I|}$ is the string obtained by keeping only the coordinates of $x$ indexed by $I$;
    \item the \emph{projection} $\exten{x}{I}\in\zo^n$ is the $n$-bit string that agrees with $x$ on $I$ and is zero elsewhere.
\end{itemize}
Equivalently, $\exten{x}{I}$ is the element-wise product of $x$ and the indicator string $1^n_I\in\zo^n$ (the string with a $1$ at every position in $I$ and a $0$ elsewhere). 

\subsubsection*{Hilbert spaces, states, and operators}
We work exclusively with finite-dimensional Hilbert spaces over $\C$. A generic Hilbert space is denoted $\H$; when a Hilbert space is associated with a physical system (or register) labeled $A$, we write $\H_A$. Pure quantum states are denoted by Dirac kets $\ket{\psi}\in\H$, and the corresponding density matrix is denoted by the bare symbol $\psi\coloneq\ketbra{\psi}{\psi}$. We will say states (on $\H$) to denote the set of normalized density matrices
\[
    \{\rho\in\Pos(\H) : \Tr(\rho) = 1\},
\]
where $\Pos(\H)$ denotes the positive semi-definite operators on $\H$ (see \cref{table:operator-set-conventions}).
\par
\medskip
For an operator $M$ on $\H$, we denote by $M^\dagger$ its \emph{Hermitian adjoint} and by $\overline{M}$ its element-wise \emph{complex conjugate} in the computational basis; the transpose in the computational basis is $M^T = \overline{M^\dagger}$. We write $\id$ for the identity operator on $\H$ (with a subscript indicating the Hilbert space whenever ambiguity may arise). The trace of an operator $\rho$ is denoted $\Tr{\rho}$. A more thorough discussion of operator and vector norms, including our default convention $\norm{\cdot}=\norm{\cdot}_\infty$, is deferred to \cref{section:distance-measures}.
\par
\medskip   
Our conventions for the most frequently used subsets of the bounded operators on $\H$ are summarized in \cref{table:operator-set-conventions}.

\begin{table}[H]
\renewcommand{\arraystretch}{1.25}
\centering
\begin{tabular}{ll}
    \toprule
    \textbf{Symbol} & \textbf{Definition} \\
    \midrule
    $L(\H)$ & Linear operators $A:\H\to\H$ (endomorphisms of $\H$)\\
    $GL(\H)$ & Invertible linear operators $A:\H\to\H$ (automorphisms of $\H$)\\
    $U(\H)$ & Unitary operators $A:\H\to\H$\\
    $\Pos(\H)$ & Positive semi-definite operators $A:\H\to\H$\\
    $\Herm(\H)$ & Hermitian (self-adjoint) operators $A:\H\to\H$\\
    $\Obs(\H)$ & Binary observables $A:\H\to\H$\\
    $\Cliff(\H)$ & Clifford unitaries $A:\H\to\H$\\
    \bottomrule
\end{tabular}
\caption{Overview of operator sets on a finite-dimensional Hilbert space $\H$.}\label{table:operator-set-conventions}
\end{table}

The following inclusions hold:
\[
    \Pos(\H)\subsetneq\Herm(\H),
    \qquad
    U(\H)\subsetneq GL(\H)\subsetneq L(\H),
    \qquad
    \Cliff(\H) \subsetneq U(\H).
\]
For $A,B\in\Pos(\H)$ we write $A\preceq B$ when $B-A\in\Pos(\H)$, the PSD order. 
\par
\medskip
Unless stated otherwise, all observables are \textit{binary}; for brevity, ``observable'' will henceforth mean binary observable.

\subsubsection*{Fourier transform}
For a function $f:\zo^n\to\C$, we denote its Fourier (or Walsh--Hadamard) transform over $\Z_2^n$ by
\[
    g(x)=\widehat f(x)\coloneq\E_{a\in\zo^n} (-1)^{x\cdot a}f(a),
\]
where the expectation is taken uniformly over $a\in\zo^n$. The inverse transform is
\[
    f(a) = \sum_{x\in\zo^n} (-1)^{a\cdot x}g(x).
\]

\subsubsection*{Measurements and observable families}
Projectors corresponding to single-symbol questions in a protocol specification are denoted by the same symbol as the question, with a superscript denoting the $n$-bit measurement outcome (where $n$ is the expected number of answer bits on that question). For example, if the verifier sends question $Q$, expecting $n$ answer bits, we denote the PVM applied during the un-encrypted interaction with the prover as $\{Q^v\}_{v\in\zo^n}$. If the question is not a single symbol but a tuple, for example $(Q,a)$, we denote the corresponding two-outcome projectors by $\{Q_a^0, Q_a^1\}$. For other tailored questions we don't have a default convention; the corresponding projectors are always introduced explicitly.
\par 
\medskip
To every such projective measurement we associate a family of observables parameterized by $a\in\zo^n$:
\[
    W(a)\coloneq \sum_{x\in\zo^n} (-1)^{a\cdot x}W^x.
\]
If the PVM has only two outcomes, we use the shorthand $W = W(1) = W^0-W^1$; thus two-outcome measurements share the symbol with their corresponding single-symbol questions. The outcome projectors can be interpreted as the Fourier transform of the observable $W(a)$, i.e.\ $W^x = \widehat W(x) = \E_{a\in\zo^n} (-1)^{x\cdot a}W(a)$.

\subsubsection*{Expectations and probability distributions}
We denote the expectation value by $\E$. If the subscript is a variable together with an indication of the set it belongs to (e.g.\ $a\in S$), the expectation is the \emph{uniform} one over that set. If the set is omitted, the distribution should be clear from context. For non-uniform expectations, we are often explicit by writing both the variable and the distribution in the subscript, for example $\E_{W\sim\mu}$. We write $U_n$ for the uniform probability distribution over a set of $n$ elements, i.e.\ $U_n(a)=1/n$. We write $x\draw \cS$ when $x$ is uniformly sampled from a set $\cS$. The arrow may carry a distribution above it, for example $x\draw{\mu}\cS$, in which case $x$ is sampled according to the distribution $\mu$ from the set $\cS$. For a predicate $P$, we write $\ind{P}$ for the indicator taking value $1$ if $P$ holds and $0$ otherwise.

\newpage
\null\vfill
\begin{table}[H]
\renewcommand{\arraystretch}{1.25}
\centering
\begin{tabular}{ll}
    \toprule
    \textbf{Symbol} & \textbf{Definition} \\
    \midrule
    \multicolumn{2}{l}{\emph{Sets and bitstrings}}\\
    $\log(x)$ & The binary logarithm of $x$ (logarithm base 2) \\
    $[n]$ & The set of integers $\{1,2,\ldots,n\}$\\
    $\zo^n$ & The set of $n$-bit binary strings\\
    $e_i$ & The $n$-bit string with a $1$ at position $i\in[n]$ and zeros elsewhere\\
    $|x|$ & The Hamming weight of the binary string $x\in\zo^n$\\
    $\overline{x}$ & The element-wise complement of the binary string $x\in\zo^n$\\
    $\restr{x}{I}$ & The restriction of $x\in\zo^n$ to the coordinates indexed by $I\subseteq[n]$\\
    $\exten{x}{I}$ & The projection of $x\in\zo^n$ onto the coordinates indexed by $I\subseteq[n]$\\
    \midrule
    \multicolumn{2}{l}{\emph{Hilbert spaces, states, and operators}}\\
    $\H$, $\H_A$ & A (finite-dimensional, complex) Hilbert space; subscript denotes the system\\
    $\ket{\psi}$, $\psi$ & A pure state and its density matrix $\psi=\ketbra{\psi}{\psi}$\\
    $\id$ & The identity operator on a Hilbert space $\H$\\
    $M^\dagger$ & The Hermitian adjoint of an operator $M$\\
    $\overline{M}$ & The complex conjugate of an operator $M$ (in the computational basis)\\
    $\Tr$ & The trace of an operator\\
    $\norm{M}$ & The operator norm of $M$, equal to the largest singular value of $M$\\
    \midrule
    \multicolumn{2}{l}{\emph{Fourier transform, expectations, and distributions}}\\
    $\widehat f(a)$ & The Fourier (Walsh--Hadamard) transform of $f$ over $\Z^n_2$\\
    $U_n$ & The uniform distribution over a set of $n$ elements, $U_n(a)=1/n$\\
    $x\draw \cS$ & $x$ is uniformly sampled from the set $\cS$\\
    $\ind{P}$ & Indicator: $1$ if predicate $P$ holds, $0$ otherwise\\
    \midrule
    \multicolumn{2}{l}{\emph{Extended Pauli group (three-tier notation; \cref{def:extended-pauli-group})}}\\
    $x(a),z(a),g(a)$ & Abstract generators of $C_n$; also $y(a)$, $\mathsf{f}(a)$ (derived)\\
    $x_j$ & Single-qubit generator $x(e_j)$ (likewise $y_j$, $z_j$, $g_j$)\\
    $\sigma_W(a)$ & Pauli/Clifford \emph{matrix} $\bigotimes_i \sigma_W^{a_i}$\\
    $W(a)$ & Prover observable for the question labelled $W$\\
    \bottomrule
\end{tabular}
\caption{Overview of notational conventions used throughout this paper.}\label{table:notation-conventions}
\end{table}
\vfill\null
\subsection{Groups}\label{sec:groups}
\subsubsection*{Pauli group} 
\label{sec:pauli-group}

The Pauli matrices are defined as,
\begin{align*}
    \sigma_X \coloneqq \begin{pmatrix} 0 & 1 \\ 1 & 0 \end{pmatrix}, \quad
    \sigma_Y \coloneqq \begin{pmatrix} 0 & -i \\ i & 0 \end{pmatrix}, \quad
    \sigma_Z \coloneqq \begin{pmatrix} 1 & 0 \\ 0 & -1 \end{pmatrix}.
\end{align*}
Additionally, we define
\begin{align*}
    \sigma_G \coloneqq \frac{1}{\sqrt{2}}(\sigma_Y+\sigma_X) =  \frac{1}{\sqrt{2}}\begin{pmatrix} 0 & 1-i \\ 1+i & 0 \end{pmatrix},\quad
    \sigma_F \coloneqq \frac{1}{\sqrt{2}}(\sigma_Y-\sigma_X) =  \frac{1}{\sqrt{2}}\begin{pmatrix} 0 & -i-1 \\ i-1 & 0 \end{pmatrix},
\end{align*}
with $\overline{\sigma}_G = -\sigma_F$, as well as
\begin{equation*}
    \sigma_S \coloneqq \begin{pmatrix} 1 & 0 \\ 0 & i \end{pmatrix}, \quad \sigma_T \coloneqq \begin{pmatrix} 1 & 0 \\ 0 & e^{i\pi/4} \end{pmatrix}.
\end{equation*}
We write $\sigma_X, \sigma_Z,\dots$ rather than $X, Z, \dots$ in order to distinguish the ``true'' matrices from the untrusted operators that the prover applies; in the case of $\sigma_S$ and $\sigma_T$, we will never need to make this distinction in our analysis so we will often use $S$ and $T$ as shorthand for these. We also note that some references, in particular~\cite{Broadbent2018}, use $P$ instead of $S$.
\par
\medskip
For $a \in \mathbb{Z}_{2}^n$ and $W\in\{X,Y,Z,F,G\}$, let
\begin{align*}
    \sigma_W(a) \coloneqq \bigotimes_{i\in[n]} (\sigma_W)^{a_i},
\end{align*}
in other words, $\sigma_W(a)$ is the $2^n$-dimensional matrix that is the tensor product of $\sigma_W$ on all the qubits where $a_i = 1$ and identity elsewhere.
Furthermore, for $\tilde W\in\{X,Y,Z,F,G\}^n$, let
\[
\sigma_{\tilde W}(a) \coloneq\bigotimes_{i\in[n]}(\sigma_{\tilde W_i})^{a_i},
\]
which is the mixed-basis extension of the above. We denote the corresponding projector as 
\[
\tau_{\tilde W}^{a} \deq \bigotimes_{i\in[n]}\tau_{\tilde W_i}^{a_i},
\]
where $\tau_{\tilde W_i}$ is the single-qubit projector defined as
\[
\tau_{\tilde W_i}^{a_i} \deq \frac{1}{2}(\id + (-1)^{a_i}\sigma_{\tilde W_i}).
\]

\subsubsection*{Extended Pauli group}\label{subsec:group}

\begin{definition}[Extended Pauli group]\label{def:extended-pauli-group}
    We use $C_n$ to denote the $n$-qubit extended Pauli group, which is a subgroup of the $n$-qubit Clifford group. It is generated by a central phase \(\omega\) and single-qubit elements \(x_j,z_j,g_j\) for \(j\in[n]\). Writing \(x(a)\coloneqq\prod_{j:a_j=1}x_j\) and likewise \(z(a)\), \(g(a)\), these satisfy the following relations for all \(a,b\in\zo^n\):
    \begin{itemize}
        \item \textit{Orders:} \(\omega^4=x(a)^2=z(a)^2=g(a)^2=\id\).
        \item \textit{Linearity:} \(w(a\oplus b)=w(a)w(b)\) for \(w\in\{x,z,g\}\).
        \item \textit{Centrality:} \(w(a)\omega=\omega w(a)\) for \(w\in\{x,z,g\}\).
        \item \textit{(Anti)commutation:} \(w(a)z(b)=\omega^{2a\cdot b}z(b)w(a)\) for \(w\in\{x,g\}\).
        \item \textit{Conjugation:} \(g(a)x(b)g(a)=\omega^{|a\cap b|}x(b)z(a\cap b)\).
    \end{itemize}
    An element of $C_n$ can be uniquely written in the normal form\footnote{This is a bijection on $C_n$: the set of normal-form words is closed under right multiplication by the generators (surjectivity), and the defining representation of \cref{lemma:clifford-representations} is injective.}: $\omega^{p} x(a) g(b) z(c)$ for $p\in\{0,1,2,3\}$ and $a,b,c \in \zo^n$. This implies that the group has $2^{3n+2}$ elements.
\end{definition}

When a label \(W\in\{X,Y,Z,F,G\}\) is used as a variable, we write \(w(a)\) for the corresponding abstract element of \(C_n\). This includes the generators \(x(a)\), \(z(a)\), \(g(a)\) from the definition, and the derived elements \(y(a)\) and \(\mathsf{f}(a)\), specified by the unique words
\[
y(a)=\omega^{|a|}x(a)z(a),\qquad
\mathsf{f}(a)=\omega^{|a|}g(a)z(a).
\]
Single-qubit generators are the special case \(x_j=x(e_j)\), and likewise \(y_j\), \(z_j\), \(\mathsf{f}_j\), \(g_j\).

\begin{lemma}\label{lemma:clifford-representations}
    The irreducible representations $\rho_\mu: C_n\to U(\C^{d_\mu})$ of the $n$-qubit extended Pauli group can be characterized by their action on the center of the group, i.e.\ $\omega\mapsto i^\ell\id$ for $\ell\in\{0,1,2,3\}$. We can use this to distinguish two classes of irreps:
    \begin{itemize}
        \item \textbf{The `classical' irreps} ($\ell\in\{0,2\}$): $2^k$-dimensional representations (for $0\leq k\leq n$), which map $\omega\mapsto(-1)^{\ell/2}\id$. They can retain anti-commutation of $x$ and $g$ (if $k\neq 0$) but always collapse to commuting $x$ and $z$ (and $g$ and $z$), since they map $\omega^2\mapsto\id$. They are given by choosing $K\subset[n]$, with $|K|=k$, and assigning: 
        \begin{align*}
            \forall j\in K, \quad &x_j \mapsto \sigma_X(e_j), \\
            &g_j \mapsto \sigma_Z(e_j),
        \end{align*} 
        followed by choosing an assignment of either 1 or -1 for each of the generators 
        \[
        x_j, g_j
        \]
        with $j\not\in K$ and using that $y_j=g_j x_j g_j$ and $z_j=\omega y_j x_j$.
        \item \textbf{The `quantum' irreps} ($\ell\in\{1,3\}$): $2^n$-dimensional representations, which map $\omega\mapsto i^\ell\id$, where
        \begin{align*}
            \forall j\in[n], \quad &x_j \mapsto \sigma_X(e_j), \\
            &g_j \mapsto \pm\frac{1}{\sqrt{2}}(\sigma_X(e_j)+i^{(\ell-1)}\sigma_Y(e_j)),\\
            &z_j \mapsto \sigma_Z(e_j);
        \end{align*}
            There is an additional freedom of sign choice on each $g_j$, which means that there are $2^n$ such irreducible representations for a fixed $\ell$. The all-plus, \(\ell=1\) member of this family is the \emph{defining} representation, sending \(x(a)\mapsto\sigma_X(a)\), \(z(a)\mapsto\sigma_Z(a)\), \(g(a)\mapsto\sigma_G(a)\), and likewise \(y(a)\mapsto\sigma_Y(a)\), \(\mathsf{f}(a)\mapsto\sigma_F(a)\). This representation is injective, so the normal form of \cref{def:extended-pauli-group} labels each group element uniquely.
    \end{itemize}
\end{lemma}
\begin{proof}
    It can be verified by direct calculation that each one of these is a representation. It is also clear that all of them are mutually non-isomorphic. The defining representation is injective: a local factor $\sigma_X^{a_j}\sigma_G^{b_j}\sigma_Z^{c_j}$ is scalar if and only if $a_j=b_j=c_j=0$, so $\rho(\omega^p x(a)g(b)z(c))=\id$ forces $p=a=b=c=0$.
    \par
    \medskip
    Recall that the sum of the squares of the degrees of the irreps of a group is equal to the order of the group. The order of $C_n$ is $2^{3n+2}$. The number of $2^k$-dimensional `classical' representations is $N(2^k)=2\cdot 4^{n-k}\cdot\binom{n}{k}$. Summing over its product with the squared dimension yields:
    \[
    \sum_{k=0}^n N(2^k)\cdot (2^k)^2= 2\sum_{k=0}^{n} \binom{n}{k} 4^{n} = 2^{3n+1}
    \]
    the $2^n$-dimensional quantum representations ($\ell = 1$), and their complex conjugates ($\ell = 3$), also contribute $2^{3n+1}$, for a total of $2^{3n+1} + 2^{3n+1} = 2^{3n+2}$. Hence these are all of the irreps of $C_n$.
\end{proof}

\subsection{Quantum Fourier transform}\label{subsec:QFT}
We define the quantum Fourier transform (QFT) over our group as 
\[
U_{\mathrm{QFT}}\coloneq \sum_{g\in C_n}\sum_{\mu}\sum_{i,j\in[d_\mu]}\sqrt{\frac{d_\mu}{|C_n|}}\left[\rho_{\mu}(g)\right]_{i,j}\ket{\mu}\ket{i,j}_{\mu}\bra{g},
\]
where $\mu$ is an index over all irreducible representations and $\rho_\mu$ is a specific irrep. indexed by $\mu$, with dimension $d_\mu$. The subscript $\mu$ on the vector $\ket{i,j}_\mu$ indicates that it lives in a $d_\mu^2$-dimensional space. We will explicitly show that $U_{\mathrm{QFT}}$ is unitary
\begin{align*}
    U_{\mathrm{QFT}}^\dag U_{\mathrm{QFT}}&=\sum_{g,g'}\sum_{\mu}\sum_{i,j\in[d_\mu]}\frac{d_\mu}{|C_n|}\overline{\left[\rho_{\mu}(g')\right]}_{i,j}\left[\rho_{\mu}(g)\right]_{i,j}\ket{g'}\bra{g}\\
    &=\frac{1}{|C_n|}\sum_{g,g'}\sum_{\mu}\sum_{i,j\in[d_\mu]}d_\mu\left[\rho_{\mu}(g')^\dag\right]_{j,i}\left[\rho_{\mu}(g)\right]_{i,j}\ket{g'}\bra{g}\\
    &=\frac{1}{|C_n|}\sum_{g,g'}\sum_{\mu}d_\mu\Tr{\rho_{\mu}(g')^\dag\rho_{\mu}(g)}\ket{g'}\bra{g}\\
    &=\frac{1}{|C_n|}\sum_{g,g'}\sum_{\mu}d_\mu\chi_\mu((g')^{-1}g)\ket{g'}\bra{g}\\
    &=\sum_{g,g'}\delta_{gg'}\ket{g'}\bra{g}\\
    &=\id,
\end{align*}
where the first term uses orthogonality of the $\mu$ and potential $i,j$ registers, the fourth line uses the exact group homomorphism property and the second-to-last line uses \cite[Proposition 1 and Corollary 2]{Serre1977}, where $\chi_\mu$ is the character of the $\mu$-th irrep defined as 
\[
\chi_\mu(g)=\Tr{\rho_\mu(g)}.
\]

\par
\medskip

Let $\H_{C_n}$ be a $|C_n|$-dimensional Hilbert space, whose basis vectors correspond to group elements. The left regular representation $\pi:C_n\to U(\H_{C_n})$ is then defined as
\[
\pi(g)=\sum_{h\in C_n}\ket{gh}\bra{h}.
\]
We will now show that the quantum Fourier transform over $C_n$ block-diagonalizes the left regular representation
\begin{align*}
    U_{\mathrm{QFT}}&(\pi(g))U_{\mathrm{QFT}}^\dag\\
    &=\sum_{\mu,\mu'}\sum_{i,j,k\in[d_\mu]}\sum_{l,m\in[d_{\mu'}]}\frac{\sqrt{d_\mu d_{\mu'}}}{|C_n|}\left[\rho_\mu(g)\right]_{i,k}\overbrace{\sum_{g'}\left[\rho_\mu(g')\right]_{k,j}\overline{\left[\rho_{\mu'}(g')\right]}_{l,m}}^{\text{orthogonality of irreps}}\ket{\mu}\ket{i,j}_\mu\bra{\mu'}\bra{l,m}_{\mu'}\\
    &=\sum_{\mu,\mu'}\sum_{i,j,k\in[d_\mu]}\sum_{l,m\in[d_{\mu'}]}\frac{\sqrt{d_\mu d_{\mu'}}}{|C_n|}\left[\rho_\mu(g)\right]_{i,k}\frac{|C_n|}{d_\mu}\delta_{\mu\mu'}\delta_{kl}\delta_{jm}\ket{\mu}\ket{i,j}_{\mu}\bra{\mu'}\bra{l,m}_{\mu'}\\
    &=\sum_{\mu}\sum_{i,j,k\in[d_\mu]}\left[\rho_\mu(g)\right]_{i,k}\ket{\mu}\ket{i,j}_\mu\bra{\mu}\bra{k,j}_\mu\\
    &=\sum_{\mu}\ket{\mu}\bra{\mu}\sum_{i,k\in[d_\mu]}\left[\rho_\mu(g)\right]_{i,k}\ket{i}\bra{k}_\mu\ot\id_{d_\mu}\\
    &=\bigoplus_\mu \rho_\mu(g)\ot\id_{d_\mu},
\end{align*}
where we used the orthogonality of irreducible representations \cite[Corollary 3]{Serre1977}.

\subsubsection*{Explicit QFT circuit}
To obtain an explicit circuit for the quantum Fourier transform over our group, we need to find an explicit expression for $\left[\rho_{\mu}(g)\right]_{i,j}$. Ideally, this would be a simple index-based expression, matching the specific irreps chosen in \cref{lemma:clifford-representations}. We start by writing out this kind of compact expression for the `classical' irreps ($\ell\in\{0,2\}$):
\begin{align*}
    \left[\rho_{C,\mu}^{(\ell)}(g)\right]_{r,s}&=\left[\rho_{C,\mu}^{(\ell)}(\omega^p x(a)g(b)z(c))\right]_{r,s}\\
    &=i^{p\ell}\underbrace{\left(\prod_{j\in S(\mu\oplus 1^n)}(-1)^{a_jr_j+b_js_j+c_j\ell/2}\right)}_{\text{1D contributions}}\underbrace{\left(\prod_{j\in S(\mu)}(-1)^{b_js_j+(\ell/2+1)c_j}\delta_{r_j,s_j\oplus a_j}\right)}_{\text{2D contributions}}.
\end{align*}
Here $S(\mu)\coloneq\{i\in[n]:\mu_i=1\}$. This indeed yields the expected structure, for example if $S(\mu)=[n]$ and $g=x(1^n)$ we obtain
\[
\left[\rho_{C,K}^{(\ell)}(x(1^n))\right]_{r,s} = \prod_{j\in[n]}\delta_{r_j,s_j\oplus 1} = \delta_{r,s\oplus 1^n} = \sigma_X(1^n),
\]
which is the defining matrix of this irrep on that group element. We can do the same thing for the `quantum' irreps ($\ell\in\{1,3\}$). Let $\omega(s)\coloneq e^{i\frac{\pi}{4}(1-2s)}$, then for $r,s\in\zo$ and $\ell=1$ we have the matrix entries of \(\sigma_G\)
\[
[\sigma_G]_{r,s} = \omega(s)\delta_{r,s\oplus 1},
\]
if $\ell=3$, this becomes
\[
[\sigma_G]_{r,s} = \overline{\omega}(s)\delta_{r,s\oplus 1}.
\]
This observation lifted to the $n$-qubit case (where $r,s\in\zo^n$) allows us to also write down a compact expression for the quantum irreps
\begin{align*}
    \left[\rho_{Q,\mu}^{(\ell)}(g)\right]_{r,s} &= \left[\rho_{Q,\mu}^{(\ell)}(\omega^p x(a)g(b)z(c))\right]_{r,s}\\
    &=i^{p\ell}\prod_{j\in[n]}(-1)^{b_j\mu_j+c_js_j}\omega_\ell(s_j)^{b_j}\delta_{r_j,s_j\oplus a_j\oplus b_j},
\end{align*}
where $\omega_\ell(s)\coloneq \mathbf{1}\{\ell = 1\}\omega(s)+\mathbf{1}\{\ell=3\}\overline{\omega}(s)$. We can combine both classes into a single indexed set of irreps
\[
\left[\rho_{\mu}^{(\ell)}(g)\right]_{r,s}=\mathbf{1}\{\ell\equiv 0 \pmod 2\}\left[\rho_{C,\mu}^{(\ell)}(g)\right]_{r,s} + \mathbf{1}\{\ell\equiv 1 \pmod 2\} \left[\rho_{Q,\mu}^{(\ell)}(g)\right]_{r,s}.
\]
It remains to design an efficient quantum circuit that implements the following operation
\[
\ket{g} = \ket{p_1,p_0,a,b,c}\mapsto \sum_{\ell_0,\ell_1\in\zo}\sum_{\mu,r,s\in\zo^n} \sqrt{\frac{d_\mu}{|C_n|}} \left[\rho_{\mu}^{(\ell_0+2\ell_1)}(g)\right]_{r,s}\ket{\ell_0,\ell_1,\mu,r,s},
\]
where we represented a group element in bit notation where $p_0,p_1\in\zo$ and $a,b,c\in\zo^n$, by the observation that $\omega^p x(a)g(b)z(c)$ uniquely encodes a group element. An efficient circuit for the quantum Fourier transform over $C_n$ is shown in \cref{fig:circuit}, where any control or gate acting on a multi-qubit wire is interpreted as parallel operations (as visualized in \cref{fig:gate-convention}). A correctness proof for this specific circuit can be found in the appendix (\cref{appendix:supplementary-material}).
\begin{figure}
    \centering
    \includegraphics[width=0.40\linewidth]{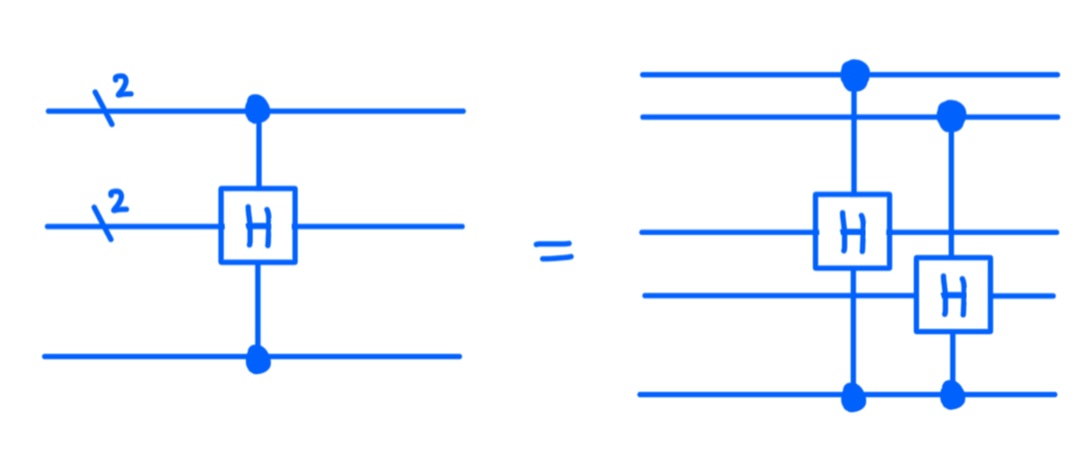}
    \caption{Parallel gate notational convention.}
    \label{fig:gate-convention}
\end{figure}

\begin{figure}
    \centering
    \includegraphics[width=0.95\linewidth]{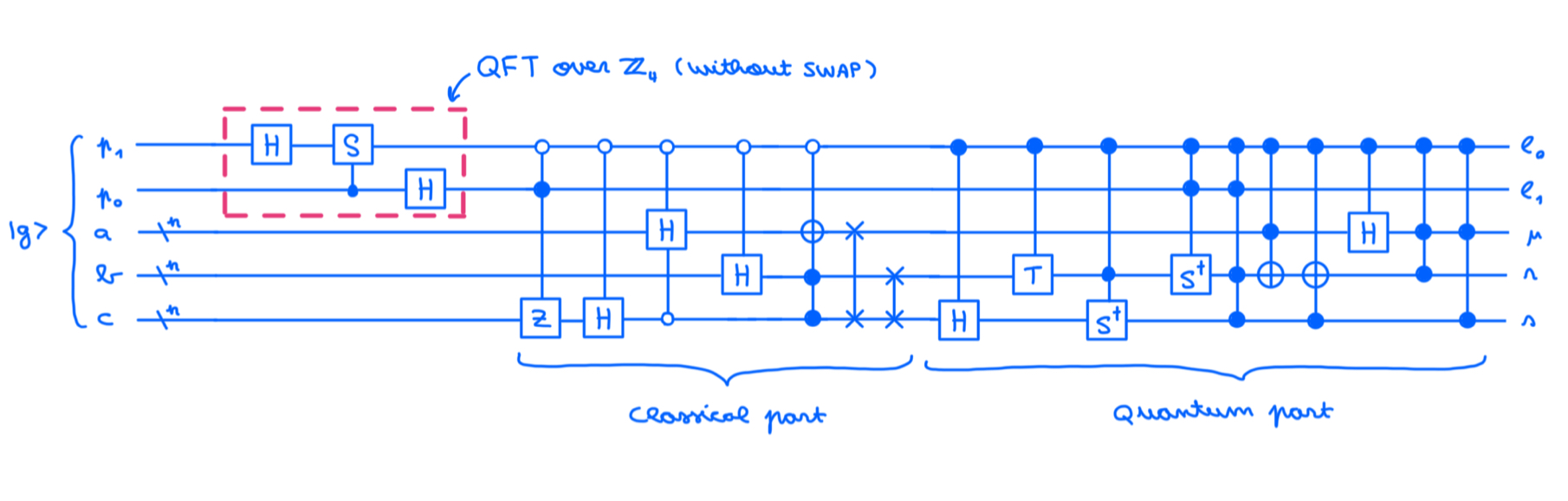}
    \caption{Explicit circuit implementation of $U_{\mathrm{QFT}}$ for $C_n$ in chosen basis. Note that we are using a slightly ambiguous notational convention for parallel gates.}
    \label{fig:circuit}
\end{figure}

\subsection{Bell states}\label{subsec:bell-states}
The four Bell states form an orthonormal basis for $\mathbb{C}^2 \otimes \mathbb{C}^2$
consisting of maximally entangled states.  They are the simultaneous eigenstates of the observables $\sigma_X \otimes \sigma_X$ and $\sigma_Z \otimes \sigma_Z$.
Explicitly, for $a, b \in \{0,1\}$ we define
\begin{equation}\label{eq:bell-def}
  \ket{\Phi^{ab}}
  \;=\;
  \frac{1}{\sqrt{2}}\sum_{s \in \{0,1\}} (-1)^{a\cdot s} \ket{s, s \oplus b},
\end{equation}
so that
\begin{equation*}
  (\sigma_X \otimes \sigma_X)\ket{\Phi^{ab}} = (-1)^{a}\ket{\Phi^{ab}},
  \qquad
  (\sigma_Z \otimes \sigma_Z)\ket{\Phi^{ab}} = (-1)^{b}\ket{\Phi^{ab}}.
\end{equation*}
In the traditional notation this gives
\begin{equation*}
  \ket{\Phi^{00}} = \ket{\Phi^+},\quad
  \ket{\Phi^{10}} = \ket{\Phi^-},\quad
  \ket{\Phi^{01}} = \ket{\Psi^+},\quad
  \ket{\Phi^{11}} = \ket{\Psi^-}.
\end{equation*}

The Bell basis is related to the EPR state
$\ket{\Phi^{00}} = \frac{1}{\sqrt{2}}(\ket{00}+\ket{11})$ by local Pauli
corrections:
\begin{equation*}
  \ket{\Phi^{ab}} \;=\; (\sigma_Z^a \sigma_X^b \otimes \id)\,\ket{\Phi^{00}}.
\end{equation*}
Equivalently, acting on the second register,
$\ket{\Phi^{ab}} = (\id \otimes \sigma_X^b \sigma_Z^a)\,\ket{\Phi^{00}}$.
More generally, for any operator $M$ acting on $\mathbb{C}^2$ the EPR state has the following transpose property
\begin{equation}\label{eqn:transpose-trick}
  (M \otimes \id)\ket{\Phi^{00}} \;=\; (\id \otimes M^T)\ket{\Phi^{00}},
\end{equation}
where $M^T$ denotes the transpose in the computational basis.

Because the Bell states form an ONB, they also resolve the identity
\begin{equation*}
  \sum_{a,b \in \{0,1\}} \ketbra{\Phi^{ab}}{\Phi^{ab}} \;=\; \id_4,
\end{equation*}
and satisfy the twirl identity: for any operator $\rho$ on $\mathbb{C}^2 \otimes \mathbb{C}^2$,
\begin{equation}\label{eq:bell-twirl}
  \frac{1}{4}\sum_{a,b}(\sigma_Z^a \sigma_X^b \otimes \id)\,\rho\,(\sigma_X^b \sigma_Z^a \otimes \id)
  \;=\;
  \sum_{a,b} \bra{\Phi^{ab}}\rho\ket{\Phi^{ab}}\,\ketbra{\Phi^{ab}}{\Phi^{ab}}.
\end{equation}
That is, twirling by the single-qubit Pauli group dephases any two-qubit state
into the Bell basis.

\subsection{Distance measures}\label{section:distance-measures}
The Schatten-$p$ norm of $A\in L(\H)$ is defined as
\[
\|A\|_p \coloneq \Tr{(A^\dag A)^{p/2}}^{1/p}.
\]
The Schatten norms satisfy H\"older's inequality for $p,q,r \in [1,\infty]$ such that $\frac{1}{r} = \frac{1}{p} + \frac{1}{q}$: 
\begin{align*}
\norm{A B}_r \leq \norm{A}_p \norm{B}_q\,.
\end{align*}
We use the shorthand notation $\|A\| = \|A\|_\infty$, i.e.\ our default operator norm is the Schatten-$\infty$ norm, which can equivalently be defined as
\[
\|A\| = \sup_{\braket{x}{x}\leq 1} \|A\ket{x}\| = \sup_{\substack{\braket{x}{x}\leq 1\\ \braket{y}{y}\leq 1}} |\bra{x}A\ket{y}|.
\]
We use the Euclidean norm as our default vector norm.
\begin{definition}[State-dependent inner product and norm]\label{def:state-dependent-norm}
Let $\H$ be a finite-dimensional Hilbert space and $A,B\in L(\H)$ be linear operators on $\H$. Let $\rho\in\mathrm{Pos}(\H)$. We define the state-dependent (semi) inner product of $A$ and $B$ w.r.t $\rho$ as
\[
\langle A,B\rangle_\rho \coloneq \Tr{A^\dagger B\rho}.
\]
This induces the state-dependent (semi) norm
\[
\|A\|^2_\rho \coloneq \langle A,A\rangle_\rho \coloneq \Tr{A^\dagger A\rho}.
\]
\end{definition}
\begin{remark}
    The state-dependent norm can also be expressed as a Schatten-2 norm
    \[
    \|A\|_\rho = \|A\rho^{1/2}\|_2.
    \]
    For a pure state, the state-dependent norm of an operator reduces to the Euclidean norm of that operator acting on the pure state vector.
\end{remark}
\begin{lemma}\label{lemma:state-dependent-norm-properties}
For all $\rho, \rho'\in\mathrm{Pos}(\H)$ and linear operators $A, B \in L(\H)$ on some finite-dimensional Hilbert space $\H$, the following identities hold:
\begin{enumerate}[label=(\roman*),ref=\roman*]
    \item $\|A\|_{B\rho B^\dag} = \|AB\|_\rho$.\label{prop:state-dependent-norm-properties,i}
    \item $\|AB\|_\rho \leq \|A\|\cdot\|B\|_\rho$.\label{prop:state-dependent-norm-properties,ii}
    \item $\forall\;U\in U(\H): \|UA\|_\rho = \|A\|_\rho$.\label{prop:state-dependent-norm-properties,iii}
    \item $\|A\|^2_{\rho+\rho'} = \|A\|^2_\rho + \|A\|^2_{\rho'}$.\label{prop:state-dependent-norm-properties,iv}
    \item $\norm{\sum_{x \in \cX} A_x}_\rho^2 \leq |\cX| \sum_{x \in \cX} \norm{A_x}_\rho^2$.\label{prop:state-dependent-norm-properties,v}
    \item $\norm{\E_{x \in \cX} A_x}_\rho^2 \leq \E_{x \in \cX} \norm{A_x}_\rho^2 \,.$\label{prop:state-dependent-norm-properties,vi}
\end{enumerate}
\end{lemma}

With our distance measure properly introduced, we can now define what it means for two operators to be close to each other.

\begin{definition}[$\delta$-isometric and $\delta$-equivalent]\label{def:delta-isometric}
Let $\H$ and $\mathcal{H'}$ be two finite-dimensional Hilbert spaces, $\delta > 0$, $R \in L(\H)$ and $S\in L(\mathcal{H'})$ linear operators. We say that $R$ and $S$ are $\delta$-isometric with respect to $\rho\in \Pos(\H)$, and write $R \simeq_{\delta} S$, if there exists an isometry $V : \H\to \mathcal{H'}$ such that
\[
    \|VR-SV\|^2_\rho = O(\delta).
\]
If $V$ is the identity (i.e.\ $\H$ = $\mathcal{H'}$), then we further say that R and S are $\delta$-equivalent, and write $R\approx_\delta S$ for $\|R-S\|_\rho^2 = O(\delta)$. With this we can equivalently write $R\simeq_\delta S$ as $VR\approx_\delta SV$, which clearly shows the asymmetry of this notation.
\end{definition}
\begin{remark}
    If we are dealing with the $\delta$-equivalence of a parameterized family of observables, i.e.\ $A(a)\approx_\delta B(a)$, the notation is overloaded to also include an implicit uniform expectation over the parameter, if not specified otherwise, thus
    \[
    \E_a\norm{A(a)-B(a)}^2_\rho = O(\delta),
    \]
    by the equivalence of $A(a)\simeq_\delta B(a)$ to $VA(a)\approx_\delta B(a)V$, this also holds for the $\delta$-isometric notation.
\end{remark}
\begin{remark}
    Our definition of $\delta$-isometric differs from the one used in \cite{Coladangelo2024}. Our notion is strictly stronger, since by \cref{lemma:closeness-equivalence} one can recover the one in \cite{Coladangelo2024} from ours without loss, but not the other way around incurs a square-root loss.
\end{remark}
The notations $R \simeq_{\delta} S$ and $R\approx_\delta S$ are ambiguous since they neither specify the state $\rho$ nor the isometry $V$. Both should always be obvious from the context and will be explicitly stated if not. The isometry relation is transitive, but not reflexive: the operator on the left acts on the state before the isometry, the operator on the right acts after the isometry. The notion of $\delta$-equivalence is both transitive and reflexive; we will use it as our main notion of distance and its transitivity will be frequently used.
\par 
\medskip
We overload the notation $a\approx_\delta b$ to mean $|a-b|=\delta$ for scalars $a$ and $b$, which is consistent with the canonical definition in this setting.
\begin{remark}
    The previous definition extends to scalars if we take $R$ and $S$ to be scalar linear maps ($R=a\cdot\id$ and $S=b\cdot\id$). In this setting $R\approx_\delta S$ would imply $(a-b)^2=O(\delta)$, which gives a square root difference to the canonical definition for approximate equality between scalars.
\end{remark}

\subsection{Nonlocal games}

\begin{definition}[Nonlocal game]\label{def:nonlocal-game}
    A two-player nonlocal game $G$ is specified by two question and answer sets ($\Qa,\;\Qb$ and $\cA_A,\;\cA_B$), a distribution $\cQ$ over pairs $(x,y)\in\Qa\times\Qb$, and a $\PPT$ verification predicate $V(x,y,a,b)\in\zo$, where $a\in\cA_A$ and $b\in\cA_B$.
\end{definition}

\begin{definition}\label{def:classical-strategy-nonlocal-game}
    A deterministic classical strategy $\mathcal{S}_c$ for a two-player nonlocal game $G$ consists of two functions: $f:\Qa\to\cA_A$ and $g:\Qb\to\cA_B$.
\end{definition}

The \emph{(classical)} value or winning probability of $\mathcal{S}_c$ in $G$ is
\begin{equation}\label{eq:classical-strategy-value}
    \omega(G,\mathcal{S}_c)
    \;\coloneqq\;
    \E_{(x,y)\sim Q}\ \sum_{a\in\cA_A}\sum_{b\in\cA_B}
    V(x,y,f(x),g(y))\,.
\end{equation}

\begin{definition}[Local value]\label{def:local-value}
    The local value of a game $G$ is
    \begin{equation}\label{eq:local-value}
        \omega_c(G) \;\coloneqq\; \sup_{\mathcal{S}_c}\ \omega(G,\mathcal{S}_c)\,,
    \end{equation}
    where the supremum is taken over all classical strategies $\mathcal{S}_c$ for $G$ (as in \cref{def:classical-strategy-nonlocal-game}). Here it is sufficient to take the supremum over all deterministic classical strategies, because the $\omega(G,\mathcal{S}_c)$ is linear and the set of classical strategies is the convex hull of the deterministic ones.
\end{definition}

\begin{definition}[Quantum strategy]\label{def:quantum-strategy-nonlocal-game}
    A quantum strategy $\mathcal{S}$ for a two-player nonlocal game $G$ consists of:
    \begin{itemize}
        \item a bipartite finite-dimensional quantum state $\ket{\psi}\in\H_{A}\otimes\H_{B}$;
        \item for every $x\in\Qa$, a projective measurement $\{A^{x}_{a}\}_{a\in\cA_A}$ on $\H_{A}$ (Alice's measurements);
        \item for every $y\in\Qb$, a projective measurement $\{B^{y}_{b}\}_{b\in\cA_B}$ on $\H_{B}$ (Bob's measurements).
    \end{itemize}
\end{definition}

The \emph{(quantum)} value or winning probability of $\mathcal{S}$ in $G$ is
\begin{equation}\label{eq:quantum-strategy-value}
    \omega^{*}(G,\mathcal{S})
    \;\coloneqq\;
    \E_{(x,y)\sim Q}\ \sum_{a\in\cA_A}\sum_{b\in\cA_B}
    V(x,y,a,b)\cdot \bra{\psi} A^{x}_{a}\ot B^{y}_{b}\ket{\psi}\,.
\end{equation}

\begin{definition}[Entangled value]\label{def:entangled-value}
    The entangled value of a game $G$ is
    \begin{equation}\label{eq:entangled-value}
        \omega^{*}(G) \;\coloneqq\; \sup_{\mathcal{S}}\ \omega^{*}(G,\mathcal{S})\,,
    \end{equation}
    where the supremum is taken over all quantum strategies $\mathcal{S}$ for $G$ (as in \cref{def:quantum-strategy-nonlocal-game}).
\end{definition}

\subsection{Cryptography}
\begin{definition}[Quantum polynomial time]
    The class \emph{quantum polynomial time} ($\QPT$) consists of all procedures that can be implemented by a logspace-uniform family of quantum circuits with size polynomial in 1) the number of qubits $n$ which they take as input, and 2) the security parameter $\lambda$. We denote by $\PPT$ the corresponding class of classical probabilistic polynomial-time procedures.
\end{definition}

\begin{definition}[Non-uniform quantum polynomial time]
    The class \emph{non-uniform quantum polynomial time} ($\nuQPT$) consists of all procedures that can be implemented by a family of quantum circuits $\{C_{\lambda}\}_{\lambda\in\Nat}$ with $|C_{\lambda}|\le\poly(\lambda)$. The circuit family $\{C_\lambda\}_{\lambda}$ need not be uniformly generated.
\end{definition}
\begin{remark}
    Non-uniform QPT often includes access to quantum advice, which is not necessarily efficiently computable. To work under more conservative hardness assumptions, we exclude it here. The indistinguishability reductions of this section run the prover only once, so they would lift to quantum-advice distinguishers if the hardness assumption was upgraded; it might be useful to assume quantum auxiliary input to get desirable guarantees for composability and sequential repetition, we show that the latter does not require quantum advice in our setting. Even though we are considering non-uniform procedures, throughout this work we will often drop the indexing by $\lambda$ for ease of exposition.
\end{remark}

\begin{definition}[Efficiency and distinguishing advantage]\label{def:crypto-small}
    Throughout this work, what qualifies as `efficient' and which distinguishing advantage is tolerated, will depend on a chosen hardness assumption. We will consider the following two hardness assumptions, which always include classical advice, which determine a class of \emph{efficient} quantum algorithms and a class of \emph{cryptographically small} functions $\eta:\Nat\to[0,1]$:
    \begin{itemize}
        \item \textbf{Polynomial hardness:} efficient means $\nuQPT$ and $\eta$ ranges over negligible functions of $\lambda$.
        \item \textbf{Sub-exponential hardness:} efficient means a non-uniform family of circuits $\{C_\lambda\}_{\lambda\in\Nat}$, with $|C_\lambda|\leq 2^{O(\lambda^\delta)}$ and $\eta$ ranges over functions satisfying $\eta(\lambda)=2^{-\Omega(\lambda^{\delta})}$. Here, $\delta\in(0,1)$ is determined by the assumed hardness of the cryptographic scheme.
    \end{itemize}
    Here, $\lambda$ denotes the security parameter of the cryptographic scheme. We write $\ea$ for a generic cryptographically small function arising from the cryptography, that may differ between statements.
\end{definition}
\begin{remark}[Setting $\lambda$ relative to the game size]\label{rem:lambda-vs-n}
    Honest provers act on $O(n)$ qubits, where $n$ is the game size. Under polynomial hardness this requires $\lambda=n^{\Omega(1)}$, so that honest, $\poly(n,\lambda)$-sized, provers lie in the $\nuQPT$ adversary class. Under sub-exponential hardness it suffices to take $\lambda=(\log n)^{\Theta(1/\delta)}$, since then $\poly(n,\lambda)=2^{O(\lambda^\delta)}$. Resource overheads of the form $\poly(\lambda)\cdot n$ are therefore almost linear in $n$ in the first instantiation (specifically $O(n^{1+\eps})$ for every $\eps>0$) and quasilinear in the second. We keep all rigidity lemmas in terms of a generic cryptographically small $\ea$ and instantiate $\lambda$ according to one of the two hardness assumptions from \cref{def:crypto-small}, only when stating concrete resource bounds (\cref{theorem:bqp-verification}, \cref{subsec:qfhe-overhead}).
\end{remark}

\begin{remark}[Non-uniformity]\label{rem:non-uniformity}
    Non-uniformity is primarily required for ease of exposition: From separate IND-CPA invocations we get multiple cryptographically small functions $\eta_p(\lambda)$. Hardwiring, for every security parameter, a worst-case parameter $p^*$ into the advice yields a single adversary, whose advantage $\nu(\lambda)$ must also be cryptographically small, with $\eta_p(\lambda) \le \nu(\lambda)$ for all $p$. Since $\mathbb{E}_p[\eta_p] \le \max_p \eta_p$, expectations over arbitrary distributions are then again cryptographically small (when considering non-uniform adversaries). We deliberately do not allow quantum advice as it is not needed for our argument and would be a strong assumption under sub-exponential hardness. 
\end{remark}

\begin{definition}[Computational indistinguishability]\label{def:comp-indistinguish}
Let $\{\rho_\lambda\}_\lambda$ and $\{\rho'_\lambda\}_\lambda$ be two families of states in $\Pos(\H)$, indexed by the security parameter $\lambda$. We say that they are computationally indistinguishable if for every efficient two-outcome POVM (corresponding to the non-uniform family $\{M_\lambda,\id-M_\lambda\}_\lambda)$ on $\H$, it holds that
\[
\Tr{M_\lambda(\rho_\lambda-\rho'_\lambda)}\leq \ea,
\]
with $\ea$ cryptographically small. This is denoted by
\[
\rho\overset{c}{\approx}\rho',
\]
where we suppress the dependence of the states on $\lambda$. If we write
\[
\rho\overset{c}{\approx}_\eps\rho',
\]
this means that every such distinguisher has an advantage of $O(\eps)$:
\[
\Tr{M_\lambda(\rho_\lambda-\rho'_\lambda)}\leq O(\eps).
\]
\end{definition}

The following definition is taken from \cite{Kalai2023}, with some minor modifications. For example, we make the evaluation key explicit as it will be considered separately in the resource accounting of \cref{subsec:qfhe-overhead}.

\begin{definition}[Quantum fully homomorphic encryption ($\QFHE$)]\label{def:quantum-fully-homomorphic-encryption}
    A quantum fully homomorphic encryption scheme $\QFHE=(\Gen,\Enc,\Eval,\Dec)$ for a class of quantum circuits $\cC$ is a tuple of algorithms with the following syntax:
    \begin{itemize}
        \item \textbf{$\Gen$} is a $\PPT$ algorithm that takes as input the security parameter $1^{\lambda}$ and the number of levels $1^L$ of the circuit, and outputs a classical secret key $\sk\in\zo^{\poly(\lambda)}$ and an evaluation key $\evk\in\zo^{\poly(\lambda,L)}$;
        \item \textbf{$\Enc$} is a $\PPT$ algorithm that takes as input a secret key $\sk$ and a classical plaintext $x$, and outputs a ciphertext $\ct$;
        \item \textbf{$\Eval$} is a $\QPT$ algorithm that takes as input a tuple $(C,\ket{\Psi},\ct_{\mathrm{in}},\evk)$, where $C:\H\times(\C^{2})^{\ot n}\to (\C^{2})^{\ot m}$ is a quantum circuit (with at most $L$ levels), $\ket{\Psi}\in\H$ is a quantum state, $\ct_{\mathrm{in}}$ is a ciphertext encrypting an $n$-bit message and $\evk$ is an evaluation key. $\Eval$ runs a quantum circuit $\Eval_{C}(\ket{\Psi}\ot\ket{0}^{\ot \poly(\lambda,n)},\ct_{\mathrm{in}},\evk)$ and outputs a ciphertext $\ct_{\mathrm{out}}$. If $C$ produces a classical output, then $\Eval_{C}$ is required to produce a classical output as well.
        \item \textbf{$\Dec$} is a $\QPT$ algorithm that takes as input a secret key $\sk$ and a ciphertext $\ct$, and outputs a quantum state $\ket{\phi}$. Moreover, if $\ct$ is a classical ciphertext, then decryption outputs a classical bit string $y$.
    \end{itemize}
    We require the following properties from $\QFHE$:
    \begin{itemize}
        \item \textbf{Correctness with auxiliary input:} For every security parameter $\lambda\in\Nat$, any quantum circuit $C:\H\times(\C^{2})^{\ot n}\to (\C^{2})^{\ot m}$ (with classical output and at most $L$ levels), any quantum state $\ket{\Psi}_{AB}\in\H_A\ot\H_B$, any message $x\in\zo^n$, any key pair $(\sk,\evk)\draw\Gen(1^\lambda,1^L)$ and any ciphertext $\ct\draw\Encsk(x)$, the following states have cryptographically small trace distance:
        \begin{enumerate}[label=\textbf{Game~\arabic*.},wide,labelindent=0pt]
            \item\label{game:qhe-corr-plain} Start from $(x,\ket{\Psi}_{AB})$, evaluate $C$ on classical input $x$ and register~$A$, obtaining a classical string $y$, and output $y$ together with the contents of register~$B$.
            \item\label{game:qhe-corr-hom} Start from $\ket{\Psi}_{AB}$ and $\ct$. Homomorphically evaluate $C$ on register~$A$ by running $\ct'\leftarrow\Eval_{C}(\cdot\ot \ket{0}^{\ot \poly(\lambda,n)},\ct,\evk)$. Compute $y'=\Decsk(\ct')$. Output $y'$ together with the contents of register~$B$.
        \end{enumerate}
        \item \textbf{IND-CPA security against efficient quantum distinguishers:} For any efficient adversary $\cA$ (in the sense of \cref{def:crypto-small}) and any two messages $x_{0},x_{1}\in\zo^n$, there exists a cryptographically small function $\ea$ such that,
        {\small
        \begin{align*}
            \Bigg|
            \Pr&\!\left[
            \cA^{\Encsk(\cdot)}(\ct_{0},\evk)=1
            \;\middle|\;
            \begin{array}{@{}l@{}}
                (\sk,\evk)\draw\Gen(1^{\lambda},1^L) \\
                \ct_{0}\draw\Encsk(x_{0})
            \end{array}
            \right]
            -\\
            &\Pr\!\left[
            \cA^{\Encsk(\cdot)}(\ct_{1},\evk)=1
            \;\middle|\;
            \begin{array}{@{}l@{}}
                (\sk,\evk)\draw\Gen(1^{\lambda},1^L) \\
                \ct_{1}\draw\Encsk(x_{1})
            \end{array}
            \right]
            \Bigg|
            \leq \ea.
        \end{align*}
        }
    Here $\Encsk(\cdot)$ denotes the encryption oracle instantiated with key $\sk$.
    \end{itemize}
\end{definition}

\begin{remark}[Instantiation]
We will use the leveled scheme due to Brakerski \cite{Brakerski2018} to instantiate QFHE in this work. The scheme has a hybrid structure, where ciphertexts consist of a QOTP encrypted state and a classical FHE encryption of the OTP keys. Once these classical encryptions are replaced by encryptions of
a fixed unrelated value---such as zero---by a hybrid argument,the OTP hides the state information-theoretically, so any adversary against the IND-CPA security of the QFHE yields an adversary against the classical FHE. One can then apply a standard argument relating FHE security to that of LWE, with a polynomial loss~\cite[Theorems~3.4 and~3.6]{Brakerski2018}. Thus, the hardness assumption on LWE (against polynomial-time or sub-exponential-time adversaries in our case) propagates cleanly, up to polynomial factors, to the semantic security of the QFHE scheme.
\end{remark}

\subsubsection{The KLVY transform}
With any scheme satisfying \cref{def:quantum-fully-homomorphic-encryption} we can instantiate the transformation from nonlocal games to single-prover arguments, proposed by \cite{Kalai2023}. \Cref{def:compiled-nonlocal-game} presents the transform for a two-player nonlocal game, since this is the only setting which will be used in this paper. The general transform applies to nonlocal games with an arbitrary number of players and is described in \cite[Section 3.2]{Kalai2023}.

\begin{definition}[Compiled nonlocal game]\label{def:compiled-nonlocal-game}
    Fix a quantum homomorphic encryption scheme $\QFHE = (\Gen, \Enc, \Eval, \Dec)$. The KLVY transform of the two-prover nonlocal game $G=(\cQ,V)$ works as follows:
    \begin{enumerate}
    \item The verifier samples $(x, y) \draw \cQ$, $(\sk,\evk) \draw\Gen(1^\lambda,1^L)$, and $c \draw \Encsk(x)$. The verifier then sends $c$ to the prover as its first message. (The evaluation key $\evk$ is used to evaluate the quantum circuit $C$, of depth $L$, on the ciphertext $c$ and is assumed to be public knowledge to the prover.)
    \item The prover replies with a message $\alpha$.
    \item The verifier sends $y$ to the prover in the clear.
    \item The prover replies with a message $b$.
    \item Define $a \deq \Decsk(\alpha)$. The verifier accepts if and only if $V(x, y, a, b) = 1$.
    \end{enumerate}
\end{definition}

\begin{remark}[Fixed-length encoding of Alice questions]\label{rem:fixed-length-encoding}
IND-CPA security (\cref{def:quantum-fully-homomorphic-encryption}) is only guaranteed for pairs of messages $x_0,x_1\in\zo^n$ of the \emph{same} length $n$; every switching argument built on it in \cref{subsec:crypto-security} (\cref{lemma:poly-state-indistinguishability,lemma:efficient-operator-indistinguishability,cor:efficient-operator-state-switching} and their uses throughout this work) is therefore only meaningful once all elements of Alice's question set $\Qa$ are encoded as bit strings of one common length before encryption. Throughout this work, we fix once and for all a canonical, efficiently invertible padding scheme and encode every Alice question---large-answer questions $W\in\Sigma^n$, small-answer questions $(W,a)$, Bell pairings, magic-square rows and columns, CHSH questions, and the delegation-game question $\tilde W\in\{X,Y,Z,F,G\}^m$ alike, across every test in \cref{protocol:cliff-test,protocol:verification}---as a bit string of the same length $n_{\Qa}=\poly(\lambda)$ before encryption, padding shorter encodings to $n_{\Qa}$. Under this convention $\ct_0,\ct_1$ are always ciphertexts of equal-length plaintexts, so the ciphertext length carries no information about which type of question was sent, and every use of \cref{lemma:poly-state-indistinguishability} or \cref{cor:efficient-operator-state-switching} to switch between Alice questions of different types is licensed by IND-CPA security in this sense.
\end{remark}

\subsubsection{Modeling prover strategies in a compiled game}\label{section:modeling}
The main benefit of compiled nonlocal games is the fact that you are only dealing with a single prover. Unfortunately, this also has some drawbacks. For example, a common technique in nonlocal games is to trace out one of the prover's systems, since by assuming non-communication and finite-dimensional Hilbert spaces, the actions of different provers can be assumed to be local to their system; because of this locality, any action which Alice performs after her measurement won't affect Bob's outcome. In the single-prover setting these assumptions no longer hold. The prover now acts on the same quantum state in the first and second round of the interaction. We will now establish the most general way to model the prover's actions in a compiled nonlocal game, using the same conventions as \cite{Natarajan2023,Metger2024}
\par 
\medskip
We start by introducing the initial prover state $\ket{\psi}$, which is assumed to be efficiently prepareable and an implicit function of the cryptography's evaluation key ($\evk$) and security parameter ($\lambda$). All actions performed by the prover before receiving the first verifier message can be absorbed into this state. Without loss of generality, we can assume that $\ket{\psi}$ is a pure state.
\par 
\medskip
When the prover receives the first encrypted message ($c\draw\Encsk(x)$ for $x\in\Qa$) from the verifier, he performs a general projective measurement $\{P^\alpha_c\}_\alpha$, where $\alpha$ goes over all possible encrypted responses to that question; again, there is an implicit dependence on $\evk$ and $\lambda$. The measurement is assumed to be projective, since we can always extend the prover's space to obtain a projective measurement by Naimark's dilation theorem. After this the prover is free to apply any unitary ($U^c_\alpha$) which depends on the question and his output; we can absorb this operation into his projective measurement and define the \textbf{non-unitary} operator $A^c_\alpha\deq U^c_\alpha P^\alpha_c$, where $\{(A^c_\alpha)^\dag A^c_\alpha\}_\alpha$ is a projector-valued measurement (PVM). Up to this point, we can identify all actions performed by the prover into the `Alice part' of the nonlocal game. The sub-normalized `post-Alice' state is defined as
\[
\ket{\psi^c_\alpha} = A^c_\alpha \ket{\psi} =  U^c_\alpha P^\alpha_c \ket{\psi},
\]
where 
\[
\Pr[\alpha] = \braket{\psi^c_\alpha}{\psi^c_\alpha} = \bra{\psi}(A^c_\alpha)^\dag A^c_\alpha\ket{\psi},
\]
is the probability of the prover outputting $\alpha$ on the first question. Lastly,
\[
\psi^c = \sum_\alpha \psi^c_\alpha = \sum_\alpha\ket{\psi^c_\alpha}\bra{\psi^c_\alpha}.
\]
After receiving the first response, the verifier sends the second message ($b\in\Qb$) in the clear and the prover's actions can be modeled by a projective measurement $\{B^b_y\}_b$, which will often be denoted by the corresponding question symbol. The prover is free to perform any unitary after this measurement, but we do not need to account for this in our model, since the interaction in a nonlocal game stops after receiving the second message from the prover.
\par 
\medskip
Throughout this work we will use the following compact notation:
\[
\pe{Q} = \sum_\alpha \pea{Q}\quad\text{and}\quad \pea{Q}=\E_{c\leftarrow\Enc(Q)}A^c_\alpha\ket{\psi}\bra{\psi}(A^c_\alpha)^\dag,
\]
where the expectation over the secret key $\sk$ is implicit. Specifically, in any expression containing one or multiple secret key dependent objects, we have an implicit expectation over the secret key. I.e.
\[
\norm{O(\Dec(\alpha))}_{\pe{Q}}^2 = \E_{\sk}\norm{O(\Decsk(\alpha))}^2_{\psi^{\Encsk(Q)}},
\]
where
\[
\psi^{\Encsk(Q)} = \E_{c\leftarrow\Encsk(Q)}\psi^c.
\]
This is very important, since the security guarantees of the cryptography (for example IND-CPA security) only hold under expectation over the secret key, as it is randomly generated by the key generation function. The same implicit expectation is assumed for the evaluation key.
\subsubsection{Security of the cryptography}\label{subsec:crypto-security}
In this section we will introduce some useful lemmas from \cite{Natarajan2023} regarding the indistinguishability of different objects under the cryptography. We had to modify some of the statements, since we are working with parametrized observables implemented by Bob, which can be correlated with the Alice question.
\par 
\medskip
The following lemma uses the IND-CPA security of the QFHE scheme to argue that two of the prover's states after the encrypted interaction---on two different questions---are computationally indistinguishable, even if the distinguishing measurement is allowed to depend on the prover's first-round outcome $\alpha$ and one takes the expectation over different questions. Since this invokes properties of the QFHE scheme, it introduces a cryptographically small function $\ea$, which represents the prover's distinguishing advantage and depends on the security parameter of the encryption. It also depends on the prover and on the distinguishing procedure (i.e.\ the distributions and the measurement family in the lemma below), each of which is understood as a $\lambda$-indexed family. This function will appear in most of our error bounds, and is always related to the security of the QFHE scheme.
\par
\paragraph{Convention (cryptographically small functions).} Throughout this work, cryptographically small functions (in the sense of \cref{def:crypto-small}) will appear. These always depend on the specific adversary. Thus when separately invoking the indistinguishability lemmas of this section, we are dealing with multiple different cryptographically small functions. When averaging multiple cryptographically small functions over a set whose size depends on $\lambda$, the result might not be cryptographically small. This is why the lemmas in this section construct a particular adversary, who samples the questions according to the desired distribution himself, such that a single cryptographically small function is obtained through a single invocation of the IND-CPA security. This is possible as long as the adversary can construct the question himself and its distribution is efficiently sampleable, which will always be the case in our work. However, explicitly carrying this average through calculations often unnecessarily complicates the notation. This is why we will often stick to point-wise bounds and resort to obtaining a bound on the expectation by leveraging the assumption of non-uniform adversaries (see \cref{rem:non-uniformity}).

\begin{remark}
    The following lemma introduces an efficient family of POVMs $\{(M_{p},\id-M_{p})\}_{p}$ that is \emph{uniform in $p$}: for every $\lambda\in\Nat$ there is a single circuit $C_{\lambda}$, of size at most the bound in \cref{def:crypto-small}, which takes $p$ as classical input and implements the measurement $(M_{p},\id-M_{p})$. Our adversaries and their operations are thus non-uniform in the security parameter $\lambda$, but uniform in additional parameters, such as $p$; in particular, the size bound does not depend on $p$, which is what \cref{rem:non-uniformity} requires.
\end{remark}

\begin{lemma}\label{lemma:poly-state-indistinguishability}
    Let $D$ be an efficiently sampleable distribution over triples $(x,z_0,z_1)$, where $x$ is any parameter and $z_0,z_1$ are plaintext Alice questions. Then, for any efficient prover (modeled as in \cref{section:modeling}) and any efficient two-outcome POVM family $\{M_{x,\alpha},\id-M_{x,\alpha}\}_{x,\alpha}$ (uniform in $x$ and $\alpha$), there exists a cryptographically small function $\ea$ such that for all $\lambda\in\Nat$,
    \[
    \left|\E_{(x,z_0,z_1)\sim D}\sum_\alpha\Tr{M_{x,\alpha}\paren{\pea{z_0}-\pea{z_1}}}\right|\leq\ea.
    \]
    In particular, if $D_1$ and $D_2$ are efficiently sampleable distributions over pairs $(x,z)$ with identical marginals on $x$ and efficiently sampleable conditionals $D_i(\cdot\mid x)$, the same bound holds with
    \[
    \left|\E_{(x,z)\sim D_1}\sum_\alpha\Tr{M_{x,\alpha}\pea{z}}-\E_{(x,z)\sim D_2}\sum_\alpha\Tr{M_{x,\alpha}\pea{z}}\right|\leq\ea.
    \]
\end{lemma}
\begin{proof}
    This follows from a standard reduction to the IND-CPA security of the QFHE scheme. For a triple $t=(x,z_0,z_1)$ and $i\in\zo$ let
    \[
    p_i(t)\deq\sum_\alpha\Tr{M_{x,\alpha}\,\psi_\alpha^{\Enc(z_i)}},
    \]
    and write $f(t)\deq p_0(t)-p_1(t)$. The quantity $p_i(t)$ is the acceptance probability of the efficient experiment $\cE(x,c)$: prepare $\ket{\psi}$, apply the prover's first-round measurement on ciphertext $c$, obtaining an outcome $\alpha$, then measure $\{M_{x,\alpha},\id-M_{x,\alpha}\}$. Define a distinguisher $\cA$ against the IND-CPA security of the scheme as follows:
    \begin{itemize}
        \item Sample $t=(x,z_0,z_1)\sim D$ and submit $(z_0,z_1)$ to the IND-CPA challenger; receive the challenge ciphertext $c\draw\Enc(z_b)$, where $b\draw\zo$ is chosen by the challenger.
        \item Run $\cE(x,c)$ and output the outcome bit.
    \end{itemize}
    The reduction is efficient in the sense of \cref{def:crypto-small}: it samples from an efficient distribution, runs the efficient prover once, and implements a uniformly efficient measurement. Conditioned on the challenge bit $b$, the acceptance probability of $\cA$ is $\E_{t\sim D}p_b(t)$. The distinguishing advantage is therefore $\bigl|\E_{t\sim D}f(t)\bigr|$, which is cryptographically small by IND-CPA security of the QFHE scheme.
    \par
    \medskip
    For the $D_1,D_2$ form, sample $x$ from the common marginal (e.g.\ sample $(x,z)\sim D_1$ and discard $z$), then sample $z_0\sim D_1(\cdot\mid x)$ and $z_1\sim D_2(\cdot\mid x)$. This is an efficiently sampleable distribution over triples, and the first claim specialises to the second.
\end{proof}

\par
\medskip
The lemma above is a strengthening of \cite[Lemma 8]{Natarajan2023}, since it allows for a dependence on $\alpha$ and introduces an additional parameter $x$. Choosing $M_{x,\alpha}=M_x$ independent of $\alpha$ yields the corresponding statement for the marginal states $\pe{z}$.

\begin{lemma}\label{lemma:povm-indistinguishability}
    Let $D_1$ and $D_2$ be two distributions over pairs $(x,z)$, where $x$ is any parameter and $z$ is a plaintext Alice question. Suppose $D_1$ and $D_2$ satisfy the following.
    \begin{itemize}
        \item $D_1$ and $D_2$ are efficiently sampleable, with identical marginals on $x$.
        \item For any $x$ in the support of their common marginal, the conditional distributions $D_1(\cdot | x)$ and $D_2(\cdot | x)$ over $z$ are also efficiently sampleable.
    \end{itemize} Then, for any efficient prover, and for any efficient family of POVMs $\{\{M_{\beta,x,\alpha}\}_\beta\}_{x,\alpha}$ (uniform in $x$ and $\alpha$) with outcomes $|\beta|\leq O(1)$, there exists a cryptographically small function $\ea$ such that for all $\lambda\in\Nat$,
    \[
    \left|\E_{(x,z)\sim D_1}\sum_\alpha\sum_\beta\beta\Tr{M_{\beta,x,\alpha}\pea{z}}-\E_{(x,z)\sim D_2}\sum_\alpha\sum_\beta\beta\Tr{M_{\beta,x,\alpha}\pea{z}}\right|\leq\ea.
    \]
\end{lemma}
\begin{proof}
    Let $C\in\Nat$ such that $|\beta|\leq C$ for all $(x,\alpha)$, and define the rescaled outcomes as,
    \[
    \tilde\beta = \frac{\beta + C}{2C}\in[0,1].
    \]
    Define the two-outcome POVM $\{N_{x,\alpha},\id-N_{x,\alpha}\}$, with
    \[
    N_{x,\alpha} \coloneqq\sum_\beta\tilde\beta M_{\beta,x,\alpha},
    \]
    intuitively this measurement can be implemented as follows: first perform the $M$ measurement (efficient given the pair $(x,\alpha)$) and obtain an outcome $\beta$, then output $1$ with probability $\tilde\beta$. $N_{x,\alpha}$ is indeed a valid POVM element, since $N_{x,\alpha} \succeq 0$ ($\tilde\beta\geq 0$ and $M_{\beta,x,\alpha}\succeq 0$) and 
    \[
    N_{x,\alpha} \preceq \sum_\beta \id \cdot M_{\beta,x,\alpha}=\id,
    \]
    it is also uniformly efficient in $(x,\alpha)$, since $\{\{M_{\beta,x,\alpha}\}_\beta\}_{x,\alpha}$ is. Applying \cref{lemma:poly-state-indistinguishability} to $N_{x,\alpha}$ in the $D_1,D_2$ form, there exists a cryptographically small function $\ea$ such that
    \[
    \left|\E_{(x,z)\sim D_1}\sum_\alpha\Tr{N_{x,\alpha}\pea{z}}-\E_{(x,z)\sim D_2}\sum_\alpha\Tr{N_{x,\alpha}\pea{z}}\right|\leq\ea,
    \]
    inserting the definition of $N_{x,\alpha}$ and $\tilde\beta$, we get
    \[
    \left|\E_{(x,z)\sim D_1}\sum_{\alpha,\beta} \paren{\frac{\beta+C}{2C}}\Tr{M_{\beta,x,\alpha}\pea{z}}-\E_{(x,z)\sim D_2}\sum_{\alpha,\beta}\paren{\frac{\beta+C}{2C}}\Tr{M_{\beta,x,\alpha}\pea{z}}\right|\leq\ea.
    \]
    The shift of $C$ contributes the same constant to both terms, since by completeness of the POVM
    \[
    \sum_\alpha\sum_\beta\frac{C}{2C}\Tr{M_{\beta,x,\alpha}\pea{z}}=\frac{1}{2}\sum_\alpha\Tr{\pea{z}}=\frac{1}{2}\Tr{\pe{z}}=\frac{1}{2},
    \]
    so it cancels, and the re-scaling of $2C$ can be absorbed in $\ea$, which remains cryptographically small. This completes the proof.
\end{proof}

\begin{definition}
    Let $H_x \in\Herm(\H)$ be a Hermitian matrix parameterized by $x$. A uniformly efficient family of block encodings of $\{H_x\}$ with scale factors $\{t_x\}$ and auxiliary dimension $k$ is an efficient circuit, that given $x$ implements a unitary matrix $U_x\in U((\C^2)^{\ot k}\otimes\H)$, such that 
    \[
    (\bra{0^k}\otimes\id)U_x(\ket{0^k}\otimes\id)=t_xH_x,
    \]
    or graphically for the case of $k=1$:
    \[
    U_x = \begin{pmatrix}
    t_xH_x & * \\
    *  & *
    \end{pmatrix}.
    \]
    That is, the set $\{U_x\}_x$ is a efficient family of unitaries, uniform in $x$, with the additional block-encoding condition.
\end{definition}

\begin{lemma}\label{lemma:observables-efficient-unitaries}
    The set of unitaries $\{B(a)\}_a$ defined from an efficient PVM $\{B^x\}_{x\in\zo^n}$, as
    \[
    B(a)=\sum_{x\in\zo^n}(-1)^{x\cdot a}B^x,
    \]
    is a efficient family of unitaries, uniform in $a$. I.e. there exists an efficient procedure, that given $a$ as input, implements $B(a)$.
\end{lemma}
\begin{proof}
     Since the PVM is efficiently implementable, there exists an efficient unitary $U$ (the measurement circuit) such that
    \[
    U\ket{\psi}\ket{0}=\sum_x (B^x\ket{\psi})\ket{x},
    \]
    where the second register is an ancilla register, which records the measurement outcome. We have 
    \[
    B(a)\ket{\psi}\ket{0} = U^\dag(\id\ot\sigma_Z(a))U\ket{\psi}\ket{0},
    \]
    which is also efficiently implementable given $a$ as input.
\end{proof}

\begin{definition}[Uniformly efficient family of linear combination of unitaries (LCU)]
A set of operators $\{A_x\}_x$, indexed by a classical parameter $x$, is a uniformly efficient family of LCUs, if there exists a constant $m \in \mathbb{N}$ and:
\begin{itemize}
    \item uniformly efficient real coefficient $\gamma_{i,x}\in[-1,1]$, with $\sum_i |\gamma_{i,x}|\neq 0\;\forall x$.
    \item a set of uniformly efficient families of unitaries $\{\{U_{i,x}\}_x\}_{i\in[m]}$,
\end{itemize}
such that for all $x$
\[
A_x=\sum_{i\in[m]}\gamma_{i,x}U_{i,x}.
\]
\end{definition}

\begin{lemma}\label{lemma:block-implementable}
    For every uniformly efficient family of LCUs ($\{A_x\}_x$), where all $A_x$ are Hermitian, there exists a uniformly efficient family of block encodings of $\{A_x\}$ with uniformly efficient scale factors $t_x=1/\nu_x$, where $\nu_x=\sum_i|\gamma_{i,x}|$. In particular $|t_x|\geq 1/m$.
\end{lemma}
\begin{proof}
    The circuit will act on $k=\ceil{\log(m)}$ ancilla qubits, where $m$ is the maximal number of terms in the LCUs. We first introduce the preparation unitary $P_x$, which given $x$ implements
    \[
    P_x\ket{0} = \sum_{i\in[m]}\sqrt{\frac{|\gamma_{i,x}|}{\nu_x}}\ket{i-1},
    \]
    where $\nu_x=\sum_i|\gamma_{i,x}|$ and $\ket{i-1}$ denotes the computational-basis encoding of the integer $i-1$ on the ancilla register. This operation is efficient since there are only a constant number of terms in the sum. Next we define the controlled application of our unitaries, based on the ancilla register, let
    \[
    S_x\coloneqq \sum_{i\in[m]}\ket{i-1}\bra{i-1}\ot s_{i,x}U_{i,x}
    \]
    where $s_{i,x}=\frac{\gamma_{i,x}}{|\gamma_{i,x}|}$ (and $s_{i,x}=1$ if $\gamma_{i,x}=0$). Again this unitary is uniformly efficient in $x$, because the family of unitaries are uniformly efficient in $x$ and because there are only constantly many of them. The operation
    \[
    U_x=(P_x^\dag\ot\id)S_x(P_x\ot\id)
    \]
    then satisfies
    \[
    (\bra{0}\ot\id)U_x(\ket{0}\ot\id)=\sum_{i\in[m]}\frac{|\gamma_{i,x}|}{\nu_x}\,s_{i,x}U_{i,x}=\frac{1}{\nu_x}A_x,
    \]
    so $U_x$ is a block encoding of $A_x$ with scale factor $t_x=1/\nu_x$. Since $\gamma_{i,x}\in[-1,1]$ and there are only $m$ terms, we have $0<\nu_x\leq m$, hence $|t_x|\geq 1/m$. Both $P_x$ and $S_x$ are efficient procedures taking $x$ as input and implementing a specific quantum circuit, so $\{U_x\}_x$ is a uniformly efficient family of unitaries, which completes the proof.
\end{proof}

\begin{lemma}\label{lemma:QPT-measurable}
    Suppose we have a uniformly efficient family of block encodings for a family of (not necessarily binary) observables $\{A_x\}_x$ with uniformly efficient scale factors $|t_x|\geq1/c$ for some constant $c$ independent of $x$. Then there exists a uniformly efficient POVM family $\{M_{\beta,x}\}_{\beta,x}$ with outcomes $\beta\in\{-1/t_x, 1/t_x\}$, satisfying $|\beta|\leq c = O(1)$, such that for any state $\rho$,
    \[
    \sum_\beta\beta\Tr{M_{\beta,x}\rho}=\Tr{A_x\rho}.
    \]
\end{lemma}
\begin{proof}
    This is a generalization of \cite[Lemma 14]{Natarajan2023} to uniformly efficient families of observables and measurements. Since $\{A_x\}$ are observables, they are Hermitian, and we can replace the canonical phase-estimation approach by the simpler Hadamard test, which yields the same result for Hermitian operators.
    \par
    \medskip
    By the lemma assumptions we know that there exists a uniformly efficient family of block encodings of $\{A_x\}$, the elements of which we can denote by $U_x$. We then construct the binary measurement $M_{\beta,x}$ as follows: it performs the Hadamard test on $\rho$ with a controlled $U_x$ (and the block encoding ancillas initialized in $\ket{0}$), denoting the control qubit measurement outcome as $b$, the associated POVM outcome is then defined as $\beta=(-1)^b/t_x$. This POVM family is uniformly efficient in $x$, since both $t_x$ and $\{U_x\}$ are uniformly efficient and we only incur a multiplicative overhead when implementing the controlled operation. By the Hadamard test identity and the block-encoding relation $(\bra{0}\ot\id)U_x(\ket{0}\ot\id)=t_xA_x$, we have
    \[
    \Pr[b] = \frac{1}{2}\bigl(1+(-1)^b\Re\Tr{U_x(\ket{0}\bra{0}\ot\rho)}\bigr) = \frac{1}{2}\bigl(1+(-1)^b t_x\Tr{A_x\rho}\bigr),
    \]
    where the last step uses that $A_x$ is Hermitian (so the trace is real) and that $t_x$ is real. The two outcomes are $\beta_{\pm}=\pm 1/t_x$, and therefore
    \[
    \sum_\beta\beta\Tr{M_{\beta,x}\rho} = \frac{1}{t_x}\Pr[b=0]-\frac{1}{t_x}\Pr[b=1] = \Tr{A_x\rho}.
    \]
    The hypothesis $|t_x|\geq 1/c$ gives $|\beta|\leq c=O(1)$, which completes the proof.
\end{proof}

\begin{lemma}\label{lemma:efficient-operator-indistinguishability}
    Let $D_1$ and $D_2$ be two distributions over pairs $(x,z)$, where $x$ is any parameter and $z$ is a plaintext Alice question. Suppose $D_1$ and $D_2$ satisfy the following.
    \begin{itemize}
        \item $D_1$ and $D_2$ are efficiently sampleable, with identical marginals on $x$.
        \item For any $x$ in the support of their common marginal, the conditional distributions $D_1(\cdot | x)$ and $D_2(\cdot | x)$ over $z$ are also efficiently sampleable.
    \end{itemize} Then, for any efficient prover, and for any uniformly efficient family of LCUs $\{A_{x,\alpha}\}_{x,\alpha}$ (with $A_{x,\alpha}=A_{x,\alpha}^\dag$), indexed by the pair $(x,\alpha)$, there exists a cryptographically small function $\ea$ such that for all $\lambda\in\Nat$,
    \[
    \left|\E_{(x,z)\sim D_1}\sum_\alpha\Tr{A_{x,\alpha}\pea{z}}-\E_{(x,z)\sim D_2}\sum_\alpha\Tr{A_{x,\alpha}\pea{z}}\right|\leq\ea.\numberthis\label{eqn:eff-indistinguishability}
    \]
\end{lemma}
\begin{proof}
    Let $\{M_{\beta,x,\alpha}\}$ be the uniformly efficient family of POVMs guaranteed by \cref{lemma:QPT-measurable} applied to the uniformly efficient family of block encodings for $\{A_{x,\alpha}\}$ given by \cref{lemma:block-implementable}, both instantiated with the classical index $(x,\alpha)$. \Cref{lemma:block-implementable} yields scale factors $t_{x,\alpha}=1/\nu_{x,\alpha}$ with $\nu_{x,\alpha}\leq m$, hence $|t_{x,\alpha}|\geq 1/m$; the Hadamard-test outcomes therefore satisfy $|\beta|=\nu_{x,\alpha}\leq m=O(1)$ uniformly over $(x,\alpha)$, as required by \cref{lemma:povm-indistinguishability}. By construction of the Hadamard test, the left-hand side of \cref{eqn:eff-indistinguishability} equals
    \[
    \left|\E_{(x,z)\sim D_1}\sum_\alpha\sum_\beta\beta\Tr{M_{\beta,x,\alpha}\pea{z}}-\E_{(x,z)\sim D_2}\sum_\alpha\sum_\beta\beta\Tr{M_{\beta,x,\alpha}\pea{z}}\right|,
    \]
    which is cryptographically small by \cref{lemma:povm-indistinguishability}.
\end{proof}

\begin{corollary}\label{cor:efficient-operator-state-switching}
    Let $D_1$ and $D_2$ be two distributions over pairs $(x,z)$, where $x$ is any parameter and $z$ is a plaintext Alice question. Suppose $D_1$ and $D_2$ satisfy the following.
    \begin{itemize}
        \item $D_1$ and $D_2$ are efficiently sampleable, with identical marginals on $x$.
        \item For any $x$ in the support of their common marginal, the conditional distributions $D_1(\cdot | x)$ and $D_2(\cdot | x)$ over $z$ are also efficiently sampleable.
    \end{itemize} Then, for any efficient prover, and for any uniformly efficient family of LCUs $\{A_{x,\alpha}\}_{x,\alpha}$, there exists a cryptographically small function $\ea$ such that for all $\lambda\in\Nat$,
    \[
    \left|\E_{(x,z)\sim D_1}\sum_\alpha\norm{A_{x,\alpha}}_{\pea{z}}^2-\E_{(x,z)\sim D_2}\sum_\alpha\norm{A_{x,\alpha}}_{\pea{z}}^2\right|\leq\ea.\numberthis\label{eqn:norm-indistinguishability}
    \]
\end{corollary}
\begin{proof}
    \Cref{eqn:norm-indistinguishability} is equivalent to
    \[
    \left|\E_{(x,z)\sim D_1}\sum_\alpha\Tr{A_{x,\alpha}^\dag A_{x,\alpha}\pea{z}}-\E_{(x,z)\sim D_2}\sum_\alpha\Tr{A_{x,\alpha}^\dag A_{x,\alpha}\pea{z}}\right|\leq\ea.
    \]
    If we can show that the set of all $C_{x,\alpha} = A_{x,\alpha}^\dag A_{x,\alpha}$ is again a uniformly efficient family of LCUs (indexed by $(x,\alpha)$), we can apply \cref{lemma:efficient-operator-indistinguishability} to $\{C_{x,\alpha}\}$ which completes the proof. We have
    \[
    C_{x,\alpha}=A_{x,\alpha}^\dag A_{x,\alpha} = \sum_{i\in[m]}\sum_{j\in[m]}\gamma_{i,x,\alpha}\gamma_{j,x,\alpha} U^\dag_{i,x,\alpha}U_{j,x,\alpha},
    \]
    which is an LCU with $m^2$ terms, indexed by pairs $(i,j)\in[m]\times[m]$. Each coefficient $\gamma_{i,x,\alpha}\gamma_{j,x,\alpha}$ lies in $[-1,1]$, and since $\{\{U_{i,x,\alpha}\}_{x,\alpha}\}_{i\in[m]}$ is a set of uniformly efficient families of unitaries, their pairwise products (cross-family for the same parameter) again form uniformly efficient families of unitaries. This shows that $\{C_{x,\alpha}\}$ is again a uniformly efficient family of LCUs, so the corollary conclusion immediately follows from \cref{lemma:efficient-operator-indistinguishability} applied to $\{C_{x,\alpha}\}$.
\end{proof}

\Cref{cor:efficient-operator-state-switching} switches the state argument of an LCU $A_{x,\alpha}\in L(\H)$, an operator \emph{on} $\H$. Several of our arguments (e.g.\ \cref{lemma:cliff-delta-self-consistent,lemma:state-rounding-it}) instead need to switch the state argument of a norm $\norm{VU_{x,\alpha}-N_{x,\alpha}V}_\peq^2$, where $V:\H\to\H'$ is the (generally non-unitary) isometry of \cref{def:clifford-isometry} and $U_{x,\alpha},N_{x,\alpha}$ are unitaries on $\H,\H'$ respectively: here the object being switched is not literally an operator on $\H$, so \cref{cor:efficient-operator-state-switching} does not apply directly. The following isometry version supplies the missing step.

\begin{corollary}[Isometry version of \cref{cor:efficient-operator-state-switching}]\label{cor:efficient-isometry-state-switching}
    Let $V:\H\to\H'$ be a QPT-implementable isometry: there exist an efficiently-sized ancilla space $\mathrm{anc}$, an efficient unitary dilation $U_V\in U(\H\ot\mathrm{anc})$, and a fixed efficient embedding $\H'\hookrightarrow\H\ot\mathrm{anc}$, such that $U_V(\ket\psi\ot\ket0_{\mathrm{anc}})=V\ket\psi$ for all $\ket\psi\in\H$. (This holds for the isometry of \cref{def:clifford-isometry}, by the same argument as in the proof of \cref{cor:gh-efficient}: the quantum Fourier transform of \cref{subsec:QFT} and the controlled-$f$ built from the prover's operators are both efficient.) Let $D_1,D_2$ be as in \cref{cor:efficient-operator-state-switching}, let $\{U_{x,\alpha}\}_{x,\alpha}$ be a uniformly efficient family of unitaries on $\H$, and let $\{N_{x,\alpha}\}_{x,\alpha}$ be a uniformly efficient family of unitaries on $\H'$. Then, for any efficient prover, there exists a cryptographically small function $\ea$ such that for all $\lambda\in\Nat$,
    \[
    \left|\E_{(x,z)\sim D_1}\sum_\alpha\norm{VU_{x,\alpha}-N_{x,\alpha}V}^2_{\pea{z}}-\E_{(x,z)\sim D_2}\sum_\alpha\norm{VU_{x,\alpha}-N_{x,\alpha}V}^2_{\pea{z}}\right|\leq\ea.
    \]
\end{corollary}
\begin{proof}
 Write $B_{x,\alpha}\coloneqq V^\dag N_{x,\alpha}V\in L(\H)$. Since $V^\dag V=\id_\H$ and $U_{x,\alpha},N_{x,\alpha}$ are unitary, we can expand the squared norm to get
    \begin{align}
    \sum_\alpha \norm{VU_{x,\alpha}-N_{x,\alpha}V}^2_{\psi_\alpha^{\Enc(x)}} &= \sum_\alpha \Tr{(U_{x,\alpha}^\dag V^\dag - V^\dag N_{x,\alpha}^\dag)(VU_{x,\alpha}-N_{x,\alpha}V)\psi_\alpha^{\Enc(x)}} \nonumber \\
    &=2\underbrace{\sum_\alpha \Tr{\psi_\alpha^{\Enc(x)}}}_{=1} -\sum_\alpha \Tr{\bigl(U_{x,\alpha}^\dag B_{x,\alpha}+B_{x,\alpha}^\dag U_{x,\alpha}\bigr)\psi_\alpha^{\Enc(x)}}. \label{eq:vu-minus-nv}
    \end{align}
    Our main task is to show that the second term in the last line corresponds to an efficient measurement of the state, such that we can apply \cref{lemma:povm-indistinguishability} to relate its expectation value on the two distributions $D_1, D_2$.
    \par
    \medskip
    To do so, extend $N_{x,\alpha}$ to a unitary $\tilde N_{x,\alpha}\coloneqq N_{x,\alpha}\oplus\id$ on $\H\ot\mathrm{anc}$, acting as the identity on the orthogonal complement of the fixed embedding $\H'\hookrightarrow\H\ot\mathrm{anc}$; since this embedding is fixed and efficient, $\{\tilde N_{x,\alpha}\}$ is again uniformly efficient. Define
    \[
    \hat U_{x,\alpha}\coloneqq (U_{x,\alpha}^\dag\ot\id_{\mathrm{anc}})\,U_V^\dag\,\tilde N_{x,\alpha}\,U_V\;\in\;U(\H\ot\mathrm{anc}),
    \]
    which is efficient, since $U_{x,\alpha}, U_V$ and $\tilde N_{x,\alpha}$ are. Using $U_V(\ket\phi\ot\ket0_{\mathrm{anc}})=V\ket\phi\in\H'$ for all $\ket\phi\in\H$, and that $\tilde N_{x,\alpha}$ restricted to $\H'$ acts as $N_{x,\alpha}$, a direct computation gives, for all $\ket\phi,\ket\chi\in\H$,
    \[
    (\bra\phi\ot\bra0_{\mathrm{anc}})\,\hat U_{x,\alpha}\,(\ket\chi\ot\ket0_{\mathrm{anc}}) = \bra\phi\,U_{x,\alpha}^\dag V^\dag N_{x,\alpha}V\,\ket\chi = \bra\phi\, U_{x,\alpha}^\dag B_{x,\alpha}\,\ket\chi,
    \]
    so $\hat U_{x,\alpha}$ is an efficient, scale-$1$ block encoding of $U_{x,\alpha}^\dag B_{x,\alpha}\in L(\H)$, with ancilla register $\mathrm{anc}$. By \cref{lemma:QPT-measurable} (with $c=1$), this yields a uniformly efficient family of two-outcome ($\beta=\pm1$) Hadamard-test POVMs $\{M_{\beta,x,\alpha}\}$ with $\sum_\beta\beta\Tr{M_{\beta,x,\alpha}\rho}=\Tr{U_{x,\alpha}^\dag B_{x,\alpha}\,\rho}$ for every state $\rho$; since the Hadamard test measures the real part of this (complex, in general) quantity, the same family also satisfies \[ \sum_\beta\beta\Tr{M_{\beta,x,\alpha}\rho}=\Tr{\tfrac12\bigl(U_{x,\alpha}^\dag B_{x,\alpha}+B_{x,\alpha}^\dag U_{x,\alpha}\bigr)\rho}.\]

    Therefore, by \cref{eq:vu-minus-nv}, we have
    \[
    \sum_\alpha\norm{VU_{x,\alpha}-N_{x,\alpha}V}^2_{\pea{z}} = 2-\sum_\alpha\sum_\beta\beta\Tr{M_{\beta,x,\alpha}\pea{z}}.
    \]
    The first term does not depend on $z$, so applying \cref{lemma:povm-indistinguishability} to the uniformly efficient POVM family $\{M_{\beta,x,\alpha}\}$ bounds the difference of the second term between $D_1$ and $D_2$ by a single cryptographically small function, which completes the proof.
\end{proof}

\subsubsection{QFHE Overhead}\label{subsec:qfhe-overhead}
For our construction we will instantiate QFHE from Brakerski's scheme \cite{Brakerski2018} (which improves the original proposal by Mahadev \cite{Mahadev2020}). We choose a specific scheme to be able to determine the overall resource requirements, which includes all QFHE encryptions, homomorphic evaluations, etc. We believe that our construction is still efficient if we replace the QFHE scheme with the one from \cite{Gupte2024} but we did not go through the resource estimation in detail for that protocol. 
\par 
\medskip
\begin{table}[ht]
\centering
\renewcommand{\arraystretch}{1.25}
\setlength{\tabcolsep}{6pt}
\begin{tabular}{@{}l l l@{}}
\toprule
\textbf{Category} & \textbf{Resource} & \textbf{Asymptotic cost} \\
\midrule
\multicolumn{3}{@{}c}{\textit{Setting:} classical input of $c$ bits;\ security parameter $\lambda$;\ leveled scheme;} \\
\multicolumn{3}{@{}c}{Clifford + Toffoli circuit of size $n$ with Toffoli depth $L\leq O(\mathrm{polylog}(n))$ and width $q$.} \\
\midrule
\multirow{4}{*}{Communication}
 & Encrypted input \ ($V\!\to\!P$)            & $O(c\cdot\poly(\lambda))$ bits \\
 & Evaluation keys \ ($V\!\to\!P$)     & $O(L\cdot\poly(\lambda))$ bits \\
 & Encrypted output \ ($P\!\to\!V$)           & $O(q\cdot\poly(\lambda))$ bits \\
 & \textbf{Total}                              & $O\bigl((c+L+q)\cdot\mathrm{poly}(\lambda)\bigr)$ bits \\
\midrule
\multirow{3}{*}{Prover}
 & Quantum circuit size                        & $O(n\cdot\mathrm{poly}(\lambda))$ gates \\
 & Ancilla qubits per Toffoli (reusable)       & $\mathrm{poly}(\lambda)$ \\
 & Time complexity                              & $O(n\cdot\mathrm{poly}(\lambda))$\\
\midrule
Verifier
 & Time complexity                          & $O((c+q)\cdot\mathrm{poly}(\lambda))$  \\
\bottomrule
\end{tabular}
\caption{%
Resource requirements for Brakerski's QFHE protocol~\cite{Brakerski2018}, leveled instantiation under LWE with polynomial modulus and no circular
security assumption.}
\label{table:brakerski-overhead}
\end{table}

\cref{table:brakerski-overhead} tells us that the QFHE encryption only introduces a multiplicative overhead of $\poly(\lambda)$ in all resources. Unfortunately, this specific scheme requires the circuit, which will be evaluated on the encrypted input, to be compiled in the Clifford + Toffoli gate set. The circuit in question is the honest Alice strategy, which consists of a circuit taking in the verifier's question and Alice's quantum state as inputs and performing a specific measurement on Alice's state, depending on the question.
\par
\medskip
In our rigidity test (\cref{protocol:cliff-test} and subtests), there are only a few different measurement types which the honest prover must implement. The most important ones are:
\begin{itemize}
    \item \textbf{Large-answer measurements:} Every qubit with index $i\in[n]$ gets measured in the basis $\tilde W_i$, with $\tilde W\in\{X,Y,Z,F,G\}^n$. Pure-basis measurements are a special case of this where $\tilde W\in\{X^n,Y^n,Z^n,F^n,G^n\}$. This measurement type appears in protocols \ref{protocol:pure-vs-mixed-test}, \ref{protocol:small-large-answer-test}, \ref{protocol:product-test} and \ref{protocol:cliff-test}.
    \item \textbf{Small-answer measurements:} The prover can be asked to perform up to three simultaneous (compatible) binary measurements. For example, in the commutation test, observables $A$ and $B$ are simultaneously measured. This measurement type appears in protocols \ref{protocol:commutation-test}, \ref{protocol:anti-commutation-test}, \ref{protocol:conjugation-test-mod}, \ref{protocol:small-large-answer-test} and \ref{protocol:product-test}.
    \item \textbf{Bell-basis measurements:} In the Clifford test (\cref{protocol:cliff-test}) the prover is asked to measure all qubits in one of the two brick-wall pairings in the Bell basis.
\end{itemize}
One can come up with a general encoding of the verifier question, which specifies the measurement type and some type-specific parameters (such as the basis or pairing choice). The honest prover strategy---acting on an encoded question and the prover's state---is highly parallelizable (given $O(g)$ ancillas) since we can ensure that different measurement blocks act on distinct qubits (i.e. qubits are separately handled), that there is only a constant depth of intersecting operations within one measurement block and
that the rest of the operations is non-intersecting. When using the natural gate set (all the operations needed for the strategy are allowed), the circuit implementing Alice's honest strategy has constant depth.
\par 
\medskip
To allow for homomorphic evaluation, the circuit needs to be compiled to the Clifford + Toffoli gate set. By the Solovay--Kitaev construction for approximating arbitrary unitaries, this will increase the depth by at most a factor of $O(\mathrm{polylog}(g))$ (see \cite{Kuperberg2023} for an efficient algorithm), since there are $g$ gates in the circuit and approximating any gate up to precision $\eps$ introduces an overhead of $\mathrm{polylog}(1/\eps)$. If we want to achieve a constant overall approximation precision, this $\eps$ will depend on $g$, which justifies the claimed blow-up. The exact circuit depth can be determined a priori and it is small ($O(\mathrm{polylog}(g))$ which is $O(\poly(\lambda))$ under both hardness assumptions), which allows us to use a leveled QFHE scheme, which does not require the circular security assumption, but where the verifier needs to send as many evaluation keys as there are levels (which are defined as a sequence of Cliffords followed by one layer of non-intersecting Toffolis) in the circuit.
\par 
\medskip
Not every gate in the honest strategy needs to be approximated, the only non-Clifford, non-Toffoli gates that appear are $CH$, $CP$ and $CT$, where the first two can be exactly implemented (see \cite{Bera2008} for $CH$ and use an ancilla with phase kick-back for $CP$) and only the latter needs to be approximated.
\par 
\medskip
If we insert a circuit size of $O(g\poly(\log g))$ (which also trivially bounds the width\footnote{True width is $O(g)+\poly(\lambda)$ since the Toffoli ancillas are reusable.}) and classical input size of $O(g)$ into \cref{table:brakerski-overhead}, we see that the overall resource cost for the QFHE step is upper-bounded by $O(\poly(\lambda)\mathrm{polylog}(g)g)$. Under polynomial hardness of LWE one must take $\lambda=g^{\Omega(1)}$ to accommodate honest computations, so the overhead becomes $O(g^{1+\eps})$ for every $\eps>0$. Under sub-exponential hardness, one can take $\lambda=(\log g)^{\Theta(1/\delta)}$ and write the cost as $\widetilde{O}(g)$.

\subsection{Approximate group representations}\label{sec:stability-theorem}
A stability theorem is a result of the following form: if a collection of objects approximately satisfies a certain set of properties, then it can be ``rounded'' to another collection of objects that satisfy these properties exactly.
We will be concerned with the stability of group representations, i.e.\ results that show that if a collection of operators approximately satisfies the relations of a group representation, then there exists another exact representation that is, in a suitable sense, close to the approximate representation.
Such stability results have been shown in generality by Gowers and Hatami~\cite{Gowers2017} as well as de Chiffre~et al.~\cite{DeChiffre2019}.
\par
\medskip
For our results, we need two modifications of these existing stability theorems: first, we show that group representations are stable even when averages are taken over arbitrary measures on the underlying group.
This also provides a much simpler proof of a weaker version of the main result from~\cite{Gowers2017} using only elementary facts from quantum information theory, as shown in \cref{theorem:GH}. Second, we introduce an explicit expression for the isometry, which is defined in \cref{def:gh-isometry}.

\begin{definition}[GH Isometry]\label{def:gh-isometry}
    Let $G$ be a finite group, $U_{\mathrm{QFT}}$ be the quantum Fourier transform over that group and $f : G \to U(\H)$. Then the Gowers--Hatami (GH) isometry is the isometry $V : \H \to \H'$ defined as:
    \[
    V \coloneqq \frac{1}{|G|} \sum_{u, t\in G} U_{\mathrm{QFT}}\ket{t}\otimes f(u t^{-1})\otimes\ket{u}.
    \] 
\end{definition}
\begin{theorem}[Strengthening of Theorem 3.1 in \cite{Metger2024}]\label{theorem:GH}
Let $G$ be a finite group and $f : G \to U(\H)$. Then for all $g\in G$ and all $\rho\in\mathrm{Pos}(\H)$, it holds that
$$ \frac{1}{|G|}\sum_{h\in G} \| f(h)f(g) - f(hg) \|^2_{\rho} = \| Vf(g) - \pi(g)V \|^2_{\rho}. $$
where $V : \H \to \H'$ is the GH isometry from \cref{def:gh-isometry} and $\pi : G \to U(\H')$ is the direct sum of all irreducible representations $\rho_\mu: G \to U(\H_\mu)$ of $G$ 
$$
    \pi(g)\deq\bigoplus_\mu \rho_\mu(g)\ot\id.
$$
\end{theorem}
\begin{proof} 
This proof is in large part inspired by \cite[Theorem 3.1]{Metger2024}, we made the modification of explicitly writing out an isometry to obtain a tighter anti-symmetric distance bound and simplify the proof. In this form the proof consists almost entirely of elementary facts from linear algebra. By embedding the function $f$ inside the isometry, we are essentially just cleverly re-writing the expression to explicitly pull out an exact unitary representation of the group $G$.
\par 
\medskip
Let $\H_G$ be a $|G|$-dimensional Hilbert space, where every basis vector is associated to a group element and define a representation $\pi' : G \to U(\H_G)$ by 
\begin{equation}\label{eqn:exact-repr-definition}
    \pi'(g) = \sum_{h\in G} \ket{gh}\bra{h}.
\end{equation} 

It is clear that $\pi'$ is a unitary representation of $G$ (in fact, $\pi'$ is the left regular representation, written in Dirac notation). Recall from \cref{subsec:QFT} that the quantum Fourier transform block-diagonalizes the left regular representation, and thus
\[
    U_{\mathrm{QFT}}\,\pi'(g)\,U_{\mathrm{QFT}}^\dag = \bigoplus_\mu \rho_\mu(g)\ot\id_{d_\mu},
\]
where $\rho_\mu: G \to U(\H_\mu)$ is an irreducible representation of $G$ with dimension $d_\mu$. As choice for the isometry, we will use the GH isometry from \cref{def:gh-isometry}. The specific form is based on an inspection of the proof of \cite[Theorem 3.1]{Metger2024}. Specifically, after determining the Kraus operators for the channel corresponding to $\phi$ in said proof, we guessed the simplified form of \cref{def:gh-isometry} which yields a stronger bound and a shorter proof. Thus we have an isometry $V: \H\to\H_G\otimes\H\otimes\H_G$ given by
\[
V \coloneqq \frac{1}{|G|} \sum_{u, t\in G} U_{\mathrm{QFT}}\ket{t}\otimes f(u t^{-1})\otimes\ket{u},
\] 
note that $V$ is not necessarily efficient, since it requires computing the quantum Fourier transform of a general group, preparing the uniform superposition over all group element basis states and implementing the inverse and multiplication group operations. It is indeed an isometry since
\begin{align*}
V^\dagger V &= \frac{1}{|G|^2} \Bigg(\sum_{u', t'} \bra{t'}U_{\mathrm{QFT}}^\dag\otimes f(u' (t')^{-1})^\dagger \otimes \bra{u'}\Bigg) \left(\sum_{u, t} U_{\mathrm{QFT}}\ket{t}\otimes f(u t^{-1}) \otimes \ket{u}\right) \\
&= \frac{1}{|G|^2} \sum_{u, t} \overbrace{f(u t^{-1})^\dagger f(u t^{-1})}^{\id}\\
&= \id,
\end{align*}
where the last step follows from the unitarity of $f$.
The following expression will be useful later
\begin{align*}
(U_{\mathrm{QFT}}\pi'(g)U_{\mathrm{QFT}}^\dag\otimes\id \otimes \id_G) V &= \frac{1}{|G|} \sum_{u, t} U_{\mathrm{QFT}}\pi'(g)\ket{t}\otimes f(u t^{-1}) \otimes \ket{u} \\
&= \frac{1}{|G|} \sum_{u, t} U_{\mathrm{QFT}}\ket{gt}\otimes f(u t^{-1}) \otimes \ket{u} \\
&= \frac{1}{|G|} \sum_{u, t} U_{\mathrm{QFT}}\ket{t} \otimes f(u t^{-1}g) \otimes \ket{u},
\end{align*}
where we used the definition of $\pi'(g)$ in the second equality and relabeled $tg\mapsto t$ for the last equality. With this and an explicit expression for the isometry we can make the following observation, for arbitrary $g\in G$
\begin{align*}
&\left\| \left( (U_{\mathrm{QFT}}\pi'(g)U_{\mathrm{QFT}}^\dag\otimes\id \otimes \id_G) V - V f(g) \right) \ket{\psi} \right\|^2 \\
&= \frac{1}{|G|^2}  \Big\| \sum_{u, t}  U_{\mathrm{QFT}}\ket{t}\otimes (f(u t^{-1}g) - f(ut^{-1})f(g))\ket{\psi} \otimes \ket{u} \Big\|^2 \\
&= \frac{1}{|G|^2} \sum_{u, t} \left\|(f(u t^{-1}g) - f(ut^{-1})f(g))\ket{\psi} \right\|^2 \\
&= \frac{1}{|G|} \sum_{h\in G} \| (f(h g) - f(h) f(g)) \ket{\psi} \|^2,
\end{align*}
where we used left unitary invariance, the orthogonality of the group element basis and the following elementary fact for the last equality
\[
\sum_{u,t\in G}f(ut^{-1}) = |G|\sum_{h\in G}f(h).
\]
This fact follows from the closure of the group under the group operations and the fact that every group element has a unique inverse. Since any $\rho\in\mathrm{Pos}(\H)$ is also Hermitian, it can be written in its orthonormal eigenbasis as
\[
\rho = \sum_i \lambda_i\ket{\psi_i}\bra{\psi_i},
\]
with real eigenvalues $\lambda_i$. By the linearity in the state of the state-dependent (semi) norm and its equivalence to the Euclidean norm for pure states, we have 
\begin{align*}
\left\| \pi(g) V - V f(g) \right\|^2_\rho &= \sum_i \lambda_i \left\| \pi(g) V - V f(g) \right\|^2_{\psi_i}\\
&= \sum_i \frac{\lambda_i}{|G|}\sum_{h\in G}  \| f(h g) - f(h) f(g)\|^2_{\psi_i}\\
&= \frac{1}{|G|}\sum_{h\in G} \| f(h g) - f(h) f(g)\|^2_\rho
\end{align*}
with $\pi: G\to U(\H')$, $\pi(g)\coloneqq U_{\mathrm{QFT}}\pi'(g)U_{\mathrm{QFT}}^\dag\otimes\id \otimes \id_G$, this shows the theorem claim and concludes the proof.
\end{proof}

\begin{remark}
    If $f$ is \emph{exactly} right-multiplicative over some $\mu$, i.e.\ if $f(h)f(g) = f(hg)$ for all $g \in \mathrm{supp}(\mu)$ and for all $h \in G$, then $Vf(g) = \pi(g) V$ for all $g \in \mathrm{supp}(\mu)$. 
\end{remark}

\begin{corollary}\label{cor:gh-efficient}
Let $G$ be a finite group and $f : G \to U(\H)$. If the quantum Fourier transform over $G$, the group operations (including inversion) and a controlled-$f$ are QPT-implementable in $\log|G|$ and $\lambda$, then there exists a finite-dimensional Hilbert space $\H'$ such that the GH isometry $V : \H \to \H'$ from \cref{def:gh-isometry} is QPT-implementable, and it holds for all $g\in G$ and all $\rho\in\mathrm{Pos}(\H)$ that
$$ \frac{1}{|G|}\sum_{h\in G} \| f(h)f(g) - f(hg) \|^2_{\rho} = \| Vf(g) - \pi(g)V \|^2_{\rho}, $$
where $\pi : G \to U(\H')$ is the direct sum of all irreducible representations $\rho_\mu: G \to U(\H_\mu)$ of $G$
$$
    \pi(g)\deq\bigoplus_\mu \rho_\mu(g)\ot\id.
$$
\end{corollary}
\begin{proof}
    This follows from \cref{theorem:GH}, where the additional requirements ensure that the isometry $V$ (from \cref{def:gh-isometry}) can be efficiently implemented.
\end{proof}

\section{Rigidity}\label{chapter:rigidity}

\subsection{Protocols}\label{sec:protocols}
The Clifford test (\cref{protocol:cliff-test}) is built hierarchically from a small set of reusable subtests. \Cref{fig:cliff-structure} gives an overview of this structure: each box is a (sub)test, and an arrow points from a test to every subtest it invokes. The remainder of this section defines each of these protocols, working bottom-up from the primitive tests to the full Clifford test.
\par

\paragraph{Convention (automatic consistency test).} The consistency test (\cref{protocol:consistency-test}) never appears explicitly in the protocol boxes below. Instead, we wrap only \emph{named} tests which themselves send questions to the provers, i.e. the \emph{leaf} tests $\textprotocol{com}$, $\textprotocol{ac}$, $\textprotocol{slc}$, $\textprotocol{prel}$, $\textprotocol{mbt}$ and $\textprotocol{conj}$. Whenever such a leaf $T$ is invoked as a subtest, it is executed as $\textprotocol{con}(T)$: with probability $1/2$ the verifier plays $T$ and with probability $1/2$ it sends a question sampled from either the Alice-marginal or the Bob-marginal of $T$ (each with equal probability) to both players and checks that the answers agree. Named tests that only dispatch to other named tests---$\textprotocol{crel}$, $\textprotocol{conj-cliff}$ and $\textprotocol{cliff-group}$---are executed unwrapped, as are the top-level tests (the Clifford test (\cref{protocol:cliff-test}) and the verification protocol (\cref{protocol:verification})), the subtests specified inline within them, and the invocation of the Clifford test inside the verification protocol. This convention has two consequences that we use throughout. For completeness: the honest strategy (\cref{table:honest-prover}) passes every consistency check with probability $1$, because Alice honestly measures the transposed observables on her halves of the shared EPR pairs, so her outcomes agree with Bob's on identical questions; since moreover every wrapped leaf has perfect honest completeness, the wrapping leaves all quoted completeness values---in particular $\omega^*=\tfrac16(5+\cos^2\tfrac{\pi}{8})$ for the Clifford test, whose only imperfect item (the inline CHSH subtest) is unwrapped---unchanged. For soundness: a prover that succeeds with probability $1-\eps$ in an invocation of a leaf test $T$ (which, by the convention, is an invocation of $\textprotocol{con}(T)$) succeeds with probability $1-2\eps$ in $T$ itself and, by \cref{lemma:compiled-consistency}, every question distribution appearing in $T$ is $\eps$-self-consistent. Success $1-\eps$ in an unwrapped dispatcher is then success $1-O(\eps)$ in each of its constantly many wrapped leaves. All constant-factor bookkeeping arising from this is absorbed into the $O(\cdot)$ notation. Whenever a soundness lemma below assumes success probability $1-\eps$ in a leaf test, it is to be read as success probability $1-\eps$ in its consistency-wrapped execution.

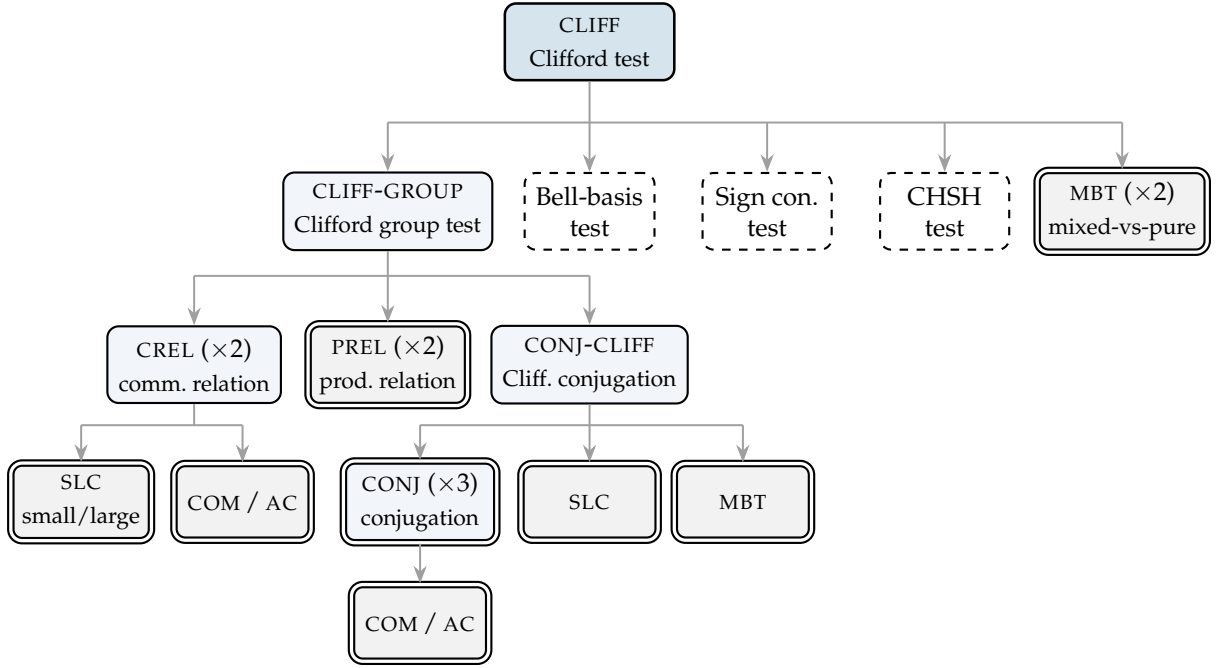
\begin{figure}[H]
    \centering
    \resizebox{\textwidth}{!}{%
    \begin{tikzpicture}[
        font=\footnotesize,
        root/.style={draw, line width=0.9pt, rounded corners, align=center,
                     fill=SteelBlue!22, minimum width=2.1cm, minimum height=0.95cm, inner sep=4pt},
        disp/.style={draw, rounded corners, align=center,
                     fill=SteelBlue!7, minimum width=1.9cm, minimum height=0.95cm, inner sep=4pt},
        proto/.style={draw, double, double distance=1.15pt, rounded corners, align=center,
                      fill=SteelBlue!7, minimum width=1.9cm, minimum height=0.95cm, inner sep=4pt},
        prim/.style={draw, double, double distance=1.15pt, rounded corners, align=center,
                     fill=black!5, minimum width=1.7cm, minimum height=0.95cm, inner sep=4pt},
        inl/.style={draw, dashed, rounded corners, align=center,
                    fill=white, minimum width=1.6cm, minimum height=0.95cm, inner sep=4pt},
        bus/.style={line width=0.7pt, gray!75},
        arr/.style={-{Stealth[length=2.2mm]}, line width=0.7pt, gray!75},
    ]
    \node[root] (cliff) at (6.5,0) {\textprotocol{cliff}\\[1pt]{\scriptsize Clifford test}};
    \node[disp] (cgroup) at (4,-2.1)  {\textprotocol{cliff-group}\\[1pt]{\scriptsize Clifford group test}};
    \node[inl]   (bell)   at (6.5,-2.1)  {Bell-basis\\test};
    \node[inl]   (cons)   at (8.7,-2.1) {Sign con.\\test};
    \node[inl]   (chsh)   at (10.9,-2.1) {CHSH\\test};
    \node[prim]  (mbt0)   at (13.1,-2.1) {\textprotocol{mbt} ($\times2$)\\[1pt]{\scriptsize mixed-vs-pure}};
    \node[disp]  (crel)   at (1.6,-4)  {\textprotocol{crel} ($\times2$)\\[1pt]{\scriptsize comm.\ relation}};
    \node[prim]  (prel)   at (4,-4)  {\textprotocol{prel} ($\times2$)\\[1pt]{\scriptsize prod.\ relation}};
    \node[disp]  (ccliff) at (6.5,-4)  {\textprotocol{conj-cliff}\\[1pt]{\scriptsize Cliff.\ conjugation}};
    \node[prim]  (slcA)   at (0.2,-5.7) {\textprotocol{slc}\\[1pt]{\scriptsize small/large}};
    \node[prim]  (comacA) at (2.2,-5.7)  {\textprotocol{com}\,/\,\textprotocol{ac}};
    \node[proto] (conj)   at (4.4,-5.7)  {\textprotocol{conj} ($\times3$)\\[1pt]{\scriptsize conjugation}};
    \node[prim]  (slcB)   at (6.5,-5.7)  {\textprotocol{slc}};
    \node[prim]  (mbtB)   at (8.4,-5.7)  {\textprotocol{mbt}};
    \node[prim]  (comacB) at (4.4,-7.2)  {\textprotocol{com}\,/\,\textprotocol{ac}};

    \draw[bus] (cliff.south) -- (6.5,-1.0);
    \draw[bus] (4,-1.0) -- (13.1,-1.0);
    \foreach \c in {cgroup,bell,cons,chsh,mbt0} {\draw[arr] (\c.north |- 0,-1.0) -- (\c.north);}
    \draw[bus] (cgroup.south) -- (4,-2.9);
    \draw[bus] (1.6,-2.9) -- (6.5,-2.9);
    \foreach \c in {crel,prel,ccliff} {\draw[arr] (\c.north |- 0,-2.9) -- (\c.north);}
    \draw[bus] (crel.south) -- (1.6,-4.75);
    \draw[bus] (0.2,-4.75) -- (2.2,-4.75);
    \foreach \c in {slcA,comacA} {\draw[arr] (\c.north |- 0,-4.75) -- (\c.north);}
    \draw[bus] (ccliff.south) -- (6.5,-4.75);
    \draw[bus] (4.4,-4.75) -- (8.4,-4.75);
    \foreach \c in {conj,slcB,mbtB} {\draw[arr] (\c.north |- 0,-4.75) -- (\c.north);}
    \draw[arr] (conj.south) -- (comacB.north);
    \end{tikzpicture}%
    }
    \caption{Protocol structure of the Clifford test. An arrow points from each test to the subtests it invokes. Solid boxes are protocols defined in this section; the darkest is the top-level Clifford test (\cref{protocol:cliff-test}), gray boxes are the primitive tests, and dashed boxes are subtests specified inline within \cref{protocol:cliff-test}. Double outlines mark leaf tests that are executed as $\textprotocol{con}(T)$ (cf.\ the consistency convention above); named dispatchers ($\textprotocol{cliff-group}$, $\textprotocol{crel}$, $\textprotocol{conj-cliff}$) have a single outline and are unwrapped, as are the Clifford test itself and the inline subtests. A multiplicity ``$\times k$'' indicates that the protocol is invoked $k$ times (with different operator sets); e.g.\ \textprotocol{crel} and \textprotocol{prel} are each run for the Pauli and the Clifford generators.}
    \label{fig:cliff-structure}
\end{figure}

\begin{protocolbox}{Consistency test}\label{protocol:consistency-test}
\textbf{Notation:} $\textprotocol{con}(T)$\\
\textbf{Input:} any two-player test $T$, with question set $\Qa\times\Qb$ and question distribution $P_T(X,Y)$ for $(X,Y)\in\Qa\times\Qb$.
\par 
\medskip
Calculate $P_{T,A}(\cdot) = \sum_{y\in\Qb} P_T(\cdot,y)$ and $P_{T,B}(\cdot) = \sum_{x\in\Qa} P_T(x,\cdot)$, i.e.\ the Alice and Bob question marginals, and draw $b,s\draw \{0,1\}$.

\begin{enumerate}
        \item If $b = 0$, execute the game $T$.
        \item If $b = 1$, sample $W\sim P_{T,A}$ if $s=0$ and $W\sim P_{T,B}$ if $s=1$, and send this question to both players. Receive answers $a,b\in\{0,1\}^n$ from both provers, respectively and accept if and only if $a=b$. Here $n$ is equal to the number of answer bits expected for the label $W$ in game $T$.
\end{enumerate}

\end{protocolbox}

\begin{protocolbox}{Commutation test}\label{protocol:commutation-test}
\textbf{Notation:} $\textprotocol{com}(A,B)$\\
\textbf{Input:} Two questions $A$ and $B$, corresponding to some Bob observables.

\begin{enumerate}
    \item Send $W=(A,B)$ to Alice and receive answer $a=(a_1,a_2)\in\{0,1\}^2$.
    \item Draw $b\draw[2]$, send $W_b$ to Bob and receive answer $d\in\{0,1\}$. Accept if and only if $a_b = d$, where $W_1=A$ and $W_2=B$.
\end{enumerate}
\end{protocolbox}

\begin{protocolbox}{Anti-commutation test}\label{protocol:anti-commutation-test}
\textbf{Notation:} $\textprotocol{ac}(A,B)$\\
\textbf{Input:} Two questions $A$ and $B$, corresponding to some Bob observables. 
\par 
\medskip
\textit{In this protocol, the verifier plays a version of the Mermin--Peres Magic Square game \cite{Mermin1990,Peres1990} with the provers, in which Alice is asked to measure three observables forming a row or column of the square, and Bob is asked to measure one observable from a single cell of the square. All labels except for cells 2 and 4 should not appear in any other game, thus if the input labels to the anti-commutation game are $X$ and $Y$, the cells will have labels $C_{IX}$, $C_{IY}$, ....}
\[
\begin{tikzpicture}[x=1.1cm,y=1.1cm,
    cellnum/.style={anchor=north west, inner sep=1.5pt, font=\tiny\bfseries},
    cellbody/.style={anchor=center, inner sep=0pt}]
\foreach \i in {0,1,2,3} {
    \draw (\i,0) -- (\i,3);
    \draw (0,\i) -- (3,\i);
}
\node[cellbody] at (0.5,2.5) {$\Gamma^1_{AB}$}; \node[cellnum] at (0,3) {1};
\node[cellbody] at (1.5,2.5) {$\boldsymbol{A}$};  \node[cellnum] at (1,3) {2};
\node[cellbody] at (2.5,2.5) {$\Gamma^3_{AB}$}; \node[cellnum] at (2,3) {3};
\node[cellbody] at (0.5,1.5) {$\boldsymbol{B}$};  \node[cellnum] at (0,2) {4};
\node[cellbody] at (1.5,1.5) {$\Gamma^5_{AB}$}; \node[cellnum] at (1,2) {5};
\node[cellbody] at (2.5,1.5) {$\Gamma^6_{AB}$}; \node[cellnum] at (2,2) {6};
\node[cellbody] at (0.5,0.5) {$\Gamma^7_{AB}$}; \node[cellnum] at (0,1) {7};
\node[cellbody] at (1.5,0.5) {$\Gamma^8_{AB}$}; \node[cellnum] at (1,1) {8};
\node[cellbody] at (2.5,0.5) {$\Gamma^9_{AB}$}; \node[cellnum] at (2,1) {9};
\end{tikzpicture}
\]
\begin{enumerate}
    \item Choose a cell index $j\in[9]$ uniformly at random; choose the row or the column which contains cell $j$ uniformly at random, on the $3\times 3$ grid. Denote the 3 cells in this row or column as $(i_1, i_2, i_3)$ (by construction one of these will be equal to $j$).
    \item Send the triple of labels in the chosen row or column to Alice and receive three answer bits $a=(a_1, a_2, a_3)\in\{0,1\}^3$. For example if $(i_1, i_2, i_3) = (1,2,3)$ then Alice receives $(\Gamma^1_{AB},A,\Gamma^3_{AB})$.
    \item Send the label corresponding to the cell index $j$ to Bob and receive one answer bit $b\in\{0,1\}$. For example, if $j=3$, Bob receives $\Gamma^3_{AB}$.
    \item Accept if $a_l = b$ for $l$ s.t. $i_l = j$ and
    \[
    \begin{cases}
        a_1\oplus a_2\oplus a_3 = 0 & \text{If }(i_1, i_2, i_3)\neq(3,6,9)\\
        a_1\oplus a_2\oplus a_3 = 1 & \text{If }(i_1, i_2, i_3)=(3,6,9)
    \end{cases}
    \]
\end{enumerate}
\end{protocolbox}

\begin{protocolbox}{Conjugation test}\label{protocol:conjugation-test-mod}
\textbf{Notation:} $\textprotocol{conj}(A,B,R)$\\
\textbf{Input:} Three questions $A$, $B$ and $R$, corresponding to equally named Bob observables.
\par 
\medskip
Execute each of the following tests with equal probability (1/3 each).

\begin{enumerate}
    \item \textbf{(Anti)commutation test:} With uniform probability (1/5 each), execute one of the following tests: $\textprotocol{com}(X_R, C_{AB})$, $\textprotocol{ac}(X_C, Z_C)$, $\textprotocol{com}(A, X_C)$, $\textprotocol{com}(B, X_C)$ or $\textprotocol{com}(R, Z_C)$.
    \item \textbf{$C_{AB}$ characterization test:} Ask Alice to measure $A, B, C_{AB}$ or $Z_C$ (with probability 1/4 each) returning answer $a\in\{0,1\}$, and Bob to measure $(A,Z_C)$ or $(B, Z_C)$ (with probability 1/2 each), returning $b=(b_1,b_2)\in\{0,1\}^2$. Reject if either:
    \begin{itemize}
        \item Alice was asked $C_{AB}$, Bob was asked $(A, Z_C)$, $a\neq b_1$ and $b_2=0$;
        \item Alice was asked $C_{AB}$, Bob was asked $(B, Z_C)$, $a\neq b_1$ and $b_2=1$;
        \item Alice was asked $A$, Bob was asked $(A,Z_C)$ and $a\neq b_1$;
        \item Alice was asked $B$, Bob was asked $(B,Z_C)$ and $a\neq b_1$;
        \item Alice was asked $Z_C$ and $a\neq b_2$.
    \end{itemize}
    \item \textbf{$X_R$ characterization test:} Ask Alice to measure $R, X_C$ or $X_R$ (with probability 1/3 each) returning answer $a\in\zo$ and Bob to measure $(R, X_C)$ returning $b=(b_1,b_2)\in\zo^2$. Reject if either:
    \begin{itemize}
        \item Alice was asked $R$ and $a\neq b_1$;
        \item Alice was asked $X_C$ and $a\neq b_2$;
        \item Alice was asked $X_R$ and $a\neq b_1\oplus b_2$.
    \end{itemize}
\end{enumerate}

\end{protocolbox}

\begin{protocolbox}{Mixed-versus-pure basis test} \label{protocol:pure-vs-mixed-test}

\textbf{Notation:} $\textprotocol{mbt}(\Sigma, n, \mu)$\\
\textbf{Input:} A question alphabet $\Sigma$, expected answer length $n$ and a distribution $\mu$ over $\Sigma^n$.
\begin{enumerate}
    \item Sample $W\draw\Sigma$, send it to Alice and receive answer $a\in\zo^n$.
    \item Sample $\tilde W\draw{\mu}\Sigma^n$, send it to Bob and receive answer $b\in\zo^n$. Accept if and only if $a_i=b_i$ for all $i: \tilde W_i=W$.
\end{enumerate}
\end{protocolbox}

\begin{protocolbox}{Small/large answer consistency test} \label{protocol:small-large-answer-test}

\textbf{Notation:} $\textprotocol{slc}(\Theta, n, \mu)$\\
\textbf{Input:} A question set $\Theta$, expected answer length $n$ and bit-string distribution $\mu$.

\begin{enumerate}
    \item Sample $a\draw{\mu}\zo^n$ and $W\draw\Theta$. Send question $(W,a)$ to Alice, receiving answer $x\in\{0,1\}$ and send question $W$ to Bob, obtaining answer $u\in\{0,1\}^n$. Accept iff $u\cdot a=x$.
\end{enumerate}
\end{protocolbox}

\begin{protocolbox}{Commutation relation test} \label{protocol:crel-test}
\textbf{Notation:} $\textprotocol{crel}(\Sigma, n)$\\
\textbf{Input:} A question alphabet $\Sigma$ and expected answer length $n$.
\par 
\medskip
Sample $a,b\draw\zo^n$, $W_1\draw \Sigma$ and $W_2\draw\Sigma\setminus \{W_1\}$. Execute each of the following tests with equal probability.
\begin{enumerate}
    \item \textbf{Small/large answer consistency test:} Execute $\textprotocol{slc}(\Sigma, n,U_{2^n})$.
    \item \textbf{Small answer (anti)commutation test:}
\begin{enumerate}
\item If $a \cdot b = 0$: Execute $\textprotocol{com}((W_1,a),(W_2,b))$.
\item If $a \cdot b = 1$: Execute $\textprotocol{ac}((W_1,a),(W_2,b))$.
\end{enumerate}
\end{enumerate}
\end{protocolbox}

\begin{protocolbox}{Product relation test} \label{protocol:product-test}

\textbf{Notation:} $\textprotocol{prel}(Z,X,Y, n)$\\
\textbf{Input:} Three question symbols $Z$, $X$, $Y$ and expected answer length $n$.
\par 
\medskip
Sample $\alpha\draw\bits^{n-1}$, $y \draw [2]$ and compute the $n$-bit string $a \coloneqq \alpha \: \| \: (\alpha \cdot \alpha)$. Note that $a$ always has even Hamming weight. Let $W=(Z,X)$, so $W_1=Z$ and $W_2=X$.
\par 
\medskip
Execute each of the following tests with equal probability.

\begin{enumerate}
\item Send question $((Z,a),(X,a))$ to Alice, receive outcome $(x_1, x_2) \in \{0,1\}^2$, and send question $Y$ to Bob, obtaining an answer $u\in\{0,1\}^n$. Accept if and only if $u \cdot a  \oplus |a|/2 = x_1 \oplus x_2\pmod 2$.
\item Send question $((Z,a),(X,a))$ to Alice, receive outcome $(x_1, x_2) \in \{0,1\}^2$ and send question $W_y$ to Bob, obtaining an answer $u \in \{0,1\}^n$. Accept if and only if  $x_y= u \cdot a $.
\end{enumerate}

\end{protocolbox} 

\begin{protocolbox}{Clifford conjugation test} \label{protocol:conj-cliff-test}

\textbf{Notation:} $\textprotocol{conj-cliff}(X,Y,G,F, n)$\\
\textbf{Input:} Four question symbols $X$, $Y$, $G$, $F$ and expected answer length $n$.
\par 
\medskip
Sample $a,b\draw\zo^{n}$ and let $A$ be the unique element in $\{X,Y\}^n$, that has $X$'s at the positions $i\in[n]$ where $a_i=0$ and $Y$'s at the positions where $a_i=1$. Note its implicit dependence on $a$.
\par 
\medskip
Execute each of the following tests with equal probability.
\begin{enumerate}
    \item \textbf{Conjugation test:} Execute each of the following tests with equal probability
    \begin{enumerate}
        \item $\textprotocol{conj}((G,a),(F,a),(Y,a))$.
        \item $\textprotocol{conj}((X,a),(Y,a),(G,a))$.
        \item $\textprotocol{conj}((X,b),(A,b),(G,a))$.
    \end{enumerate}
    \item \textbf{Small/large answer consistency test:} Execute $\textprotocol{slc}(\{A,X,Y,G,F\},n, U_{2^n})$.
    \item \textbf{Mixed-versus-pure basis test:} Execute $\textprotocol{mbt}(\{X,Y\},n, U_{2^n})$.
\end{enumerate}

\end{protocolbox}

\begin{protocolbox}{Clifford group test} \label{protocol:cliff-group}
\textbf{Notation:} $\textprotocol{cliff-group}(X,Y,Z,F,G,n)$\\
\textbf{Input:} Five large-answer question symbols, $X$, $Y$, $Z$, $F$ and $G$ and expected answer length $n$.
\par 
\medskip
Execute each of the following tests with equal probability.

\begin{enumerate}
\item \textbf{Pauli (anti)commutation test:} Execute $\textprotocol{crel}(\{X,Y,Z\},n)$.
\item \textbf{Clifford (anti)commutation test:} Execute $\textprotocol{crel}(\{Z,F,G\},n)$.
\item \textbf{Pauli product relation test:} Execute $\textprotocol{prel}(X,Z,Y,n)$.
\item \textbf{Clifford product relation test:} Execute $\textprotocol{prel}(G,Z,F,n)$.
\item \textbf{Clifford conjugation test:} Execute $\textprotocol{conj-cliff}(X,Y,G,F,n)$.
\end{enumerate}
\end{protocolbox}

\begin{protocolbox}{Clifford test}\label{protocol:cliff-test}
\textbf{Notation:} $\textprotocol{cliff}(X, Y, Z, F, G, n)$\\
\textbf{Input:} Observable symbols $X,Y,Z,F,G$ and qubit count $n$.
\par\medskip
Let $P_0, P_1$ be the brick-wall pairings on $[n]$ (see \cref{def:brick-wall}). 
\par
\medskip
Execute each of the following tests with equal probability.

\begin{enumerate}
  \item \textbf{Clifford group test:} Execute $\textprotocol{cliff-group}(X,Y,Z,F,G,n)$.
  \item \textbf{Bell-basis test:}
    Let $W=(X,Z)$. Sample $b,c\draw\zo$.
    \begin{enumerate}
      \item Send $P_b$ to Alice, receiving Bell-measurement outcomes
            $v = (v_0,v_1) \in \zo^{|P_b|}\times\zo^{|P_b|}$.
      \item Send $W_c$ to Bob, receiving outcome $u\in\zo^n$.
      \item Accept iff for every $(j,k)\in P_b$:
            $u_j \oplus u_k = (v_c)_{\ceil{j/2}}$,
            where $\ceil{j/2}$ indexes the pair $(j,k)$ in $v_c$.
    \end{enumerate}

  \item \textbf{Sign consistency test:}
    Sample $b,c\draw\zo$.
    Let $W=(G, W_1)$, where $W_1\in\{G,F\}^n$ is the string of alternating $G$ and $F$ symbols, starting with $G$.
    \begin{enumerate}
      \item Send $P_b$ to Alice, receiving Bell-measurement outcomes
            $v = (v_0,v_1) \in \zo^{|P_b|}\times\zo^{|P_b|}$.
      \item Send $W_c$ to Bob, receiving outcome $u\in\zo^n$.
      \item Reject if there is a pair $(j,k)\in P_b$ with $(v_1)_{\ceil{j/2}}\neq c$ and $u_j \oplus u_k = (v_0)_{\ceil{j/2}}\oplus(v_1)_{\ceil{j/2}}$.
    \end{enumerate}

  \item \textbf{Mixed-versus-pure basis test:} Execute $\textprotocol{mbt}(\{G,F\},n,\mu)$, where $\mu$ is the point distribution on $W_1\in\{G,F\}^n$, the string of alternating $G$ and $F$ symbols, starting with $G$.

  \item \textbf{CHSH test:} Let $A\coloneqq(Y,X)$ and $B\coloneqq(F,G)$ and sample $x,y\draw\zo$.

\begin{enumerate}
    \item Send $A_x$ to Alice and receive answer $a\in\zo^n$.
    \item Send $B_y$ to Bob and receive answer $b\in\zo^n$. \item Accept if and only if
    \[
    a_1 \oplus b_1 \equiv x\cdot (1-y)\pmod 2\,.
    \]
\end{enumerate}
\item \textbf{Mixed-versus-pure basis test:} Execute $\textprotocol{mbt}(\{X,Y,Z,F,G\}, n, U_{5^n})$.
\end{enumerate}
\end{protocolbox}

\subsubsection{Honest prover strategy}
Most questions that a prover can receive in the Clifford test or any of its subtests are summarized in \cref{table:honest-prover}. A measurement in the basis indicated by $W\in\Sigma$ means that the prover measures in $\sigma_W$, as introduced in \cref{sec:pauli-group}. A \emph{large-answer} question is a single symbol $W\in\Sigma$; a \emph{small-answer} question is a pair $(W,a)$ with $W\in\Sigma$ and a binary string $a\in\zo^n$; and a \emph{mixed-basis} question is a length-$n$ string $\tilde W\in\Sigma^n$ specifying one symbol per qubit. In addition, the conjugation test (\cref{protocol:conjugation-test-mod}) internally introduces the auxiliary observable symbols $C_{AB}$, $X_C$, $Z_C$ and $X_R$, which are either constant ($X_C$ and $Z_C$) or depend on the test inputs $A,B$ and $R$. The honest prover responds to each question type as described below. By our consistency convention, Alice should perform the transpose of the specified operations.
\par
\medskip
Some questions, such as those internal to the anti-commutation test, are not listed here explicitly, since the magic-square test is standard in the self-testing literature, we refer the reader to \cite{Cleve2004} for more details. We do note that the prover also needs one additional EPR pair to pass the magic-square test (in addition to the control qubit required for the conjugation test).

\vspace{0.4\baselineskip}
\begin{table}[H]
\centering
\small
\renewcommand{\arraystretch}{1.28}
\setlength{\tabcolsep}{5.5pt}
\begin{tabular}{@{}>{\RaggedRight\arraybackslash}p{0.28\textwidth}>{\RaggedRight\arraybackslash}p{0.68\textwidth}@{}}
    \toprule
    \textbf{Question} & \textbf{Honest prover action} \\
    \midrule
    $W\in\Sigma$ &
    Measure every qubit $i\in[n]$ separately in the basis indicated by $W$ and return the $n$-bit outcome $x\in\zo^n$. \\
    \midrule
    $(W,a)\in\Sigma\times\zo^n$ &
    Simultaneously measure the qubits at positions $i\in[n]$ where $a_i=1$ in the basis indicated by $W$ and return the single-bit outcome $x\in\zo$. \\
    \midrule
    $\tilde W\in\Sigma^n$ &
    Measure each qubit $i\in[n]$ separately in the basis specified by the symbol $\tilde W_i$ and return the $n$-bit outcome $x\in\zo^n$. \\
    \midrule
    $(\tilde W,a)\in\Sigma^n\times\zo^n$ &
    Simultaneously measure the qubits at positions $i\in[n]$ where $a_i=1$ in the basis indicated by $\tilde W_i$ and return the single-bit outcome $x\in\zo$. \\
    \midrule
    $(Q_1,Q_2)$ \par
    {\itshape e.g.\ $(X_C,Z_C)$} &
    Jointly measure the two named small-answer observables and return the pair of bits $(x_1,x_2)\in\zo^2$. Extended to three observables in the anti-commutation test. \\
    \midrule
    $P\in\{P_0,P_1\}$ &
    Perform a Bell-basis measurement on each qubit pair in $P$ (indicating one of the two brick-wall pairings from \cref{def:brick-wall}) and return the outcomes $v\in\zo^{|P|}\times\zo^{|P|}$. \\
    \midrule
    $X_R$ &
    Implement the observable $X_R$ as defined in \Cref{eqn:C-X-R-definition} and return the single-bit outcome $x\in\zo$. \\
    \midrule
    $C_{AB}$ &
    Implement the observable $C_{AB}$ as defined in \Cref{eqn:C-X-R-definition} and return the single-bit outcome $x\in\zo$. \\
    \midrule
    $X_C$ or $Z_C$ &
    Measure qubit $n+1$ (the control) in the Hadamard or computational basis (depending on the question) and return the single-bit outcome $x\in\zo$. \\
    \bottomrule
\end{tabular}
\caption{Questions in the Clifford test and its subtests, with the honest prover action for each. Here $\Sigma=\{X,Y,Z,F,G\}$, $n$ is the qubit count, and $A,B,R$ are the conj. test inputs.}
\label{table:honest-prover}
\end{table}

\begin{remark}
    We will use $\Qa$ to denote the set of all Alice questions in the Clifford test (\cref{protocol:cliff-test}) and its subtests. Through the automatic consistency test, this includes every Bob question of a wrapped leaf; conversely, every Alice question of a wrapped leaf is also asked of Bob.
\end{remark}

\subsection{Compiled prover switching}

\begin{definition}[Question families]
    We define question families as the sets from which individual questions are sampled during protocol execution, for either player. The size of these sets can either depend on protocol parameters, such as the number of qubits or can be constant. All questions which get passed as input arguments to a game or which are fixed and not sampled from a set are also treated as individual question families.
\end{definition}

\begin{definition}[Question distributions]
    A question distribution consists of a tuple $(Q,\mu)$ where the first entry is a question family and the second is an efficiently sampleable probability distribution over the elements of $Q$. The question distribution for single-element question families is the trivial point-distribution ($U_1$).
\end{definition}

\begin{definition}[$\eps$-self-consistency]\label{def:self-consistent}
    Any question distribution $(Q_W,\mu_W)$, over questions $W\in Q_W$, corresponding to PVMs, with efficient projectors $\{W^v\}_{v\in\{0,1\}^{n}}$, which produce $n$-bit answers and satisfy the following expression
    \[
    \E_{W\sim\mu_W}\sum_\alpha\Tr{W^{\Dec(\alpha)}\pea{W}}\geq 1-O(\eps)
    \]
    is called $\eps$-self-consistent. Here $\Dec(\alpha)\in\{0,1\}^n$ and the sum over $\alpha$ goes over all valid encryptions of all $n$-bit strings. If $\mu_W$ is the uniform distribution we replace the expectation by $\E_{W\in Q_W}$ by our notational convention.
\end{definition}

\begin{definition}[$\eps$-cross-consistency]\label{def:cross-consistent}
    Any uniformly efficient family of PVMs $\{W_V\}_{V\in Q_V}$, with projectors $\{W_V^v\}_{v\in\{0,1\}^{n}}$, which produce $n$-bit answers and satisfy the following expression
    \[
    \E_{V\sim \mu_V}\sum_\alpha\Tr{W_V^{\Dec(\alpha)}\pea{V}}\geq 1-O(\eps)
    \]
    is called $\eps$-cross-consistent with the question distribution $(Q_V,\mu_V)$. Here $\Dec(\alpha)\in\{0,1\}^n$ and the sum over $\alpha$ goes over all valid encryptions of all $n$-bit strings.
\end{definition}

\begin{lemma}\label{lemma:compiled-consistency}
    For any computationally efficient prover modeled as in Section \ref{section:modeling} that succeeds with probability $1-\eps$ in the compiled consistency game $\textprotocol{con}(T)$, obtained from Protocol \ref{protocol:consistency-test}, it holds that the same strategy succeeds with probability $1-2\eps$ in game $T$ and that all question distributions in $T$ are $\eps$-self-consistent. 
\end{lemma}

\begin{proof}
The first claim of the lemma is immediate, since with probability $1/2$ the verifier just plays game $T$ and the averaged winning probability needs to be $1-\eps$.
\par
\medskip
For the second claim, winning the consistency game with probability $1-\eps$ requires
    \[
    \underset{s\in\{0,1\}}{\mathbb{E}}\,\underset{W\sim P_{T,s}}{\mathbb{E}} \sum_\alpha\Tr{W^{\Dec(\alpha)}\pea{W}} \geq 1-2\eps,
    \]
    where $P_{T,0}=P_{T,A}$ and $P_{T,1}=P_{T,B}$ are the Alice and Bob question marginals of $T$. Restricting to either value of $s$, each marginal therefore satisfies the same bound up to a constant factor. Each of these distributions can be decomposed into two separate distributions: one representing the probability of choosing a specific question family and the other representing the probability of choosing a specific question inside that family (according to the corresponding question distribution). Since there will always be only a constant number of question families within a game, on either side, we can disregard the expectation over this first distribution and conclude that for all question distributions $(Q_W,\mu_W)$ in game $T$ it holds that
    \[
    \E_{W\sim\mu_W}\sum_\alpha\Tr{W^{\Dec(\alpha)}\pea{W}} \geq 1-O(\eps),
    \]
    which completes the proof.
\end{proof}

\begin{lemma}\label{lemma:generic-cross-isometric-equality}
    For any $\eps$-self-consistent question distribution $(Q_A,\mu_A)$ over questions $A\in Q_A$, corresponding to efficient PVMs with projectors $\{A^v\}_{v\in\{0,1\}^{m}}$, and for every uniformly efficient family of PVMs $\{B_{A}\}_{A\in Q_A}$, with projectors $\{B_{A}^v\}_{v\in\{0,1\}^{n}}$, which satisfies the following equation for any efficiently computable function $f: \{0,1\}^m\times\{0,1\}^n\to\{0,1\}$ and any efficiently sampleable set $S\subseteq\zo^n$ (which can depend on $A$)
    \[
    \E_{A\sim\mu_A}\E_{a\in S}\sum_\alpha (-1)^{f(\Dec(\alpha),a)}\Tr{B_{A}(a){\pea{A}}}\geq 1- O(\delta),
    \]
    where $\Dec(\alpha)\in\{0,1\}^m$, the following holds: for all $q\in\Qa$ there exists a cryptographically small function $\ea$ such that
    \[
    \E_{A\sim\mu_A}\E_{a\in S}\norm{\sum_{v\in\{0,1\}^m}(-1)^{f(v,a)}A^v-B_{A}(a)}^2_{\psi^{\Enc(q)}}\leq O(\eps + \delta)+\ea.
    \]
\end{lemma}
\begin{proof}
     We can rewrite the self-consistency of $(Q_A,\mu_A)$ as:
    \begin{equation}\label{eqn:projector-close-to-id-2}
    \E_{A\sim\mu_A}\sum_\alpha \|A^{\Dec(\alpha)} - \id\|_{\pea{A}}^2 = 1 - \E_{A\sim\mu_A}\sum_\alpha\Tr{A^{\Dec(\alpha)}\pea{A}} \leq O(\eps),
    \end{equation}
    by the normalization property of the PVM (all projectors sum to the identity), we also know that
    \begin{equation}\label{eqn:projector-close-to-0-2}
    \E_{A\sim\mu_A}\sum_{\substack{a,\alpha\\\Dec(\alpha)\neq a}}\|A^{a}\|^2_{\pea{A}} = \E_{A\sim\mu_A}\sum_{\substack{a,\alpha\\\Dec(\alpha)\neq a}}\Tr{A^a\pea{A}} \leq O(\eps).
    \end{equation}
    Then,
    \begin{align*}
        &\E_{A\sim\mu_A}\E_{a\in S}\Bigg\|\overbrace{\sum_{v\in\{0,1\}^m}(-1)^{f(v,a)}A^v}^{A_f(a)}-B_{A}(a)\Bigg\|^2_{\psi^{\Enc(A)}} \\
        \intertext{We can separate the term that aligns with $\Dec(\alpha)$ from the remaining sum over $v$:}
        &= \E_{A\sim\mu_A}\E_{a\in S}\sum_\alpha \norm{(-1)^{f(\Dec(\alpha),a)}A^{\Dec(\alpha)}-B_{A}(a) + \sum_{\substack{v\in\{0,1\}^m\\v\neq\Dec(\alpha)}}(-1)^{f(v,a)}A^v}^2_{\pea{A}}\\
        \intertext{Applying the triangle inequality for the squared state-dependent norm (\cref{lemma:state-dependent-norm-properties}, point (\textit{\ref{prop:state-dependent-norm-properties,v}})):}
        &\leq 2\E_{A\sim\mu_A}\E_{a\in S}\sum_\alpha\left(\norm{(-1)^{f(\Dec(\alpha),a)}A^{\Dec(\alpha)}-B_{A}(a)}^2_{\pea{A}} + \norm{\sum_{\substack{v\in\{0,1\}^m\\v\neq\Dec(\alpha)}}(-1)^{f(v,a)}A^v}^2_{\pea{A}}\right)\\
        \intertext{Applying a telescopic sum, the same triangle inequality, left unitary invariance and using the fact that $\{A^v\}$ is a PVM:}
        &\leq 2\E_{A\sim\mu_A}\E_{a\in S}\sum_\alpha\Bigg( 2\Big(\norm{A^{\Dec(\alpha)}-\id}^2_{\pea{A}} + \norm{(-1)^{f(\Dec(\alpha), a)}\id-B_{A}(a)}^2_{\pea{A}}\Big)\\
        &\quad +\sum_{\substack{v\in\{0,1\}^m\\v\neq\Dec(\alpha)}}\norm{A^v}^2_{\pea{A}}\Bigg)\\
        &\leq O(\eps+\delta),
    \end{align*}
    where the last line follows from \eqref{eqn:projector-close-to-id-2}, \eqref{eqn:projector-close-to-0-2} and the lemma hypothesis (rewritten as a state-dependent norm). Lastly, the argument of the state-dependent norm in the lemma conclusion is $A_f(a)-B_{A}(a)$, where $\{A_f(a)\}_{A,a}$ is a uniformly efficient family of unitaries parameterized by $A\in\Qa$ and $a\in\zo^n$, since the prover projectors are efficiently measurable given the question and we can reuse the argument of \cref{lemma:observables-efficient-unitaries}, now with an efficient function in the exponent. Similarly, $\{B_{A}(a)\}_{A,a}$ is a uniformly efficient family of unitaries (parameterized by $A$ and $a$), by \cref{lemma:observables-efficient-unitaries}. Thus the argument of the norm is a uniformly efficient family of LCUs (which is also Hermitian) and we can apply \cref{cor:efficient-operator-state-switching}, where the sampled parameter is the tuple $(a,A)$, with $a$ uniformly from $S$, $A\sim\mu_A$, and $z=A$ deterministically for $D_1$. For $D_2$ the distribution over the sampled parameter is identical (recall identical marginal constraint) and $z=q$ deterministically. The point distribution on any fixed $q\in\Qa$ is efficiently sampleable, so the corollary yields, for every such $q$, a single cryptographically small function $\ea$ such that
    \[
    \E_{A\sim\mu_A}\E_{a\in S}\norm{\sum_{v\in\{0,1\}^m}(-1)^{f(v,a)}A^v-B_{A}(a)}^2_{\psi^{\Enc(q)}}\leq O(\eps+\delta)+\ea,
    \]
    which completes the proof.
\end{proof}

\begin{corollary}\label{cor:self-consistent-equality}
    For any $\eps$-self-consistent question distribution $(Q_A,\mu_A)$ over questions $A$, corresponding to PVMs, with efficient projectors $\{A^a\}_{a\in\{0,1\}^{n}}$ and any uniformly efficient family of PVMs $\{B_{A}\}_{A\in Q_A}$, with projectors $\{B_{A}^b\}_{b\in\{0,1\}^{n}}$, which is $\delta$-cross-consistent with $(Q_A,\mu_A)$, for all $q\in\Qa$ there exists a cryptographically small function $\ea$ such that
    \[
    \E_{A\sim\mu_A}\sum_{v\in\{0,1\}^n}\norm{A^v-B_{A}^v}^2_{\psi^{\Enc(q)}} = \E_{A\sim\mu_A}\E_{a\in\{0,1\}^n} \norm{A(a)-B_{A}(a)}^2_{\psi^{\Enc(q)}} \leq O(\eps+\delta)+\ea.
    \]
\end{corollary}
\begin{proof}
    If we set $S=\zo^n$, $m=n$ and $f(a,b)=a\cdot b$ in \cref{lemma:generic-cross-isometric-equality}, the lemma hypothesis is equivalent to the displayed $\delta$-cross-consistency requirement, since
    \begin{align*}
        \E_{A\sim\mu_A}\E_{a\in \bits^n}&\sum_\alpha (-1)^{\Dec(\alpha)\cdot a} \Tr{B_{A}(a)\pea{A}} \\
        &= \E_{A\sim\mu_A}\sum_\alpha\sum_{v\in\{0,1\}^n} \overbrace{\E_{a\in\{0,1\}^n} (-1)^{(\Dec(\alpha) +v)\cdot a}}^{\delta_{v,\Dec(\alpha)}}\Tr{B_{A}^v\pea{A}}\\
        &=\E_{A\sim\mu_A}\sum_\alpha\Tr{B_{A}^{\Dec(\alpha)}\pea{A}}\\
        &\geq 1-O(\delta).
    \end{align*}
    The hypothesis is thus satisfied by the cross-consistency and the conclusion holds. We again provide an equivalent form for the conclusion, which follows from Parseval's identity (\cref{cor:parseval}), since $A^v - B_{A}^v$ is the Fourier transform of $A(a)-B_{A}(a)$.
\end{proof}

\begin{lemma}[Restatement of \cite{Metger2024} Lemma 6.1]\label{lemma:self-consistency-identity}
    For any $\eps$-self-consistent question distribution $(Q_B,\mu_B)$ over questions $B$, corresponding to efficient PVMs $\{B^v\}_{v\in\{0,1\}^n}$, it holds that for all $a \in \{0,1\}^n$,
    \[
    \E_{B\sim\mu_B}\sum_\alpha\norm{B(a) - (-1)^{a \cdot \Dec(\alpha)}\id}_{\psi^{\Enc(B)}_\alpha}^2\leq O(\eps).
    \]
\end{lemma}
\begin{proof}
    The condition that $(Q_B,\mu_B)$ is $\eps$-self-consistent means that
    \begin{align*}
        \E_{B\sim\mu_B}\sum_{\alpha} \Tr[B^{\Dec(\alpha)} \psi^{\Enc(B)}_\alpha] &\geq 1 - O(\eps) \\
        \E_{B\sim\mu_B}\sum_\alpha\sum_{\substack{w\\w \neq \Dec(\alpha)}} \Tr[B^w \psi^{\Enc(B)}_\alpha] &\leq O(\eps),
    \end{align*}
    where the second expression follows immediately from the fact that a PVM is normalized. Thus, for all $a\in\zo^n$
    \begin{align*}
    &\E_{B\sim\mu_B}\sum_\alpha\|B(a) - (-1)^{a \cdot \Dec(\alpha)}\id \|_{\psi^{\Enc(B)}_\alpha}^2\\
    &= 2- 2\E_{B\sim\mu_B} \sum_\alpha (-1)^{a \cdot \Dec(\alpha)} \Tr{B(a) \psi^{\Enc(B)}_\alpha}\\ 
    &=2- 2 \E_{B\sim\mu_B}\sum_\alpha \Tr{B^{\Dec(\alpha)} \psi^{\Enc(B)}_\alpha} - 2\E_{B\sim\mu_B} \sum_\alpha\sum_{w\neq \Dec(\alpha)} (-1)^{a \cdot (w+\Dec(\alpha))} \overbrace{\Tr{B^w \psi^{\Enc(B)}_\alpha}}^{\geq 0}\\
    &\leq 2- 2 \E_{B\sim\mu_B}\sum_\alpha \Tr{B^{\Dec(\alpha)} \psi^{\Enc(B)}_\alpha} + 2 \E_{B\sim\mu_B}\sum_\alpha\sum_{w\neq \Dec(\alpha)}\Tr{B^w \psi^{\Enc(B)}_\alpha}\\
    & \leq O(\eps).
    \end{align*}
\end{proof}

\begin{lemma}[Restatement \cite{Metger2024} Lemma 6.2]\label{lemma:self-consistent-to-trace-norm-crit}
    For any $\eps$-self-consistent question distribution ($Q_W,\mu_W$) over questions $W\in Q_W$, which correspond to efficient PVMs $\{W^v\}_{v\in\{0,1\}^{n}}$, for all $a\in\bits^n$ it holds that
    \[
    \E_{W\sim\mu_W}\sum_\alpha \norm{W(a)\pea{W} - \pea{W}W(a)}_1 \leq O(\eps^{1/2}),
    \]
\end{lemma}
\begin{proof}
    Consider the quantity we want to bound
    \begin{align*}
       &\E_{W\sim\mu_W}\sum_\alpha \norm{W(a)\pea{W} - \pea{W} W(a)}_1\\
       &\leq \E_{W\sim\mu_W}\sum_\alpha\paren{\norm{(W(a)-(-1)^{a\cdot\Dec(\alpha)}\id)\pea{W}}_1+\norm{\pea{W}((-1)^{a\cdot\Dec(\alpha)}-W(a))}_1}\\
       \intertext{Since the Schatten-1 norm is invariant under Hermitian adjoint and both $\pea{W}$ and $W(a)$ are Hermitian:}
        &= 2\E_{W\sim\mu_W}\sum_\alpha\norm{(W(a)-(-1)^{a\cdot\Dec(\alpha)}\id)\pea{W}}_1\\
       \intertext{Applying \cite[Lemma 2.10]{Metger2024}, where the sums and the expectation are combined into one large sum:}
       &\leq 2\,\sqrt{\E_{W\sim\mu_W}\sum_\alpha\norm{W(a)-(-1)^{a\cdot\Dec(\alpha)}\id}^2_\pea{W}}\\
       &\leq O(\eps^{1/2}),
    \end{align*}
    the last line follows from \cref{lemma:self-consistency-identity} and the fact that the question distribution $(Q_W,\mu_W)$ is assumed to be $\eps$-self-consistent.
\end{proof}

\begin{lemma}[Compiled prover switching]\label{lemma:compiled-prover-switching}
    For any efficient LCU $A_{B}\in L(\H)$, with $\norm{A_{B}}\leq O(1)$ (for all $B\in Q_B$), and $\eps$-self-consistent question distribution $(Q_B,\mu_B)$ over questions $B$, corresponding to efficient PVMs $\{B^v\}_{v\in\{0,1\}^n}$, and for all $q\in\Qa$ and all $a \in \{0,1\}^n$ there exists a cryptographically small function $\ea$ such that
    \[
    \E_{B\sim\mu_B}\norm{A_{B}B(a)}^2_{\psi^{\Enc(q)}}\leq 2 \E_{B\sim\mu_B}\|A_{B}\|^2_{\psi^{\Enc(q)}} + O(\eps) + \ea.
    \]
\end{lemma}
\begin{proof}
    Because of the $\eps$-self-consistency of $(Q_B,\mu_B)$, we can immediately apply \cref{lemma:self-consistency-identity}, yielding (for all $a\in\bits^n$)
    \[
    \E_{B\sim\mu_B}\sum_\alpha\norm{B(a) - (-1)^{a \cdot \Dec(\alpha)}\id}_{\psi^{\Enc(B)}_\alpha}^2\leq O(\eps).
    \]
    By the above, we can interpret $\psi^{\Enc(B)}_\alpha$ as an approximate eigenstate of $B(a)$. Since $B(a)$ is an observable, applying it on that state will only yield a factor of either $-1$ or $1$, which is irrelevant if we just want to make use of an upper-bound on the norm of $A_{B}$. This is the intuition behind the fact that we can completely disregard self-consistent operators on the right side of the state-dependent norm.
    \par 
    \medskip
    By \cref{cor:efficient-operator-state-switching}, where $D_1$ is the point distribution on $q$ and $D_2=\mu_b$ and the fact that $A_BB(a)$ is again an efficient LCU (since $B(a)$ is an efficient observable): 
    \begin{equation*}
        \E_{B\sim\mu_B}\|A_{B} B(a)\|_{\psi^{\Enc(q)}}^2 \approx_{\eta(\lambda)} \E_{B\sim\mu_B}\|A_{B}B(a)\|_{\psi^{\Enc(B)}}^2\,,
    \end{equation*}
    so we can focus on bounding the latter quantity. To do so, we calculate:
    \begin{align*}
        \E_{B\sim\mu_B}&\| A_{B}B(a) \|_{\psi^{\Enc(B)}}^2 \\
        &= \E_{B\sim\mu_B}\sum_{\alpha}  \| A_{B}B(a) \|_{\psi^{\Enc(B)}_\alpha}^2 \\
        &= \E_{B\sim\mu_B}\sum_{\alpha} \| A_{B}B(a) - (-1)^{a \cdot \Dec(\alpha)} A_{B} + (-1)^{a \cdot \Dec(\alpha)} A_{B} \|_{\psi^{\Enc(B)}_\alpha}^2 \\
        &\leq 2\E_{B\sim\mu_B}\sum_{\alpha} \left( \|(-1)^{a \cdot \Dec(\alpha)} A_{B} \|_{\psi^{\Enc(B)}_\alpha}^2+ \|A_{B}(B(a) - (-1)^{a \cdot \Dec(\alpha)}\id) \|_{\psi^{\Enc(B)}_\alpha}^2\right) \\
        &= 2\E_{B\sim\mu_B}\|A_{B}\|_{\psi^{\Enc(B)}}^2 + 2 \E_{B\sim\mu_B}\sum_\alpha \|A_{B}(B(a) - (-1)^{a \cdot \Dec(\alpha)}\id) \|_{\psi^{\Enc(B)}_\alpha}^2 \\
        &\leq 2\E_{B\sim\mu_B}\|A_{B}\|_{\psi^{\Enc(B)}}^2 + 2 \E_{B\sim\mu_B} \|A_{B}\|^2\sum_\alpha \|B(a) - (-1)^{a \cdot \Dec(\alpha)}\id \|_{\psi^{\Enc(B)}_\alpha}^2 \\
        &\leq 2\E_{B\sim\mu_B}\|A_{B}\|_{\psi^{\Enc(B)}}^2 + 2 O(1) \E_{B\sim\mu_B} \sum_\alpha \|B(a) - (-1)^{a \cdot \Dec(\alpha)}\id \|_{\psi^{\Enc(B)}_\alpha}^2 \\
        &\leq 2\E_{B\sim\mu_B}\|A_{B}\|_{\psi^{\Enc(B)}}^2 + O(\eps).
    \end{align*}
    We can then use \cref{cor:efficient-operator-state-switching} again (and the fact that $A_{B}$ is an efficient LCU), which completes the proof.
\end{proof}

\begin{corollary}\label{cor:compiled-prover-switching-proj}
    For any efficient LCU $A\in L(\H)$ with $\norm{A}\leq O(1)$ and $\eps$-self-consistent question distribution $(Q_B,\mu_B)$ over questions $B$, corresponding to efficient PVMs $\{B^v\}_{v\in\{0,1\}^n}$, for all $q\in\Qa$ there exists a cryptographically small function $\ea$ such that
    \[
    \E_{B\sim\mu_B}\sum_{v\in\zo^n}\norm{AB^v}^2_{\psi^{\Enc(q)}}\leq 2 \|A\|^2_{\psi^{\Enc(q)}} + O(\eps) + \ea.
    \]
\end{corollary}
\begin{proof}
The same argument as in \cref{lemma:compiled-prover-switching}, with the additional parameter $a$ drawn uniformly from $\zo^n$ included in the sampled index of \cref{cor:efficient-operator-state-switching}, yields a single cryptographically small function $\ea$ (for this $q$) such that
    \begin{align*}
        \E_{B\sim\mu_B}\E_{a\in\zo^n} \|AB(a)\|^2_{\psi^{\Enc(q)}} &\leq 2\| A\|_{\psi^{\Enc(q)}}^2 + \|A\|^2O(\eps)+
        \ea.
    \end{align*}
    Parseval's identity (\cref{cor:parseval}), since $AB^v=\widehat{AB}(a)$, then gives
    \begin{align*}
    \E_{B\sim\mu_B}\sum_{v\in\zo^n} \|AB^v\|^2_{\psi^{\Enc(q)}} &= \E_{B\sim\mu_B}\E_{a\in\zo^n} \|AB(a)\|^2_{\psi^{\Enc(q)}} \\
    &\leq 2\| A\|_{\psi^{\Enc(q)}}^2 + \|A\|^2O(\eps)+
        \ea.
    \end{align*}
\end{proof}

\begin{corollary}\label{cor:compiled-prover-switching-obs}
    For any $\eps$-self-consistent question distribution $(Q_B,\mu_B)$ over questions $B$, corresponding to efficient PVMs $\{B^v\}_{v\in\{0,1\}^n}$ and efficient LCUs $A,C\in L(\H)$, with $A\approx_\delta C$ and $\norm{A-C}\leq O(1)$. For all $a\in\zo^n$ and all $q\in\Qa$ there exists a cryptographically small function $\ea$ such that
    \[
    AB(a)\approx_{\delta+\eps+\ea} CB(a),
    \]
    on $\peq$ and under implicit expectation over $B\sim\mu_B$.
\end{corollary}
\begin{proof}
    \begin{align*}
        \E_{B\sim\mu_B}\| (A - C) B(a)\|^2_{\psi^{\Enc(q)}} &\leq 2\| A - C \|_{\psi^{\Enc(q)}}^2 + O(\eps)+
        \eta(\lambda)\\
        &\leq O(\delta+\eps+\ea),
    \end{align*}
    where the first inequality follows from \cref{lemma:compiled-prover-switching} since $A-C$ is again an efficient LCU and the fact that $B(a)$ is constructed from an efficient PVM, the last inequality uses the hypothesis.
\end{proof}

\begin{lemma}\label{lemma:reverse-compiled-prover-switching}
    For any efficient LCU $A\in L(\H)$ with $\norm{A}\leq O(1)$ and $\eps$-self-consistent question distribution $(Q_B,\mu_B)$ over questions $B$, corresponding to efficient PVMs $\{B^v\}_{v\in\{0,1\}^n}$, for all $q\in\Qa$ and all $a \in \{0,1\}^n$ there exists a cryptographically small function $\ea$ such that
    \[
    \norm{A}^2_{\psi^{\Enc(q)}}\leq 2\E_{B\sim\mu_B} \norm{AB(a)}^2_{\psi^{\Enc(q)}} + O(\eps)+\ea.
    \]
\end{lemma}
\begin{proof}
    Because of the $\eps$-self-consistency of $(Q_B,\mu_B)$, we can immediately apply \cref{lemma:self-consistency-identity}, yielding (for all $a\in\bits^n$)
    \[
    \E_{B\sim\mu_B}\sum_\alpha\norm{B(a) - (-1)^{a \cdot \Dec(\alpha)}\id}_{\psi^{\Enc(B)}_\alpha}^2\leq O(\eps).
    \]
    By \cref{cor:efficient-operator-state-switching} and the fact that $A$ is an efficient LCU: 
    \begin{equation*}
        \|A\|_{\psi^{\Enc(q)}}^2 \approx_{\eta(\lambda)} \E_{B\sim\mu_B}\|A\|_{\psi^{\Enc(B)}}^2\,,
    \end{equation*}
    so we can focus on bounding the latter quantity. To do so, we calculate:
    \begin{align*}
        \E_{B\sim\mu_B}\| A\|_{\psi^{\Enc(B)}}^2 &= \E_{B\sim\mu_B}\sum_{\alpha}  \| A \|_{\psi^{\Enc(B)}_\alpha}^2 \\
        &= \E_{B\sim\mu_B}\sum_{\alpha} \norm{AB(a) - A\paren{B(a) - (-1)^{a \cdot \Dec(\alpha)} \id }}_{\psi^{\Enc(B)}_\alpha}^2 \\
        &\leq 2\E_{B\sim\mu_B}\sum_{\alpha} \paren{\|AB(a) \|_{\psi^{\Enc(B)}_\alpha}^2+ \|A(B(a) - (-1)^{a \cdot \Dec(\alpha)}\id) \|_{\psi^{\Enc(B)}_\alpha}^2} \\
        &\leq 2\E_{B\sim\mu_B}\|AB(a)\|_{\psi^{\Enc(B)}}^2 + 2 \|A\|^2 \E_{B\sim\mu_B}\sum_\alpha \|B(a) - (-1)^{a \cdot \Dec(\alpha)}\id \|_{\psi^{\Enc(B)}_\alpha}^2 \\
        &\leq 2\E_{B\sim\mu_B}\|AB(a)\|_{\psi^{\Enc(B)}}^2 + O(\eps).
    \end{align*}
    We can then use \cref{cor:efficient-operator-state-switching} (state switching) again since the product of an efficient LCU and an efficient observable is again an efficient LCU, which completes the proof.
\end{proof}

\begin{remark}
    \cref{lemma:compiled-prover-switching} and \cref{lemma:reverse-compiled-prover-switching} allow us to both introduce and remove observables corresponding to $\eps$-self-consistent question distributions on the right of the state-dependent norm. The bound in both directions is quite loose, but this gives the best $\eps$ dependence. For a tighter two-sided bound, refer to the following lemma, which has a square-root dependence on $\eps$.
 \end{remark}

\begin{lemma}\label{lemma:compiled-prover-switching-povm}
    Let $\{M_x,\id-M_x\}_x$ be a uniformly efficient family of two-outcome POVMs indexed by a classical parameter $x$, let $\nu$ be an efficiently sampleable distribution over pairs $(x,q)$ with $q\in\Qa$, and let $(Q_B,\mu_B)$ be an $\eps$-self-consistent question distribution over questions $B$, corresponding to efficient PVMs $\{B^v\}_{v\in\{0,1\}^n}$. Then for all $a\in\zo^n$ there exists a cryptographically small function $\ea$ such that
    \[
    \left|\E_{(x,q)\sim\nu}\E_{B\sim\mu_B}\Bigl(\Tr{B(a)M_xB(a)\,\psi^{\Enc(q)}}-\Tr{M_x\,\psi^{\Enc(q)}}\Bigr)\right|\leq C\eps^{1/2}+2\ea,
    \]
    for some constant $C\in\Nat$.
\end{lemma}
\begin{proof}
    Expanding through a telescoping sum and using linearity of the trace and the expectation,
    \begin{align*}
    &\E_{(x,q)\sim\nu}\E_{B\sim\mu_B}\Bigl(\Tr{B(a)M_xB(a)\,\psi^{\Enc(q)}}-\Tr{M_x\,\psi^{\Enc(q)}}\Bigr)\\
    &\quad=\E_{(x,q)\sim\nu}\E_{B\sim\mu_B}\Tr{B(a)M_xB(a)\paren{\psi^{\Enc(q)}-\psi^{\Enc(B)}}}\\
    &\quad\quad+\E_{(x,q)\sim\nu}\E_{B\sim\mu_B}\Bigl(\Tr{B(a)M_xB(a)\,\psi^{\Enc(B)}}-\Tr{M_x\,\psi^{\Enc(B)}}\Bigr)\\
    &\quad\quad+\E_{(x,q)\sim\nu}\E_{B\sim\mu_B}\Tr{M_x\paren{\psi^{\Enc(B)}-\psi^{\Enc(q)}}}.
    \end{align*}
    For the middle term, the triangle inequality, the cyclicity of the trace, the tracial H\"older inequality \cite[Theorem 2]{Baumgartner2011}, $\norm{M_xB(a)}\leq 1$ and \cref{lemma:self-consistent-to-trace-norm-crit} give
    \begin{align*}
    \Bigl|\E_{B\sim\mu_B}\Bigl(\Tr{B(a)M_xB(a)\,\psi^{\Enc(B)}}-\Tr{M_x\,\psi^{\Enc(B)}}\Bigr)\Bigr|&\leq \E_{B\sim\mu_B}\norm{B(a)\psi^{\Enc(B)}-\psi^{\Enc(B)}B(a)}_1\\
    &\leq C\eps^{1/2},
    \end{align*}
    pointwise in $x$ and $a$, hence also after taking $\E_{(x,q)\sim\nu}$. For the first and third terms, both $\{M_x\}$ and $\{B(a)M_xB(a)\}_{(x,B,a)}$ are uniformly efficient two-outcome POVM families (conjugating an efficient POVM element by the efficient unitary $B(a)$ from \cref{lemma:observables-efficient-unitaries} preserves uniform efficiency), so two applications of \cref{lemma:poly-state-indistinguishability}---with $D$ sampling $(x,q)\sim\nu$ and $B\sim\mu_B$, and plaintext pair $(z_0,z_1)=(q,B)$---bound their absolute values by $\ea$ each. Combining the three bounds completes the proof.
\end{proof}

\begin{lemma}\label{lemma:compiled-prover-switching-2}
    For any efficient LCU $A\in L(\H)$ and $\eps$-self-consistent question distribution $(Q_B,\mu_B)$ over questions $B$, corresponding to efficient PVMs $\{B^v\}_{v\in\{0,1\}^n}$, for all $a\in\zo^n$ and all $q\in\Qa$ there exists a cryptographically small function $\ea$ such that
    \[
    \E_{B\sim\mu_B}\norm{AB(a)}^2_{\psi^{\Enc(q)}}\approx_{C\norm{A}^2\eps^{1/2}+2\ea} \|A\|^2_{\psi^{\Enc(q)}},
    \]
    with some constant $C\in\Nat$.
\end{lemma}
\begin{proof}
    By Definition \ref{def:state-dependent-norm}, the cyclicity of the trace, and the fact that $B(a)$ is an observable:
\begin{align*}
\E_{B\sim\mu_B}\Big| \|AB(a)\|^2_{\psi^{\Enc(B)}} -\|A\|^2_{\psi^{\Enc(B)}}\Big| &\leq \E_{B\sim\mu_B}\Big|\Tr{B(a)A^\dag A(B(a)\psi^{\Enc(B)} - \psi^{\Enc(B)}B(a))}\Big|\\
\intertext{now we apply the tracial Hölder's inequality \cite[Theorem 2]{Baumgartner2011}:}
&\leq \E_{B\sim\mu_B} \|A^\dag AB(a)\|\cdot \norm{B(a)\psi^{\Enc(B)} - \psi^{\Enc(B)}B(a)}_1\\
\intertext{Using that $B(a)$ is an observable and \cref{lemma:op-norm-adjoint-invariance}:}
&\leq \E_{B\sim\mu_B} \|A\|^2\cdot \|B(a)\psi^{\Enc(B)} - \psi^{\Enc(B)}B(a)\|_1\\
&\leq \norm{A}^2O(\eps^{1/2}),
\end{align*}
the last line uses the fact that $(Q_B,\mu_B)$ is a $\eps$-self-consistent question distribution and \cref{lemma:self-consistent-to-trace-norm-crit} (with a triangle inequality and the definition of $\pe{B}$). This tells us that for all $a\in\zo^n$
\begin{equation}\label{eqn:hoelder-closeness}
\E_{B\sim\mu_B} \|AB(a)\|^2_{\psi^{\Enc(B)}} \approx_{C\norm{A}^2\eps^{1/2}} \E_{B\sim\mu_B}\|A\|^2_{\psi^{\Enc(B)}},
\end{equation}
for some constant $C\in\Nat$. This bound is stated for the state obtained on question $B$; we need the same estimate for an arbitrary initial Alice question. Since $A$ is an efficient LCU and $B(a)$ is an efficient observable, their product is an efficient LCU and we can use \cref{cor:efficient-operator-state-switching} to switch out the state in the state-dependent norm
\begin{align*}
\E_{B\sim\mu_B}\|AB(a)\|^2_{\psi^{\Enc(q)}} 
&\approx_{\ea}\E_{B\sim\mu_B}\|AB(a)\|^2_{\psi^{\Enc(B)}}\\
&\approx_{C\norm{A}^2\eps^{1/2}}\E_{B\sim\mu_B}\|A\|^2_{\psi^{\Enc(B)}}\\
&\approx_{\ea}\|A\|^2_{\psi^{\Enc(q)}},
\end{align*}
where the middle line uses \eqref{eqn:hoelder-closeness} and the last line uses \cref{cor:efficient-operator-state-switching} again, combined with the fact that the operators are efficient. 
\end{proof}

\subsection{Soundness}
\subsubsection{Conjugation relation}\label{subsec:conj}
One way in which the extended Pauli group differs from the Heisenberg--Weyl group is in the group relations that define it. The $F$ and $G$ operators lie in the $XY$-plane and are related to the latter by a conjugation relation. To apply the GH theorem, we need to certify that the prover's operators approximately satisfy all group relations, which means that we need a self-test for conjugation relations. Such a self-test is defined in \cref{protocol:conjugation-test-mod} and proven sound in \cref{lemma:conj}. The idea (which was first proposed in \cite{Coladangelo2024}) is the following: we need to test that three binary observables $A,B,R$ approximately satisfy the following relations
\[
RA\approx_{\eps}BR\quad\text{and}\quad AR\approx_{\eps}RB.
\]
To this end, we introduce two new binary observables $C_{AB}$ and $X_R$, which are certified to have the following form, in the basis spanned by an additional control qubit (defined by anti-commuting observables $Z_C$ and $X_C$):
\[
C_{AB}=\begin{pmatrix}
    A & 0\\
    0 & B\\
\end{pmatrix},\quad\text{and}\quad
X_R=\begin{pmatrix}
    0 & R\\
    R & 0\\
\end{pmatrix}.\numberthis\label{eqn:C-X-R-definition}
\]
It remains to show that $C_{AB}$ and $X_R$ approximately anti-commute, which yields the desired conjugation relations. An anti-commutation test is standard in the literature, so the main challenge is to ensure the correct form for the observables $C_{AB}$ and $X_R$.

\begin{lemma}\label{lemma:conj}
    Let $P^*$ be any computationally efficient prover modeled as in Section \ref{section:modeling} that succeeds with probability $1-\eps$ in the compiled conjugation game $\textprotocol{conj}(A,B,R)$, obtained from Protocol \ref{protocol:conjugation-test-mod}. Then for all $q\in\Qa$ there exists a cryptographically small function $\ea$ such that on $\peq$
    \[
    RA\approx_{\eps+\ea}BR\quad\text{and}\quad AR\approx_{\eps+\ea}RB.
    \] 
\end{lemma}

\begin{proof}
The operators $Z_C$ and $X_C$ act exclusively as reference or ``control'' operators: they do not appear in any other test and are solely introduced to describe the basis in which $C_{AB}$ is block-diagonal and $X_R$ is anti-diagonal.
\par 
\medskip
By our convention on automatic consistency tests (see the beginning of \cref{sec:protocols}), the leaf game $\textprotocol{conj}(A,B,R)$ is executed wrapped in the consistency test $\textprotocol{con}(\cdot)$ from \cref{protocol:consistency-test}. This certifies that every question distribution which appears in $\textprotocol{conj}(A,B,R)$ is $\eps$-self-consistent, by \cref{lemma:compiled-consistency}. 
\par 
\medskip
We start by analyzing Item 2 of \cref{protocol:conjugation-test-mod}, to characterize the structure of $C_{AB}$. Let $\{W_{AZ}^{ij}\}_{i,j \in \{0,1\}}$ be the 4-outcome projector corresponding to the Bob question $(A,Z_C)$ and similarly for $\{W_{BZ}^{ij}\}_{i,j \in \{0,1\}}$ and $(B,Z_C)$. Introduce the following compact notation:
\[
W_P \coloneqq W_P^0 - W_P^1,\quad W_P^{i} \coloneqq \sum_z W_{PZ}^{iz},\quad W_{Z|P} \coloneqq W_{Z|P}^0 - W_{Z|P}^1\quad\text{ and }\quad W_{Z|P}^i \coloneqq \sum_{p} W_{PZ}^{pi},
\]
for $P \in \{A,B\}$. Then a success probability of $1-O(\eps)$ in the cross-consistency test (Alice was asked $A$ or $B$ and Bob was asked ($A$, $Z_C$) or ($B$, $Z_C$)), implies that:
\begin{equation}\label{eqn:consistency-AB}
\begin{aligned}
\sum_\alpha \Tr{W_{P}^{\Dec(\alpha)}\pea{P}} \geq 1 - O(\eps),\\
\sum_\alpha \Tr{W_{Z|P}^{\Dec(\alpha)}\pea{Z_C}} \geq 1 - O(\eps),
\end{aligned}
\end{equation}
with $P\in\{A,B\}$ and by the protocol specification we can assume that $\alpha$ is the encryption of a single bit. Equation \eqref{eqn:consistency-AB} tells us that $W_P$ is $\eps$-cross-consistent with $P$ and $W_{Z|P}$ is $\eps$-cross-consistent with $(\{Z_C\},U_1)$, for $P\in\{A,B\}$. Recall that the automatic self-consistency check guarantees $\eps$-self-consistency of $(\{A\},U_1),(\{B\},U_1)$ and $(\{Z_C\},U_1)$. With that we can invoke \cref{cor:self-consistent-equality} (both times the argument of the norm is an efficient LCU), to conclude that for $P\in\{A,B\}$ and all $q\in\Qa$ we have
\begin{align}\label{eqn:W-P-closeness}
\frac{1}{2}\norm{W_P-P}^2_\peq = \sum_{a\in\{0,1\}}\|W^a_P-P^a\|^2_{\psi^{\Enc(q)}} \leq O(\eps) + \ea,
\end{align}
and
\begin{equation}\label{eqn:W-Z-closeness}
\frac{1}{2}\norm{W_{Z|P}-Z_C}^2_\peq = \sum_{a\in\{0,1\}}\|W^a_{Z|P}-Z_C^a\|^2_{\psi^{\Enc(q)}}\leq O(\eps)+ \ea.
\end{equation}
We have $W_{AZ}^{az} = W_A^aW_{Z|A}^z=W_{Z|A}^zW_A^a$, which we use to show that $A$ and $Z_C$ approximately commute. 
\begin{align*}
    AZ_C&\approx_{\eps+\ea} W_A Z_C\\
    &\approx_{\eps+\ea} W_AW_{Z|A}\\
    &= W_{Z|A}W_A\\
    &\approx_{\eps+\ea} W_{Z|A}A\\
    &\approx_{\eps+\ea} Z_CA,
\end{align*}
where we used \eqref{eqn:W-P-closeness}, \eqref{eqn:W-Z-closeness}, the $\eps$-self-consistency of $(\{A\},U_1)$, the fact that all products are efficient unitaries and $(\{Z_C\},U_1)$ and \cref{lemma:compiled-prover-switching}. Thus
\[
[A,Z_C]\approx_{\eps+\ea} 0.
\]
Repeating the same steps for $W^{bz}_{BZ}$, we find that \[
[B,Z_C]\approx_{\eps+\ea}0.
\]
Both expressions hold on the state $\peq$ for all $q\in\Qa$.
\par 
\medskip
Success in the characterization tests, where Alice is asked to measure $C_{AB}$, guarantees
\begin{align*}
    \sum_\alpha \Tr{\left(W_{AZ}^{\Dec(\alpha)0}+W_{Z|A}^1\right)\pea{C_{AB}}}\geq 1-O(\eps),
\end{align*}
and
\begin{align*}
    \sum_\alpha \Tr{\left(W_{BZ}^{\Dec(\alpha)1}+W_{Z|B}^0\right)\pea{C_{AB}}}\geq 1-O(\eps),
\end{align*}
or equivalently, by the normalization of the $\{W_{AZ}^{ij}\}$ and $\{W_{BZ}^{ij}\}$ PVMs and the factorization into marginals
\begin{align*}
    \sum_\alpha\norm{W_{A}^{\overline{\Dec(\alpha)}}W_{Z|A}^0}_{\pea{C_{AB}}}^2 = \sum_\alpha \Tr{\left(W_{A}^{\overline{\Dec(\alpha)}}W_{Z|A}^0\right)\pea{C_{AB}}}\leq O(\eps),\\
    \sum_\alpha\norm{W_{B}^{\overline{\Dec(\alpha)}}W_{Z|B}^1}_{\pea{C_{AB}}}^2 = \sum_\alpha \Tr{\left(W_{B}^{\overline{\Dec(\alpha)}}W_{Z|B}^1\right)\pea{C_{AB}}}\leq O(\eps),
\end{align*}
here the equalities follow from the fact that the marginal projectors commute. An application of a telescoping sum and the triangle inequality, using \cref{eqn:W-P-closeness} and \cref{eqn:W-Z-closeness} (already on $\psi^{\Enc(C_{AB})}$, since those bounds hold for all $q\in\Qa$) and the $\eps$-self-consistency of $(\{Z_C\}, U_1)$ with \cref{cor:compiled-prover-switching-proj}, gives
\[
\sum_\alpha\norm{A^{\overline{\Dec(\alpha)}}Z_C^0}_{\pea{C_{AB}}}^2 \leq O(\eps)+\ea\quad\text{ and }\quad \sum_\alpha\norm{B^{\overline{\Dec(\alpha)}}Z_C^1}_{\pea{C_{AB}}}^2\leq O(\eps)+\ea.
\]
By $A^0+A^1=B^0+B^1=\id$ it immediately follows that
\begin{align*}
    \sum_\alpha\norm{(\id-A^{\Dec(\alpha)})Z_C^0}_{\pea{C_{AB}}}^2\leq O(\eps)+\ea,\\
    \sum_\alpha\norm{(\id-B^{\Dec(\alpha)})Z_C^1}_{\pea{C_{AB}}}^2\leq O(\eps)+\ea.
\end{align*}
These two observations are sufficient to conclude that
\[
\sum_\alpha\norm{AZ_C^0-(-1)^{\Dec(\alpha)}Z_C^0}_{\pea{C_{AB}}}^2\leq O(\eps)+\ea,\numberthis\label{eqn:C-W-characterization-1}
\]
and
\[ 
\sum_\alpha\norm{BZ_C^1-(-1)^{\Dec(\alpha)}Z_C^1}_{\pea{C_{AB}}}^2\leq O(\eps)+\ea.\numberthis\label{eqn:C-W-characterization-2}
\]
Combining \cref{eqn:C-W-characterization-1} and \cref{eqn:C-W-characterization-2} with the guarantee from \cref{lemma:self-consistency-identity} (since $(\{C_{AB}\}, U_1)$ is $\eps$-self-consistent) and an application of the triangle inequality gives
\[
\norm{AZ_C^0+BZ_C^1-C_{AB}}_{\pe{C_{AB}}}^2\leq O(\eps)+\ea.\numberthis\label{eqn:C-structure}
\]
Now we will analyze Item 3 of \cref{protocol:conjugation-test-mod}, to characterize the structure of $X_R$.
Let the 4-outcome PVM, executed when Bob is asked to measure $(R,X_C)$, be denoted by $\{W^{ij}_{RX}\}_{i,j\in\zo}$. Define 
\[
W_R^i\coloneqq\sum_j W^{ij}_{RX},\quad W_R\coloneqq W_R^0 - W_R^1,\quad W_X^j\coloneqq\sum_i W^{ij}_{RX}\quad\text{ and }\quad W_X\coloneqq W_X^0-W_X^1.
\]
A success probability of $1-O(\eps)$ in the cross-consistency test (Alice measures $R$ or $X_C$) implies
\begin{align*}
    \sum_\alpha \Tr{W_R^{\Dec(\alpha)}\pea{R}}\ge 1-O(\eps)\\
    \sum_\alpha \Tr{W_X^{\Dec(\alpha)}\pea{X_C}}\ge 1-O(\eps)
\end{align*}
The first equation tells us that $\{W_R^i\}_i$ is $\eps$-cross-consistent with $(\{R\},U_1)$ and since the latter is $\eps$-self-consistent, we can apply \cref{cor:self-consistent-equality} to conclude that
\begin{equation}\label{eqn:closeness-W}
\frac{1}{2}\norm{W_R-R}^2_\peq=\sum_{i\in\zo}\norm{W_R^{i} - R^i}_\peq^2\leq O(\eps)+\ea.
\end{equation}
Analogously, from the second equation, $\eps$-self-consistency of $(\{X_C\}, U_1)$ and \cref{cor:self-consistent-equality} we conclude that
\begin{equation}\label{eqn:closeness-X}
\frac{1}{2}\norm{W_X-X_C}^2_\peq=\sum_{i\in\zo}\norm{W_X^{i} - X_C^i}_\peq^2\leq O(\eps)+\ea.
\end{equation}
Again we can use this and the observation that $W^{ij}_{RX} = W^i_RW^j_X=W^j_XW^i_R$ to conclude that $X_C$ and $R$ approximately commute
\begin{align*}
    X_CR&\approx_{\eps+\ea}W_XR\\
    &\approx_{\eps+\ea}W_XW_R\\
    &=W_RW_X\\
    &\approx_{\eps+\ea}W_RX_C\\
    &\approx_{\eps+\ea}RX_C,
\end{align*}
where we used \eqref{eqn:closeness-W}, \eqref{eqn:closeness-X}, the $\eps$-self-consistency of $R$ and $X_C$ and \cref{cor:compiled-prover-switching-obs}.
\par 
\medskip
Success in the characterization test, where Alice is asked to measure $X_R$, implies
\begin{align*}
    \sum_\alpha \Tr\Bigg[\overbrace{\paren{W^{\Dec(\alpha)0}_{RX} + W^{(1-\Dec(\alpha))1}_{RX}}}^{=W_{X_R}^{\Dec(\alpha)}}\pea{X_R}\Bigg]\geq 1-O(\eps).
\end{align*}
Similarly, since $\{W_{X_R}^{i}\}_i$ forms a PVM and $(\{X_R\},U_1)$ is $\eps$-self-consistent, we can conclude by \cref{cor:self-consistent-equality} that
\[
\sum_{i\in\zo}\norm{W_{X_R}^{i} - X_R^i}_\peq^2\leq O(\eps)+\ea,
\]
inserting the definition of $W_{X_R}^i$, we get
\[
\sum_{i\in\zo}\norm{W_{R}^{i}W_X^0 + W_R^{(1-i)}W_X^1 - X_R^i}_\peq^2\leq O(\eps)+\ea,
\]
skipping some steps where we use \eqref{eqn:closeness-W}, \eqref{eqn:closeness-X} and the $\eps$-self-consistency of $(\{X_C\},U_1)$, we obtain
\[
\frac{1}{2}\norm{\tilde X_R - X_R}^2_\peq = \sum_{i\in\zo}\Big\|\overbrace{R^{i}X_C^0 + R^{(1-i)}X_C^1}^{\tilde X_R^i} - X_R^i\Big\|_\peq^2\leq O(\eps+\ea),
\]
where $X_R=X_R^0-X_R^1$, $\tilde X_R\coloneqq \tilde X_R^0-\tilde X_R^1 = RX_C$, and the left-hand side was obtained through Parseval's identity (\cref{cor:parseval}). Since $\tilde X_R = RX_C$ we can immediately conclude that for all $q\in\Qa$, on $\peq$ we have
\[
X_R\approx_{\eps+\ea} RX_C.
\]

It remains to show the conjugation relation between $A$ and $B$. Success with probability $1-O(\eps)$ in Item 1 of the protocol certifies, by \cref{lemma:compiled-commutation}, that
\[
    [X_R,C_{AB}]\approx_{\eps+\ea}0,\;\;[A,X_C]\approx_{\eps+\ea}0,[B,X_C]\approx_{\eps+\ea}0,\;[R,Z_C]\approx_{\eps+\ea}0,
\]
and by \cref{lemma:compiled-anti-commutation}
\[
    \{X_C,Z_C\}\approx_{\eps+\ea}0,
\]
note that all of bounds above hold for distinct cryptographically small functions, but since there is only a constant number of them, we take the maximum and denote it by a single symbol.
\par 
\medskip
We can use the (anti)commutator relations, together with our characterizations of $X_R$ and $C_{AB}$ (see \cref{eqn:C-structure}) to obtain the desired conjugation relations. Concretely, expanding the commutator with $X_R\approx RX_C$ and $C_{AB}\approx AZ_C^0+BZ_C^1$, commuting $X_C$ past $A$ and $B$, $Z_C$ past $R$, and using the exact relation $X_CZ_C^0=Z_C^1X_C$ (which follows from $\{X_C,Z_C\}=0$; in the approximate setting the anti-commutation bound is used instead), we obtain
\[
[X_R,C_{AB}]\approx \paren{Z_C^0(RB-AR)+Z_C^1(RA-BR)}X_C,
\]
and removing the unitary $X_C$ on the right via \cref{lemma:compiled-prover-switching} (using the $\eps$-self-consistency of $(\{X_C\},U_1)$) yields
\[
\norm{Z_C^0(RB-AR)+Z_C^1(RA-BR)}_\peq^2\leq O(\eps+\ea),
\]
since $\{Z_C^i\}_i$ is a PVM (we specifically need orthogonality and idempotence) this is equivalent to
\begin{equation}\label{eqn:double-norm-bound}
    \norm{Z_C^0(RB-AR)}_\peq^2+\norm{Z_C^1(RA-BR)}_\peq^2\leq O(\eps+\ea),
\end{equation}

since both terms are lower-bounded by 0, this gives us a separate upper bound for each of them. The two bounded blocks pair $Z_C^0$ with $RB-AR$ and $Z_C^1$ with $RA-BR$; to bound the full relations we also need the two opposite pairings, which follow by inserting a factor of $R$. Now
\begin{align*}
    \norm{Z_C^0(RA-BR)}_\peq^2 &= \norm{RZ_C^0(RA-BR)}_\peq^2\\
    \intertext{where the equality is exact, since $R$ is a binary observable and hence $R^2=\id$. Using the fact that $Z_C$ and $R$ approximately commute (as certified by Item 1 of the protocol), the $\eps$-self-consistency of $(\{A\},U_1),(\{B\},U_1)$ and $(\{R\},U_1)$, combined with repeated applications of \cref{lemma:compiled-prover-switching}:}
    &\leq \norm{Z_C^0R(RA-BR)}_\peq^2 + O(\eps+\ea)\\
    &= \norm{Z_C^0(RB-AR)R}_\peq^2 + O(\eps+\ea)\\
    \intertext{where we used $R(RA-BR)=A-RBR=-(RB-AR)R$. Again using $\eps$-self-consistency of $(\{R\},U_1)$ and \cref{lemma:compiled-prover-switching}:}
    &\leq 2\norm{Z_C^0(RB-AR)}_\peq^2 + O(\eps+\ea)\\
    &\leq O(\eps+\ea),
\end{align*}
the last line follows from \eqref{eqn:double-norm-bound}. This finally allows us to bound the desired quantity, by using the above and \eqref{eqn:double-norm-bound}
\begin{align*}
    \norm{RA-BR}^2_\peq &= \norm{(Z_C^0+Z_C^1)(RA-BR)}^2_\peq \\
    &= \norm{Z_C^0(RA-BR)}^2_\peq + \norm{Z_C^1(RA-BR)}^2_\peq\\
    &\leq O(\eps+\ea).
\end{align*}
To prove the second conjugation relation, one repeats the same $R$-insertion argument on the $Z_C^1$ block to find 
\[
    \norm{Z_C^1(AR-RB)}_\peq^2 \leq O(\eps+\ea),
\] 
which combines with $\norm{Z_C^0(AR-RB)}_\peq^2=\norm{Z_C^0(RB-AR)}_\peq^2\leq O(\eps+\ea)$ from \eqref{eqn:double-norm-bound} to give $\norm{AR-RB}_\peq^2\leq O(\eps+\ea)$. This completes the proof.

\end{proof}

\subsubsection{Mixed-versus-pure basis}\label{subsec:mixed-vs-pure}
In much of the initial certification of the extended Pauli group, including the application of GH, we will be restricting the prover to applying pure-basis measurements, i.e.\ measuring all qubits in the same basis. Not only does this simplify the analysis, it allows us to add on the certification of mixed-basis measurements later on, in a modular way. This is achieved through the mixed-versus-pure basis test, defined in \cref{protocol:pure-vs-mixed-test} and proven sound in \cref{lemma:mixed-vs-pure-guarantee}.
\par
\medskip

The idea is the following: the verifier sends a pure-basis question to Alice and a mixed-basis question to Bob; it then accepts if and only if the answers are equal on the subset of qubits where the bases align. This simple consistency test allows us to generalize our pure-basis results to the mixed-basis case. The verifier can also choose which distribution of mixed-basis questions to use, which provides an easy route to certifying different mixed-basis distributions through a simple protocol modification---potentially useful for downstream applications.

\begin{lemma}\label{lemma:mixed-vs-pure-guarantee}
    For any computationally efficient prover modeled as in Section \ref{section:modeling} that succeeds with probability $1-\eps$ in the compiled mixed-versus-pure basis game $\textprotocol{mbt}(\Sigma,n,\mu)$ (where $\Sigma$ is a constant size alphabet), obtained from Protocol \ref{protocol:pure-vs-mixed-test}, the following holds. For all $q\in\Qa$, all $W\in\Sigma$ and all distributions $\mu'$ over $\zo^n$, there exists a cryptographically small function $\ea$ such that
    \[
    \E_{\tilde W\sim\mu}\E_{a\sim\mu'}\norm{W(\exten{a}{J}) - \tilde W(\exten{a}{J})}_\peq^2\leq O(\eps)+\ea,
    \]
    where $J = \{j\in[n]:\tilde W_j=W\}$.
\end{lemma}
\begin{proof}
Let $k\coloneqq|J|$. We will use the following shorthand notation in this proof $d\coloneqq\restr{\Dec(\alpha)}{J}$ (the restriction of the decryption of $\alpha$ to the indices included in $J$), and
\[
S_{d}\coloneqq \{a\in\zo^n: \restr{a}{J}=d\},
\]
these notations will be convenient but they leave a lot of indirect dependencies implicit. For example, $S_d$ depends on the set $J$, which in turn depends on both $W$ and $\tilde W$.
\par 
\medskip
Winning in $\textprotocol{mbt}(\Sigma,n,\mu)$, means that for all $W\in\Sigma$ (as long as the size of $\Sigma$ is constant)
\begin{align}\label{eqn:mbt-win-cond}
    \E_{\tilde W\sim\mu}\sum_{\alpha}\sum_{a\in S_d}\Tr{\tilde W^{a}\pea{W}}= \E_{\tilde W\sim\mu} N(\tilde W, W) \geq 1-O(\eps),
\end{align}
where we introduced
\[
N(\tilde W, W)\coloneqq \sum_{\alpha}\sum_{a\in S_d}\Tr{\tilde W^{a}\pea{W}} = \sum_\alpha \Tr{P^{d}\pea{W}},
\]
for the left-hand side we pulled the sum over $a\in S_d$ into the trace and defined
\[
P^x\coloneqq \sum_{a\in S_x}\tilde W^a,
\]
for $x\in\zo^k$. Clearly, $\{P^x\}_{x\in\zo^k}$ again forms a PVM, since
\[
\paren{P^x}^2 = \paren{\sum_{a\in S_x}\tilde W^a}\paren{\sum_{b\in S_x}\tilde W^b}= \sum_{a\in S_x}\tilde W^a=P^x\quad\forall x\in\zo^k,
\]
which follows from the fact that $\{\tilde W^a\}_a$ is a PVM. Moreover,
\[
P^x P^y = \paren{\sum_{a\in S_x}\tilde W^a}\paren{\sum_{b\in S_y}\tilde W^b} = 0\quad\forall x\neq y\in\zo^k,
\]
which follows from the fact that $S_x\cap S_y = \emptyset$ for $x\neq y$ and the orthogonality of $\{\tilde W^a\}_a$. Third, we have
\[
\sum_{x\in\zo^k}P^x=\sum_{x\in\zo^k}\sum_{a\in S_x}\tilde W^a=\id,
\]
since $\bigcup_{x\in\zo^k}S_x = \zo^n$ and $\{\tilde W^a\}_a$ is normalized. Lastly, $\{P^x\}_x$ is an efficient PVM, since $\{\tilde W^a\}_a$ is, $J$ can be efficiently computed and the sum of these projectors can coherently be prepared. At this point it is important to note that $P^x$ depends on $W$ and $\tilde W$, so we are assuming both to be fixed for now, i.e.\ $P^x$ will only appear inside a specific $N(\tilde W,W)$. 
\par 
\medskip
By the convexity of $O(\eps)$, we can use \cref{eqn:mbt-win-cond} to obtain the following general pointwise lower bound for a fixed $W$ and for all $\tilde W\in\Sigma^n$
\[
\sum_\alpha \Tr{P^{d}\pea{W}} = N(\tilde W, W)\geq 1-\eps_{\tilde W},
\]
where $\eps_{\tilde W}\geq 0$ and $\E_{\tilde W}\eps_{\tilde W}=O(\eps)$. Define
\[
P(a) = \sum_{x\in\zo^k}(-1)^{a\cdot x}P^x,
\]
if we then also choose $f(\Dec(\alpha), a)\coloneqq \restr{\Dec(\alpha)}{J}\cdot a$, we can equivalently write
\begin{align*}
    \E_{a\in\zo^k}\sum_\alpha (-1)^{f(\Dec(\alpha), a)}\Tr{P(a)\pea{W}} = \sum_\alpha \Tr{P^{\restr{\Dec(\alpha)}{J}}\pea{W}} \geq 1-\eps_{\tilde W},
\end{align*}
this is now in the form of the hypothesis of \cref{lemma:generic-cross-isometric-equality} (since we also know that $(\Sigma,U_{|\Sigma|})$ is $\eps$-self-consistent by the automatic consistency check, as this is the Alice-marginal of the leaf), $S=\bits^k$, the $m$ in the lemma is equal to $n$ here and the $n$ in the lemma is equal to $k$ here. Applying that lemma (without the final state switch) yields, for each $\tilde W$,
\begin{align*}
    \E_{a\in\zo^k}\norm{\sum_{v\in\zo^n}(-1)^{\restr{v}{J}\cdot a}W^v - P(a)}_{\psi^{\Enc(W)}}^2\leq O(\eps_{\tilde W}+\eps).
\end{align*}
Averaging over $\tilde W\sim\mu$ and using $\E_{\tilde W}\eps_{\tilde W}=O(\eps)$,
\begin{align*}
    \E_{\tilde W\sim\mu}\E_{a\in\zo^k}\norm{\sum_{v\in\zo^n}(-1)^{\restr{v}{J}\cdot a}W^v - P(a)}_{\psi^{\Enc(W)}}^2\leq O(\eps).
\end{align*}
The argument of the norm is a uniformly efficient family of LCUs parameterized by $(a,\tilde W)$. Applying \cref{cor:efficient-operator-state-switching} with this pair as the common marginal, $D_1$ having $z=W$ and $D_2$ the point distribution on the fixed $q\in\Qa$, we obtain a single cryptographically small function $\ea$ (depending on $q$) such that
\begin{align}\label{eqn:raw-mixed-vs-pure-guarantee}
    \E_{\tilde W\sim\mu}\E_{a\in\zo^{|J|}}\norm{\sum_{v\in\zo^n}(-1)^{\restr{v}{J}\cdot a}W^v - \sum_{x\in\zo^{|J|}}\sum_{\substack{b\in\bits^n\\\restr{b}{J}=x}}(-1)^{x\cdot a}\tilde W^b}_\peq^2 \leq O(\eps)+\ea,
\end{align}
here we replaced $S_x$ by its definition $\{b\in\bits^n:\restr{b}{J}=x\}$, since this makes the dependence on $J$ explicit. These restriction sets form a partitioning of $\bits^n$ and thus we can write the double sum on the left as a sum over $v$, obtaining
\begin{align*}
    \E_{\tilde W\sim\mu}\E_{a\in\zo^{|J|}}\norm{\sum_{v\in\zo^n}(-1)^{\restr{v}{J}\cdot a}\paren{W^v - \tilde W^v}}_\peq^2 \leq O(\eps)+\ea.
\end{align*}
Lastly, we pass from $a\in\zo^{|J|}$ to $a\in\zo^n$ by observing that $v\cdot\exten{a}{J}=\restr{v}{J}\cdot\restr{a}{J}$, so averaging $W(\exten{a}{J})$ over $n$-bit $a$ reproduces the previous average (each projection is repeated $2^{n-|J|}$ times). This yields
\begin{align*}
    \E_{\tilde W\sim\mu}\E_{a\in\bits^n}\norm{W(\exten{a}{J}) - \tilde W(\exten{a}{J})}_\peq^2&=\E_{\tilde W\sim\mu}\E_{a\in \bits^n}\norm{\sum_{v\in\zo^n}(-1)^{v\cdot \exten{a}{J}}\paren{W^v - \tilde W^v}}_\peq^2 \\
    &\leq O(\eps)+\ea.\numberthis\label{eqn:mbt-uniform-conclusion}
\end{align*}
Finally, we will use the linearity of the $W(a)$ and $\tilde W(a)$ operators and their $\eps$-self-consistency to extend this result to arbitrary distributions over $a$. This trick will be used in the Pauli basis test analysis as well.
\begin{align*}
    \E_{\tilde W\sim\mu}&\E_{a\sim\mu'}\norm{W(\exten{a}{J}) - \tilde W(\exten{a}{J})}_\peq^2 \\
    \intertext{By the $\eps$-self-consistency of $(\Sigma, U_{|\Sigma|})$ and \cref{lemma:reverse-compiled-prover-switching} (where we dropped the expectation over $\Sigma$):}
    &\leq 2\E_{\tilde W\sim\mu}\E_{a\sim\mu'}\E_{b\in\bits^n}\norm{W(\exten{a}{J}+\exten{b}{J}) - \tilde W(\exten{a}{J})W(\exten{b}{J})}_\peq^2 + O(\eps) + \ea\\
    \intertext{By the fact that $\exten{a}{J}+\exten{b}{J}$ is uniformly distributed for uniform $b$ and arbitrarily distributed $a$ and \cref{eqn:mbt-uniform-conclusion}:}
    &\leq 4\E_{\tilde W\sim\mu}\E_{a\sim\mu'}\E_{b\in\bits^n}\norm{\tilde W(\exten{a}{J}+\exten{b}{J}) - \tilde W(\exten{a}{J})W(\exten{b}{J})}_\peq^2 + O(\eps) + 3\ea\\
    \intertext{By the linearity of $\tilde W(a)$, left unitary invariance and \cref{eqn:mbt-uniform-conclusion}:}
    &\leq O(\eps) + 7\ea,\\
\end{align*}
this concludes the proof.
\end{proof}

\subsubsection{Small/large answer consistency}\label{subsec:slc}
Since we are not trying to achieve succinctness at this point, we can tolerate a linear communication complexity, so we can use large-answer questions in much of our analysis. These are questions, on which the prover answers with $n$ bits, which---in an honest setting---correspond to a measurement of all his qubits in the basis specified by the question. The benefit of these large-answer questions, is the fact that we can construct the corresponding observables in the analysis/on the verifier's side (as the Fourier transform of the projectors) and get exact linearity for free, which simplifies much of the analysis.
\par
\medskip
However, these large-answer questions don't work in every setting, since in an honest execution they correspond to a separate measurement of every qubit. In some cases, such as the commutation test, we want to simultaneously measure different qubits and only obtain one outcome. Think for example of $\sigma_X\otimes\sigma_X$, which commutes with $\sigma_Z\otimes\sigma_Z$, only if the two qubits are measured simultaneously. Thus to achieve completeness, we will sometimes have to rely on small-answer questions. 
\par
\medskip
The small/large answer consistency test, defined in \cref{protocol:small-large-answer-test} and proven sound in \cref{lemma:slc-guarantee}, allows us to relate the prover operators corresponding to both question types. The test has the same simple consistency check structure as the mixed-versus-pure basis test. In particular, the verifier sends a small-answer question to Alice and a large-answer question to Bob; it then accepts if and only if the XOR of the large answers (restricted to the qubits included in the small-answer question) is equal to the small answer.

\begin{lemma}\label{lemma:slc-guarantee}
    For any computationally efficient prover modeled as in \Cref{section:modeling} that succeeds with probability $1-\eps$ in the small/large answer consistency test $\textprotocol{slc}(\Theta, n, \mu)$, obtained from \Cref{protocol:small-large-answer-test}, the prover's efficient operators satisfy the following: for all $q\in\Qa$ there exists a cryptographically small function $\ea$ such that
   \begin{align}
   \E_{W\in\Theta}\E_{a\sim\mu}\norm{(W_a^0-W_a^1) - W(a)}^2_{\psi^{\Enc(q)}}\leq O(\eps)+\ea.
   \end{align}
\end{lemma}
\begin{proof}
Remember that we denote the small-answer projectors, corresponding to the question $(W,a)$ as $\{W_a^v\}_{v\in\{0,1\}}$. The large-answer projectors corresponding to question $W$ (on which an $n$-bit response is expected) are denoted by $\{W^v\}_{v\in\{0,1\}^n}$, which are combined to form $W(a)$. The winning condition can be rewritten as
\begin{align*}
   \E_{W\in\Theta}\E_{a\sim\mu}\sum_\alpha (-1)^{\Dec(\alpha)} \Tr{W(a)\pea{(W,a)}} \geq 1 - O(\eps),
\end{align*}
where $\Dec(\alpha)\in\{0,1\}$. This corresponds to the hypothesis of \cref{lemma:generic-cross-isometric-equality}, if we set $m=1$, $f(a,b) = a$ and we let $S$ be the set that only exactly the $a$ that appears in the Alice question $(W,a)$. By the automatic self-consistency check that is executed for this leaf test, every Alice question also exists on Bob's side and a winning probability of $1-\eps$ implies $\eps$-self-consistency of the question distribution $(Q_{s},\mu_s)$, over all small-answer questions $(W,a)$ (corresponding to projectors $\{W_a^0,W_a^1\}$). Thus we can apply \cref{lemma:generic-cross-isometric-equality} and conclude that
\[
\overbrace{\E_{W\in\Theta}\E_{a\sim\mu}}^{=\E_{(W,a)\sim\mu_s}}\norm{W_a^0 - W_a^1 - W(a)}^2_{\psi^{\Enc(q)}}\leq O(\eps)+\ea,
\]
where $\{W_a^0,W_a^1\}$ are the projectors applied by Bob on question $Q=(W,a)$, this completes the proof.
\end{proof}

\subsubsection{Commutation relation}\label{subsec:acom}
Multi-qubit (anti-)commutation relations are an important part of the group relations of the extended Pauli group. \Cref{protocol:crel-test} is designed to certify these and is proven sound in \cref{lemma:compiled-crel}. The relation we want to test is the following:
\[
W(a)W'(b)\approx_{\eps} (-1)^{a\cdot b}W'(b)W(a)\quad\text{for all }a,b\in\zo^n\text{ and }W\neq W'\in\{X,Y,Z\}.
\]
and the same for the observables in $\{Z,F,G\}$. The test is essentially a wrapper on the commutation and anti-commutation tests, which both act on the small-answer level. The test thus combines these two subtests with the small/large answer consistency test to obtain a guarantee on the large-answer level.
\begin{lemma}\label{lemma:compiled-crel}
    For any computationally efficient prover modeled as in \Cref{section:modeling} that succeeds with probability $1-\eps$ in the commutation relation test $\textprotocol{crel}(\Sigma,n)$, obtained from \Cref{protocol:crel-test}, the prover's efficient operators satisfy the following: for all $q\in\Qa$ there exists a cryptographically small function $\ea$ such that
    \begin{align}
    \E_{W\neq W'\in\Sigma}\E_{a\sim\mu,b\sim\mu'}\norm{W(a)W'(b) - (-1)^{a\cdot b}W'(b)W(a)}^2_{\psi^{\Enc(q)}} \leq O(\eps+\ea).
    \end{align}
\end{lemma}
\begin{proof}
If $a\cdot b = 0$, success in $\textprotocol{com}((W,a),(W',b))$ tells us (by \cref{lemma:compiled-commutation}) that for all $q\in\Qa$, all $a,b\in\zo^n$ with $a\cdot b=0$ and all $W\neq W'\in\Sigma$ there exists a cryptographically small function $\eta_{W,W',a,b}(\lambda)$ such that
    \[
    \norm{[W_a^0-W_a^1,{W'_b}^0-{W'_b}^1]}^2_{\psi^{\Enc(q)}} \leq O(\eps_{W,W',a,b})+\eta_{W,W',a,b}(\lambda),
    \]
    where
    \[
    \E_{W\neq W'\in\Sigma}\E_{\substack{a,b\in \zo^n\\a\cdot b=0}}\eps_{W,W',a,b} = \eps.
    \]
    To apply the technique from \cref{rem:non-uniformity}, we imagine an adversary against the commutation test, when averaging over $W\neq W'\in\Sigma$ and $a,b\in \zo^n$ with $a\cdot b=0$. This adversary receives the values of $W,W',a$ and $b$ which maximize $\eta_{W,W',a,b}(\lambda)$ (per $\lambda$) as classical advice. Since the commutation test is uniform in $W,W',a$ and $b$, the advantage of this new adversary is also bounded by a cryptographically small function $\ea$ such that
    \[
    \E_{W\neq W'\in\Sigma}\E_{\substack{a,b\in \zo^n\\a\cdot b=0}}\norm{[W_a^0-W_a^1,{W'_b}^0-{W'_b}^1]}^2_{\psi^{\Enc(q)}} \leq O(\eps)+\ea.
    \]
    Similarly, if $a\cdot b=1$, success in $\textprotocol{ac}((W,a),(W',b))$ tells us (by \cref{lemma:compiled-anti-commutation}) that for all $q\in\Qa$, all $a,b\in\zo^n$ with $a\cdot b=1$ and all $W\neq W'\in\Sigma$ there exists a cryptographically small function $\eta_{W,W',a,b}(\lambda)$ such that
    \[
    \norm{\{W_a^0-W_a^1,{W'_b}^0-{W'_b}^1\}}^2_{\psi^{\Enc(q)}} \leq O(\eps_{W,W',a,b})+\eta_{W,W',a,b}(\lambda),
    \]
    where
    \[
    \E_{W\neq W'\in\Sigma}\E_{\substack{a,b\in \zo^n\\a\cdot b=1}}\eps_{W,W',a,b} = \eps.
    \]
    As with the commutation guarantee, we again apply the technique from \cref{rem:non-uniformity}, to conclude that there exists a single negligible function $\ea$ such that
    \[
    \E_{W\neq W'\in\Sigma}\E_{\substack{a,b\in \zo^n\\a\cdot b=1}}\norm{\{W_a^0-W_a^1,{W'_b}^0-{W'_b}^1\}}^2_{\psi^{\Enc(q)}} \leq O(\eps)+\ea.
    \]
    Together, we thus have
    \begin{align*}
    \E_{W\neq W'\in\Sigma}\E_{a,b\in \zo^n}\norm{(W_a^0-W_a^1)({W'_b}^0-{W'_b}^1) - (-1)^{a\cdot b}({W'_b}^0-{W'_b}^1)(W_a^0-W_a^1)}^2_{\psi^{\Enc(q)}}& \\
    \leq O(\eps)+\ea&.
    \end{align*}
    Since the prover passes the small/large answer consistency test, we can invoke \Cref{lemma:slc-guarantee}, and conclude that
    \[
    \E_{W\in\Sigma}\E_{a\in \zo^n}\norm{W(a)- (W_a^0-W_a^1)}^2_{\psi^{\Enc(q)}}\leq O(\eps)+\ea,
    \]
    which also implies for some $W\in\Sigma$:
    \[
    \E_{W'\in(\Sigma\setminus W)}\E_{a\in \zo^n}\norm{W'(a)- ({W'_a}^0-{W'_a}^1)}^2_{\psi^{\Enc(q)}}\leq \frac{|\Sigma|}{|\Sigma|-1}\paren{O(\eps)+\ea},
    \]
    With this we can write out the following approximate equalities, under expectation over $W\neq W'\in\Sigma$ and $a,b\in \zo^n$
    \begin{align*}
        W(a)W'(b)&\approx_{\eps+\ea} (W_a^0-W_a^1) W'(b)\\
        &\approx_{\eps+\ea} (W_a^0-W_a^1)({W'_b}^0-{W'_b}^1)\\
        &\approx_{\eps+\ea} (-1)^{ab}({W'_b}^0-{W'_b}^1)(W_a^0-W_a^1) \\
        &\approx_{\eps+\ea} (-1)^{ab}({W'_b}^0-{W'_b}^1)W(a) \\
        &\approx_{\eps+\ea} (-1)^{ab}W'(b) W(a),
    \end{align*}
    here the first and the last line use compiled prover-switching (\cref{cor:compiled-prover-switching-obs}), since $(\Sigma,U_{|\Sigma|})$ and $(\Sigma\setminus W,U_{|\Sigma|-1})$ are $\eps$-self-consistent question distributions, as certified by the wrapped $\textprotocol{slc}$ subtest and \cref{lemma:compiled-consistency}. The chain of inequalities tells us that for all $q\in\Qa$  we have
    \[
    \E_{W\neq W'\in\Sigma}\E_{a,b\in \zo^n}\norm{W(a)W'(b) - (-1)^{a\cdot b}W'(b)W(a)}^2_{\psi^{\Enc(q)}} \leq O(\eps+\ea).\numberthis\label{eqn:uniform-commutation-relation-pbt}
    \]
    Now we will exploit the exact linearity of the observables and their $\eps$-self-consistency to extend the result to arbitrary distributions. Consider the quantity that needs to be bounded
    \begin{align*}
        &\E_{W\neq W'\in\Sigma}\E_{a\sim\mu,b\in \zo^n}\norm{W(a)W'(b) - (-1)^{a\cdot b}W'(b)W(a)}^2_{\psi^{\Enc(q)}}\\
        \intertext{By the $\eps$-self-consistency of $(\Sigma,U_{|\Sigma|})$ and \cref{lemma:reverse-compiled-prover-switching}:}
        &\leq 2 \E_{W\neq W'\in\Sigma}\E_{a\sim\mu,b,c\in \zo^n}\norm{W(a)W'(b)W(c) - (-1)^{a\cdot b}W'(b)W(a+c)}_\peq^2 \\
        &\quad +\, O(\eps) + \ea\\
        &\leq 4 \E_{W\neq W'\in\Sigma}\E_{a\sim\mu,b,c\in \zo^n}\norm{W(a)(W'(b)W(c) - (-1)^{c\cdot b}W(c)W'(b))}_\peq^2\\
        &\quad + 4 \E_{W\neq W'\in\Sigma}\E_{a\sim\mu,b,c\in \zo^n}\norm{W(a+c)W'(b) - (-1)^{(a+c)\cdot b}W'(b)W(a+c)}_\peq^2\\
        &\quad +\, O(\eps) + \ea\\
        \intertext{By left unitary invariance and since the distribution on $a+c$ is uniform again, we can apply \cref{eqn:uniform-commutation-relation-pbt} and conclude:}
        &\leq O(\eps+\ea).
    \end{align*}
    Repeating the same steps, starting from the expression above, one can show that for any $\mu$ and $\mu'$
    \[
    \E_{W\neq W'\in\Sigma}\E_{a\sim\mu,b\sim\mu'}\norm{W(a)W'(b) - (-1)^{a\cdot b}W'(b)W(a)}^2_{\psi^{\Enc(q)}} \leq O(\eps+\ea).
    \]
\end{proof}

\subsubsection{Product relation}\label{subsec:prod}
Product relations of the type certified in \cref{protocol:product-test} are not directly part of the group relations, but they are crucial for pinning down the abstract group elements $y$ and $\mathsf{f}$, which are not generators themselves, but are defined as products of the generators. 
\par
\medskip
The test only acts on even numbers of qubits; this is due to the fact that in an honest execution $Y$ and $F$ are defined as 
\[
Y = iXZ,\quad\text{and}\quad F = iGZ.
\]
Since we can't associate a prover operator to the imaginary phase unit $i$, because it wouldn't be an observable, it is not possible to test this product relation on the single-qubit level. Instead, if we pair an even number of qubits, the phase will factor out (since it is in the center of the group) and we obtain a relation that is directly testable. 
\par 
\medskip
These relations are needed only in the multi-qubit setting: for $n=1$ the test would not apply, but it would also be unnecessary, since the correct form of $Y$ and $F$ would then follow from the analysis directly.
\begin{lemma}\label{lemma:prod-test}
    For any computationally efficient prover modeled as in \Cref{section:modeling} that succeeds with probability $1-\eps$ in the product relation test $\textprotocol{prel}(X,Z,Y,n)$, obtained from \Cref{protocol:product-test}, the prover's efficient operators satisfy the following: for all $q\in\Qa$ there exists a cryptographically small function $\ea$ such that
    \begin{align}
    \Eazz\norm{(-1)^{|a|/2}X(a)Z(a)-Y(a)}_{\psi^{\Enc(q)}}^2\leq O(\eps)+\ea.
    \end{align}
\end{lemma}
\begin{proof}
Let $Q=((Z,a),(X,a))$, its dependence on $a$ is left implicit, but it is important to remember that the expectation over $a$ is an expectation over the Alice question. Since the subtests $(a)$ and $(b)$ are selected with uniform probability and the overall winning probability is $1-\eps$, we have
\begin{align*}
1-2\eps&\leq\Pr[u\cdot a + |a|/2=x_1+x_2 | \mathsf{subtest }(a)]\\
&=\frac{1}{2}+\frac{1}{2}\Eazz\sum_\alpha (-1)^{\Dec(\alpha)_1 + \Dec(\alpha)_2 + |a|/2} \Tr{Y(a)\pea{Q}},
\end{align*}
repeating the same reasoning for the other passing conditions yields
\begin{align*}
    \Eazz\sum_\alpha (-1)^{\Dec(\alpha)_1 + \Dec(\alpha)_2 + |a|/2} \Tr{Y(a)\pea{Q}} \geq 1 - 4\eps\\
    \Eazz\sum_\alpha (-1)^{\Dec(\alpha)_1} \Tr{Z(a)\pea{Q}} \geq 1 - 8\eps \\
    \Eazz\sum_\alpha (-1)^{\Dec(\alpha)_2} \Tr{X(a)\pea{Q}} \geq 1 - 8\eps. 
\end{align*}
From the above we can conclude that
\begin{equation}\label{eqn:Y-closeness-id}
\begin{aligned}
    \Eazz&\sum_\alpha\norm{(-1)^{\Dec(\alpha)_1 + \Dec(\alpha)_2 + |a|/2}Y(a)-\id}^2_\pea{Q}\\
    &= 2-2\Eazz\sum_\alpha (-1)^{\Dec(\alpha)_1 + \Dec(\alpha)_2 + |a|/2} \Tr{Y(a)\pea{Q}}\\
    &\leq 8\eps
\end{aligned}
\end{equation}
and similarly we have
\begin{equation}\label{eqn:xz-closeness-id}
\begin{aligned}
\Eazz\sum_\alpha\norm{(-1)^{\Dec(\alpha)_1}Z(a)-\id}^2_\pea{Q}\leq 16\eps\\
\Eazz\sum_\alpha\norm{(-1)^{\Dec(\alpha)_2}X(a)-\id}^2_\pea{Q}\leq 16\eps.
\end{aligned}
\end{equation}
Now we have all ingredients to bound the desired quantity under uniform distributions. For brevity, define
\[
s_\alpha \coloneqq (-1)^{\Dec(\alpha)_1+\Dec(\alpha)_2},\qquad
t_\alpha \coloneqq (-1)^{\Dec(\alpha)_1},\qquad
u_\alpha \coloneqq (-1)^{\Dec(\alpha)_2}.
\]
\begin{align*}
    &\Eazz\norm{(-1)^{|a|/2}X(a)Z(a)-Y(a)}_{\psi^{\Enc(Q)}}^2\\
    &= \Eazz\sum_\alpha\norm{X(a)Z(a)-(-1)^{|a|/2}Y(a)}^2_\pea{Q}\\
    &= \Eazz\sum_\alpha\norm{X(a)Z(a)-s_\alpha\id+s_\alpha\id-(-1)^{|a|/2}Y(a)}^2_\pea{Q}\\
    &\leq 2\Eazz\sum_\alpha\norm{X(a)Z(a)-t_\alpha X(a)+t_\alpha X(a)-s_\alpha\id}^2_\pea{Q}\\
    &\quad+2\Eazz\sum_\alpha\norm{s_\alpha\id-(-1)^{|a|/2}Y(a)}^2_\pea{Q}\\
    &\leq 16\eps
      +4\Eazz\sum_\alpha\norm{Z(a)-t_\alpha\id}^2_\pea{Q}
      +4\Eazz\sum_\alpha\norm{X(a)-u_\alpha\id}^2_\pea{Q}\\
    &\leq 144\eps,
\end{align*}
where the last two inequalities follow from \eqref{eqn:Y-closeness-id} and \eqref{eqn:xz-closeness-id}, respectively. We can then use the fact that the argument of the norm is a uniformly efficient family of LCUs (parameterized by $a$) to apply \cref{cor:efficient-operator-state-switching}. Specifically, the expectation over $a$ is the common marginal (where $a$ plays the role of $x$ in the corollary), in the first case $Q$ is fully determined by $a$ and we switch to a setting where $q$ is fixed and independent of $a$. The point distribution on any fixed $q\in\Qa$ is efficiently sampleable, so the corollary yields, for every such $q$, a cryptographically small function $\ea$ such that
\[
\Eazz\norm{(-1)^{|a|/2}X(a)Z(a)-Y(a)}_{\psi^{\Enc(q)}}^2\leq 144\eps + \ea,
\]
which concludes the proof.
\end{proof}

\subsubsection{Clifford conjugation}\label{subsec:conj-cliff}
\begin{lemma}\label{lemma:conj-cliff}
    For any computationally efficient prover modeled as in \Cref{section:modeling} that succeeds with probability $1-\eps$ in the compiled Clifford conjugation test $\textprotocol{conj-cliff}(X,Y,G,F, n)$, obtained from \Cref{protocol:conj-cliff-test}, the prover's efficient operators satisfy the following relationships: for all $q\in\Qa$ there exists a cryptographically small function $\ea$ such that
    \[
    \E_{a\in\zo^n}\norm{Y(a)X(a)-F(a)G(a)}_\peq^2\leq O(\eps+\ea),
    \]
    and
    \[
    \E_{a,b\in\zo^n}\norm{G(a)X(b) - X(b\setminus a)Y(a\cap b)G(a)}_\peq^2\leq O(\eps+\ea).
    \]
    \end{lemma}
    \begin{proof}
        Since the prover passes the small/large answer consistency test with probability $1-O(\eps)$ we can apply \cref{lemma:slc-guarantee} and conclude that for all $W\in\{A,X,Y,G,F\}$ 
        \[
        \E_{a\in\zo^n}\norm{W(a)- (W_a^0-W_a^1)}^2_{\psi^{\Enc(q)}}\leq O(\eps)+\ea.\numberthis\label{eqn:conj-cliff-slc-guarantee}
        \]
        Meanwhile, \cref{lemma:conj} and the prover succeeding in the third conjugation test, tells us that for all $a,b\in\zo^n$ there exists a cryptographically small function $\eta_{a,b}(\lambda)$ such that
        \[
        \norm{(G_a^0-G_a^1)(X_b^0-X_b^1) - (A_b^0-A_b^1)(G_a^0-G_a^1)}_\peq^2\leq O(\eps_{a,b})+\eta_{a,b}(\lambda),
        \]
        with $\E_{a,b\in\zo^n}\eps_{a,b} = \eps$.
        To apply the technique from \cref{rem:non-uniformity}, we can imagine an adversary against the averaged conjugation test, who, per $\lambda$, receives the values of $a,b$ as classical advice, which maximize $\eta_{a,b}(\lambda)$. Since the whole test is uniform in $a$ and $b$, this adversaries advantage is also upper-bounded by a single negligible function $\ea$ such that
        \[
        \E_{a,b\in\zo^n}\norm{(G_a^0-G_a^1)(X_b^0-X_b^1) - (A_b^0-A_b^1)(G_a^0-G_a^1)}_\peq^2\leq O(\eps)+\ea.
        \]
        We can use \cref{eqn:conj-cliff-slc-guarantee} to obtain an expression in terms of the large-answer observables. Consider the following chain of approximate equalities, which holds under an implicit expectation over $a,b\in\zo^n$:
        \begin{align*}
            G(a)X(b)&\approx_{\eps+\ea}(G_a^0-G_a^1)X(b)\\
            &\approx_{\eps+\ea}(G_a^0-G_a^1)(X_b^0-X_b^1)\\
            &\approx_{\eps+\ea} (A_b^0-A_b^1)(G_a^0-G_a^1)\\
            &\approx_{\eps+\ea} (A_b^0-A_b^1)G(a)\\
            &\approx_{\eps+\ea} A(b)G(a),
        \end{align*}
        here the first and last line used \cref{lemma:compiled-prover-switching} and the $\eps$-self-consistency of $(\{Y,A,G\}, U_3)$ (as certified by the automatic consistency check performed for $\textprotocol{slc}$), which implies self-consistency of the separate observables. Thus,
        \[
        \E_{a,b\in\zo^n}\norm{G(a)X(b) - A(b)G(a)}_\peq^2\leq O(\eps+\ea).\numberthis\label{eqn:raw-conj-cliff-guarantee}
        \]
        Now we want to relate the mixed-basis large-answer observable $A(b)$ to a product of pure-basis large-answer observables. For this we will use the guarantee obtained from the pure-versus-mixed basis test and the exact linearity of the large-answer observables. Success in $\textprotocol{mbt}(\{X,Y\},n, U_{2^n})$ tells us that for all $W\in\{X,Y\}$ and all distributions $\mu$ over $\zo^n$:
        \[
        \E_{A\in\{X,Y\}^n}\E_{b\sim\mu}\norm{W(\exten{b}{J_W}) - A(\exten{b}{J_W})}_\peq^2\leq O(\eps)+\ea,
        \]
        where
        \[
        J_W\coloneqq\{i\in[n]: A_i = W\}.
        \]
        By the way that $A$ is defined from $a$ (see \cref{protocol:conj-cliff-test}), we know that
        \[
        J_X = \{i\in[n]: a_i = 0\}\quad\text{and}\quad J_Y=\{i\in[n]: a_i=1\},
        \]
        stepping over to more convenient notation, this means that
        \[
        \exten{b}{J_X} = b\setminus a\quad\text{and}\quad \exten{b}{J_Y}=a\cap b.
        \]
        
        We can now bound the following quantity
        \begin{align*}
            \E_{a,b\in\zo^n}&\norm{A(b)-X(b\setminus a)Y(a\cap b)}_\peq^2
            \intertext{Using exact linearity of $A$:}
            &=\E_{a,b\in\zo^n}\norm{A(b\setminus a)A(a\cap b)-X(b\setminus a)Y(a\cap b)}_\peq^2
            \intertext{By a telescopic sum and the triangle inequality:}
            &\leq 2\E_{a,b\in\zo^n}\norm{A(b\setminus a)A(a\cap b) - X(b\setminus a)A(a\cap b)}_\peq^2\\
            &\quad+2\E_{a,b\in\zo^n}\norm{X(b\setminus a)A(a\cap b)-X(b\setminus a)Y(a\cap b)}_\peq^2 + O(\eps) +\ea
            \intertext{By left unitary invariance, the $\eps$-self-consistency of $(\{A,X,Y,G,F\}, U_5)$ and \cref{lemma:compiled-prover-switching}:}
            &\leq 4\E_{a,b\in\zo^n}\norm{A(b\setminus a)-X(b\setminus a)}_\peq^2\\
            &\quad+  2\E_{a,b\in\zo^n}\norm{A(a\cap b) - Y(a\cap b)}_\peq^2 + O(\eps) +2\ea
            \intertext{By our earlier observations, these are exactly the guarantees we get out of the mixed-versus-pure basis test, which lets us conclude:}
            &\leq O(\eps+\ea).\numberthis\label{eqn:conj-cliff-splitting}
        \end{align*}
    
        Combining \eqref{eqn:raw-conj-cliff-guarantee} with \eqref{eqn:conj-cliff-splitting}, using $\eps$-self-consistency of $G$ and \cref{lemma:compiled-prover-switching}, we obtain
        \[
        \E_{a,b\in\zo^n}\norm{G(a)X(b) - X(b\setminus a)Y(a\cap b)G(a)}_\peq^2\leq O(\eps+\ea).
        \]
        Using the same combination of the conjugation test guarantee (\cref{lemma:conj}) with the non-uniform adversary trick and the small-versus-large answer test (\cref{lemma:slc-guarantee}), that were used to obtain \cref{eqn:raw-conj-cliff-guarantee}, we can conclude from the first and second conjugation test in \cref{protocol:conj-cliff-test} that
        \[
        \E_{a\in\zo^n}\norm{X(a)G(a) - G(a)Y(a)}_\peq^2\leq O(\eps+\ea),
        \]
        and
        \[
        \E_{a\in\zo^n}\norm{G(a)Y(a) - Y(a)F(a)}_\peq^2\leq O(\eps+\ea).
        \]
        With this we can make the following observation, under an implicit uniform average over $a\in\zo^n$
        \begin{align*}
            Y(a)X(a)&=Y(a)G(a)(G(a)X(a))\\
            &\approx_{\eps+\ea}Y(a)(G(a)Y(a))G(a)\\
            &\approx_{\eps+\ea}Y(a)^2F(a)G(a)\\
            &=F(a)G(a),
        \end{align*}
        the third line uses the self-consistency of $G$ and \cref{lemma:compiled-prover-switching}. This completes the proof.
    \end{proof}

\subsubsection{Clifford group relations}\label{subsec:cliff}
\begin{lemma}\label{lemma:epbt-basic-guarantee}
    For any computationally efficient prover modeled as in \Cref{section:modeling} that succeeds with probability $1-\eps$ in Items 1 to 5 of the compiled Clifford group test, obtained from \Cref{protocol:cliff-group}, the prover's efficient observables satisfy the following relationships: for all distributions $\mu$, $\mu'$ on $\zo^n$ and all $q\in\Qa$ there exists a cryptographically small function $\ea$ such that
    \begin{enumerate}
    \item $(\Sigma, U_3)$ is $\eps$-self-consistent for $\Sigma=\{X,Y,Z\}$ and $\Sigma=\{Z,F,G\}$.
    \item $W(a) W(b) = W(a+b)\;$ for all $a,b \in \{0,1\}^n$ and all $W \in \{X, Y, Z, F, G\}$.\label{item:linearity}
    \item For all $W\neq W'\in\{X,Y,Z\}$ and all $W\neq W'\in\{Z, F, G\}$:
    \[
    \E_{a\sim\mu, b \sim\mu'} \norm{ W(a) W'(b) - (-1)^{a \cdot b} W'(b)
      W(a) }_{\psi^{\Enc(q)}}^2 \leq O(\eps +\ea).
    \]\label{item:comrel}
    \item For $W=(X,Z,Y)$ and $W=(G,Z,F)$:
    \[
    \E_{a \sim\mu : |a| = 0 \pmod 2} \norm{(-1)^{|a|/2}W_1(a)W_2(a) -W_3(a)}^2_{\psi^{\Enc(q)}} \leq O(\eps + \ea).
    \]\label{item:prodrel}
    \item $\E_{a\in\zo^n}\E_{b\sim\mu}\norm{G(a)X(b) - X(b\setminus a)Y(a\cap b)G(a)}_\peq^2\leq O(\eps+\ea).$\label{item:conjrel}
    \item $\E_{a\sim\mu}\norm{Y(a)X(a)-F(a)G(a)}_\peq^2\leq O(\eps+\ea).$
    \end{enumerate}
    \end{lemma}
    \begin{proof}
    
      We prove each item in turn.
      \begin{enumerate}
        \item This follows from the automatic execution of $\textprotocol{con}(\textprotocol{slc}(\{X,Y,Z\}), n))$,\\ $\textprotocol{con}(\textprotocol{slc}(\{Z,F,G\}), n))$ and \cref{lemma:compiled-consistency}.
        \item This follows directly from the definition of $W(a)$ and the
          projectivity of the Bob PVMs $\{W^u\}_{u\in\{0,1\}^n}$.
        \item Let $\Sigma=\{X,Y,Z\}$, then the first execution of $\textprotocol{crel}$ in \cref{protocol:cliff-group} ensures that for arbitrary distributions $\mu,\mu'$ over $\zo^n$:
        \[
        \E_{W\neq W'\in\Sigma}\E_{a\sim\mu, b \sim\mu'} \norm{ W(a) W'(b) - (-1)^{a \cdot b} W'(b)
      W(a) }_{\psi^{\Enc(q)}}^2 \leq O(\eps +\ea).
        \]
        Since the size of $\Sigma$ is small, we can drop the expectation over $W$ and $W'$, for a constant blow-up of the error. The same argument works if we set $\Sigma=\{Z,F,G\}$ by the second execution of $\textprotocol{crel}$ in \cref{protocol:cliff-group}. 
        
        \item From \Cref{lemma:prod-test}, and the fact that $\textprotocol{prel}$ is executed for $(X,Z,Y)$ and $(G,Z,F)$ with constant probability, we can conclude that (for $W=(X,Z,Y)$ and $W=(G,Z,F)$):
        \[
        \Eazz \norm{(-1)^{|a|/2}W_1(a)W_2(a) -W_3(a)}^2_{\psi^{\Enc(q)}} \leq O(\eps + \ea).\numberthis\label{eqn:prodrel-uniform-guarantee}
        \]
        Using the commutation guarantees 1 to 3 and a shifting trick we can extend this result to arbitrary distributions. Using left unitary invariance we can introduce uniformly distributed observables. For brevity, let
        \[
        V(b) \coloneqq (-1)^{|b|/2}W_1(b)W_2(b).
        \]
    \begin{align*}
        &\Eazzm\norm{(-1)^{|a|/2}W_1(a)W_2(a)-W_3(a)}_{\psi^{\Enc(q)}}^2\\
        &= \Eazzm\E_{\substack{b \in \zo^n\\ |b| \equiv 0\pmod 2 }}
        \norm{(-1)^{|a|/2}V(b)W_1(a)W_2(a)-V(b)W_3(a)}_{\psi^{\Enc(q)}}^2\\
        \intertext{By a telescoping sum, the triangle inequality, left unitary invariance, $\eps$-self-consistency of $(\{W_1, W_2, W_3\},U_3)$ with \cref{lemma:compiled-prover-switching} and the fact that $a+b$ is uniformly distributed and $|a+b|=|a|+|b|-2a\cdot b$:}
        &\leq 6\Eazzm\E_{\substack{b \in \zo^n\\ |b| \equiv 0\pmod 2 }}
        \norm{W_2(b)W_1(a)-(-1)^{a\cdot b}W_1(a)W_2(b)}_{\psi^{\Enc(q)}}^2\\
        &\quad +9\E_{\substack{c \in \zo^n\\ |c| \equiv 0\pmod 2 }}
        \norm{V(c)-W_3(c)}_{\psi^{\Enc(q)}}^2\\
        &\leq O(\eps+\ea),
    \end{align*}
    here the last line follows from \cref{item:comrel} and \cref{eqn:prodrel-uniform-guarantee}.
    
        \item Success in the $\textprotocol{conj-cliff}$ test and \cref{lemma:conj-cliff} guarantee that
        \[
        \E_{a,b\in\zo^n}\norm{G(a)X(b) - X(b\setminus a)Y(a\cap b)G(a)}_\peq^2\leq O(\eps+\ea).\numberthis\label{eqn:vanilla-conj-guarantee}
        \]
        It only remains to show that we can replace the uniform expectation over $b\in\zo^n$ by one under an arbitrary distribution on $T$. Again, we will use the familiar ``shifting trick'' from earlier calculations, which exploits the exact linearity of the large-answer observables and their self-consistency. For any distribution $\mu$ over $T$, consider
        \begin{align*}
            &\E_{a\in\zo^n}\E_{b\sim\mu}\norm{G(a)X(b) - X(b\setminus a)Y(a\cap b)G(a)}_\peq^2
            \intertext{Using the self-consistency of $(\{X,Y,Z\}, U_3)$ and \cref{lemma:reverse-compiled-prover-switching} to insert the unitary $X(c)$ on the right (note $X(b)X(c)=X(b+c)$ exactly, by \cref{item:linearity}):}
            &\leq 6\E_{a,c\in\zo^n}\E_{b\sim\mu}\norm{G(a)X(b+c) - X(b\setminus a)Y(a\cap b)G(a)X(c)}_\peq^2 + O(\eps) + \ea\\
            \intertext{By a telescoping sum and the triangle inequality:}
            &\leq 12\E_{a,c\in\zo^n}\E_{b\sim\mu}\norm{G(a)X(b+c) - X((b+c)\setminus a)Y(a\cap (b+c))G(a)}_\peq^2\\
            &\quad+ 12\E_{a,c\in\zo^n}\E_{b\sim\mu}\norm{X((b+c)\setminus a)Y(a\cap (b+c))G(a) - X(b\setminus a)Y(a\cap b)G(a)X(c)}_\peq^2\\
            &\quad+ O(\eps) + \ea\\
            \intertext{By exact linearity (e.g.\ $Y((b+c)\setminus a) = Y(b\setminus a)Y(c\setminus a)$), left unitary invariance and the fact that $b+c$ is uniformly distributed:}
            &\leq 24\E_{a,d\in\zo^n}\norm{G(a)X(d) - X(d\setminus a)Y(a\cap d)G(a)}_\peq^2\\
            &\quad+ 12\E_{a,c\in\zo^n}\E_{b\sim\mu}\norm{X(c\setminus a)Y(a\cap b) Y(a\cap c)G(a) - Y(a\cap b)G(a)X(c)}_\peq^2\\
            &\quad+ O(\eps) + \ea\\
            \intertext{Using \cref{eqn:vanilla-conj-guarantee}, a telescoping sum and left unitary invariance:}
            &\leq 24\E_{a,c\in\zo^n}\E_{b\sim\mu}\norm{(X(c\setminus a)Y(a\cap b) - Y(a\cap b)X(c\setminus a)) Y(a\cap c)G(a)}_\peq^2\\
            &\quad+ 24\E_{a,c\in\zo^n}\norm{X(c\setminus a) Y(a\cap c)G(a) - G(a)X(c)}_\peq^2\\
            &\quad+ O(\eps+\ea)\\
            \intertext{We can bound the second term using \cref{eqn:vanilla-conj-guarantee}, the first term can be bounded by applying \cref{lemma:compiled-prover-switching} twice (since $X$, $Y$ and $G$ are all $\eps$-self-consistent), using \cref{item:comrel} and the fact that $(c\setminus a)\cdot(a\cap b)=0$, which yields the final bound of:}
            &\leq O(\eps+\ea),
        \end{align*}
        \item Again $\textprotocol{conj-cliff}$ guarantees that
        \[
        \E_{a\in\zo^n}\norm{Y(a)X(a)-F(a)G(a)}_\peq^2\leq O(\eps+\ea),\numberthis\label{eqn:conj-cliff-raw2}
        \]
        which we will extend to arbitrary distributions using Items 1 to 3. Consider the quantity which we want to bound
        \begin{align*}
            &\E_{a\sim\mu}\norm{Y(a)X(a)-F(a)G(a)}_\peq^2\\
            &=\E_{a\sim\mu}\E_{b\in\zo^n}\norm{Y(b)X(b)Y(a)X(a)-Y(b)X(b)F(a)G(a)}_\peq^2\\
            \intertext{Using self-consistency of $X$ and $G$ with \cref{lemma:compiled-prover-switching}, \cref{item:linearity} and the fact that $a+b$ is again uniformly distributed:}
            &\leq 8 \E_{a\sim\mu}\E_{b\in\zo^n}\norm{X(b)Y(a)-(-1)^{a\cdot b}Y(a)X(b)}_\peq^2\\
            &\quad+ 8 \E_{a\sim\mu}\E_{b\in\zo^n}\norm{G(b)F(a)-(-1)^{a\cdot b}F(a)G(b)}_\peq^2\\
            &\quad+\E_{c\in\zo^n}\norm{Y(c)X(c)-F(c)G(c)}_\peq^2 + O(\eps+\ea)\\
            &\leq O(\eps+\ea),
        \end{align*}
        the last line follows from \cref{item:comrel} and \cref{eqn:conj-cliff-raw2}.
        \end{enumerate}
    
    \end{proof}

    \subsubsection{Rounding to an exact representation}
    We have now introduced all tests required to certify the group relations and proved their soundness. Thus, a successful prover in all of these tests (specifically \cref{protocol:cliff-group}) must be applying operations which approximately satisfy these group relations. What remains is to introduce an explicit approximate representation in terms of the prover's operators and use the certified relations to show that it indeed is an approximate unitary representation of the extended Pauli group. We can then round this approximate representation to an exact one using a stability theorem.
    This framework of approximate representation theory has become the standard approach for many recent self-testing proofs. 
    \par
    \medskip
    To define an approximate representation of the group, we recall that we can represent any group element in terms of its generators as $\omega^p x(a)g(b)z(c)$, for $p\in\{0,1,2,3\}$ and $a,b,c\in\zo^n$.
    Intuitively, we want to define our approximate representation as 
    \begin{align*}
    f(\omega^p x(a)g(b)z(c)) = (-1)^s\Delta^r X(a)G(b)Z(c),
    \end{align*}
    with $s\coloneqq \floor{p/2}\pmod 2$ and $r\coloneqq p\pmod 2$. Here $\Delta$ is a ``phase operator'' that is meant to represent the phase $i$ and will be defined in terms of the prover's observables (in particular the $Y$ and $F$ observables, which only enter into the approximate representation via $\Delta$).
    
    As a first step, we construct this phase operator and show that it satisfies some useful properties, which we would expect from a scalar.
    
    \paragraph*{Constructing the phase operator \label{subsec:phase-op}}

    \begin{lemma}
    \label{claim:delta-commutes}
        Let $\Delta(a) = Y(a) Z(a) X(a)$ and $\tilde\Delta(a)=F(a)Z(a)G(a)$. Suppose that for all distributions $\mu,\mu'$ on $\zo^n$ and all $q\in\Qa$ there exists a cryptographically small function $\ea$ such that:
        \begin{enumerate}
            \item $(\Sigma, U_3)$ is $\eps$-self-consistent for $\Sigma=\{X,Y,Z\}$ and $\Sigma=\{Z,F,G\}$.\label{item:self-con2}
            \item $W(a) W(b) = W(a+b)\;$ for all $a,b \in \{0,1\}^n$ and all $W \in \{X, Y, Z, F, G\}$.\label{item:linearity2}
            \item For all $W\neq W'\in\{X,Y,Z\}$ and all $W\neq W'\in\{Z, F, G\}$:
            \[
            \E_{a\sim\mu, b \sim\mu'} \norm{ W(a) W'(b) - (-1)^{a \cdot b} W'(b)
              W(a) }_{\psi^{\Enc(q)}}^2 \leq O(\eps +\ea).
            \]\label{item:comrel2}
            \item $\E_{a\sim\mu}\norm{Y(a)X(a)-F(a)G(a)}_\peq^2\leq O(\eps+\ea).$\label{item:ingredient-close-deltas}
        \end{enumerate} Then for all $W\in\{X,Y,Z,F,G\}$, all distributions $\mu$, $\mu'$ on $\bits^n$ and all $q\in\Qa$, there exists a cryptographically small function $\ea$ such that
        \begin{enumerate}
            \item $\norm{\Delta^{a+b}-(-1)^{a\cdot b}\Delta^{a\oplus b}}_\peq^2\leq O(\eps+\ea)$ for $a,b\in\zo$.
            \item $\E_{a\sim\mu}\norm{\Delta(a)-\tilde\Delta(a)}_\peq^2\leq O(\eps+\ea).$
            \item $\E_{a\sim\mu,b\sim\mu'}\norm{\Delta(a) W(b) - W(b) \Delta(a)}_\peq^2\leq O(\eps+\ea)$.
        \end{enumerate}
    \end{lemma}
    \begin{proof}
        To obtain the first conclusion, we first show that $\Delta$ is approximately a root of unity:
        \begin{align*}
            \Delta^2 &= Y(e_1)Z(e_1)X(e_1)Y(e_1)(Z(e_1)X(e_1))\\
            &\approx_{\eps+\ea}-Y(e_1)Z(e_1)X(e_1)(Y(e_1)X(e_1))Z(e_1)\\
            &\approx_{\eps+\ea}Y(e_1)(Z(e_1)Y(e_1))Z(e_1)\\
            &\approx_{\eps+\ea}-\id
        \end{align*}
        here we repeatedly used \cref{item:self-con2} with \cref{lemma:compiled-prover-switching} and \cref{item:comrel2}. With this we can conclude the following, for $a,b\in\zo$:
        \[
        \norm{\Delta^{a+b}-(-1)^{a\cdot b}\Delta^{a\oplus b}}_\peq^2\leq O(\eps+\ea).
        \]
        \par
        \medskip
        For the third conclusion (its proof involves showing the second conclusion), suppose $W = Y$. We make the following observations, with an implicit expectation over $a\sim\mu$ and $b\sim\mu'$:
        \begin{align*}
            W(b) \Delta(a) &= Y(b) Y(a) Z(a) X(a)  \\
            \intertext{By exact linearity (hypothesis 2):}
            &= Y(a) (Y(b) Z(a)) X(a) \\
        \intertext{By hypotheses 1 (self-consistency) and 3 (anti-commutation) and \cref{cor:compiled-prover-switching-obs}:}
            &\approx_{\eps+\ea} (-1)^{b \cdot a} Y(a) Z(a) (Y(b) X(a)) \\
            &\approx_{\eps+\ea} (-1)^{2(b \cdot a)} Y(a) Z(a) X(a) Y(b) \\
            &= \Delta(a) W(b).
        \end{align*}
        So we have $W(b) \Delta(a) \approx_{\eps+\ea} \Delta(a) W(b)$ in the case $W = Y$. The other cases $W = X,Z$ are analogous. For the cases where $W=G,F$, we will use \cref{item:ingredient-close-deltas} to show that $\Delta(a)\approx\tilde\Delta(a)$, then we can perform the same steps as above with $\tilde\Delta(a)$ and switch it back out for $\Delta(a)$ in the final expression, using compiled prover switching. Consider the following chain of approximate inequalities under an implicit expectation over $a\sim\mu$:
        \begin{align*}
            \Delta(a)&=Y(a)(Z(a)X(a))\\
            &\approx_{\eps+\ea}(-1)^{|a|}(Y(a)X(a))Z(a)\\
            &\approx_{\eps+\ea}(-1)^{|a|}F(a)(G(a)Z(a))\\
            &\approx_{\eps+\ea}(-1)^{2|a|}F(a)Z(a)G(a)\\
            &=\tilde\Delta(a),
        \end{align*}
        the second and third line use \cref{item:comrel2} and the third line uses \cref{item:ingredient-close-deltas} and \cref{lemma:compiled-prover-switching} with the self-consistency of $Z$. This completes the proof.
    \end{proof}

    \paragraph*{Defining an approximate representation}
    
    \begin{definition} \label{def:prover-approx-rep}
    For a prover with observables $W(a)$ (for $W \in \{X,Y,Z,F,G\}$ and $a \in \bits^n$), we define a function $f: C_{n}\to U(\cH)$ as follows (recalling from \cref{def:extended-pauli-group} that we can represent an arbitrary group element as $\omega^p x(a)g(b)z(c)$ for $p \in \{0,1,2,3\}$ and $a,b,c \in \bits^{n}$):
    \begin{align*}
    f(\omega^p x(a) g(b)z(c)) = (-1)^{s(p)}\Delta^{r(p)} X(a)G(b) Z(c) \,,
    \end{align*}
    where $s(p)\coloneqq \floor{p/2}\pmod 2$ and $r(p)\coloneqq p\pmod 2$.
    \end{definition}
    
    Note that the definition of $f$ depends on the $Y$-observables only through $\Delta \coloneqq\Delta(e_1) = Y(e_1)Z(e_1)X(e_1)$, where $e_1$ is as in \cref{section:notation}.

    \begin{lemma}\label{lemma:approx-repr-close-to-Y}
        For any computationally efficient prover modeled as in Section \ref{section:modeling} that wins with probability $1 - \eps$ in Items 1 to 5 of the compiled Clifford group test, obtained from \cref{protocol:cliff-group}, it holds for the function $f$, from \cref{def:prover-approx-rep}, that for all $q\in\Qa$ and all distributions $\mu$ over $\zo^n$ there exists a cryptographically small function $\ea$ such that
        \[
        \E_{a\sim\mu}\norm{f(y(a))-Y(a)}_\peq^2\leq O(\eps+\ea).
        \]
    \end{lemma}
    \begin{proof}
    Recall that, by the group relations, we can alternatively represent $y(a)$ as 
    \[
    \omega^{|a|}x(a)z(a).
    \]
    Let $s(a)\coloneqq\floor{a/2}\pmod 2$ and $r(a)\coloneqq a\pmod 2$, so
    \begin{align*}
        &\E_{a\sim\mu}\norm{f(y(a))-Y(a)}_\peq^2 \\
        &= \E_{a\sim\mu}\norm{(-1)^{s(|a|)}\Delta^{r(|a|)}X(a)Z(a)-Y(a)}_\peq^2\\
        \intertext{We can split the expectation over $a\in\zo^n$ based on the Hamming weight parity:}
        &= \Eazzm\norm{(-1)^{|a|/2}X(a)Z(a)-Y(a)}_\peq^2\\
        &\quad+ \E_{\substack{a\sim\mu\\|a|\equiv 1 \pmod 2}}\norm{(-1)^{(|a|-1)/2}\Delta X(a)Z(a)-Y(a)}_\peq^2\\
        &\leq O(\eps+\ea),\numberthis\label{eqn:y-closeness-uniform}
    \end{align*}
    where we bounded both terms separately. The bound on the first term automatically follows from \cref{lemma:epbt-basic-guarantee}, \cref{item:prodrel}. To bound the second term, we have to do a bit more work. Expanding $\Delta$ and using left unitary invariance and the exact linearity of $Y(a)$ and $X(a)$:
    \begin{align*}
        \E_{\substack{a\sim\mu\\|a|\equiv 1 \pmod 2}}&\norm{(-1)^{(|a|-1)/2}\Delta X(a)Z(a)-Y(a)}_\peq^2\\
        &=\E_{\substack{a\sim\mu\\|a|\equiv 1 \pmod 2}}\norm{(-1)^{(|a|-1)/2}Z(e_1)X(a+e_1)Z(a)-Y(a+e_1)}_\peq^2\\
        \intertext{Using compiled prover switching left unitary invariance and the triangle inequality, we obtain:}
        &\leq 4\E_{\substack{a\sim\mu\\|a|\equiv 1 \pmod 2}}\norm{Z(e_1)X(a+e_1) - (-1)^{a\cdot e_1+1}X(a+e_1)Z(e_1)}_\peq^2\\
        &\quad + 2\E_{\substack{a\sim\mu\\|a|\equiv 1 \pmod 2}}\norm{(-1)^{(|a|+2a\cdot e_1 + 1)/2}X(a+e_1)Z(a+e_1) - Y(a+e_1)}_\peq^2\\
        &\quad+O(\eps+\ea)\\
        \intertext{The first term can be bounded using \cref{lemma:epbt-basic-guarantee}, Conclusion 3. The second term can be rewritten by introducing $\tilde a=a+e_1$; this string has even Hamming weight and $|\tilde a|\equiv |a|+2a\cdot e_1 +1 \pmod 4$:}
        &\leq 2\E_{\substack{\tilde a\sim\mu\\|\tilde a|\equiv 0 \pmod 2}}\norm{(-1)^{|\tilde a|/2}X(\tilde a)Z(\tilde a) - Y(\tilde a)}_\peq^2 + O(\eps+ \ea)\\
        \intertext{The final bound is then obtained by applying \cref{lemma:epbt-basic-guarantee}, \cref{item:prodrel},}
        &\leq O(\eps+\ea).
    \end{align*}

    \end{proof}
    
    \begin{lemma}\label{lemma:approx-repr-close-to-F}
        For any computationally efficient prover modeled as in Section \ref{section:modeling} that wins with probability $1 - \eps$ in Items 1 to 5 of the compiled Clifford group test, obtained from \cref{protocol:cliff-group}, it holds for the function $f$, from \cref{def:prover-approx-rep}, that for all $q\in\Qa$ and all distributions $\mu$ over $\zo^n$ there exists a cryptographically small function $\ea$ such that
        \[
        \E_{a\sim\mu}\norm{f(\mathsf{f}(a))-F(a)}_\peq^2\leq O(\eps+\ea).
        \]
    \end{lemma}
    \begin{proof}
    Plugging the guarantees from \cref{lemma:epbt-basic-guarantee} into \cref{claim:delta-commutes}, we know that 
    \[
    \E_{a\sim\mu}\norm{\Delta(a)-\tilde\Delta(a)}_\peq^2\leq O(\eps+\ea).\numberthis\label{eqn:deltas-close}
    \]
    Recall that, by the group relations, we can alternatively represent $\mathsf{f}(a)$ as
    \[
    \omega^{|a|}g(a)z(a).
    \]
    Let $s(a)\coloneqq\floor{a/2}\pmod 2$ and $r(a)\coloneqq a\pmod 2$, so
    \begin{align*}
        &\E_{a\sim\mu}\norm{f(\mathsf{f}(a))-F(a)}_\peq^2 \\
        &= \E_{a\sim\mu}\norm{(-1)^{s(|a|)}\Delta^{r(|a|)}G(a)Z(a)-F(a)}_\peq^2\\
        \intertext{We can split the expectation over $a\in\zo^n$ based on the Hamming weight parity:}
        &= \Eazzm\norm{(-1)^{|a|/2}G(a)Z(a)-F(a)}_\peq^2\\
        &\quad+ \E_{\substack{a\sim\mu\\|a|\equiv 1 \pmod 2}}\norm{(-1)^{(|a|-1)/2}\Delta G(a)Z(a)-F(a)}_\peq^2\\
        \intertext{Using \cref{lemma:epbt-basic-guarantee}, \cref{item:prodrel} and \cref{eqn:deltas-close} with compiled prover switching on $G$ and $Z$:}
         &\leq \E_{\substack{a\sim\mu\\|a|\equiv 1 \pmod 2}}\norm{(-1)^{(|a|-1)/2}\tilde\Delta G(a)Z(a)-F(a)}_\peq^2+O(\eps+\ea)\\
        &\leq O(\eps+\ea),
    \end{align*}
    here the second term can be bounded in exactly the same way as in the proof of \cref{lemma:approx-repr-close-to-Y}, where $F$ takes the role of $Y$ and $G$ takes the role of $X$.
    \end{proof}
    
    \begin{lemma}\label{lemma:tailored-conjrel}
    Suppose that for all distributions $\mu,\mu'$ on $\zo^n$ and all $q\in\Qa$ there exists a cryptographically small function $\ea$ such that:
    \begin{enumerate}
    \item $(\Sigma, U_3)$ is $\eps$-self-consistent for $\Sigma=\{X,Y,Z\}$ and $\Sigma=\{Z,F,G\}$.\label{item:self-con3}
    \item $W(a) W(b) = W(a+b)\;$ for all $a,b \in \{0,1\}^n$ and all $W \in \{X, Y, Z, F, G\}$.\label{item:linearity3}
    \item For all $W\neq W'\in\{X,Y,Z\}$ and all $W\neq W'\in\{Z, F, G\}$:
    \[
    \E_{a\sim\mu, b \sim\mu'} \norm{ W(a) W'(b) - (-1)^{a \cdot b} W'(b)
      W(a) }_{\psi^{\Enc(q)}}^2 \leq O(\eps +\ea).
    \]\label{item:comrel3}
    \item $\E_{a\in\zo^n}\E_{b\sim\mu}\norm{G(a)X(b) - X(b\setminus a)Y(a\cap b)G(a)}_\peq^2\leq O(\eps+\ea).$\label{item:conjrel3}
    \item For all $W\in\{X,Y,Z,F,G\}$: $\E_{a\sim\mu,b\sim\mu'}\norm{\Delta(a) W(b) - W(b) \Delta(a)}_\peq^2\leq O(\eps+\ea).$\label{item:delta-commutes3}
    \item $\E_{a\sim\mu}\norm{f(y(a))-Y(a)}_\peq^2\leq O(\eps+\ea).$\label{item:y-closeness}
    \end{enumerate}
    Then for all $q\in\Qa$ there exists a cryptographically small function $\ea$ such that
    \begin{enumerate}
        \item $\E_{a\in\zo^n}\E_{b\sim\mu}\norm{G(a)X(b) - (-1)^{a\cdot b+s(|a\cap b|)}\Delta^{r(|a\cap b|)}X(b)G(a)Z(a\cap b)}_\peq^2\leq \gamma.$
        \item $\E_{a\in\zo^n}\E_{b\sim\mu}\norm{G(a)Y(b) - (-1)^{s(|b\setminus a|)}\Delta^{r(|b\setminus a|)}X(b)G(a)Z(b\setminus a)}_\peq^2\leq \gamma.$
    \end{enumerate}
    where $\gamma=O(\eps+\ea)$, $s(a)\equiv\floor{a/2}\pmod 2$ and $r(a)\equiv a\pmod 2$.
    \end{lemma}
    \begin{proof}
    Towards proving the first conclusion, we show that under implicit expectation over $a\in\zo^n$ and $b\sim\mu$, we have
    \begin{align*}
        G(a)X(b)&\approx_{\eps+\ea}X(b\setminus a)Y(a\cap b)G(a)\\
        &\approx_{\eps+\ea}X(b\setminus a)(-1)^{s(|a\cap b|)}\Delta^{r(|a\cap b|)}X(a\cap b)Z(a\cap b)G(a)\\
        &\approx_{\eps+\ea}(-1)^{s(|a\cap b|)}\Delta^{r(|a\cap b|)}X(b)Z(a\cap b)G(a)\\
        &\approx_{\eps+\ea}(-1)^{a\cdot b +s(|a\cap b|)}\Delta^{r(|a\cap b|)}X(b)G(a)Z(a\cap b),
    \end{align*}
    here the first line uses \cref{item:conjrel3}, the second line uses \cref{item:y-closeness}, the third one follows from \cref{item:delta-commutes3} with \cref{item:linearity3} and the last one uses \cref{item:comrel3}. For most of these steps we also need \cref{item:self-con3} with compiled prover switching (\cref{lemma:compiled-prover-switching}). This shows the first conclusion. 
    \par
    \medskip
    Now, towards the second conclusion, we consider the following approximate equalities under an implicit expectation over $a\in\zo^n$ and $b\sim\mu$:
    \begin{align*}
        G(a)Y(b)&\approx_{\eps+\ea}G(a)(-1)^{s(|b|)}\Delta^{r(|b|)}X(b)Z(b)\\
        &\approx_{\eps+\ea}(-1)^{s(|b|)}\Delta^{r(|b|)}G(a)X(b)Z(b)\\
        &\approx_{\eps+\ea}(-1)^{a\cdot b + s(|b|)+s(|a\cap b|)}\Delta^{r(|b|)+r(|a\cap b|)}X(b)G(a)Z(b\setminus a)\\
        &\approx_{\eps+\ea}(-1)^{a\cdot b + s(|b|)+s(|a\cap b|)+r(|b|) r(|a\cap b|)}\Delta^{r(|b|)\oplus r(|a\cap b|)}X(b)G(a)Z(b\setminus a)\\
        &=(-1)^{s(|b\setminus a|)}\Delta^{r(|b\setminus a|)}X(b)G(a)Z(b\setminus a),
    \end{align*}
    the first approximate equality follows from \cref{item:y-closeness}, the second one uses \cref{item:delta-commutes3}, for the third line we used our previous observation about commuting $X$ past $G$ and \cref{item:linearity3}. The second-to-last line uses \cref{item:comrel3} and the final equality follows from the definitions of $s(a)$, $r(a)$ and the fact that $|b|=|a\cap b|+|b\setminus a|$ and the following identity for $u,v\in\Nat$:
    \[
    s(u+v)=s(u)+s(v)+r(u)r(v).
    \]
    \end{proof}
    
    \begin{lemma}\label{lemma:f-approx-pure-basis}
    For any computationally efficient prover modeled as in Section \ref{section:modeling} that wins with probability $1 - \eps$ in Items 1 to 5 of the compiled Clifford group test, obtained from \cref{protocol:cliff-group}, it holds for the function $f$, from \cref{def:prover-approx-rep}, that for all $q\in\Qa$, all $W\in\{X,Y,Z,F,G\}$ and all distributions $\mu$ over $\zo^n$ there exists a cryptographically small function $\ea$ such that
    \begin{align}
     \E_{a \sim \mu, h \in C_{n}} \| f(h) f(w(a)) - f(h\cdot w(a)) \|_\peq^2 \leq O(\eps+\ea) \,, \label{eqn:f-prover-approx-rep}
    \end{align}
    \end{lemma}
    \begin{proof}
        We first apply \cref{lemma:epbt-basic-guarantee}, which guarantees that the following relations hold for all distributions $\mu$, $\mu', \mu''$ on $\zo^n$ and all $q\in\Qa$:
        \begin{enumerate}
        \item $(\Sigma, U_3)$ is $\eps$-self-consistent for $\Sigma=\{X,Y,Z\}$ and $\Sigma=\{Z,F,G\}$.\label{item:self-consistent}
        \item $W(a) W(b) = W(a+b)\;$ for all $a,b \in \{0,1\}^n$ and all $W \in \{X, Y, Z, F, G\}$.\label{item:final-exact-lin}
        \item For all $W\neq W'\in\{X,Y,Z\}$ and all $W\neq W'\in\{Z, F, G\}$:
        \[
        \E_{a\sim\mu, b \sim\mu'} \norm{ W(a) W'(b) - (-1)^{a \cdot b} W'(b)
          W(a) }_{\psi^{\Enc(q)}}^2 \leq O(\eps +\ea).
        \]\label{item:final-anticomm}
        \end{enumerate}
        Using the conclusions of \cref{lemma:epbt-basic-guarantee} (conclusions $1,2$ and 3) and applying them to \cref{claim:delta-commutes}, we can also conclude that for all $W\in\{X,Y,Z,F,G\}$:
        \[
        \E_{a\sim\mu,b\sim\mu'}\norm{\Delta(a) W(b) - W(b) \Delta(a)}_\peq^2\leq O(\eps+\ea).\numberthis\label{eqn:tailored-conjrel-delta-commutes}
        \]
        and for all $a,b\in\zo$ 
        \[
        \norm{\Delta^{a+b}-(-1)^{a\cdot b}\Delta^{a\oplus b}}_\peq^2\leq O(\eps+\ea).\numberthis\label{eqn:final-delta-root}
        \]
        From \cref{lemma:approx-repr-close-to-Y} and \cref{lemma:approx-repr-close-to-F} we know that
        \[
        \E_{a\sim\mu}\norm{f(y(a))-Y(a)}_\peq^2\leq O(\eps+\ea),\numberthis\label{eqn:repr-close-Y-F-1}
        \]
        and
        \[
        \E_{a\sim\mu}\norm{f(\mathsf{f}(a))-F(a)}_\peq^2\leq O(\eps+\ea).\numberthis\label{eqn:repr-close-Y-F-2}
        \]
        With all of the above, we can apply \cref{lemma:tailored-conjrel} and conclude that:
        \begin{equation}\label{eqn:tailored-conjrel}
        \begin{aligned}
        \E_{a\in\zo^n}\E_{b\sim\mu}\norm{G(a)X(b) - (-1)^{a\cdot b+s(|a\cap b|)}\Delta^{r(|a\cap b|)}X(b)G(a)Z(a\cap b)}_\peq^2\leq O(\eps+\ea),\\
        \E_{a\in\zo^n}\E_{b\sim\mu}\norm{G(a)Y(b) - (-1)^{s(|b\setminus a|)}\Delta^{r(|b\setminus a|)}X(b)G(a)Z(b\setminus a)}_\peq^2\leq O(\eps+\ea),
        \end{aligned}
        \end{equation}
        where $s(a)\equiv\floor{a/2}\pmod 2$ and $r(a)\equiv a\pmod 2$.

        All of these guarantees together, which we obtained through the appropriate self-tests, give us enough information about the prover operators to show \cref{eqn:f-prover-approx-rep}. Since $\omega$ is in the center of the group and the approximate representation (\cref{def:prover-approx-rep}) maps $\omega^2$ to $-1$ we can assume, without loss of generality, that $h=\omega^px(a)g(b)z(c)$, with $p\in\{0,1\}$ and $a,b,c\in\zo^n$.
        \par
        \medskip
        Right multiplication on the normal form is given by the following identities, which follow from the relations in \cref{def:extended-pauli-group}. Anticommutation yields $z(c)x(d)=\omega^{2c\cdot d}x(d)z(c)$ and $z(c)g(d)=\omega^{2c\cdot d}g(d)z(c)$. Conjugation together with $g(b)^2=\id$ and a further anticommutation to restore the $xgz$ order yields $g(b)x(d)=\omega^{3|b\cap d|}x(d)g(b)z(b\cap d)$. Linearity of $x$, $g$ and $z$ then gives the claims for $x(d)$, $z(d)$ and $g(d)$; the claims for $y(d)$ and $\mathsf{f}(d)$ follow by substituting $y(d)=\omega^{|d|}x(d)z(d)$ and $\mathsf{f}(d)=\omega^{|d|}g(d)z(d)$. For all $d\in\zo^n$ and $h\in C_n$,
        \begin{align*}
            h\cdot x(d)&=\omega^{p+2c\cdot d+3|b\cap d|}x(a+d)g(b)z(c+b\cap d),\\
            h\cdot y(d)&=\omega^{p+|d|+2c\cdot d+3|b\cap d|}x(a+d)g(b)z(c+d\setminus b),\\
            h\cdot z(d)&=\omega^{p}x(a)g(b)z(c+d),\\
            h\cdot g(d)&=\omega^{p+2c\cdot d}x(a)g(b+d)z(c),\\
            h\cdot \mathsf{f}(d)&=\omega^{p+|d|+2c\cdot d}x(a)g(b+d)z(c+d).
        \end{align*}
        Applying \cref{def:prover-approx-rep}, reducing the central phase via $s(k+2m)=s(k)+m$, $r(k+2m)=r(k)$ and using the fact that $|d|=|b\cap d|+|d\setminus b|$, we obtain, for $p\in\{0,1\}$,
        \begin{align*}
            f(h\cdot x(d))&=(-1)^{(b+c)\cdot d+s(p+|b\cap d|)}\Delta^{r(p+|b\cap d|)}X(a+d)G(b)Z(c+b\cap d),\\
            f(h\cdot y(d))&=(-1)^{c\cdot d+s(p+|d\setminus b|)}\Delta^{r(p+|d\setminus b|)}X(a+d)G(b)Z(c+d\setminus b),\\
            f(h\cdot z(d))&=\Delta^{p}X(a)G(b)Z(c+d),\\
            f(h\cdot g(d))&=(-1)^{c\cdot d}\Delta^{p}X(a)G(b+d)Z(c),\\
            f(h\cdot \mathsf{f}(d))&=(-1)^{c\cdot d+s(p+|d|)}\Delta^{r(p+|d|)}X(a)G(b+d)Z(c+d).
        \end{align*}
        It remains to show that $f(h)f(w(d))$ is close to the right-hand side in each case. We start with $W=X$ and $W=Y$. Under uniform expectation over $a,b,c\in\zo^n$, $p\in\zo$ and $d\sim\mu$,
        \begin{align*}
            \Delta^{p}X(a)G(b)(Z(c)X(d))&\approx_{\eps+\ea}(-1)^{c\cdot d}\Delta^{p}X(a)G(b)X(d)Z(c)\\
            &\approx_{\eps+\ea}(-1)^{(b+c)\cdot d+s(|b\cap d|)}\Delta^{p+r(|b\cap d|)}X(a+d)G(b)Z(c+b\cap d)\\
            &\approx_{\eps+\ea}(-1)^{(b+c)\cdot d+s(|b\cap d|)+pr(|b\cap d|)}\Delta^{r(p+|b\cap d|)}X(a+d)G(b)Z(c+b\cap d)\\
            &=(-1)^{(b+c)\cdot d+s(p+|b\cap d|)}\Delta^{r(p+|b\cap d|)}X(a+d)G(b)Z(c+b\cap d),
    \\
            \Delta^{p}X(a)G(b)(Z(c)Y(d))&\approx_{\eps+\ea}(-1)^{c\cdot d}\Delta^{p}X(a)G(b)Y(d)Z(c)\\
            &\approx_{\eps+\ea}(-1)^{c\cdot d+s(|d\setminus b|)}\Delta^{p+r(|d\setminus b|)}X(a+d)G(b)Z(c+d\setminus b)\\
            &\approx_{\eps+\ea}(-1)^{c\cdot d+s(|d\setminus b|)+pr(|d\setminus b|)}\Delta^{r(p+|d\setminus b|)}X(a+d)G(b)Z(c+d\setminus b)\\
            &=(-1)^{c\cdot d+s(p+|d\setminus b|)}\Delta^{r(p+|d\setminus b|)}X(a+d)G(b)Z(c+d\setminus b)
        \end{align*}
        In both cases the first line follows from \cref{item:final-anticomm}, the second line uses \cref{eqn:tailored-conjrel} with \cref{eqn:tailored-conjrel-delta-commutes} and \cref{item:final-exact-lin}, the third line uses \cref{eqn:final-delta-root} and the last line follows from the definitions of $s(a)$, $r(a)$ and the fact that $s(p)=0$ for $p\in\zo$. The right-hand sides are $f(h\cdot x(d))$ and $f(h\cdot y(d))$ as displayed above, which shows \cref{eqn:f-prover-approx-rep} for $W=X,Y$.
        \par
        \medskip
        For $W=Z$, exact linearity immediately gives the exact equality 
        \[
        f(h)f(z(d))=\Delta^{p}X(a)G(b)Z(c)Z(d)=\Delta^{p}X(a)G(b)Z(c+d)=f(h\cdot z(d)).
        \]
        For $W=G$, under the same implicit expectation,
        \begin{align*}
            \Delta^{p}X(a)G(b)(Z(c)G(d))&\approx_{\eps+\ea}(-1)^{c\cdot d}\Delta^{p}X(a)G(b)G(d)Z(c)\\
            &=(-1)^{c\cdot d}\Delta^{p}X(a)G(b+d)Z(c),
        \end{align*}
        where the first step uses \cref{item:final-anticomm} and the second uses \cref{item:final-exact-lin}; the right-hand side is $f(h\cdot g(d))$. Thus it only remains to show \cref{eqn:f-prover-approx-rep} for $W=F$. Again we take an implicit expectation over $a,b,c\in\zo^n$, $p\in\zo$ and $d\sim\mu$:
        \begin{align*}
            \Delta^{p}X(a)G(b)Z(c)F(d)&\approx_{\eps+\ea}(-1)^{s(|d|)}\Delta^{p+r(|d|)}X(a)G(b)(Z(c)G(d))Z(d)\\
            &\approx_{\eps+\ea}(-1)^{c\cdot d+s(|d|)}\Delta^{p+r(|d|)}X(a)G(b+d)Z(c+d)\\
            &\approx_{\eps+\ea}(-1)^{c\cdot d+s(|d|)+pr(|d|)}\Delta^{r(p+|d|)}X(a)G(b+d)Z(c+d)\\
            &=(-1)^{c\cdot d+s(p+|d|)}\Delta^{r(p+|d|)}X(a)G(b+d)Z(c+d)
        \end{align*}
    
        Here the first line follows from \cref{eqn:repr-close-Y-F-1}, \cref{eqn:repr-close-Y-F-2} and \cref{eqn:tailored-conjrel-delta-commutes}, the second line uses \cref{item:final-anticomm} and \cref{item:final-exact-lin}, the third line uses \cref{eqn:final-delta-root} and the last line follows from the definitions of $s(a)$, $r(a)$ and the fact that $s(p)=0$ for $p\in\zo$. The right-hand side is $f(h\cdot \mathsf{f}(d))$, which shows \cref{eqn:f-prover-approx-rep} for $W=F$ and completes the proof.
    
    \end{proof}

    \paragraph*{From approximate representation to rigidity}
    \begin{definition}[Clifford test isometry]\label{def:clifford-isometry}
    Let $C_n$ be the $n$-qubit extended Pauli group (\cref{def:extended-pauli-group}), $U_{\mathrm{QFT}}$ be the quantum Fourier transform over that group (see \cref{subsec:QFT}) and $f : C_n \to U(\H)$ be the approximate unitary representation built from the prover operators in \cref{def:prover-approx-rep}. Then the Clifford isometry is the isometry $V : \H \to \H'$ defined as:
    \[
    V \coloneqq \frac{1}{|C_n|} \sum_{u, t\in C_n} U_{\mathrm{QFT}}\ket{t}\otimes f(u t^{-1})\otimes\ket{u}.
    \] 
    \end{definition}
    
    \begin{lemma}\label{lemma:EPBT-pure-basis}
    For any computationally efficient prover modeled as in Section \ref{section:modeling} that wins with probability $1 - \eps$ in the compiled Clifford group test $\textprotocol{cliff-group}(X,Y,Z,F,G,n)$, obtained from \cref{protocol:cliff-group}, there exists a complex Hilbert space $\cH'$ of finite dimension $d' = d|C_n|^2$ and an efficient isometry $V: \cH \to \cH'$ (from \cref{def:clifford-isometry}) such that for all $W \in \{X,Y,Z,G,F\}$, all distributions $\mu$ on $\bits^n$ and all $q\in\Qa$ there exists a cryptographically small function $\ea$ such that
    \begin{align*}
    \E_{a\sim\mu} \norm{ V W(a) -((\sigma_W(a)\oplus\overline{\sigma}_W(a)) \ot\Lambda_W(a)) V }_\peq^2 \leq O(\eps+\ea),
    \end{align*}
    where $\Lambda_X(a)=\Lambda_Y(a)=\Lambda_Z(a)=\id$ and $\Lambda_F(a)=\Lambda_G(a)=\sum_{k\in\zo^n}(-1)^{a\cdot k}\ket{k}\bra{k}\ot\id$ for all $a\in\zo^n$.
    \end{lemma}
    
    \begin{proof}
    Consider any prover strategy that succeeds with probability $1-\eps$ and choose an arbitrary $q\in\Qa$ and distribution $\mu$ over $\bits^n$.
    Let $f$ be as in \cref{def:prover-approx-rep}.
    From \cref{lemma:f-approx-pure-basis}, we get that for any $W \in \{X,Y,Z,G,F\}$,
    \begin{align*}
    \E_{a \sim \mu, h \in C_n} \| f(h) f(w(a)) - f(h \cdot w(a)) \|_\peq^2 \leq O(\eps+\ea) \,. 
    \end{align*}
    Any element of $C_n$ can be encoded in $3n+2$ qubits and group multiplication is efficient under this encoding. Moreover, $f$ is constructed from prover operators with classical post-processing, which enables an efficient implementation of the controlled-$f$ operation. As shown in \cref{subsec:QFT} the quantum Fourier transform over $C_n$ can be efficiently implemented and thus we can apply \cref{cor:gh-efficient} and conclude that under the isometry $V: \cH \to \cH'$ from \cref{def:clifford-isometry} it holds that
    \begin{align*}
    \E_{a\sim\mu} \norm{Vf(w(a)) -  (\pi(w(a))\ot\id) V}_\peq^2 \leq O(\eps+\ea) \,,
    \end{align*}
    where the identity tensor factor is $d|C_n|$-dimensional and $\pi: C_n\to U(\H')$ is the direct sum of all irreducible representations $\rho_\mu: C_n \to U(\H_\mu)$ of $C_n$ (with dimension $d_\mu$),
    \[
    \pi(g)=\bigoplus_\mu \rho_\mu(g)\ot\id_{d_\mu}.
    \]
    Furthermore, from the definition of $f$ and the equivalence guaranteed by \cref{cor:gh-efficient}, we know that
    \begin{equation}\label{eqn:pi-close-fundamental}
    \begin{aligned}
    0&=\E_{h \in C_n} \norm{f(h) f(-\id) - f(-h)}^2_\peq \\
    &= \norm{Vf(-\id) -  (\pi(-\id)\ot\id) V}_\peq^2 \\
    &= \norm{\id + \pi(-\id)\ot\id }^2_{V\peq V^\dag} \,.
    \end{aligned}
    \end{equation}
    where $V$ is the same isometry as before.
    \par 
    \medskip
    Finally, to show that $X(a)$, $Y(a)$, $Z(a)$, $G(a)$ and $F(a)$ are actually close to the corresponding Clifford operators (up to potential choice of faithful representation) after isometry, we need to further characterize the unitary representation $\pi$. Specifically, we will show that there exists a $\hat{\pi}$, such that $\hat{\pi}$ is of the form stated in the lemma, and
    \begin{align}
    \E_{a\sim\mu} \|V W(a) - (\hat{\pi}(w(a))\ot\id) V \|_\peq^2 \leq O(\eps+\ea).
    \end{align}
    \par 
    \medskip
    To this end, recall that $\pi$ is the block-diagonalization of the left regular representation of $C_n$. By the specific choice of ordering for the encoding basis (see \cref{subsec:QFT}), we enforce that the `classical' representations occupy the top block, so we can write $\pi(g) = \pi_+(g)\oplus\pi_-(g)$, where $\pi_\pm(g)$ are representations satisfying $\pi_\pm(-g) = \pm\pi_\pm(g)$ for all $g\in C_n$. I.e. 
    \[
    \pi(g) = \begin{pmatrix}
    \pi_+(g) & 0 \\
    0 & \pi_-(g)
    \end{pmatrix}.
    \]
    We will now round $\pi(g)$ to a representation that only maps the negative identity group-element to $-\id$, since this is exactly the unique algebraic property that distinguishes the classical from the fundamental irreps. We know that every irreducible representation appears in the left regular representation with a multiplicity that is equal to its dimension. From \cref{lemma:clifford-representations} we thus know that $\dim(\pi_{+}) = \dim(\pi_{-}) = 2^{3n+1}$. As such, we define
    \[
    \hat{\pi}(g) = \begin{pmatrix}
    \pi_-(g)& 0 \\
    0 & \pi_-(g)
    \end{pmatrix} = \pi_-(g)\otimes\id_{2}.
    \]
    For all $g\in C_n$,
    \begin{align*}
        \|(\pi(g)-\hat{\pi}(g))\ot\id\|_{V\psi^{\Enc(q)}V^\dag} &= \frac{1}{2}\|(\pi(g)-\hat{\pi}(g))\cdot(\id+\pi(-\id))\ot\id\|_{V\psi^{\Enc(q)}V^\dag}\\
        &\leq \frac{1}{2}\|\pi(g)-\hat{\pi}(g)\|\cdot \|(\id+\pi(-\id))\ot\id\|_{V\psi^{\Enc(q)}V^\dag}\\
        &\leq \frac{1}{2}\left(\|\pi(g)\| + \|\hat{\pi}(g)\|\right)\cdot \|(\id+\pi(-\id))\ot\id\|_{V\psi^{\Enc(q)}V^\dag}\\
        &=\|(\id+\pi(-\id))\ot\id\|_{V\psi^{\Enc(q)}V^\dag} = 0,
    \end{align*}
    where the first equality exploits the exact structure of the rounding (since only the top left block of the difference is non-zero). The second line follows from Lemma \ref{lemma:state-dependent-norm-properties} point ($\ref{prop:state-dependent-norm-properties,ii}$), the third line makes use of the triangle inequality for the Schatten-$\infty$ norm, the last line follows from the fact that both $\pi(g)$ and $\hat{\pi}(g)$ are representations and thus unitaries and the last equality uses \cref{eqn:pi-close-fundamental}. Since the state-dependent norm is non-negative we can conclude that both representations are equivalent, when restricted to the post-isometry state and thus
    \begin{align*}
    \E_{a\in\mu} \norm{Vf(w(a)) -  (\hat\pi(w(a))\ot\id) V}_\peq^2 \leq O(\eps+\ea) \,.\numberthis\label{eqn:GH-conclusion-EPBT}
    \end{align*}
    The unitary representation $\hat\pi$ is constructed from a direct sum of copies of the $2^{n+1}$ `quantum' irreducible representations of $C_n$. We know how these look and by the specific basis ordering chosen for the QFT, the canonical ones appear before the conjugate ones and we can write
    \[
    \E_{a\in\mu} \norm{Vf(w(a)) -  ((\sigma_W(a)\oplus\overline{\sigma}_W(a))\ot\Lambda_W(a)) V}_\peq^2 \leq O(\eps+\ea) \,,
    \]
    where $\Lambda_X(a)=\Lambda_Y(a)=\Lambda_Z(a)=\id$ and $\Lambda_F(a)=\Lambda_G(a)=\sum_{k\in\zo^n}(-1)^{a\cdot k}\ket{k}\bra{k}\ot\id$ for all $a\in\zo^n$.
    \par 
    \medskip
    It remains to show that $f(w(a))=W(a)$ for all $W\in\{X,Y,Z\}$ and $a\in\bits^n$. By \cref{def:prover-approx-rep} it is immediately clear that $f(x(a))=X(a)$, $f(z(a))=Z(a)$ and $f(g(a))=G(a)$ for all $a\in\zo^n$, however, since the prover's $Y$ operator is hidden inside $\Delta$, and similarly for $F$. Earlier we already proved the required relations (see \cref{lemma:approx-repr-close-to-Y} and \cref{lemma:approx-repr-close-to-F}), plugging this into \cref{eqn:GH-conclusion-EPBT} (with a triangle inequality), we can conclude that
    \begin{align*}
    \E_{a\sim\mu} \norm{ V W(a) -((\sigma_W(a)\oplus\overline{\sigma}_W(a))\ot\Lambda_W(a)) V}_\peq^2 \leq O(\eps+\ea) \,,\numberthis\label{eqn:EPBT-conc1}
    \end{align*}
    where the $\Lambda$ operators are defined as above. This concludes the proof.
    \end{proof}

\subsubsection{Clifford test}\label{section:corr}
The guarantee of \cref{lemma:EPBT-pure-basis} doesn't yet pin down the sign on the prover's $G$ and $F$ observables, which requires a test that certifies an algebraic relation that is not one of the group relations. As sketched in \cref{section:rigidity}, this test will use entanglement swapping to enforce correct pairwise correlations on the observables corresponding to the unencrypted interaction.

\paragraph*{Bell-basis test}\label{subsec:epr-analysis}
The first step is to convert entanglement shared between Alice and Bob into entanglement internal to Bob's register. To this end, Alice receives a specific question, which instructs her to measure a predefined pairing of her qubits in the Bell basis. The second subtest of \cref{protocol:cliff-test} ensures that Alice really performs the requested measurement, by using Bob to check her reported outcomes.

\begin{definition}[Brick-wall pairings]\label{def:brick-wall}
For $n$ qubits, define the two \textbf{brick-wall pairings}:
\[
  P_0 \coloneqq \bigl\{(2i-1,\, 2i)\bigr\}_{i\in[\floor{n/2}]},
  \qquad
  P_1 \coloneqq \bigl\{(2i,\, 2i+1)\bigr\}_{i\in[\floor{(n-1)/2}]}.
\]
\end{definition}

\begin{definition}[Pairwise Bell and Pauli projectors]\label{def:pair-projectors}
For a single qubit, define $\tau_W^u \coloneqq \frac{1}{2}(\id + (-1)^u \sigma_W)$ for $W\in\{X,Z,G,F\}$ and $u\in\zo$; for multi-bit superscripts, we mean the natural extension to tensor products of projectors.
For a pair of qubits $(j,k)$, define the single-qubit Bell-basis projector as
\[
  \Phi_{j,k}^{u,v}
  \coloneqq \paren{\sum_{l\oplus m = u}\left(\tau_{X}^l\right)_j \left(\tau_{X}^m\right)_k} \cdot \paren{\sum_{l\oplus m = v}\left(\tau_{Z}^l\right)_j \left(\tau_{Z}^m\right)_k}
  = \ketbra{\Phi^{uv}}{\Phi^{uv}}_{j,k}\,.
\]
For a pairing $P_b$ and label strings $\alpha,\beta\in\zo^{|P_b|}$, define
$\Phi_b^{(\alpha,\beta)} \coloneqq \bigotimes_{i\in[|P_b|]} \Phi_{(P_b)_i}^{\alpha_i,\beta_i}$, with identity on all un-paired qubits.
\end{definition}

\begin{lemma}\label{lemma:epr-bell}
For any computationally efficient prover modeled as in \Cref{section:modeling} that succeeds with probability $1-\eps$ in Items 1 and 2 of the compiled Clifford test, obtained from \Cref{protocol:cliff-test}, Alice's post-measurement state (after the isometry from \cref{def:clifford-isometry}) satisfies, for all $b\in\zo$:
\[
  \sum_{\alpha}
  \norm{
    \paren{\Phi_b^{\Dec(\alpha)}\otimes \id} - \id
  }_{V\pea{P_b} V^\dag}^2
  \leq O(\eps+\ea).
\]
\end{lemma}
\begin{proof}
Let $W=(X,Z)$ as in the protocol specification. 
Fix $b\in\zo$ and let $f_b:\zo^n\to\zo^{|P_b|}$ map an outcome string to its pairwise parities according to the corresponding brick-wall pairing, i.e.\ 
\[
f_b(u)_i=u_{2i-1+b}\oplus u_{2i+b}
\]
for $i\in[|P_b|]$. We also define $h_b:\zo^{|P_b|}\to\zo^n$ for mapping into pair-consistent, $n$-bit strings as
\[
  h_b(s)_{2i-1+b}=h_b(s)_{2i+b}=s_i
\]
for every $(2i-1+b,2i+b)\in P_b$, and set all unpaired coordinates to zero. Notice that: $s\cdot f_b(u)=h_b(s)\cdot u$ for all $s\in\zo^{|P_b|}$ and $u\in\zo^n$.
For $v\in\zo^{|P_b|}$ and $c\in\zo$, define the projectors
\[
  P_{W_c,b}^v\coloneqq\sum_{\substack{u\in\zo^n\\f_b(u)=v}}W_c^u
  \qquad\text{and}\qquad
  \tilde\tau^v_{W_c,b}\coloneqq\sum_{\substack{u\in\zo^n\\f_b(u)=v}}\tau^u_{W_c}.
\]
The left projector can be used to characterize the prover's winning probability. For a prover to pass the Bell-basis test, the parities of Bob's reported outcomes ($u\in\zo^n$) need to match pair-wise measurement outcome string reported by Alice ($v_c=\Dec(\alpha)_c\in\zo^{|P_b|}$). We want to transition from projectors to binary observables, so the following character identity will be useful:
\[
  \ind{f_b(u)=v}
  =\E_{s\in\zo^{|P_b|}}(-1)^{s\cdot(f_b(u)\oplus v)}.
\]
Using the above with $s\cdot f_b(u)=h_b(s)\cdot u$ and the observable expressions $W_c(a)$ and $\sigma_{W_c}(a)$, we obtain
\begin{align*}
  P_{W_c,b}^v
  &=\sum_u\E_{s\in\zo^{|P_b|}}
    (-1)^{s\cdot(f_b(u)\oplus v)}W_c^u\\
  &=\E_{s\in\zo^{|P_b|}}(-1)^{v\cdot s}
    \sum_u(-1)^{h_b(s)\cdot u}W_c^u\\
  &=\E_{s\in\zo^{|P_b|}}(-1)^{v\cdot s}W_c(h_b(s)),\\
\end{align*}
and similarly
\begin{align*}
  \tilde\tau_{W_c,b}^v=\E_{s\in\zo^{|P_b|}}(-1)^{v\cdot s}\sigma_{W_c}(h_b(s)).
\end{align*}

The two Fourier expressions above (linking projectors to observables) give, for every $\alpha$,
\begin{align*}
  &VP_{W_c,b}^{\Dec(\alpha)_c}
  -(\tilde\tau_{W_c,b}^{\Dec(\alpha)_c}\otimes\id)V=\E_{s\in\zo^{|P_b|}}(-1)^{\Dec(\alpha)_c\cdot s}
  \left(VW_c(h_b(s))-(\sigma_{W_c}(h_b(s))\otimes\id)V\right).
\end{align*}
Applying \cref{lemma:state-dependent-norm-properties}(\ref{prop:state-dependent-norm-properties,vi}), followed by linearity in the state, we obtain
\begin{align*}
  \sum_\alpha&
  \norm{VP_{W_c,b}^{\Dec(\alpha)_c}
  -(\tilde\tau_{W_c,b}^{\Dec(\alpha)_c}\otimes\id)V}_{\pea{P_b}}^2\\
  &\leq \E_{s\in\zo^{|P_b|}}\sum_\alpha
  \norm{VW_c(h_b(s))-(\sigma_{W_c}(h_b(s))\otimes\id)V}_{\pea{P_b}}^2\\
  &=\E_{s\in\zo^{|P_b|}}
  \norm{VW_c(h_b(s))-(\sigma_{W_c}(h_b(s))\otimes\id)V}_{\pe{P_b}}^2\\
  &\leq O(\eps+\ea),
\end{align*}
where the last inequality follows from \cref{lemma:EPBT-pure-basis} by taking $q=P_b$ and choosing $\mu$ to be the distribution of $h_b(s)$ for uniform $s\in\zo^{|P_b|}$. This allows us to push the aggregated Bob projector through the isometry.
\par
\medskip
Since the Bell-basis test succeeds with probability $1-O(\eps)$, we have for all $c\in\zo$
\begin{equation}\label{eqn:epr-succ}
  \sum_{\alpha}\Tr\bigl[P_{W_c,b}^{\Dec(\alpha)_c}\,\pea{P_b}\bigr] \geq 1-O(\eps).
\end{equation}

Since $P_{W_c,b}^v$ is a projector for every $v$, \cref{eqn:epr-succ} implies
\[
  \sum_\alpha\norm{P_{W_c,b}^{\Dec(\alpha)_c}-\id}_{\pea{P_b}}^2
  =1-\sum_\alpha\Tr\bigl[P_{W_c,b}^{\Dec(\alpha)_c}\pea{P_b}\bigr]
  \leq O(\eps).
\]
Combining the previous two bounds with a triangle inequality and using that $V$ is an isometry yields, for all $c\in\zo$,
\[
\sum_\alpha\norm{
    (\tilde\tau^{\Dec(\alpha)_c}_{W_c,b} \otimes\id) - \id}_{V\pea{P_b}V^\dag}^2\leq O(\eps+\ea).\numberthis\label{eqn:partial-bell-projector}
\]
By the tensor product structure (everything factorizes into pairs) we can easily verify that
\[
\tilde\tau^{\Dec(\alpha)_0}_{X,b}\tilde\tau^{\Dec(\alpha)_1}_{Z,b}=\Phi_b^{\Dec(\alpha)}.\numberthis\label{eqn:bell-projector-observation}
\]
Concretely, for each pair $(j,k)\in P_b$, one checks that
\[
  \left(\tilde\tau^{u}_{X,b}\right)_j\left(\tilde\tau^{v}_{Z,b}\right)_k = \paren{\sum_{l\oplus m=u}
\left(\tau_{X}^{l}\right)_j\left(\tau_{X}^{m}\right)_k} \cdot \paren{\sum_{l\oplus m=v}\left(\tau_{Z}^{l}\right)_j\left(\tau_{Z}^{m}\right)_k}=\Phi_{j,k}^{u,v},
\]
and thus the full tensor product of all pairs yields $\tilde\tau_{X,b}^\alpha\cdot\tilde\tau_{Z,b}^\beta = \Phi_b^{(\alpha,\beta)}$, a
multi-qubit Bell-basis projection with pairing $P_b$. We can now use the two cases of \cref{eqn:partial-bell-projector} with \cref{eqn:bell-projector-observation} to prove the lemma statement
\begin{align*}
    \sum_{\alpha}&
  \norm{
    \paren{\Phi_b^{\Dec(\alpha)}\otimes \id} - \id
  }_{V\pea{P_b} V^\dag}^2\\
  &= \sum_{\alpha}
  \norm{
    \paren{\tilde\tau^{\Dec(\alpha)_0}_{X,b}\tilde\tau^{\Dec(\alpha)_1}_{Z,b}\otimes \id} - \id
  }_{V\pea{P_b} V^\dag}^2\\
  &\leq 2\sum_{\alpha}
  \norm{
    \paren{\tilde\tau^{\Dec(\alpha)_0}_{X,b}\ot\id}\paren{\paren{\tilde\tau^{\Dec(\alpha)_1}_{Z,b}\otimes \id} - \id}
  }_{V\pea{P_b} V^\dag}^2 \\
  &\quad+ 2\sum_{\alpha}
  \norm{
   \paren{\tilde\tau^{\Dec(\alpha)_0}_{X,b}\otimes \id} - \id
  }_{V\pea{P_b} V^\dag}^2\\
  &\leq 2\sum_{\alpha}
  \norm{
    \tilde\tau^{\Dec(\alpha)_0}_{X,b}\ot\id}^2_\infty\norm{\paren{\tilde\tau^{\Dec(\alpha)_1}_{Z,b}\otimes \id} - \id
  }_{V\pea{P_b} V^\dag}^2\\
  &\quad + 2\sum_{\alpha}
  \norm{
   \paren{\tilde\tau^{\Dec(\alpha)_0}_{X,b}\otimes \id} - \id
  }_{V\pea{P_b} V^\dag}^2\\
  &\leq O(\eps+\ea).
\end{align*}
\end{proof}

\paragraph*{Sign consistency test}\label{subsec:corr-analysis}
We will now analyze the sign consistency subtest (\cref{protocol:cliff-test} item 3), which ensures that Bob's reported outcomes are consistent with the underlying entanglement on his side, which Alice generated by collapsing her qubits pairwise into one of four Bell states.
\par
\medskip
Because the correlations either remain in the same observable ($G$) or cross observables ($G$ and $F$), we have two consistency subtests. We will first analyze the single-observable case.
\par 
\medskip
For ease of analysis we define the following set
\begin{definition}[Accepted answer sets]\label{def:accepted-answer-set}
For a decrypted Alice answer $(u,v)$ (for the Bell measurement question) and pairing index $b\in\zo$, we define
\[
  S_{(u,v),b} \coloneqq
  \begin{cases}
    \bigl\{a\in\zo^n : \forall (j,k)\in P_b,\;
           v_{\ceil{j/2}} = 1 \implies a_j\oplus a_k = u_{\ceil{j/2}}\bigr\}
    & \text{if } v \neq 0^{|P_b|},\\
    \zo^n & \text{if } v = 0^{|P_b|}.
  \end{cases}
\]
and
\[
  R_{(u,v),b} \coloneqq
  \begin{cases}
    \bigl\{a\in\zo^n : \forall (j,k)\in P_b,\;
           v_{\ceil{j/2}} = 0 \implies a_j\oplus a_k \neq u_{\ceil{j/2}}\bigr\}
    & \text{if } v \neq 1^{|P_b|},\\
    \zo^n & \text{if } v = 1^{|P_b|}.
  \end{cases}
\]
\end{definition}
Here $v$ serves as mask, which indicates on which pairs of the chosen pairing ($P_0$ or $P_1$, depending on $b$) a specific parity constraint is checked. The set contains only bit-strings which satisfy the parity constraints on all pairs that are checked and if the mask is trivial, the set contains all possible bit-strings. To justify this choice of sets we have to look at the behavior of our Clifford operators on the four Bell states (by \cref{eqn:transpose-trick}).
\begin{equation}\label{eqn:cliff-bell-correlations}
\begin{aligned}
-\sigma_G\otimes\sigma_F\ket{\Phi^{00}} = \ket{\Phi^{00}},\quad\sigma_G\otimes\sigma_F\ket{\Phi^{10}} = \ket{\Phi^{10}},\\
\quad\sigma_G\otimes\sigma_G\ket{\Phi^{01}} = \ket{\Phi^{01}},\quad -\sigma_G\otimes\sigma_G\ket{\Phi^{11}} = \ket{\Phi^{11}}.
\end{aligned}
\end{equation}
From there the pattern is clear: if the $ZZ$ outcome (corresponding to $v$) is 1, we are in the single-observable case and the parity of the pair directly corresponds to the $XX$ outcome. If the $ZZ$ outcome is 0, we are in the mixed-observable case and the parity of the pair is opposite of  the $XX$ outcome. The $S_{(u,v),b}$ set encodes these constraints in the single-observable case and the $R_{(u,v),b}$ set does the same in the mixed-observable setting.

\begin{definition}[$G$ and $GF$ pair projectors]\label{def:pair-projectors-G-GF}
For a decrypted Alice answer $(u,v)$ (for the Bell measurement question) and pairing index $b\in\zo$, we define
\[
  \tilde\tau^{(u,v)}_{G,b}=\sum_{a\in S_{(u,v),b}}\sum_{k\in\zo^n}\paren{\tau_G^{a\oplus k}\oplus\overline{\tau}_G^{a\oplus k}}\ot\ket{k}\bra{k}\ot\id
\]
and
\begin{align*}
  \tilde\tau^{(u,v)}_{GF,b}=\sum_{a\in R_{(u,v),b}}\sum_{k\in\zo^n}\paren{\tau_{GF}^{a\oplus k}\oplus\overline{\tau}_{GF}^{a\oplus k}}\ot\ket{k}\bra{k}\ot\id&
\end{align*}
where 
\[
\tau_{GF}^{a} = \prod_{i\in I}\left(\tau_{G}^{a_i}\right)_i\cdot\prod_{j\in J}\left(\tau_{F}^{a_j}\right)_j
\]
with $I=\{2i-1 \,|\, i\in[\floor{(n+1)/2}]\}$ and $J=\{2i \,|\, i\in[\floor{n/2}]\}$.
\end{definition}

\begin{lemma}\label{lemma:pair-projector-eigenstates}
For any computationally efficient prover modeled as in \Cref{section:modeling} that succeeds with probability $1-\eps$ in Items 1 to 4 of the compiled Clifford test, obtained from \Cref{protocol:cliff-test}, it holds that for all $b\in\zo$:
\[
  \sum_{\alpha}
  \norm{
    \tilde\tau^{\Dec(\alpha)}_{G,b} - \id
  }_{V\pea{P_b} V^\dag}^2
  \leq O(\eps+\ea),
  \]
  and
  \[
  \sum_{\alpha}
  \norm{
    \tilde\tau^{\Dec(\alpha)}_{GF,b} - \id
  }_{V\pea{P_b} V^\dag}^2
  \leq O(\eps+\ea),
\]
where $V$ is the isometry from \cref{def:clifford-isometry}.
\end{lemma}
\begin{proof}
Fix $b\in\zo$. From success in the Clifford test and \cref{lemma:EPBT-pure-basis}, taking $q=P_b$, we have for all distributions $\mu$ over $\zo^n$
\[
  \E_{a\sim\mu} \norm{VG(a) - (\sigma_{G}(a)\oplus\overline{\sigma}_{G}(a))\otimes\Lambda_{G}(a))V}_{\pe{P_b}}^2 \leq O(\eps+\ea),
\]
and similarly for $F(a)$, where $\Lambda_F(a)=\Lambda_G(a)=\sum_k (-1)^{a\cdot k}\ket{k}\bra{k}\ot\id$.
Define
\[
  \tilde\tau_G^a\coloneqq
  \sum_k(\tau_{G}^{a\oplus k}\oplus\overline{\tau}_{G}^{a\oplus k})
  \otimes\ket{k}\bra{k}\ot\id.
\]
By Parseval (\cref{cor:parseval} with $\mu$ a uniform expectation over $a\in\zo^n$):
\begin{equation}\label{eqn:parseval-g}
  \sum_a \norm{V G^a - \tilde\tau_G^aV}_{\pe{P_b}}^2 \leq O(\eps+\ea).
\end{equation}
By success in the pure-vs-mixed basis test, we know from \cref{lemma:mixed-vs-pure-guarantee} that for any distribution $\mu$ on $\zo^n$
\[
\E_{a\sim\mu}\norm{\tilde W(a)-G(\exten{a}{I})F(\exten{a}{J})}^2_{\pe{P_b}}\leq O(\eps)+\ea,
\]
where $I\subset[n]$ is the set of odd indices, $J\subset [n]$ is the set of even indices and we define $\tilde W\in\{G,F\}^n$ as the string of alternating $G$ and $F$ symbols, starting with $G$. Chaining this with the Clifford test guarantee and the automatic self-consistency of $F$, a routine calculation yields
\[
\E_{a\sim\mu} \norm{V\tilde W(a) - ((\sigma_{G}(\exten{a}{I})\sigma_{F}(\exten{a}{J})\oplus\overline{\sigma}_{G}(\exten{a}{I})\overline{\sigma}_{F}(\exten{a}{J}))\otimes\Lambda_{G}(a))V}_{\pe{P_b}}^2 \leq O(\eps+\ea).
\]
Let $P^a=\prod_{i\in I}\tau_{G,i}^{a_i}\cdot\prod_{j\in J}\tau_{F,j}^{a_j}$ and define
\[
  \tilde\tau_{GF}^a\coloneqq
  \sum_k(P^{a\oplus k}\oplus\overline{P}^{a\oplus k})
  \otimes\ket{k}\bra{k}\ot\id.
\]
Applying Parseval again gives
\[
\sum_{a}\norm{V\tilde W^a - \tilde\tau_{GF}^aV}_{\pe{P_b}}^2
\leq O(\eps+\ea).\numberthis\label{eqn:parseval-gf}
\]
The families $\{\tilde\tau_G^a\}_{a\in\zo^n}$ and $\{\tilde\tau_{GF}^a\}_{a\in\zo^n}$ are PVMs. For the first family this follows from the fact that $\{\tau_G^a\}_a$ is a PVM. For the second family, it follows from the fact that $\{P^a\}_a$ is a PVM (the $G$ and $F$ projectors act on disjoint qubits).
\par 
\medskip

Since the sign consistency test succeeds with probability $1-O(\eps)$, we have
\begin{equation}\label{eqn:consistency-succ}
  \sum_{\alpha}\sum_{a\in S_{\Dec(\alpha),b}}\Tr\bigl[G^{a}\, \pea{P_b}\bigr] \geq 1-O(\eps),\quad\text{ and }\quad  \sum_{\alpha}\sum_{a\in R_{\Dec(\alpha),b}}\Tr\bigl[\tilde W^{a}\, \pea{P_b}\bigr] \geq 1-O(\eps),
\end{equation}
where $S_{(u,v),b}$ and $R_{(u,v),b}$ are defined as in \cref{def:accepted-answer-set}. 
\par
\medskip
We will now transition to bounding the support of the state on the complementary, rejected, outcomes instead of summing the distance to the identity over all accepted outcomes. By \cref{def:pair-projectors-G-GF},
\[
  \tilde\tau_{G,b}^{\Dec(\alpha)}
  =\sum_{a\in S_{\Dec(\alpha),b}}\tilde\tau_G^a.
\]
Since $\{\tilde\tau_G^a\}_a$ is a PVM, orthogonality and completeness give
\begin{align*}
  \sum_\alpha\norm{\tilde\tau_{G,b}^{\Dec(\alpha)}-\id}_{V\pea{P_b}V^\dag}^2
  &=\sum_\alpha\sum_{a\notin S_{\Dec(\alpha),b}}
    \norm{\tilde\tau_G^aV}_{\pea{P_b}}^2\\
  &\leq 2\sum_\alpha\sum_{a\notin S_{\Dec(\alpha),b}}
    \norm{\tilde\tau_G^aV-VG^a}_{\pea{P_b}}^2
    +2\sum_\alpha\sum_{a\notin S_{\Dec(\alpha),b}}
    \norm{VG^a}_{\pea{P_b}}^2\\
  &\leq 2\sum_a\norm{\tilde\tau_G^aV-VG^a}_{\pe{P_b}}^2\\
  &\quad+2\left(1-\sum_\alpha\sum_{a\in S_{\Dec(\alpha),b}}
    \Tr\bigl[G^a\pea{P_b}\bigr]\right)\\
  &\leq O(\eps+\ea).
\end{align*}
Here the second-to-last inequality uses linearity in the state, $V^\dag V=\id$, and the fact that the summands are nonnegative; the final inequality follows from \cref{eqn:parseval-g,eqn:consistency-succ}. This proves the first claim.
\par
\medskip
For the mixed-observable case, we analogously have by \cref{def:pair-projectors-G-GF}
\[
\tilde\tau_{GF,b}^{\Dec(\alpha)}
=\sum_{a\in R_{\Dec(\alpha),b}}\tilde\tau_{GF}^a.
\]
Using the complementary outcomes of the PVM $\{\tilde\tau_{GF}^a\}_a$ gives
\begin{align*}
  \sum_\alpha\norm{\tilde\tau_{GF,b}^{\Dec(\alpha)}-\id}_{V\pea{P_b}V^\dag}^2
  &=\sum_\alpha\sum_{a\notin R_{\Dec(\alpha),b}}
    \norm{\tilde\tau_{GF}^aV}_{\pea{P_b}}^2\\
  &\leq 2\sum_\alpha\sum_{a\notin R_{\Dec(\alpha),b}}
    \norm{\tilde\tau_{GF}^aV-V\tilde W^a}_{\pea{P_b}}^2
    +2\sum_\alpha\sum_{a\notin R_{\Dec(\alpha),b}}
    \norm{V\tilde W^a}_{\pea{P_b}}^2\\
  &\leq 2\sum_a\norm{\tilde\tau_{GF}^aV-V\tilde W^a}_{\pe{P_b}}^2\\
  &\quad+2\left(1-\sum_\alpha\sum_{a\in R_{\Dec(\alpha),b}}
    \Tr\bigl[\tilde W^a\pea{P_b}\bigr]\right)\\
  &\leq O(\eps+\ea),
\end{align*}
where the last line follows from \cref{eqn:parseval-gf,eqn:consistency-succ}. This proves the second claim.
\end{proof}

\paragraph*{Sign ambiguity characterization}
It remains to fix the remaining global sign choice on the $F$ and $G$ observables to $+$, which is the whole reason for having the correlation test. The design was informed by the following observation: by performing Bell basis measurements on her qubits, Alice performs entanglement swapping, transforming the inter-prover entanglement into intra-prover entanglement on Bob's side, if we consider the nonlocal view for a moment. 
\par
\medskip
From Alice's reported measurement outcomes, the verifier knows which Bell state every pair on Bob's side collapsed into. This information is crucial, because by \cref{eqn:cliff-bell-correlations} the verifier knows exactly which correlations should be present if Bob, for example, performs a $G$ measurement on every qubit separately. If within a checked pair Bob would be using $-\sigma_G$ and $\sigma_G$, the verifier could detect this through a correlation mismatch. This is also why we need to use the brick-wall pairings, since with one type of pairing alone, Bob could still use inconsistent sign choices across distinct pairs, testing these staggered pairings allows the verifier to force Bob to use a consistent sign choice for all of its $G$ observables, either all $+$ or all $-$.
\par 
\medskip
The two lemmas above are critical in concluding that the prover has to be consistent in his sign choice. The first tells us that Alice's post-measurement state is really projected into a Bell basis eigenstate; the second identifies which operations Bob implements when asked to perform an all $G$ or a mixed $G$ and $F$ measurement, and that he passes the verifier's consistency checks. Combining these facts, we can use \Cref{eqn:cliff-bell-correlations} to extract information about the sign choice of the prover.

\begin{lemma}\label{lemma:G-factor}
For any $b\in\zo$ and $(u,v)\in\zo^{|P_b|}\times\zo^{|P_b|}$,
\[
  \tilde\tau_{G,b}^{(u,v)}\bigl(\Phi_b^{(u,v)}\otimes\id\bigr)
  =
  \paren{\sum_{k\in K_{v,b}} \id\otimes\ketbra{k}{k}\ot\id}
  \bigl(\Phi_b^{(u,v)}\otimes\id\bigr),
\]
where
$K_{v,b} \coloneqq \bigl\{a\in\zo^n : a_{j} = a_{k}\;\;\forall\, (j,k)\in P_b \text{ where } v_{\ceil{j/2}}=1\bigr\}$.
\end{lemma}
\begin{proof}
Since everything factors into tensor products of pairs, we can analyze the action of $\tilde\tau_{G,b}^{(u,v)}$ pair by pair.
Fix a pair $(j,l)\in P_b$ and let $s = u_{\ceil{j/2}}$, $t = v_{\ceil{j/2}}$.
On positions where $v$ vanishes, the strings in $S_{(u,v),b}$ are unconstrained (because the qubits got projected into a Bell state which requires mixed-observable measurement) and the sum over that set makes projectors sum to identity; we are thus only interested in the case $t\neq 0$ (there is a parity constraint on this pair). Let $a\in S_{(u,v),b}$ be any bit-string which we are summing over, for which we can define $c=a_j\oplus k_j$ and $d=a_l\oplus k_l$.
Using \cref{eqn:cliff-bell-correlations}, we compute:
\begin{align*}
  (\tau_{G,j}^c\otimes\tau_{G,l}^d)\ket{\Phi^{s1}}\bra{\Phi^{s1}}
  &= \frac{1}{4}\bigl(
       (1+(-1)^{c+d+s})\id\otimes\id
       + ((-1)^c + (-1)^{d+s})\sigma_{G,j}\otimes\id
     \bigr)\ket{\Phi^{s1}}\bra{\Phi^{s1}}\\
  &= \ind{c\oplus d\oplus s = 0}\;
     (\tau_{G,j}^c\otimes\id)\,\ketbra{\Phi^{s1}}{\Phi^{s1}}\\
  &= \ind{k_j\oplus k_l = 0}\;
     (\tau_{G,j}^c\otimes\id)\,\ketbra{\Phi^{s1}}{\Phi^{s1}}.
\end{align*}
where the last line follows from the fact that membership in $S_{(u,v),b}$ guarantees that $a_j\oplus a_l=s$. The remaining $G$ projector on qubit $j$ vanishes once we sum over all valid $a$, since only the parity of the pair is constrained, which leaves two options for the value of $c$, so the projectors sum to the identity. The same guarantee holds for the complex conjugate block and in both cases, every pair yields an indicator for the parity constraint on $k$. Tensoring over all pairs in $P_b$ and recalling that $\overline{\sigma}_G = -\sigma_F$, we can conclude:
\begin{align*}
\sum_{a\in S_{(u,v),b}}&\sum_{k\in\zo^n}\paren{\tau_G^{a\oplus k}\Phi_b^{(u,v)}\oplus\tau_F^{a\oplus k\oplus 1}\Phi_b^{(u,v)}}\ot\ket{k}\bra{k}\ot\id \\
&= \sum_{k\in\zo^n}\prod_{\substack{(j,l)\in P_b:\\ v_{\ceil{j/2}}=1}}\ind{k_j\oplus k_l = 0}\Phi_b^{(u,v)}\otimes\ket{k}\bra{k}\ot\id,
\end{align*}
which is equal to $\bigl(\sum_{k\in K_{v,b}} \id\otimes\ketbra{k}{k}\ot\id\bigr)
  \bigl(\Phi_b^{(u,v)}\otimes\id\bigr)$ and thus completes the proof.
\end{proof}

The identical computation with $\sigma_F$ replacing $\sigma_G$ (and conjugated Bell states $\ket{\Phi^{u_0}}$
replaced by $\ket{\tilde\Phi^{u_0}}$) yields:

\begin{lemma}\label{lemma:GF-factor}
For any $b\in\zo$ and $(u,v)\in\zo^{|P_b|}\times\zo^{|P_b|}$,
\[
  \tilde\tau_{GF,b}^{(u,v)}\bigl(\Phi_b^{(u,v)}\otimes\id\bigr)
  =
  \paren{\sum_{k\in K_{\overline{v},b}} \id\otimes\ketbra{k}{k}\ot\id}
  \bigl(\Phi_b^{(u,v)}\otimes\id\bigr),
\]
where
$K_{v,b} \coloneqq \bigl\{a\in\zo^n : a_{j} = a_{k}\;\;\forall\, (j,k)\in P_b \text{ where } v_{\ceil{j/2}}=1\bigr\}$.
\end{lemma}
\begin{proof}
  The argument is identical to \cref{lemma:G-factor}, except that we now consider a tensor product of two different observables per pair. Here we use $R_{(u,v),b}$, and unconstrained pairs are indicated by a vanishing entry in $v$, which explains the appearance of $\overline{v}$.
\end{proof}

\begin{lemma}\label{lemma:k-support}
For any computationally efficient prover modeled as in \Cref{section:modeling} that succeeds with probability $1-\eps$ in Items 1 to 4 of the compiled Clifford test, obtained from \Cref{protocol:cliff-test}, it holds that for all $b\in\{0,1\}$:
\[
  \sum_\alpha
  \sum_{k\notin K_{\Dec(\alpha)_1,b}}\norm{
    \id\otimes\ketbra{k}{k}\ot\id
  }_{V\pea{P_b} V^\dag}^2
  \leq O(\eps+\ea),
\]
and
\[
  \sum_\alpha
  \sum_{k\notin K_{\overline{\Dec(\alpha)}_1,b}}\norm{
    \id\otimes\ketbra{k}{k}\ot\id
  }_{V\pea{P_b} V^\dag}^2
   \leq O(\eps+\ea),
\]
where
$K_{v,b} \coloneqq \bigl\{a\in\zo^n : a_{j} = a_{k}\;\;\forall\, (j,k)\in P_b \text{ where } v_{\ceil{j/2}}=1\bigr\}$ and $V$ is the isometry from \cref{def:clifford-isometry}.
\end{lemma}
\begin{proof}
We recall our previous results
\begin{enumerate}
    \item $\sum_{\alpha}
  \norm{
    \paren{\Phi_b^{\Dec(\alpha)}\otimes \id} - \id
  }_{V\pea{P_b} V^\dag}^2
  \leq O(\eps+\ea)$ (by \cref{lemma:epr-bell}).\label{item:bell-proj}
  \item $\sum_{\alpha}
  \norm{
    \tilde\tau^{\Dec(\alpha)}_{G,b} - \id
  }_{V\pea{P_b} V^\dag}^2
  \leq O(\eps+\ea)$ (by \cref{lemma:pair-projector-eigenstates}).\label{item:G-id}
  \item $\sum_{\alpha}
  \norm{
    \tilde\tau^{\Dec(\alpha)}_{GF,b} - \id
  }_{V\pea{P_b} V^\dag}^2
  \leq O(\eps+\ea)$ (by \cref{lemma:pair-projector-eigenstates}).\label{item:GF-id}
\end{enumerate}
For $W\in\{G,GF\}$, we know that $\tilde\tau_{W,b}^{(u,v)}$ is a projector, thus $\norm{\tilde\tau_{W,b}^{(u,v)}}_\infty \leq 1$, and by a triangle inequality:
\begin{align*}
  \sum_\alpha&
  \norm{
    (\tilde\tau_{W,b}^{\Dec(\alpha)} - \id)(\Phi_b^{\Dec(\alpha)}\otimes\id)
  }_{V\pea{P_b} V^\dag}^2\\
  &\leq
  2\, \sum_\alpha
    \norm{\tilde\tau_{W,b}^{\alpha} - \id}^2_\infty
    \norm{\Phi_b^{\Dec(\alpha)}\otimes\id - \id}^2_{V\pea{P_b} V^\dag}
  + 2\,\sum_\alpha
    \norm{\tilde\tau_{W,b}^{\alpha} - \id}^2_{V\pea{P_b} V^\dag}\\
  &\leq O(\eps+\ea).
\end{align*}
Where we used \cref{item:bell-proj} and \cref{item:G-id} or \ref{item:GF-id}. From \cref{lemma:G-factor} and the above:
\begin{align*}
\sum_\alpha& \norm{\paren{\sum_{k\notin K_{\Dec(\alpha)_1,b}} \id\otimes\ketbra{k}{k}\ot\id}\paren{\Phi_b^{\Dec(\alpha)}\otimes\id}}^2_{V\pea{P_b}V^\dag}\\
&=\sum_\alpha\norm{\paren{\sum_{k\in K_{\Dec(\alpha)_1,b}} \id\otimes\ketbra{k}{k}\ot\id-\id}
  \paren{\Phi_b^{\Dec(\alpha)}\otimes\id}}_{V\pea{P_b}V^\dag}^2\\
  &=\sum_\alpha
  \norm{
    (\tilde\tau_{G,b}^{\Dec(\alpha)} - \id)(\Phi_b^{\Dec(\alpha)}\otimes\id)
  }_{V\pea{P_b} V^\dag}^2\\
  &\leq O(\eps+\ea).
\end{align*}
Using \cref{lemma:GF-factor} we similarly obtain:
\[
\sum_\alpha \norm{\paren{\sum_{k\notin K_{\overline{\Dec(\alpha)}_1,b}} \id\otimes\ketbra{k}{k}\ot\id}\paren{\Phi_b^{\Dec(\alpha)}\otimes\id}}^2_{V\pea{P_b}V^\dag}\leq O(\eps+\ea)
\]

Now we just need to drop the Bell projector again, which we can do using \cref{item:bell-proj} and the same triangle inequality based argument that we used to introduce it (again the pre-factor is a projector with bounded operator norm). The orthogonality of the projectors then immediately completes the proof.
\end{proof}

\begin{lemma}[Mismatch set characterization]\label{lemma:mismatch-set}
Let $M_b(a) \coloneqq \{(j,k)\in P_b : a_j \neq a_k\}$ denote the set of pairs in $P_b$ on which $a$ has mismatching bits.
For any $v\in\zo^{|P_b|}$ we have
\[
  \{a\in\zo^n:M_b(a)\neq\emptyset\}\subseteq
  \{a\in\zo^n:a\notin K_{v,b}\} \cup \{a\in\zo^n:a\notin K_{\overline{v},b}\}.
\]
\end{lemma}
\begin{proof}
Take any $a\in\zo^n$ with $M_b(a)\neq\emptyset$, then there exists a pair $(j,k)\in P_b$ for which $a_j\neq a_k$ and thus $a$ is either included in $\{a\notin K_{v,b}\}$ or in $\{a\notin K_{\overline{v},b}\}$, since $v\oplus\overline{v}=1^{|P_b|}$, thus every pair is checked in one of the two sets.
\end{proof}

\begin{lemma}\label{lemma:sign-consistency}
Let $P^*$ be any computationally efficient prover modeled as in \Cref{section:modeling} that succeeds with probability $1-\eps$ in Items 1 to 4 of the compiled Clifford test, obtained from \Cref{protocol:cliff-test}. Then for all $q\in\Qa$ there exists a cryptographically small function $\ea$ such that
\[
  \sum_{k\in\zo^n\setminus\{0^n,1^n\}}
  \norm{\id\otimes\ketbra{k}{k}\ot\id}_{V\peq V^\dag}^2
  \leq O(\eps+\ea),
\]
where $V$ is the isometry from \cref{def:clifford-isometry}.
\end{lemma}
\begin{proof}
Fix $b\in\zo$. By \cref{lemma:mismatch-set}, we can set $v=\Dec(\alpha)_1$ and conclude that
\begin{align*}
  \sum_{k:\,M_b(k)\neq\emptyset}&
  \norm{\id\otimes\ketbra{k}{k}\ot\id}_{V\pea{P_b}V^\dag}^2\\
  &\leq
  \sum_{k\notin K_{\Dec(\alpha)_1,b}}
  \norm{\id\otimes\ketbra{k}{k}\ot\id}_{V\pea{P_b}V^\dag}^2
  +
  \sum_{k\notin K_{\overline{\Dec(\alpha)}_1,b}}
  \norm{\id\otimes\ketbra{k}{k}\ot\id}_{V\pea{P_b}V^\dag}^2.
\end{align*}

This holds for every valid $\alpha$ separately, so we can sum both sides over all $\alpha$:
\begin{align*}
  \sum_\alpha&
  \sum_{k:\,M_b(k)\neq\emptyset}
  \norm{\id\otimes\ketbra{k}{k}\ot\id}_{V\pea{P_b}V^\dag}^2\\
  &\leq
  \sum_\alpha
  \sum_{k\notin K_{\Dec(\alpha)_1,b}}
  \norm{\id\otimes\ketbra{k}{k}\ot\id}_{V\pea{P_b}V^\dag}^2
  +
  \sum_\alpha
  \sum_{k\notin K_{\overline{\Dec(\alpha)}_1,b}}
  \norm{\id\otimes\ketbra{k}{k}\ot\id}_{V\pea{P_b}V^\dag}^2\\
  &\leq O(\eps+\ea),
\end{align*}
where the last line follows from the two bounds in \cref{lemma:k-support}. The set $\{k : M_b(k)\neq\emptyset\}$ does not depend on $\alpha$, and the state-dependent norm is linear in the state, so the left-hand side equals
\begin{align*}
  \sum_{k:\,M_b(k)\neq\emptyset}
  \sum_\alpha\norm{\id\otimes\ketbra{k}{k}\ot\id}_{V\pea{P_b}V^\dag}^2
  &=
  \sum_{k:\,M_b(k)\neq\emptyset}
  \norm{\id\otimes\ketbra{k}{k}\ot\id}_{V\pe{P_b}V^\dag}^2\\
  &=\norm{\id\ot\Pi_b\ot\id}_{V\pe{P_b}V^\dag}^2,
\end{align*}
with 
\[
  \Pi_b\deq \sum_{k:\,M_b(k)\neq\emptyset}\ketbra{k}{k}.
\]
Expanding out the state-dependent norm, defining the POVM element $M_b\deq V^\dag\Pi_b V$ and then swapping out the state by \cref{lemma:povm-indistinguishability}, we obtain:
\[
  \norm{\id\ot\Pi_b\ot\id}_{V\peq V^\dag}^2
  \leq O(\eps+\ea),
  \qquad \forall\;b\in\{0,1\}.
\]
The only $k\in\zo^n$ for which both $M_0(k)=\emptyset$ and $M_1(k)=\emptyset$
are $0^n$ and $1^n$, thus $\{k : M_0(k)\neq\emptyset\}\cup\{k : M_1(k)\neq\emptyset\} = \zo^n\setminus\{0^n,1^n\}$,
and we can conclude:
\begin{align*}
  \sum_{k\in\zo^n\setminus\{0^n,1^n\}}&
  \norm{\id\otimes\ketbra{k}{k}\ot\id}_{V\peq V^\dag}^2\\
  &\leq
  \sum_{k:\,M_0(k)\neq\emptyset}
  \norm{\id\otimes\ketbra{k}{k}\ot\id}_{V\peq V^\dag}^2
  +
  \sum_{k:\,M_1(k)\neq\emptyset}
  \norm{\id\otimes\ketbra{k}{k}\ot\id}_{V\peq V^\dag}^2\\
  &\leq O(\eps+\ea).\qedhere
\end{align*}
\end{proof}

\begin{corollary}\label{cor:sign-consistency}
    Let $P^*$ be any computationally efficient prover modeled as in \Cref{section:modeling} that succeeds with probability $1-\eps$ in Items 1 to 4 of the compiled Clifford test, obtained from \Cref{protocol:cliff-test}. Then for all $q\in\Qa$ there exists a cryptographically small function $\ea$ such that:
    \[
    \norm{\id\ot(\Lambda_G(a)-\Delta_G(a))}^2_{V\peq V^\dag}\leq O(\eps+\ea).
    \]
    where $\Lambda_G(a)=\sum_{k\in\zo^n}(-1)^{a\cdot k}\ketbra{k}{k}\ot\id$ and $\Delta_G(a)=\sigma_{Z}^{|a|}\ot\id$, here the Pauli $Z$ operator acts only on the first qubit of the $k$ register and $V$ is the isometry from \cref{def:clifford-isometry}.
\end{corollary}
\begin{proof}
    In the PSD order we have
    \begin{align*}
        \id\ot (\Lambda_G(a)-\Delta_G(a))^2\preceq 4\sum_{k\in\zo^n\setminus\{0^n,1^n\}}\id\ot\ketbra{k}{k}\ot\id,
    \end{align*}
    which immediately allows us to conclude that
    \[
    \norm{\id\ot(\Lambda_G(a)-\Delta_G(a))}^2_{V\peq V^\dag}\leq 4\sum_{k\in\zo^n\setminus\{0^n,1^n\}}\norm{\id\ot\ketbra{k}{k}\ot\id}^2_{V\peq V^\dag}.
    \]
    Using \cref{lemma:sign-consistency} we obtain the desired bound.
\end{proof}

Using \cref{cor:sign-consistency} and the fact that the prover passes the Clifford subtest, we can conclude the following
\begin{lemma}\label{lemma:cliff-sign-consistent}
    For any computationally efficient prover modeled as in Section \ref{section:modeling} that wins with probability $1 - \eps$ in Items 1 to 4 of the compiled Clifford test $\textprotocol{cliff}(X,Y,Z,F,G,n)$, obtained from \cref{protocol:cliff-test}, there exists a complex Hilbert space $\cH'$ of finite dimension $d' = d|C_n|^2$ and an efficient isometry $V: \cH \to \cH'$ (from \cref{def:clifford-isometry}) such that for all $W \in \{X,Y,Z,G,F\}$, all distributions $\mu$ on $\bits^n$ and all $q\in\Qa$ there exists a cryptographically small function $\ea$ such that
    \begin{align*}
    \E_{a\sim\mu} \norm{ V W(a) -((\sigma_W(a)\oplus\overline{\sigma}_W(a)) \ot\Delta_W(a)) V }_\peq^2 \leq O(\eps+\ea),
    \end{align*}
    where $\Delta_X(a)=\Delta_Y(a)=\Delta_Z(a)=\id$ and $\Delta_F(a)=\Delta_G(a)=\sigma_{Z}^{|a|}\ot\id$ for all $a\in\zo^n$.
\end{lemma}
\begin{proof}
    Combining \cref{lemma:EPBT-pure-basis} and \cref{cor:sign-consistency} by a triangle inequality.
\end{proof}

Now that the prover's freedom to choose a sign has been reduced to a global sign, we can execute a CHSH game, which by the form of the $F$ and $G$ observables in terms of the $X$ and $Y$ observables, is ideally suited to distinguish between a positive or negative global sign choice.
 
\begin{lemma}\label{lemma:chsh}
For any computationally efficient prover modeled as in \Cref{section:modeling} that succeeds with probability $\omega^*-\eps$ in Items 1 to 5 of the compiled Clifford test, obtained from \Cref{protocol:cliff-test}, it holds under the isometry from \cref{def:clifford-isometry} that for all $q\in\Qa$, all $W\in\{X,Y,Z,F,G\}$ and all distributions $\mu$ over $\zo^n$ there exists a cryptographically small function $\ea$ such that:
\[
  \E_{a\sim\mu}\norm{VW(a) - ((\sigma_W(a)\oplus\overline{\sigma}_W(a))\ot\id)V}_{\peq}^2
  \leq O(\eps+\ea).
\]
Here $\omega^*=\tfrac{1}{5}(4+\cos^2(\frac{\pi}{8}))$ is the optimal winning probability of Items 1 to 5\footnote{It should be possible to bring completeness exponentially close to 1 by exploiting the fact that we are effectively running $n$ CHSH tests in parallel with the large-answer measurements.}: an honest prover achieves it up to a cryptographically small loss, and no efficient prover can exceed it by more than $\ea$ (cf.\ \cref{eqn:chsh-pi}).
\end{lemma}
\begin{remark}
    The CHSH test of \cref{protocol:cliff-test} is designed under the assumption that Alice applies the transpose (in the computational basis) of the operators that she is asked to measure. This indeed has to be the case, by the answer equality between Alice and Bob, enforced by the consistency test (i.e.\ we enforce $Y\otimes Y\ket{\psi}=\ket{\psi}$ instead of the natural $Y\otimes Y\ket{\psi}=-\ket{\psi}$).
\end{remark}
\begin{proof}
We begin with the completeness of the CHSH game (\cref{protocol:cliff-test} subtest 5). A simple calculation confirms that if both parties share $n$ EPR pairs, Alice measures $-\sigma_Y$ and $\sigma_X$, and Bob measures $\sigma_F$ and $\sigma_G$, the Tsirelson bound is achieved, i.e.\ ideal provers win with probability $\cos^2(\frac{\pi}{8})$ in subtest 5. Since all other subtests have perfect completeness, honest provers win with probability $\omega^*=\tfrac{1}{5}(4+\cos^2(\frac{\pi}{8}))$, up to a cryptographically small correctness loss of the QFHE scheme; conversely, \cref{eqn:chsh-pi} below shows that no efficient prover can win with probability greater than $\omega^*+\ea$.
\par 
\medskip
We next give a short sum-of-squares proof that the quantum value of the CHSH game is preserved under compilation with KLVY. This proof is inspired by \cite{Natarajan2023}, however, it follows a simplified approach, without pseudo-expectations. The argument applies to any flavor of CHSH game (different game polynomial); we carry it out for our choice of game polynomial, since this version of the game certifies the operators that we are interested in. Throughout, we use the convention that Alice implements the transpose of the operation on Bob's side, as this is the only way to achieve completeness in the consistency test.
\par 
\medskip
The probability that the verification predicate is satisfied at any position $i\in[n]$, is given by
\begin{align*}
    \Pr[\text{win}]=\frac{1}{2}+\frac{1}{8}\sum_\alpha(-1)^{\Dec(\alpha)_1}\paren{\Tr{(G(e_1)+F(e_1))\pea{Y}}+\Tr{(G(e_1)-F(e_1))\pea{X}}}.\numberthis\label{eqn:chsh-pwin}
\end{align*}
A direct calculation gives
\begin{align*}
p^2&=\sum_\alpha\norm{\frac{1}{\sqrt{2}}(G(e_1)+F(e_1))-(-1)^{\Dec(\alpha)_1}\id}_{\pea{Y}}^2 \\
&= 2-\sqrt{2}\sum_\alpha(-1)^{\Dec(\alpha)_1}\Tr{(G(e_1)+F(e_1))\pea{Y}}+\frac{1}{2}\sum_\alpha\Tr{\{G(e_1),F(e_1)\}\pea{Y}},
\end{align*}
and 
\begin{align*}
q^2&=\sum_\alpha\norm{\frac{1}{\sqrt{2}}(G(e_1)-F(e_1))-(-1)^{\Dec(\alpha)_1}\id}_{\pea{X}}^2\\
&= 2-\sqrt{2}\sum_\alpha(-1)^{\Dec(\alpha)_1}\Tr{(G(e_1)-F(e_1))\pea{X}}-\frac{1}{2}\sum_\alpha\Tr{\{G(e_1),F(e_1)\}\pea{X}},
\end{align*}
thus, plugging in \cref{eqn:chsh-pwin}
\begin{align*}
    p^2+q^2 = 4(1+\sqrt{2})-8\sqrt{2}\Pr[\text{win}]+\frac{1}{2}\Tr{\{G(e_1),F(e_1)\}\paren{\pe{Y}-\pe{X}}},
\end{align*}
we note that the operator $M=\frac{1}{2}\{G(e_1),F(e_1)\}$ is a Hermitian LCU and has bounded eigenvalues. Hence, we can use it as POVM in \cref{lemma:efficient-operator-indistinguishability} and let $D_1,D_2$ be the appropriate point distributions. The lemma then guarantees that
\[
\left|\Tr{M\paren{\pe{Y}-\pe{X}}}\right|=\frac{1}{4}\left|\Tr{\{G(e_1),F(e_1)\}\paren{\pe{Y}-\pe{X}}}\right|\leq\ea.
\]
Plugging this into the expression and rearranging, we get
\[
\Pr[\text{win}]=\frac{1}{2\sqrt{2}}(1+\sqrt{2})-\frac{1}{8\sqrt{2}}(p^2+q^2)+\ea\leq\frac{1}{2}+\frac{\sqrt{2}}{4}+\ea=\cos^2\paren{\frac{\pi}{8}}+\ea,\numberthis\label{eqn:chsh-pi}
\]
thus the quantum value of the compiled game respects the Tsirelson bound, up to cryptographically small advantage. 
\par 
\medskip
By the \cref{lemma:cliff-sign-consistent} we know that for all $W\in\{X,Y,Z,F,G\}$, all $q\in\Qa$ and all distributions $\mu$ over $\zo^n$:
\begin{align*}
    \E_{a\sim\mu} \norm{ V W(a) -((\sigma_W(a)\oplus\overline{\sigma}_W(a)) \ot\Delta_W(a)) V }_\peq^2 \leq O(\eps+\ea),\numberthis\label{eqn:chsh-isometry-guarantee}
\end{align*}
where $\Delta_X(a)=\Delta_Y(a)=\Delta_Z(a)=\id$ and $\Delta_F(a)=\Delta_G(a)=\sigma_{Z}^{|a|}\ot\id$. Using the $\eps$-self-consistency of $(\{X,Y,Z\}, U_3)$ as certified by Item 1 (specifically the $\textsc{slc}$ within its first $\textsc{crel}$ subtest) of \cref{protocol:cliff-test}, \cref{lemma:self-consistency-identity} gives
\[
\sum_\alpha\norm{X(e_1)-(-1)^{\Dec(\alpha)_1}}_{\pea{X}}^2\leq O(\eps).\numberthis\label{eqn:chsh-x-consistency}
\]
From the lemma assumption, we know that the prover succeeds in the CHSH subtest with probability at least $\cos^2(\frac{\pi}{8})-O(\eps)$. Hence, we can conclude that
\[
q^2 = \sum_\alpha\norm{\frac{1}{\sqrt{2}}(G(e_1)-F(e_1))-(-1)^{\Dec(\alpha)_1}\id}_{\pea{X}}^2\leq O(\eps+\ea).
\]
Combining this with \cref{eqn:chsh-x-consistency} through a triangle inequality, we get
\[
\norm{\frac{1}{\sqrt{2}}(G(e_1)-F(e_1))-X(e_1)}_{\pe{X}}^2\leq O(\eps+\ea).
\]
Since the norm argument is a LCU, we can apply \cref{cor:efficient-operator-state-switching} to switch out the state for any $\peq$. Pushing this through the isometry using \cref{eqn:chsh-isometry-guarantee}, we get
\[
\norm{\id\ot(\sigma_Z-\id)\ot\id}^2_{V\peq V^\dag}\leq O(\eps+\ea),
\]
where the first identity has the dimension of twice the qubit register (i.e.\ acting on $(\C^2)^{\ot n}\oplus(\C^2)^{\ot n}$). Combining this with \cref{eqn:chsh-isometry-guarantee} completes the proof.
\end{proof}

What remains is to extend our characterization from pure-basis to mixed-basis observables. This can be achieved in a very modular way by invoking the mixed-versus-pure basis test (\cref{protocol:pure-vs-mixed-test}), which yields the following theorem:

\begin{theorem}\label{theorem:clifford-mixed-basis}
    For any computationally efficient prover modeled as in Section \ref{section:modeling} that wins with probability $\omega^* - \eps$ in the compiled Clifford test $\textprotocol{cliff}(X,Y,Z,F,G,n)$, obtained from \cref{protocol:cliff-test}, there exists a complex Hilbert space $\cH'$ of finite dimension $d' = d|C_n|^2$ and an efficient isometry $V: \cH \to \cH'$ (from \cref{def:clifford-isometry}) such that for all distributions $\mu$ on $\bits^n$ and all $q\in\Qa$ there exists a cryptographically small function $\ea$ such that
    \begin{align*}
\E_{\tilde W\in\{X,Y,Z,F,G\}^n}\E_{a\sim\mu} \norm{ V \tilde W(a) -((\sigma_{\tilde W}(a)\oplus \overline{\sigma}_{\tilde W}(a)) \otimes \id) V }_\peq^2 \leq O(\eps+ \ea).
\end{align*}
    Here $\omega^*=\tfrac{1}{6}(5+\cos^2(\frac{\pi}{8}))$ is the optimal winning probability of the Clifford test (an honest prover achieves it up to a cryptographically small loss, and no efficient prover can exceed it by more than $\ea$).
\end{theorem}
\begin{proof}
An overall winning probability of $\omega^*-\eps=\tfrac16(5+\cos^2(\tfrac\pi8))-\eps$ in the six equally weighted items of \cref{protocol:cliff-test} guarantees (using the trivial bound of $1$ on the CHSH item) that the prover wins Items 1 to 5 with probability at least $\tfrac15(4+\cos^2(\tfrac\pi8))-O(\eps)$, and each of the perfect-completeness items (in particular Item 6, the mixed-versus-pure basis test) with probability $1-O(\eps+\ea)$, where the $\ea$ term comes from bounding the CHSH item by the compiled Tsirelson bound $\cos^2(\tfrac\pi8)+\ea$ of \cref{eqn:chsh-pi}. Hence we can invoke \cref{lemma:chsh} and conclude that for all $W\in\{X,Y,Z,F,G\}$, any distribution $\mu$ over $\bits^n$ and all $q\in\Qa$
\begin{align}\label{eqn:epbt-pure-guarantee}
\E_{a\sim\mu} \norm{ V W(a) -((\sigma_W(a)\oplus\overline{\sigma}_W(a))\ot\id) V}_\peq^2 \leq O(\eps+\ea) \,,
\end{align}
We introduce the following notation
\begin{align*}
    J_W \coloneqq \{i\in [n]: \tilde W_i=W\},
\end{align*}
where we leave the $\tilde W$-dependence implicit. The following chain of approximate equalities holds under implicit expectation over $a\sim\mu$ and over $\tilde W\in\{X,Y,Z,F,G\}^n$ (uniform):
\begin{align*}
    &VX(\exten{a}{J_X})Z(\exten{a}{J_Z})Y(\exten{a}{J_Y})F(\exten{a}{J_F})G(\exten{a}{J_G})\\
    &\approx_{\tilde\eps}(\sigma_X(\exten{a}{J_X})\ot\id)VZ(\exten{a}{J_Z})Y(\exten{a}{J_Y})F(\exten{a}{J_F})G(\exten{a}{J_G})\\
    &\approx_{\tilde\eps}(\sigma_X(\exten{a}{J_X})\sigma_Z(\exten{a}{J_Z})\ot\id)VY(\exten{a}{J_Y})F(\exten{a}{J_F})G(\exten{a}{J_G})\\
    &\approx_{\tilde\eps} (\sigma_X(\exten{a}{J_X})\sigma_Z(\exten{a}{J_Z})(\sigma_Y(\exten{a}{J_Y})\oplus\overline{\sigma_Y}(\exten{a}{J_Y}))\ot\id) VF(\exten{a}{J_F})G(\exten{a}{J_G})\\
    &\approx_{\tilde\eps} (\sigma_X(\exten{a}{J_X})\sigma_Z(\exten{a}{J_Z})(\sigma_Y(\exten{a}{J_Y})\sigma_F(\exten{a}{J_F})\oplus\overline{\sigma_Y}(\exten{a}{J_Y})\overline{\sigma_F}(\exten{a}{J_F}))\ot\id) VG(\exten{a}{J_G})\\
    &\approx_{\tilde\eps} (\sigma_X(\exten{a}{J_X})\sigma_Z(\exten{a}{J_Z})(\sigma_Y(\exten{a}{J_Y})\sigma_F(\exten{a}{J_F})\sigma_G(\exten{a}{J_G})\oplus\overline{\sigma_Y}(\exten{a}{J_Y})\overline{\sigma_F}(\exten{a}{J_F})\overline{\sigma_G}(\exten{a}{J_G}))\ot\id) V,
\end{align*}
where $\tilde\eps=\eps+\ea$ and every line follows by \cref{eqn:epbt-pure-guarantee}, where we explicitly need the fact that this holds for arbitrary distributions for $a$. We also used left unitary invariance and $\eps$-self-consistency of $(\{W\},U_1)$ for $W\in\{X,Y,Z,F,G\}$, combined with \cref{cor:compiled-prover-switching-obs}. From this we can conclude that for all $q\in\Qa$ and under expectation over $a\sim\mu$ and $\tilde W\in\{X,Y,Z,F,G\}$
\[
\prod_{W\in\{X,Y,Z,F,G\}^n}W(\exten{a}{J_W})\simeq_{\eps+\ea} (\sigma_{\tilde W}(a)\oplus\overline{\sigma}_{\tilde W}(a)) \ot \id.\numberthis\label{eqn:mixed-basis-switch1}
\]
It remains to bound the following term
\begin{align*}
    &\E_{\tilde W\in\{X,Y,Z,F,G\}^n}\E_{a\sim\mu}\norm{V\tilde W(a) - V\paren{\prod_{W\in\{X,Y,Z,F,G\}}W(\exten{a}{J_W})}}_\peq^2\\
    \intertext{By left unitary invariance, because $\tilde W(a)$ is linear and because $\bigcup_{W\in\{X,Y,Z,F,G\}}J_W=[n]$:}
    &=\E_{\tilde W\in\{X,Y,Z,F,G\}^n}\E_{a\sim\mu}\norm{\paren{\prod_{W\in\{X,Y,Z,F,G\}}\tilde W(\exten{a}{J_W})} - \paren{\prod_{W\in\{X,Y,Z,F,G\}}W(\exten{a}{J_W})}}_\peq^2\\
    &\leq O(\eps+\ea),\numberthis\label{eqn:mixed-basis-switch2}
\end{align*}
here the final upper bound is obtained by repeatedly using \cref{lemma:mixed-vs-pure-guarantee}, since the prover succeeds in $\textprotocol{mbt}(\{X,Y,Z,F,G\},n, U_{5^n})$. Furthermore we used left unitary invariance and compiled prover switching (\cref{lemma:compiled-prover-switching}), since the wrapped $\textprotocol{slc}$ subtests of $\textprotocol{crel}$ guarantee $\eps$-self-consistency of $(\{W\},U_1)$ for $W\in\{X,Y,Z,F,G\}$, by \cref{lemma:compiled-consistency}.
\par 
\medskip
Combining \cref{eqn:mixed-basis-switch1} and \cref{eqn:mixed-basis-switch2}, yields
\begin{align*}
\E_{\tilde W\in\{X,Y,Z,F,G\}^n}\E_{a\sim\mu} \norm{ V \tilde W(a) - ((\sigma_{\tilde W}(a)\oplus\overline{\sigma}_{\tilde W}(a)) \ot \id) V }_\peq^2 \leq O(\eps+ \ea),
\end{align*}
which completes the proof.
\end{proof}

\subsection{State characterization}
We want to show that during both parts of the interaction (encrypted and unencrypted), the prover must be consistently using the canonical or complex conjugated irreducible representation of the extended Pauli group. The following lemma plays an important part in this as it shows shows that coherences in the flag register are computationally indetectable.

\begin{lemma}\label{lemma:cliff-delta-self-consistent}
    For any computationally efficient prover modeled as in Section \ref{section:modeling} that wins with probability $1 - \eps$ in Items 1 to 4 of the compiled Clifford test $\textprotocol{cliff}(X,Y,Z,F,G,n)$, obtained from \cref{protocol:cliff-test}, there exist complex Hilbert spaces $\cH',\hat\cH$ of finite dimension and an efficient isometry $V: \cH \to \cH' \cong(\C^2)^n\ot \C^2\ot\hat\cH$ (from \cref{def:clifford-isometry}) such that for every uniformly efficient family of two-outcome POVMs $\{M_q,\id-M_q\}_{q\in\Qa}$ in $L(\cH')$, indexed by the Alice question, and every efficiently sampleable distribution $\nu$ over $\Qa$, there exists a cryptographically small function $\ea$ such that
    \begin{align*}
    \left|\E_{q\sim\nu}\Tr{M_q\left((\id\ot\sigma_Z\ot\id) V\peq V^\dag(\id\ot\sigma_Z\ot\id) - V\peq V^\dag\right)}\right|\leq O((\eps+\ea)^{1/2}).
    \end{align*}
\end{lemma}
\begin{proof}
    Throughout the proof write $\tilde Z\coloneqq\id\ot\sigma_Z\ot\id$ and fix a specific $a\in\zo^n$ with $|a|=1$. From \cref{lemma:cliff-sign-consistent} we know that there exists an efficient isometry $V:\cH\to\cH'\cong (\C^2)^n\ot\C^2\ot\hat\cH$, such that for all $W\in\{X,Y,Z\}$ and all $q\in\Qa$ there exists a cryptographically small function $\ea$ with
    \[
    \norm{ V W(a) -(\sigma_W(a) \ot\Sigma_W^{|a|}\ot\id) V }_\peq^2 \leq O(\eps+\ea),
    \]
    with $\Sigma_X=\Sigma_Z=\id$ and $\Sigma_Y=\sigma_Z$. Here we used the identification $A\oplus A \cong A\otimes \C^2$ and the fact that for $W\in\{X,Y,Z\}$ the complex conjugate becomes a scalar factor on the second system. The norm is of the isometry form $\norm{VW(a)-NV}_\peq^2$ with $W(a)$ and $N\coloneqq\sigma_W(a)\ot\Sigma_W^{|a|}\ot\id$ uniformly efficient unitaries (independent of the Alice question) on $\cH$ and $\cH'$ respectively, so \cref{cor:efficient-isometry-state-switching} applied with $D_1$ a point distribution on an arbitrary fixed question and $D_2$ sampling $q\sim\nu$ upgrades the bound to an expectation over $\nu$, with a single cryptographically small function depending on $\nu$:
    \[
    \E_{q\sim\nu}\norm{ V W(a) -(\sigma_W(a) \ot\Sigma_W^{|a|}\ot\id) V }_\peq^2 \leq O(\eps+\ea)\numberthis\label{eqn:cliff-delta-self-consistent-bound}.
    \]
    By Cauchy--Schwarz, for every family of operators $\{N_q\}$ with $\norm{N_q}\leq 1$ and every $W\in\{X,Y,Z\}$,
    \begin{align*}
    &\left|\Tr{N_q\left(VW(a)\peq W(a)V^\dag - (\sigma_W(a)\ot\Sigma_W^{|a|}\ot\id)V\peq V^\dag(\sigma_W(a)\ot\Sigma_W^{|a|}\ot\id)\right)}\right|\\
    &\leq 2\norm{ V W(a) -(\sigma_W(a) \ot\Sigma_W^{|a|}\ot\id) V }_\peq.
    \end{align*}
    Taking the expectation over $q\sim\nu$ and using $\bigl|\E f\bigr|\le\E|f|$ together with Jensen's inequality, $\E_\nu\norm{\cdot}_\peq\leq\sqrt{\E_\nu\norm{\cdot}^2_\peq}$, we obtain with \cref{eqn:cliff-delta-self-consistent-bound} that
    \begin{align*}
    &\left|\E_{q\sim\nu}\Tr{N_q\left(VW(a)\peq W(a)V^\dag - (\sigma_W(a)\ot\Sigma_W^{|a|}\ot\id)V\peq V^\dag(\sigma_W(a)\ot\Sigma_W^{|a|}\ot\id)\right)}\right|\\
    &\leq O((\eps+\ea)^{1/2}).\numberthis\label{eqn:push-through-V}
    \end{align*}
    Furthermore, for every uniformly efficient family of POVM elements $\{N_q\}_q$ in $L(\cH')$ the family $\{V^\dag N_q V\}_q$ is a uniformly efficient family of POVM elements in $L(\cH)$ ($V$ is an efficient isometry and $0\preceq V^\dag N_q V\preceq V^\dag V=\id$), so \cref{lemma:compiled-prover-switching-povm} (with the $\eps$-self-consistent question distributions certified by the consistency test, \cref{lemma:compiled-consistency}) lets us insert or remove the observables $X(a)$, $Z(a)$ and $Y(a)$ around $\peq$ at cost $O(\eps^{1/2}+\ea)$ in the signed average over $q\sim\nu$. With this we can make the following observation:
    \begin{align*}
        &\E_{q\sim\nu}\Tr{M_q\tilde Z V\peq V^\dag\tilde Z}\\  &\approx_{\sqrt{\eps}+\ea} \E_{q\sim\nu} \Tr{M_q\tilde Z VX(a)\peq X(a)V^\dag\tilde Z } & \text{by \cref{lemma:compiled-prover-switching-povm}}\\
        &\approx_{(\eps+\ea)^{1/2}} \E_{q\sim\nu}\Tr{(\sigma_X(a)\ot\sigma_Z\ot\id)^\dag M_q(\sigma_X(a)\ot\sigma_Z\ot\id) V\peq V^\dag} & \text{by \cref{eqn:push-through-V}}\\
        &\approx_{\sqrt{\eps}+\ea} \E_{q\sim\nu}\Tr{(\sigma_X(a)\ot\sigma_Z\ot\id)^\dag M_q(\sigma_X(a)\ot\sigma_Z\ot\id) VZ(a)\peq Z(a)V^\dag} & \text{by \cref{lemma:compiled-prover-switching-povm}}\\
        &\approx_{(\eps+\ea)^{1/2}} \E_{q\sim\nu}\Tr{(\sigma_Y(a)\ot\sigma_Z\ot\id)^\dag M_q(\sigma_Y(a)\ot\sigma_Z\ot\id) V\peq V^\dag} & \text{by \cref{eqn:push-through-V}} \\
        &\approx_{(\eps+\ea)^{1/2}} \E_{q\sim\nu}\Tr{M_q VY(a)\peq Y(a)V^\dag} & \text{by \cref{eqn:push-through-V}}\\
        &\approx_{\sqrt{\eps}+\ea} \E_{q\sim\nu}\Tr{M_qV\peq V^\dag}. & \text{by \cref{lemma:compiled-prover-switching-povm}}
    \end{align*}
    In the second-to-last step $\sigma_Y(a)=i\sigma_X(a)\sigma_Z(a)$ is used and \cref{eqn:push-through-V} is applied with $W=Y$: since $|a|=1$, the post-isometry action of $Y(a)$ is exactly $\sigma_Y(a)\ot\Sigma_Y\ot\id=\sigma_Y(a)\ot\sigma_Z\ot\id$. All families of POVM elements appearing on the left of the traces are uniformly efficient in $q$, since $M_q$ is and all conjugating unitaries are fixed efficient operators. Combining the six steps yields
    \[
    \left|\E_{q\sim\nu}\Tr{M_q\tilde Z V\peq V^\dag\tilde Z}-\E_{q\sim\nu}\Tr{M_qV\peq V^\dag}\right|\leq O((\eps+\ea)^{1/2}),
    \]
    which completes the proof.
\end{proof}

The encryption forbids the prover from centrally keeping track of his sign choice. The consistency can thus only be achieved by something similar to an entangled pair, which keeps track of the choice, and is allowed by the encryption as it is `transparent' to entanglement (cf.\ the correctness with auxiliary input property in \cite{Kalai2023}). Since Bob only has partial information about this shared system, this effectively (under computational assumptions) dephases the `flag' register, which allows us to conclude that the prover can't be applying a coherent superposition of the canonical and the complex conjugate of the Pauli $Y$ observable.

\begin{corollary}\label{cor:flag-dephasing}
    For any computationally efficient prover modeled as in Section \ref{section:modeling} that wins with probability $1 - \eps$ in Items 1 to 4 of the compiled Clifford test $\textprotocol{cliff}(X,Y,Z,F,G,n)$, obtained from \cref{protocol:cliff-test}, there exist complex Hilbert spaces $\cH',\hat\cH$ of finite dimension and an efficient isometry $V: \cH \to \cH' \cong(\C^2)^n\ot \C^2\ot\hat\cH$ (from \cref{def:clifford-isometry}) such that for every uniformly efficient family of two-outcome POVMs $\{M_q,\id-M_q\}_{q\in\Qa}$ in $L(\cH')$ and every efficiently sampleable distribution $\nu$ over $\Qa$, there exists a cryptographically small function $\ea$ such that
    \begin{align*}
    \left|\E_{q\sim\nu}\Tr{M_q\left(\mathcal{D}_Z(V\peq V^\dag) - V\peq V^\dag\right)}\right|\leq O((\eps+\ea)^{1/2}),
    \end{align*}
    where $\mathcal{D}_Z$ is the dephasing operation on the phase ambiguity flag register, i.e.\
    \[
    \mathcal{D}_Z(\rho) = \frac{1}{2}\left(\rho+(\id\ot\sigma_Z\ot\id)\rho(\id\ot\sigma_Z\ot\id)\right).
    \]
\end{corollary}
\begin{proof}
    Inserting the definition of $\mathcal{D}_Z$,
    \begin{align*}
    &\left|\E_{q\sim\nu}\Tr{M_q\left(\mathcal{D}_Z(V\peq V^\dag) - V\peq V^\dag\right)}\right|\\
    &=\frac12\left|\E_{q\sim\nu}\Tr{M_q\left((\id\ot\sigma_Z\ot\id)V\peq V^\dag (\id\ot\sigma_Z\ot\id) - V\peq V^\dag\right)}\right|,
    \end{align*}
    so the claim follows directly from \cref{lemma:cliff-delta-self-consistent} (absorbing the factor $\tfrac12$ into the asymptotic notation).
\end{proof}

The next lemma is an important tool for our later analysis, as it allows us to decouple the qubit register from the verifier's encrypted question and the prover's encrypted answer. It shows that, up to efficient distinguishers, the post-Alice state satisfies an important property, which would be expected from EPR pairs being shared between the two parts of the prover. It is the closest we come to some kind of replacement for certifying EPR pairs (inter-prover), as can be achieved in the nonlocal setting.

\begin{lemma}\label{lemma:state-characterization}
    For any computationally efficient prover modeled as in Section \ref{section:modeling} that wins with probability $1 - \eps$ in the compiled Clifford group test $\textprotocol{cliff-group}(X,Y,Z,F,G,n)$, obtained from \cref{protocol:cliff-group}, there exists a complex Hilbert space $\cH'$ of finite dimension and an efficient isometry $V: \cH \to (\C^2)^n \ot\cH'$ (from \cref{def:clifford-isometry}) such that for all uniformly efficient POVM families $\{M_\alpha,\id-M_\alpha\}_\alpha$ and all $q\in\Qa$, there exists a cryptographically small function $\ea$ such that
    \begin{align*}
    \left|\sum_\alpha\Tr{M_\alpha\paren{V\pea{q} V^\dag - \frac{\id}{2^n}\ot\Tr_1\left[V\pea{q} V^\dag\right]}}\right|\leq O((\eps+\ea)^{1/2}),
    \end{align*}
    where $\Tr_1[\cdot]$ denotes the partial trace over the $(\C^2)^n$ register. In words, the post-isometry post-Alice state is computationally indistinguishable (with some advantage) from being fully mixed on the qubit register.
\end{lemma}
\begin{proof}
    Let $q\in\Qa$ and the uniformly efficient POVM family $\{M_\alpha,\id-M_\alpha\}_\alpha$ with $0\preceq M_\alpha\preceq \id$ be arbitrary. By \cref{lemma:EPBT-pure-basis}, for $W\in\{X,Z\}$ and all $a\in\zo^n$,
    \[
    \norm{VW(a)-(\sigma_W(a)\ot\id)V}_\peq^2\leq O(\eps+\ea)
    \]
    By Cauchy--Schwarz, the triangle inequality, and the above we have for any $M_\alpha$ with $\|M_\alpha\|_\infty\leq 1$, $W\in\{X,Z\}$ and any $a\in\zo^n$
    \begin{align*}
    \Big|\sum_\alpha&\Tr\Big[M_\alpha\paren{VW(a)\pea{q} W(a)V^\dag-(\sigma_W(a)\ot\id)V\pea{q}V^\dag(\sigma_W(a)\ot\id)}\Big]\Big|\\
    &\leq 2\left(\sum_\alpha\norm{VW(a)-(\sigma_W(a)\ot\id)V}_\pea{q}^2\right)^{1/2} \left(\sum_\alpha\norm{M_\alpha^\dag VW(a)}_\pea{q}^2\right)^{1/2}\\
    &\leq O((\eps+\ea)^{1/2}).\numberthis\label{eqn:intrace-isometry-switch}
    \end{align*}
    With this we can make the following observation
    \begin{align*}
        &\E_{a,b}\sum_\alpha\Tr{ M_\alpha V\pea{q} V^\dag}\\
        &\approx_{\ea}\E_{a,b}\sum_\alpha\Tr{M_\alpha V\pea{X}V^\dag} & \text{by \cref{lemma:poly-state-indistinguishability}}\\
        &\approx_{\sqrt{\eps}}\E_{a,b}\sum_\alpha\Tr{M_\alpha VX(a)\pea{X}X(a)V^\dag} & \text{by \cref{lemma:self-consistent-to-trace-norm-crit}}\\
        &\approx_{\sqrt{\eps+\ea}}\E_{a,b}\sum_\alpha\Tr{M_\alpha(\sigma_X(a)\ot\id)V\pea{X} V^\dag(\sigma_X(a)\ot\id)} & \text{by \cref{eqn:intrace-isometry-switch}}\\
        &\approx_{\ea}\E_{a,b}\sum_\alpha\Tr\Big[\overbrace{(\sigma_X(a)\ot\id)M_\alpha(\sigma_X(a)\ot\id)}^{N_{a,\alpha}}V\pea{Z} V^\dag\Big] & \text{by \cref{lemma:poly-state-indistinguishability}}\\
        &\approx_{\sqrt{\eps}}\E_{a,b}\sum_\alpha\Tr{N_{a,\alpha} VZ(b)\pea{Z}Z(b)V^\dag} & \text{by \cref{lemma:self-consistent-to-trace-norm-crit}}\\
        &\approx_{\sqrt{\eps+\ea}}\E_{a,b}\sum_\alpha\Tr{N_{a,\alpha}(\sigma_Z(b)\ot\id)V\pea{Z} V^\dag(\sigma_Z(b)\ot\id)} & \text{by \cref{eqn:intrace-isometry-switch}}\\
        &\approx_{\ea}\E_{a,b}\sum_\alpha\Tr{M_{(a,b),\alpha}' V\pea{q}V^\dag} & \text{by \cref{lemma:poly-state-indistinguishability}}\\
    \end{align*}
    where $M_{(a,b),\alpha}'\deq (\sigma_Z(b)\sigma_X(a)\ot\id)M_\alpha(\sigma_X(a)\sigma_Z(b)\ot\id)$. Evaluating the uniform expectation over $a,b\in\zo^n$ and noticing that only $M'_{(a,b),\alpha}$ depends on $a,b$, we obtain
    \[
        \sum_\alpha\Tr{ M_\alpha V\pea{q} V^\dag}\approx_{\sqrt{\eps+\ea}} \sum_\alpha\Tr{\left(\frac{\id}{2^n}\ot\Tr_1[M_\alpha]\right)V\pea{q} V^\dag},\numberthis\label{eqn:trace-state-second-step}
    \]
    by the Pauli twirl identity:
    \[
    \E_{a,b\in\zo^n}(\sigma_Z(b)\sigma_X(a)\ot\id)M_\alpha(\sigma_X(a)\sigma_Z(b)\ot\id) = \frac{\id}{2^n}\ot\Tr_1[M_\alpha].
    \]
    Finally, we can use \cref{eqn:trace-state-second-step} to conclude that 
    \[
        \sum_\alpha\Tr{ M_\alpha V\pea{q} V^\dag}\approx_{\sqrt{\eps+\ea}} \sum_\alpha\Tr{M_\alpha\left(\frac{\id}{2^n}\ot\Tr_1\left[V\pea{q} V^\dag\right]\right)},
    \]
    which follows from the following property of the partial trace 
    \[
    \Tr{(\id\ot \Tr_1[A])B} = \Tr\left[A\left(\id\ot\Tr_1\left[B\right]\right)\right],
    \]
    completing the proof.
\end{proof}

An important part of obtaining a remote state preparation guarantee in our setting is characterizing the post-Alice state. From the different rigidity tests we already obtained characterizations of the Bob observables in the previous section. Now we will relate these observable characterizations to a state characterization, which we can do through the state-dependence of the norm we are using (\cref{def:state-dependent-norm}). 

\begin{lemma}\label{lemma:state-rounding-it}
    For any computationally efficient prover modeled as in Section \ref{section:modeling} that wins with probability $\omega^* - \eps$ in the compiled Clifford test $\textprotocol{cliff}(X,Y,Z,F,G,n)$, obtained from \cref{protocol:cliff-test}, there exists a complex Hilbert space $\hat\cH$ of finite dimension, an efficient isometry $V: \cH \to (\C^2)^n\ot \C^2 \ot\hat\cH$ (from \cref{def:clifford-isometry}), and a cryptographically small function $\ea$ such that
    \begin{align*}
    \E_{\tilde W}\sum_\alpha\norm{P_{\tilde W}^{\Dec(\alpha)} V\pea{\tilde W}V^\dag P_{\tilde W}^{\Dec(\alpha)} - V\pea{\tilde W}V^\dag}_1\leq O((\eps+\ea)^{1/2}),
    \end{align*}
    where $\omega^*=\tfrac{1}{6}(5+\cos^2(\frac{\pi}{8}))$ is the optimal winning probability of the Clifford test (an honest prover achieves it up to a cryptographically small loss, and no efficient prover can exceed it by more than $\ea$, cf.\ \cref{lemma:chsh}) and
    \[
    P_{\tilde W}^{a} = \sum_{k\in\zo}\tau_{\tilde W,k}^{a}\ot\ket{k}\bra{k}\ot\id,
    \]
    with $\tau_{\tilde W, 0}^{a}=\tau_{\tilde W}^{a}$, $\tau_{\tilde W, 1}^{a} = \overline{\tau}_{\tilde W}^{a}$. In words, the post-isometry post Alice measurement state is information theoretically close to fully lying inside the image of the Bob projectors corresponding to the reported outcome of the requested measurement.
\end{lemma}
\begin{proof}
    From \cref{theorem:clifford-mixed-basis} we know that there exists an isometry $V: \H\to ((\C^2)^n\ot\H')^{\oplus 2}$, such that for all $q\in\Qa$ there exists a cryptographically small function $\ea$ with
    \[
    \E_{\tilde W\in\{X,Y,Z,F,G\}^n}\E_{a\in\zo^n} \norm{ V \tilde W(a) -((\sigma_{\tilde W}(a)\oplus \overline{\sigma}_{\tilde W}(a)) \otimes \id) V }_\peq^2 \leq O(\eps+ \ea),
    \]
    choosing some fixed $q\in\Qa$. This is of the isometry form $\norm{V\tilde W(a)-NV}_\peq^2$ with $\tilde W(a)$ and $N\coloneqq(\sigma_{\tilde W}(a)\oplus\overline\sigma_{\tilde W}(a))\ot\id$ a uniformly efficient family of unitaries, parameterized by $a$ and $\tilde W$, so \cref{cor:efficient-isometry-state-switching} lets us switch out the state and conclude that
    \[
    \E_{\tilde W\in\{X,Y,Z,F,G\}^n}\E_{a\in\zo^n} \norm{ V \tilde W(a) -((\sigma_{\tilde W}(a)\oplus \overline{\sigma}_{\tilde W}(a)) \otimes \id) V }_\pe{\tilde W}^2 \leq O(\eps+ \ea),
    \]
    applying \cref{cor:parseval} (Parseval's identity) yields
    \[
    \E_{\tilde W\in\{X,Y,Z,F,G\}^n}\sum_a \norm{ V \tilde W^a -((\tau_{\tilde W}^a\oplus \overline{\tau}_{\tilde W}^a) \otimes \id) V }_\pe{\tilde W}^2 \leq O(\eps+ \ea),
    \]
    where
    \[
    \tau_{\tilde W}^a \coloneqq\bigotimes_{i\in[n]}\tau_{\tilde W_i}^{a_i}.
    \]
    Interpreting the isometry co-domain as $(\C^2)^n\ot\C^2\ot\hat\cH$ (under the canonical isomorphism between $A\oplus A$ and $A\otimes \C^2$) we can rewrite this as
    \[
    \E_{\tilde W\in\{X,Y,Z,F,G\}^n}\sum_a\Bigg\|V\tilde W^a - \overbrace{\paren{\sum_{k\in\zo}\tau_{\tilde W, k}^{a}\ot\ket{k}\bra{k}\ot\id}}^{P_{\tilde W}^a}V\Bigg\|_\pe{\tilde W}^2\leq O(\eps+\ea),
    \]
    where $\tau_{\tilde W, 0}^{a}=\tau_{\tilde W}^{a}$ and $\tau_{\tilde W, 1}^{a}=\overline{\tau}_{\tilde W}^{a}$.
    By the $\eps$-self-consistency of $(\{X,Y,Z,F,G\}^n,U_{5^n})$, 
    \[
    \E_{\tilde W}\sum_\alpha\norm{\tilde W^{\Dec(\alpha)}-\id}_{\pea{\tilde W}}^2\leq O(\eps).
    \]
    Chaining these two observations (using the fact that $\pe{\tilde W}=\sum_\alpha \pea{\tilde W}$), we obtain
    \[
    \E_{\tilde W}\sum_\alpha\Big\|P_{\tilde W}^{\Dec(\alpha)}-\id\Big\|_{V\pea{\tilde W}V^\dag}^2\leq O(\eps+\ea),\numberthis\label{eqn:post-isometry-projector-char}
    \]

    Using \eqref{eqn:post-isometry-projector-char} and \cite[Lemma 2.10]{Metger2024}, we can bound the following trace norm
    \begin{align*}
        \E_{\tilde W}\sum_\alpha\norm{(P_{\tilde W}^{\Dec(\alpha)}-\id)V\pea{\tilde W}V^\dag}_1 &\leq \sqrt{\E_{\tilde W}\sum_\alpha\norm{P_{\tilde W}^{\Dec(\alpha)}-\id}_{V\pea{\tilde W}V^\dag}^2}\\
        &\leq O((\eps+\ea)^{1/2}).\numberthis\label{eqn:projector-trace-norm-char}
    \end{align*}
    For any projector $P$ and Hermitian operator $\rho$, it holds by the triangle inequality, Hölder's inequality and invariance of the trace norm under Hermitian adjoint, that
    \[
    \norm{P\rho P-\rho}_1\leq \norm{P\rho(P-1)}_1 + \norm{(P-1)\rho}_1\leq 2\norm{(P-1)\rho}_1
    \]
    applying the above to $P_{\tilde W}^{a}$, we can use \eqref{eqn:projector-trace-norm-char} to conclude that:
    \[
    \E_{\tilde W}\sum_\alpha\norm{P_{\tilde W}^{\Dec(\alpha)} V\pea{\tilde W}V^\dag P_{\tilde W}^{\Dec(\alpha)} - V\pea{\tilde W}V^\dag}_1\leq O((\eps+\ea)^{1/2}),
    \]
    which completes the proof.
\end{proof}

Instead of having a characterization in terms of encrypted prover answers, we want to relate the post-Alice state directly to the decrypted answer from the first interaction round. Even though decryption is an inefficient operation (without knowledge of the secret key), this can be achieved because, at the current phase of the argument, all guarantees are information theoretic.

\begin{corollary}\label{cor:state-rounding-it-marginal}
    For any computationally efficient prover modeled as in Section \ref{section:modeling} that wins with probability $\omega^* - \eps$ in the compiled Clifford test $\textprotocol{cliff}(X,Y,Z,F,G,n)$, obtained from \cref{protocol:cliff-test}, there exists a complex Hilbert space $\hat\cH$ of finite dimension, an efficient isometry $V: \cH \to (\C^2)^n\ot \C^2 \ot\hat\cH$ (from \cref{def:clifford-isometry}), and a cryptographically small function $\ea$ such that
    \begin{align*}
    \E_{\tilde W}\sum_{v\in\zo^n}\norm{P_{\tilde W}^v V\pe{\tilde W}V^\dag P_{\tilde W}^v - V\phi^{\Enc(\tilde W)}_v V^\dag}_1\leq O((\eps+\ea)^{1/2}),
    \end{align*}
    where $\omega^*=\tfrac{1}{6}(5+\cos^2(\frac{\pi}{8}))$ is the optimal winning probability of the Clifford test (an honest prover achieves it up to a cryptographically small loss, and no efficient prover can exceed it by more than $\ea$, cf.\ \cref{lemma:chsh}) and
    \[
    P_{\tilde W}^{a} = \sum_{k\in\zo}\tau_{\tilde W,k}^{a}\ot\ket{k}\bra{k}\ot\id\quad\text{and}\quad\phi^{\Enc(\tilde W)}_v = \sum_{\alpha:\Dec(\alpha)=v} \pea{\tilde W},
    \]
    with $\tau_{\tilde W, 0}^{a}=\tau_{\tilde W}^{a}$, $\tau_{\tilde W, 1}^{a} = \overline{\tau}_{\tilde W}^{a}$. In words, the post-isometry post Alice measurement state (grouped per decrypted outcome) is information theoretically close to the fully marginalized state, projected into the image space of the projector corresponding to the reported outcome of the requested measurement.
\end{corollary}
\begin{proof}
    By \cref{lemma:state-rounding-it}, there exists an isometry $V: \cH \to (\C^2)^n\ot \C^2 \ot\hat\cH$ such that
    \begin{align*}
        \E_{\tilde W}\sum_\alpha\norm{P_{\tilde W}^{\Dec(\alpha)} V\pea{\tilde W}V^\dag P_{\tilde W}^{\Dec(\alpha)} - V\pea{\tilde W}V^\dag}_1\leq O((\eps+\ea)^{1/2}),
    \end{align*}
    where
    \[
    P_{\tilde W}^v = \sum_{k\in\zo}\tau_{\tilde W,k}^{v}\ot\ket{k}\bra{k}\ot\id.
    \]
    By the triangle inequality, this implies
    \[
        \E_{\tilde W}\sum_{v\in\zo^n}\norm{P_{\tilde W}^v V\phi^{\Enc(\tilde W)}_v V^\dag P_{\tilde W}^v - V\phi^{\Enc(\tilde W)}_v V^\dag}_1\leq O((\eps+\ea)^{1/2}).\numberthis\label{eqn:state-rounding-it-marginal-1}
    \]
    note that for a set of orthogonal projectors $\{P^v\}_v$, and arbitrary operator $A$, by the pinching inequality it holds that
    \[
    \sum_v \norm{P^v A P^v}_1\leq \norm{A}_1.
    \]
    With this we can write
    \begin{align*}
        \E_{\tilde W}&\sum_{v\in\zo^n}\norm{P_{\tilde W}^v V\psi^{\Enc(\tilde W)}V^\dag P_{\tilde W}^v - P_{\tilde W}^v V\phi^{\Enc(\tilde W)}_v V^\dag P_{\tilde W}^v}_1 \\
        &= \E_{\tilde W}\sum_{v\in\zo^n}\norm{P_{\tilde W}^v\left(\sum_{v'} \left(V\phi^{\Enc(\tilde W)}_{v'} V^\dag - P_{\tilde W}^{v'} V\phi^{\Enc(\tilde W)}_{v'} V^\dag P_{\tilde W}^{v'}\right)\right)P_{\tilde W}^v}_1 \\
        &\leq \E_{\tilde W}\norm{\sum_{v'} \left(V\phi^{\Enc(\tilde W)}_{v'} V^\dag- P_{\tilde W}^{v'} V\phi^{\Enc(\tilde W)}_{v'} V^\dag P_{\tilde W}^{v'}\right)}_1 \\
        &\leq \E_{\tilde W}\sum_{v'}\norm{V\phi^{\Enc(\tilde W)}_{v'} V^\dag- P_{\tilde W}^{v'} V\phi^{\Enc(\tilde W)}_{v'} V^\dag P_{\tilde W}^{v'}}_1 \\
        &\leq O((\eps+\ea)^{1/2}),
    \end{align*}
    where the last inequality follows from \cref{eqn:state-rounding-it-marginal-1}. Combining the above with \cref{eqn:state-rounding-it-marginal-1}, through the triangle inequality, we obtain
    \[
    \E_{\tilde W}\sum_{v\in\zo^n}\norm{P_{\tilde W}^v V\psi^{\Enc(\tilde W)} V^\dag P_{\tilde W}^v - V\phi^{\Enc(\tilde W)}_v V^\dag}_1\leq O((\eps+\ea)^{1/2}),
    \]
    which completes the proof.
\end{proof}

We also want to apply the dephasing observation from \cref{cor:flag-dephasing} to the decrypted-outcome post-Alice state. This is a bit delicate, because generally we can't apply the inefficient grouping per decrypted outcome to a computational statement. However, we can use the tools established in the lemmas above to combine computational and information-theoretic guarantees such that the desired statement follows.
\begin{lemma}\label{lemma:full-flag-dephasing}
For any computationally efficient prover modeled as in Section \ref{section:modeling} that wins with probability $\omega^* - \eps$ in the compiled Clifford test $\textprotocol{cliff}(X,Y,Z,F,G,n)$, obtained from \cref{protocol:cliff-test}, there exist complex Hilbert spaces $\cH',\hat\cH$ of finite dimension and an efficient isometry $V: \cH \to \cH' \cong(\C^2)^n\ot \C^2\ot\hat\cH$ (from \cref{def:clifford-isometry}) such that for all uniformly efficient (in $v$ and $\tilde W$) POVMs $\{M_{v,\tilde W},\id-M_{v,\tilde W}\}_{v,\tilde W}\in L(\cH')$, there exists a cryptographically small function $\ea$ such that
\begin{align*}
\left|\E_{\tilde W\in\{X,Y,Z,F,G\}^n}\sum_{v\in\zo^n}\Tr{M_{v,\tilde W}\left(\mathcal{D}_Z(V\phi_v^{\Enc(\tilde W)} V^\dag) -  V\phi_v^{\Enc(\tilde W)} V^\dag\right)}\right| \leq O((\eps+\ea)^{1/2})
\end{align*}
where
\[
\phi_v^{\Enc(\tilde W)} \deq \sum_{\alpha:\Dec(\alpha)=v} \pea{\tilde W},
\]
$\mathcal{D}_Z$ is the dephasing operation on the phase ambiguity flag register, i.e.\
\[
\mathcal{D}_Z(\rho) = \frac{1}{2}\left(\rho+(\id\ot\sigma_Z\ot\id)\rho(\id\ot\sigma_Z\ot\id)\right).
\]
Here $\omega^*=\tfrac{1}{6}(5+\cos^2(\frac{\pi}{8}))$ is the optimal winning probability of the Clifford test.
\end{lemma}
\begin{proof}
    Let $M_{v,\tilde W}$ be any uniformly efficient family of POVM elements in $v\in\zo^n$ and $\tilde W\in\{X,Y,Z,F,G\}^n$.
    The lemma's hypothesis (success probability of $\omega^*-\eps$ on the full Clifford test) combined with the Tsirelson bound on the CHSH subtest and the equal weighting of the six items implies that the prover wins each non-CHSH subtest, and items 1 to 4 in particular, with probability $1-O(\eps+\ea)$. Hence both \cref{cor:state-rounding-it-marginal} and \cref{cor:flag-dephasing} apply, and make use of the same isometry $V$ (from \cref{def:clifford-isometry}).
    \par
    \medskip
    The proof idea is to replace the outcome-dependent state by a projection of the marginal, which is given by
    \[
    P^v_{\tilde W}V\pe{\tilde W}V^\dag P^v_{\tilde W}=P^v_{\tilde W}V\left(\sum_{v'} \phi_{v'}^{\Enc(\tilde W)} \right)V^\dag P^v_{\tilde W}. 
    \]
    From \cref{cor:state-rounding-it-marginal}, we know that both states are close in trace norm
    \[
    \E_{\tilde W}\sum_v\norm{V\phi_v^{\Enc(\tilde W)} V^\dag-P^v_{\tilde W} V\pe{\tilde W} V^\dag P^v_{\tilde W}}_1\leq \delta,\numberthis\label{eqn:rounding-bound}
    \]
    with $\delta=O((\eps+\ea)^{1/2})$. Using $\|M_{v,\tilde W}\|\leq 1$, the tracial Hölder inequality \cite[Theorem 2]{Baumgartner2011} and the cyclicity of the trace, we have for each $\tilde W$ and $v$
    \begin{align*}      
    \big|\Tr&\big[M_{v,\tilde W} V\phi_v^{\Enc(\tilde W)} V^\dag\big]-\Tr\big[P^v_{\tilde W}M_{v,\tilde W} P^v_{\tilde W}V\pe{\tilde W} V^\dag\big]\big|\\
    &=\big|\Tr\big[M_{v,\tilde W}\big(V\phi_v^{\Enc(\tilde W)} V^\dag-P^v_{\tilde W} V\pe{\tilde W} V^\dag P^v_{\tilde W}\big)\big]\big|\\
    &\leq\norm{V\phi_v^{\Enc(\tilde W)} V^\dag-P^v_{\tilde W} V\pe{\tilde W} V^\dag P^v_{\tilde W}}_1.\numberthis\label{eqn:non-dephased-error}
    \end{align*}
    The projector $P^v_{\tilde W}=\sum_k\tau_{\tilde W,k}^v\ot\ketbra{k}{k}\ot\id$ commutes with $\id\ot\sigma_Z\ot\id$ and thus with the channel $\mathcal{D}_Z$:
    \[
    \mathcal{D}_Z(P^v_{\tilde W}\rho P^v_{\tilde W})=P^v_{\tilde W}\mathcal{D}_Z(\rho)P^v_{\tilde W}\qquad\forall\rho.
    \]
    Combining this commutation relation with the trace-norm contractivity of $\mathcal{D}_Z$ gives
    \begin{align*}
    \big|\Tr&\big[M_{v,\tilde W}\mathcal{D}_Z(V\phi_v^{\Enc(\tilde W)} V^\dag)\big]-\Tr\big[P^v_{\tilde W}M_{v,\tilde W} P^v_{\tilde W}\mathcal{D}_Z(V\pe{\tilde W} V^\dag)\big]\big|\\
    &=\big|\Tr\big[M_{v,\tilde W}\,\mathcal{D}_Z\big(V\phi_v^{\Enc(\tilde W)} V^\dag-P^v_{\tilde W} V\pe{\tilde W} V^\dag P^v_{\tilde W}\big)\big]\big|\\
    &\leq\norm{V\phi_v^{\Enc(\tilde W)} V^\dag-P^v_{\tilde W} V\pe{\tilde W} V^\dag P^v_{\tilde W}}_1.\numberthis\label{eqn:dephased-error}
    \end{align*}
    Combining \cref{eqn:non-dephased-error} and \cref{eqn:dephased-error} via the triangle inequality, summing over $v$, taking the expectation over $\tilde W$, and applying \cref{eqn:rounding-bound} yields
    \begin{align*}
    \E_{\tilde W}\bigg|&\sum_v\Tr\big[M_{v,\tilde W}\big(\mathcal{D}_Z(V\phi_v^{\Enc(\tilde W)} V^\dag)-V\phi_v^{\Enc(\tilde W)} V^\dag\big)\big]\\
    &\quad-\sum_v\Tr\big[P^v_{\tilde W}M_{v,\tilde W} P^v_{\tilde W}\big(\mathcal{D}_Z(V\pe{\tilde W} V^\dag)-V\pe{\tilde W} V^\dag\big)\big]\bigg|\leq 2\delta.\numberthis\label{eqn:rounded-negation}
    \end{align*}
    Write $A$ and $B$ for the two summed traces inside \cref{eqn:rounded-negation}, so $\E_{\tilde W}|A-B|\leq 2\delta$. It remains to bound the signed average of $B$. Since $\mathcal{D}_Z(V\pe{\tilde W} V^\dag)-V\pe{\tilde W} V^\dag$ does not depend on $v$, we can pull the sum over $v$ inside the trace and define:
    \[
     M_{\tilde W}\deq\sum_vP^v_{\tilde W}M_{v,\tilde W} P^v_{\tilde W}.
    \]
    We claim that $\{M_{\tilde W},\id-M_{\tilde W}\}$ is a binary POVM and that $M_{\tilde W}$ is uniformly efficient given $\tilde W$.
    \par
    \medskip
    The family $\{P^v_{\tilde W}\}_v$ is a complete set of orthogonal projectors:
    \[
    P^v_{\tilde W}P^{v'}_{\tilde W}=\delta_{vv'}P^v_{\tilde W},\qquad \sum_v P^v_{\tilde W}=\id\ot(\ketbra{0}{0}+\ketbra{1}{1})\ot\id=\id.
    \]
    Combined with $0\preceq M_{v,\tilde W}\preceq\id$ this gives $0\preceq P^v_{\tilde W}M_{v,\tilde W} P^v_{\tilde W}\preceq P^v_{\tilde W}$, hence $0\preceq M_{\tilde W}\preceq\sum_vP^v_{\tilde W}=\id$.
    \par
    \medskip
    Given $\tilde W$, the binary POVM $\{M_{\tilde W},\id-M_{\tilde W}\}$ can be implemented sequentially: first projectively measure $\{P^v_{\tilde W}\}_v$ to obtain an outcome $v\in\zo^n$. Then apply $\{M_{v,\tilde W},\id-M_{v,\tilde W}\}$, which is uniformly efficient given $v$ and $\tilde W$. With this identification, \cref{cor:flag-dephasing} applied with the question-indexed family $\{M_{\tilde W}\}_{\tilde W}$ (uniformly efficient by the sequential implementation above) and $\nu$ the uniform distribution over $\tilde W\in\{X,Y,Z,F,G\}^n$ gives
    \[
    \left|\E_{\tilde W}\Tr\big[M_{\tilde W}\big(\mathcal{D}_Z(V\pe{\tilde W} V^\dag)-V\pe{\tilde W} V^\dag\big)\big]\right|\leq O((\eps+\ea)^{1/2}).\numberthis\label{eqn:contradiction-final}
    \]
    The triangle inequality $\bigl|\E_{\tilde W}A\bigr|\le\E_{\tilde W}|A-B|+\bigl|\E_{\tilde W}B\bigr|$ together with \cref{eqn:rounded-negation} and \cref{eqn:contradiction-final} yields the claim.
\end{proof}

With the above we have all required tools in place to show our random remote state preparation guarantee. We will use \cref{lemma:state-rounding-it} to project the post-Alice state into the eigenspace of the outcome projectors corresponding to the certified observables. Then we will use \cref{lemma:full-flag-dephasing} to get rid of potential coherences in the flag register. Along the way we will need to argue that different junk states on the auxiliary system, are computationally indistinguishable, making sure that they can't leak information about the state to the prover; this can be achieved by a reduction to IND-CPA security of the QFHE, through \cref{lemma:state-characterization}.  

\begin{theorem}\label{theorem:rsp-guarantee}
    For any computationally efficient prover modeled as in Section \ref{section:modeling} that wins with probability $\omega^* - \eps$ in the compiled Clifford test $\textprotocol{cliff}(X,Y,Z,F,G,n)$, obtained from \cref{protocol:cliff-test}, there exists a complex Hilbert space $\hat\cH$ of finite dimension, an efficient isometry $V: \cH \to (\C^2)^n\ot \C^2 \ot\hat\cH$ (from \cref{def:clifford-isometry}), and a cryptographically small function $\ea$ such that
    \begin{align*}
        \E_{\tilde W}\sum_{v\in\zo^n} \ketbra{v,\tilde W}{v, \tilde W}_W\ot V\phi_v^{\Enc(\tilde W)}V^\dag \overset{c}{\approx}_{\gamma} \E_{v,\tilde W}\ketbra{v, \tilde W}{v, \tilde W}_W\ot\left(\tau_{\tilde W}^v\otimes \rho_{0} + \overline{\tau}_{\tilde W}^v\otimes \rho_{1}\right),
    \end{align*}
    where $\gamma = O((\eps+\ea)^{1/2})$, the expectation on $\tilde W$ is uniform over $\{X,Y,Z,F,G\}^n$ and that on $v$ is uniform over $\zo^n$. The verifier holds the $W$ register, let
    \[
    \phi_v^{\Enc(\tilde W)} \deq \sum_{\alpha:\Dec(\alpha)=v} \pea{\tilde W},
    \]
    and $\rho_{0}$ and $\rho_{1}$ are sub-normalized states with orthogonal support, such that $\Tr{\rho_0+\rho_1}=1$. Here $\omega^*=\tfrac{1}{6}(5+\cos^2(\frac{\pi}{8}))$ is the optimal winning probability of the Clifford test.
\end{theorem}
\begin{proof}
    Applying \cref{lemma:state-rounding-it} yields an isometry $V: \cH \to (\C^2)^n\ot \C^2 \ot\hat\cH$ such that
    \begin{align*}
        \E_{\tilde W}\sum_\alpha\norm{P_{\tilde W}^{\Dec(\alpha)} V\pea{\tilde W}V^\dag P_{\tilde W}^{\Dec(\alpha)} - V\pea{\tilde W}V^\dag}_1\leq O((\eps+\ea)^{1/2}),
    \end{align*}
    where
    \[
    P_{\tilde W}^{a} = \sum_{k\in\zo}\tau_{\tilde W,k}^{a}\ot\ket{k}\bra{k}\ot\id,
    \]
    next we can combine this with \cref{lemma:full-flag-dephasing} to obtain that for all uniformly efficient families of POVMs $\{M_{v,\tilde W},\id-M_{v,\tilde W}\}_{v,\tilde W}\in L((\C^2)^n\ot \C^2 \ot\hat\cH)$,
    \[
    \left|\E_{\tilde W}\sum_{v}\Tr{M_{v,\tilde W}\left(\theta_v^{\Enc(\tilde W)} - \mathcal{D}_Z\left(P_{\tilde W}^v\,\theta_v^{\Enc(\tilde W)} P_{\tilde W}^v\right)\right)}\right|\leq O((\eps+\ea)^{1/2}),\numberthis\label{eqn:post-isometry-dephasing-char} 
    \]
    where
    \[
    \theta_v^{\Enc(\tilde W)} \deq \sum_{\alpha:\Dec(\alpha)=v} V\pea{\tilde W}V^\dag.
    \]
    We will start by expanding the right term:
    \begin{align*}
    P_{\tilde W}^v\,\theta_v^{\Enc(\tilde W)} P_{\tilde W}^v &= \tau_{\tilde W}^v\ot \ketbra{0}{0} \ot\Tr_{1,2}\left[\left(\tau_{\tilde W}^v\ot \ketbra{0}{0}\ot \id\right)\theta_v^{\Enc(\tilde W)}\right]\\
    &\quad+ \overline{\tau}_{\tilde W}^v\ot \ketbra{1}{1} \ot\Tr_{1,2}\left[\left(\overline{\tau}_{\tilde W}^v\ot \ketbra{1}{1}\ot \id\right)\theta_v^{\Enc(\tilde W)}\right]\\
    &\quad+ \ketbra{\overline{\tau}_{\tilde W}^v}{\tau_{\tilde W}^v}\ot \ketbra{1}{0} \ot\Tr_{1,2}\left[\left(\ketbra{\tau_{\tilde W}^v}{\overline{\tau}_{\tilde W}^v}\ot \ketbra{0}{1}\ot \id\right)\theta_v^{\Enc(\tilde W)}\right]\\
    &\quad + h.c.
\end{align*}
where $h.c.$ denotes the Hermitian conjugate, $\tau_{\tilde W}^v = \ketbra{\tau_{\tilde W}^v}{\tau_{\tilde W}^v}$ and $\overline{\tau}_{\tilde W}^v = \ketbra{\overline{\tau}_{\tilde W}^v}{\overline{\tau}_{\tilde W}^v}$. The notation $\Tr_{1,2}$ denotes the partial trace over the first ($(\C^2)^n$) and second ($\C^2$) register. Applying the dephasing channel removes the coherences and we obtain
\begin{align*}
    \mathcal{D}_Z\left(P_{\tilde W}^v\,\theta_v^{\Enc(\tilde W)} P_{\tilde W}^v\right) &= \tau_{\tilde W}^v\ot \ketbra{0}{0} \ot\Tr_{1,2}\left[\left(\tau_{\tilde W}^v\ot \ketbra{0}{0}\ot \id\right)\theta_v^{\Enc(\tilde W)}\right]\\
    &\quad+ \overline{\tau}_{\tilde W}^v\ot \ketbra{1}{1} \ot\Tr_{1,2}\left[\left(\overline{\tau}_{\tilde W}^v\ot \ketbra{1}{1}\ot \id\right)\theta_v^{\Enc(\tilde W)}\right]\\
    \intertext{Using the definition of $P^v_{\tilde W}$ and the fact that it commutes with computational basis projectors, we obtain}
    &= \tau_{\tilde W}^v\ot \ketbra{0}{0} \ot\Tr_{1,2}\left[\left(\id\ot \ketbra{0}{0}\ot \id\right)P_{\tilde W}^v\,\theta_v^{\Enc(\tilde W)}P_{\tilde W}^v\right]\\
    &\quad+ \overline{\tau}_{\tilde W}^v\ot \ketbra{1}{1} \ot\Tr_{1,2}\left[\left(\id\ot \ketbra{1}{1}\ot \id\right)P_{\tilde W}^v\,\theta_v^{\Enc(\tilde W)}P_{\tilde W}^v\right].
\end{align*}
Applying \cref{lemma:state-rounding-it} (using the data-processing of the trace norm) we can conclude that
\begin{align*}
    \E_{\tilde W}\sum_{v}\norm{\mathcal{D}_Z\left(P_{\tilde W}^v\,\theta_v^{\Enc(\tilde W)} P_{\tilde W}^v\right) - \left(\tau_{\tilde W}^v\ot \theta^{\Enc(\tilde W)}_{v,0} + \overline{\tau}_{\tilde W}^v\ot \theta^{\Enc(\tilde W)}_{v,1}\right)}_1\leq O((\eps+\ea)^{1/2}),   
\end{align*}
where 
\[
    \theta^{\Enc(\tilde W)}_{v, k} \deq \ketbra{k}{k} \ot\overbrace{\Tr_{1,2}\left[\left(\id\ot \ketbra{k}{k}\ot \id\right)\theta_v^{\Enc(\tilde W)}\right]}^{\rho^{\Enc(\tilde W)}_{v, k}}
\]
for $k\in\zo$. Combining this trace-norm bound with \cref{eqn:post-isometry-dephasing-char} via the triangle inequality for signed averages, we obtain
\begin{align*}
    \left|\E_{\tilde W}\sum_{v}\Tr{M_{v,\tilde W}\left(\theta_v^{\Enc(\tilde W)} - \left(\tau_{\tilde W}^v\ot \theta^{\Enc(\tilde W)}_{v,0} + \overline{\tau}_{\tilde W}^v\ot \theta^{\Enc(\tilde W)}_{v,1}\right)\right)}\right|\leq \delta,\numberthis\label{eqn:post-isometry-dephasing-char-2}
\end{align*}
where $\delta\leq O((\eps+\ea)^{1/2})$. What remains is to show that the auxiliary states are computationally indistinguishable from a $\tilde W$- and $v$-independent reference; i.e.\ no computationally bounded prover can learn any information about the basis choice or measurement outcomes from the auxiliary state. Choose any fixed $q\in\Qa$ and define the reference auxiliary states on $\C^2\ot\hat\cH$
\[
\theta_k \deq \theta_k^{\Enc(q)} = \ketbra{k}{k}\ot\rho_k^{\Enc(q)},\qquad \rho_k^{\Enc(q)}\deq\Tr_{1,2}\left[\left(\id\ot\ketbra{k}{k}\ot\id\right)V\peq V^\dag\right],
\]
for $k\in\zo$. Note that $\theta_0$ and $\theta_1$ have orthogonal supports on the flag register and that 
\[
\Tr{\theta_0+\theta_1}=\Tr{\rho_0^{\Enc(q)}+\rho_1^{\Enc(q)}}=\Tr{V\peq V^\dag}=1.
\]
Since the compiled Clifford group test is a subtest of the compiled Clifford test (and the lemma hypothesis ensures a winning probability of $1-O(\eps+\ea)$ in that subtest) we can apply a weakened version of \cref{lemma:state-characterization} (ignoring the possible $\alpha$ dependence of the POVM) and conclude that for the same isometry $V$ and any efficient POVM $\{M,\id-M\}\in L((\C^2)^n\ot\C^2\ot\hat\cH)$,
\begin{align*}
\left|\Tr{M\left(V\peq V^\dag-\frac{\id}{2^n}\ot\rho^{\Enc(q)}\right)}\right|\leq\xi,\numberthis\label{eqn:state-char-q}
\end{align*}
where $\rho^{\Enc(q)}\deq\Tr_1\left[V\peq V^\dag\right]$ is the post-isometry qubit-marginal on $\C^2\ot\hat\cH$, and $\xi\leq O((\eps+\ea)^{1/2})$. Recalling that $\Tr_{2}$ denotes the partial trace over the flag register, note that
\[
\Tr_2\left[\left(\ketbra{k}{k}\ot\id\right)\rho^{\Enc(q)}\right]=\rho_k^{\Enc(q)}.
\]

The next steps are inspired by the proof of Proposition 4.32 in \cite{Gheorghiu2022}, but we give a direct argument instead of a proof by contradiction. We will show that for any uniformly efficient family of two-outcome POVMs $\{\Lambda_{v,\tilde W},\id-\Lambda_{v,\tilde W}\}_{v,\tilde W}\in L((\C^2)^n\ot\C^2\ot\hat\cH)$,
\begin{align*}
    \left|\E_{\tilde W}\sum_{v}\Tr{\Lambda_{v,\tilde W}\left(\tau_{\tilde W}^v\ot \left(\theta^{\Enc(\tilde W)}_{v, 0}-\tfrac{1}{2^n}\theta_{0}\right) + \overline{\tau}_{\tilde W}^v\ot \left(\theta^{\Enc(\tilde W)}_{v, 1}-\tfrac{1}{2^n}\theta_{1}\right)\right)}\right|\leq \delta+\xi+\ea.\numberthis\label{eqn:post-isometry-dephasing-char-3}
\end{align*}
Define the mutually orthogonal projectors
\[
P^v_{\tilde W,0}\deq\tau_{\tilde W}^v\ot\ketbra{0}{0}\ot\id,
\qquad
P^v_{\tilde W,1}\deq\overline{\tau}_{\tilde W}^v\ot\ketbra{1}{1}\ot\id.
\]
The operator against which $\Lambda_{v,\tilde W}$ is traced is block diagonal and supported on these two projectors. Therefore, replacing $\Lambda_{v,\tilde W}$ by its pinching
\[
\Lambda_{v,\tilde W}\longmapsto \sum_{k\in\zo}P^v_{\tilde W,k} \Lambda_{v,\tilde W}P^v_{\tilde W,k}
\]
leaves the quantity in \cref{eqn:post-isometry-dephasing-char-3} unchanged. This replacement preserves uniform efficiency, since the projectors are efficiently measurable given $v$ and $\tilde W$. Consequently, without loss of generality,
\[
\Lambda_{v,\tilde W} = \tau^v_{\tilde W}\ot\ketbra{0}{0}\ot\Lambda_{v,\tilde W,0}+ \overline{\tau}^v_{\tilde W}\ot\ketbra{1}{1}\ot\Lambda_{v,\tilde W,1},\numberthis\label{eqn:lambda-classical-mixture}
\]
with $0\preceq\Lambda_{v,\tilde W,k}\preceq\id$ for all $v\in\zo^n$, $k\in\zo$, and the family $\{\Lambda_{v,\tilde W,k}\}_{v,\tilde W}$ uniformly efficient given $\tilde W$. Using \cref{eqn:lambda-classical-mixture} we can derive the following identity:
\[
    \Tr{\Lambda_{v,\tilde W} \paren{\tfrac{\id}{2^n}\ot\rho^{\Enc(q)}}}=\Tr{\Lambda_{v,\tilde W} \paren{\tau_{\tilde W}^v\ot \tfrac{1}{2^n}\theta_{0} + \overline{\tau}_{\tilde W}^v\ot \tfrac{1}{2^n}\theta_{1}}}\numberthis\label{eqn:lambda-sum-rho-enc-q}
\]
We will now expand the summand of \cref{eqn:post-isometry-dephasing-char-3} using a telescoping sum:
\begin{align*}
    \sum_v&\Tr{\Lambda_{v,\tilde W}\left(\tau_{\tilde W}^v\ot \left(\theta^{\Enc(\tilde W)}_{v, 0}-\tfrac{1}{2^n}\theta_{0}\right) + \overline{\tau}_{\tilde W}^v\ot \left(\theta^{\Enc(\tilde W)}_{v, 1}-\tfrac{1}{2^n}\theta_{1}\right)\right)}\\
    \intertext{Telescoping over $\sum_v\theta_v^{\Enc(\tilde W)}$ and introducing a second sum on the first term, using the exact structure of $\Lambda_{v,\tilde W}$ given in \cref{eqn:lambda-classical-mixture} (specifically orthogonality), we obtain:}
    &=\sum_{v}\Tr{\Bigl(\sum_{v'}\Lambda_{v',\tilde W}\Bigr)\paren{\bigl(\tau_{\tilde W}^{v}\ot \theta^{\Enc(\tilde W)}_{v,0} + \overline{\tau}_{\tilde W}^{v}\ot \theta^{\Enc(\tilde W)}_{v,1}\bigr)-\theta_{v}^{\Enc(\tilde W)}}}\\
    \intertext{Using $\sum_{v}\theta_{v}^{\Enc(\tilde W)}=V\pe{\tilde W}V^\dag$ and telescoping over $V\peq V^\dag$:}
    &\quad+\Tr{\Bigl(\sum_{v'}\Lambda_{v',\tilde W}\Bigr)\paren{V\pe{\tilde W}V^\dag -V\peq V^\dag}}\\
    \intertext{Correcting for $V\peq V^\dag$ and using \cref{eqn:lambda-sum-rho-enc-q} for the last term:}
    &\quad+\Tr{\Bigl(\sum_{v'}\Lambda_{v',\tilde W}\Bigr)\Bigl(V\peq V^\dag-\frac{\id}{2^n}\ot\rho^{\Enc(q)}\Bigr)}.
\end{align*}

Let 
\[
\Gamma_{\tilde W}\deq V^\dag\overbrace{\Big(\sum_v \Lambda_{v,\tilde W} \Big)}^{M_{\tilde W}}V\in L(\cH).
\]
By the form of \cref{eqn:lambda-classical-mixture} (orthogonal projectors on first two systems), $\{\Gamma_{\tilde W},\id-\Gamma_{\tilde W}\}$ is indeed a POVM and it's efficient given $\tilde W$, since one can apply the efficient isometry $V$, measure the second register in the computational basis (obtaining outcome $k$), then measure the first register in the bases specified by $\tilde W$ (or their complex conjugate depending on the flag outcome $k$) and then apply $\Lambda_{v,\tilde W, k}$ which is efficient given both outcomes and $\tilde W$. With this, the telescoping sum (taking the expectation over $\tilde W$) and the triangle inequality, we have
\begin{align*}
    \Biggl|\E_{\tilde W}\sum_v&\Tr{\Lambda_{v,\tilde W}\left(\tau_{\tilde W}^v\ot \left(\theta^{\Enc(\tilde W)}_{v, 0}-\tfrac{1}{2^n}\theta_{0}\right) + \overline{\tau}_{\tilde W}^v\ot \left(\theta^{\Enc(\tilde W)}_{v, 1}-\tfrac{1}{2^n}\theta_{1}\right)\right)}\Biggr|\\
    &\leq\Biggl|\E_{\tilde W}\sum_{v}\Tr{M_{\tilde W}\paren{\bigl(\tau_{\tilde W}^{v}\ot \theta^{\Enc(\tilde W)}_{v,0} + \overline{\tau}_{\tilde W}^{v}\ot \theta^{\Enc(\tilde W)}_{v,1}\bigr)-\theta_{v}^{\Enc(\tilde W)}}}\Biggr|\\
    &\quad+\Biggl|\E_{\tilde W}\Tr{\Gamma_{\tilde W}\paren{\pe{\tilde W}-\peq}}\Biggr|\\
    &\quad+\Biggl|\Tr{\Bigl(\E_{\tilde W}M_{\tilde W}\Bigr)\Bigl(V\peq V^\dag-\frac{\id}{2^n}\ot\rho^{\Enc(q)}\Bigr)}\Biggr|\\
    &\leq \delta + \xi + \ea,
\end{align*}
where the last inequality follows from \cref{eqn:post-isometry-dephasing-char-2}, \cref{eqn:state-char-q} (since $\E_{\tilde W}M_{\tilde W}$ is an efficient POVM) and \cref{lemma:poly-state-indistinguishability} applied with $D$ sampling $\tilde W$ uniformly and $(z_0,z_1)=(\tilde W, q)$, with the family $\{\Gamma_{\tilde W}\}_{\tilde W}$ uniformly efficient in $\tilde W$.
\par
\medskip
This shows \cref{eqn:post-isometry-dephasing-char-3}. Combining \eqref{eqn:post-isometry-dephasing-char-3} with \eqref{eqn:post-isometry-dephasing-char-2} via the triangle inequality for signed averages, we obtain
\begin{align*}
    \left|\E_{\tilde W}\sum_{v}\Tr{M_{v,\tilde W}\left(\theta_v^{\Enc(\tilde W)} - \tfrac{1}{2^n}\left(\tau_{\tilde W}^v\ot \theta_{0} + \overline{\tau}_{\tilde W}^v\ot \theta_{1}\right)\right)}\right|\leq O((\eps+\ea)^{1/2}),\numberthis\label{eqn:state-guarantee-final}
\end{align*}
Identifying $\theta_v^{\Enc(\tilde W)} = V\phi_v^{\Enc(\tilde W)} V^\dag$ and $\rho_k=\theta_k$ for $k\in\zo$ in the theorem statement concludes the proof, since the $W$ register in the theorem statement is classical, so without loss of generality the efficient POVM (implicit in the computational indistinguishability notation) can be written as 
\[
M=\sum_{v\in\zo^n}\ketbra{v,\tilde W}{v,\tilde W}\ot M_{v,\tilde W}.
\]
\end{proof}

\section{BQP verification}\label{chapter:verification}

\subsection{Protocol}
\begin{protocolbox}{Verification protocol} \label{protocol:verification}
    \textbf{Notation:} $\textprotocol{verify}(C, m, p)$\\
    \textbf{Input:} An $n$-qubit circuit $C$ compiled in the universal gate set $\{\sigma_X, \sigma_Z,T,H,CNOT\}$, where every $H$-gate is replaced by $H(TTH)^3$ (i.e.\ $HPHPHPH$), and a subtest probability $0<p<1$.
    \par
    \medskip
    Let $t$ be the number of $T$ gates in $C$ after this replacement ($t_0$ of which have even parity, meaning an even number of Hadamard gates precede them in the compiled circuit). The argument $m=\Theta(n+t)$ should be such that every symbol appears at least $n+t$ times in a uniformly random $\tilde W\in\{X,Y,Z,F,G\}^m$ with probability $1-e^{-O(m)}$.
    \par 
    \medskip
    Execute the following tests with probability $p$ and $1-p$, respectively.

    \begin{enumerate}
    \item \textbf{State test:} Execute $\textprotocol{cliff}(X,Y,Z,F,G,m)$.
    \item \textbf{Delegation game:} Sample $\tilde W\draw\{X,Y,Z,F,G\}^m$ and send it to Alice, receiving answer $v\in\zo^m$. If $\tilde W$ does not contain every symbol at least $n+t$ times, reject. Otherwise, execute each of the following subtests with equal probability (1/3 each).
    \par
    \medskip
    In each subtest the index sets $N,T_0,T_1$ are sampled uniformly among all tuples of disjoint subsets of the stated sizes whose positions carry the required symbols. We write $P_S \coloneqq \{i: \tilde W_i\in S\}$ for $S\subseteq\{X,Y,Z,F,G\}$ to denote the different symbol classes.
    \begin{enumerate}
        \item \textbf{Computation run:} Choose $N$ uniformly among all $n$-subsets of $P_{\{Z\}}$. Sample $T_0$ and $T_1$ uniformly among disjoint subsets of $P_{\{G,F\}}$ of sizes $t_0$ and $t-t_0$, respectively.
        
        Let $a=\restr{v}{N}$, $b=0^n$, $d\draw\zo^t$, and let $y'\in\zo^t$ with $y'_i=0$ if $(\restr{\tilde W}{T_0\cup T_1})_i=G$ and $y'_i=1$ otherwise; then $y=y'\oplus d$ and $e=\restr{v}{T_0\cup T_1}\oplus d$. 
        \par
        \medskip
        Send $(C,N,T_0, T_1)$ to Bob and execute \textit{Interactive Proof System 1} from \cite{Broadbent2018} starting at step A.3 with keys $(a,b,d,e,y)$ on the qubits in $N$, $T_0$ and $T_1$.

        \item \textbf{X-test run:} Choose $N$ and $T_0$ uniformly among disjoint subsets of $P_{\{Z\}}$ of sizes $n$ and $t_0$, respectively. Sample $T_1$ uniformly among all $(t-t_0)$-subsets of $P_{\{X,Y\}}$.
        
        Let $a=\restr{v}{N}$, $b=0^n$, $d_0=\restr{v}{T_0}$, $d_1=\restr{v}{T_1}$ and let $y\in\zo^{t-t_0}$ with $y_i=0$ if $(\restr{\tilde W}{T_1})_i=X$ and $y_i=1$ otherwise. 
        \par
        \medskip
        Send $(C,N,T_0, T_1)$ to Bob and execute \textit{Interactive Proof System 1} from \cite{Broadbent2018} starting at step B.3 with keys $(a,b,d_0,d_1,y)$ on the qubits in $N$, $T_0$ and $T_1$.

        \item \textbf{Z-test run:} Choose $T_0$ uniformly among all $t_0$-subsets of $P_{\{X,Y\}}$ first. Then sample $N$ uniformly among all $n$-subsets of $P_{\{X\}}\setminus T_0$ and $T_1$ uniformly among all $(t-t_0)$-subsets of $P_{\{Z\}}$.
        
        Let $a=0^n$, $b=\restr{v}{N}$, $d_0=\restr{v}{T_0}$, $d_1=\restr{v}{T_1}$ and let $y\in\zo^{t_0}$ with $y_i=0$ if $(\restr{\tilde W}{T_0})_i=X$ and $y_i=1$ otherwise.
        \par
        \medskip
        Send $(C,N,T_0, T_1)$ to Bob and execute \textit{Interactive Proof System 1} from \cite{Broadbent2018} starting at step C.3 with keys $(a,b,d_0,d_1,y)$ on the qubits in $N$, $T_0$ and $T_1$.
    \end{enumerate}

    \end{enumerate}
    \end{protocolbox}

\subsection{Soundness}
\begin{lemma}[Soundness against malicious `Bob']\label{lemma:soundness-against-malicious-bob}
    Suppose the verifier executes the compiled Verification protocol $\textprotocol{verify}(C,m,p)$, obtained from \cref{protocol:verification}, where\footnote{$\Pi_0$ denotes the projector onto the $\ket{0}$ state of the first output qubit of the circuit, this indicates acceptance.} $\norm{\Pi_0C\ket{0^n}}^2\leq 1/3$, with a prover $P^*$ (not necessarily efficient) such that the shared state after the encrypted interaction is
    \[
    \E_{v\in\zo^m}\E_{\tilde W\in\{X,Y,Z,F,G\}^m}\ketbra{v, \tilde W}{v, \tilde W}_W\ot\left(\tau_{\tilde W}^v\otimes \rho_{0} + \overline{\tau}_{\tilde W}^v\otimes \rho_{1}\right),
    \]
    where the verifier holds the $W$ register and $\rho_0$ and $\rho_1$ are sub-normalized states with orthogonal support, such that $\Tr{\rho_0+\rho_1}=1$. Then the verifier accepts in the delegation game with probability at most $7/9$.
\end{lemma}
\begin{proof}
    The proof follows from the soundness of the Broadbent protocol. Write $\sigma_{\tilde W}^v$ for the restriction of $\tau_{\tilde W}^v$ to the qubits in $N$, $T_0$ and $T_1$. Since $\tau_{\tilde W}^v$ is a product state and $v$ is uniform, the unused qubits are independent of the keys derived from $\restr{v}{N\cup T_0\cup T_1}$ and may be absorbed into the prover's private space. Thus, on the support of $\rho_0$, the prover holds (up to this private space) the prepare-and-send state of Interactive Proof System 1 from \cite{Broadbent2018} with the keys of \cref{protocol:verification}; on the support of $\rho_1$ he holds the complex conjugate of that state in the computational basis. The expectations over $\tilde W$ and $v$ supply the uniformly random keys required by Broadbent's analysis. If $\tilde W$ lacks sufficiently many of each symbol the verifier rejects, which can only decrease the acceptance probability, so we may condition on the index sets being well-defined.
    \par
    \medskip
    The remaining interaction is a (possibly adaptive) measurement of this cq-state, so by linearity the acceptance probability is a convex combination of the two branches. On the $\rho_0$ branch the instance is a genuine execution of the Broadbent protocol. 
    \par
    \medskip
    On the $\rho_1$ branch, let $\widetilde P$ be the complex conjugate of $P^*$'s operations in the unencrypted interaction, i.e.\ the strategy obtained by conjugating every Kraus operator in the computational basis (again a valid strategy, since conjugation preserves $\sum_j K_j^\dagger K_j = \mathbb{1}$, though not necessarily efficiency). Since the prover's answers are outcomes of computational-basis measurements, whose projectors are real, and the verifier's messages and acceptance predicate are classical functions of the transcript, $\widetilde P$ on $\sigma_{\tilde W}^v \otimes \overline{\rho_1}$ induces the same transcript distribution as $P^*$ on $\overline{\sigma_{\tilde W}^v} \otimes \rho_1$, by the following identity:
    \[
    \Tr{M\bigl(\overline{\sigma}_{\tilde W}^v\otimes\rho_1\bigr)}=\Tr{\overline{M}\bigl(\sigma_{\tilde W}^v\otimes\overline{\rho}_1\bigr)},
    \]
    Hence $\widetilde P$ is a (possibly inefficient) Broadbent prover on the correct plaintext states. Broadbent's soundness is information-theoretic and gives acceptance at most $7/9$ on a no-instance for Interactive Proof System 1, where the different run types are indistinguishable \cite[Section 7.6]{Broadbent2018}, which is the case in \cref{protocol:verification} since the prior interaction is encrypted and the extra message which Bob receives $(C,N,T_0, T_1)$ is equally distributed in all three run types. Indeed, the sampling procedure of \cref{protocol:verification} is invariant under any permutation of $[m]$ that fixes the symbol classes $P_S$ setwise, since $N,T_0,T_1$ are each chosen uniformly among subsets of a fixed size within fixed symbol classes; hence, conditioned on $\tilde W$ containing enough of each symbol, $(N,T_0,T_1)$ is uniformly distributed over all admissible triples of disjoint subsets of $[m]$ of the required sizes and symbol classes. The rejection event depends only on the symbol counts of $\tilde W$, not on which positions realize them, so it is independent of $(N,T_0,T_1)$ given those counts; thus the law of $(N,T_0,T_1)$, conditioned on non-rejection, is uniform over admissible triples independently of the run type, and so is the law of $(C,N,T_0,T_1)$. The winning probability of both branches is thus bounded by $7/9$, and so is $P^*$.
\end{proof}

\begin{lemma}\label{lemma:soundness-rounding}
    Suppose the verifier executes the compiled Verification protocol $\textprotocol{verify}(C,m,p)$, obtained from \cref{protocol:verification}, where $\norm{\Pi_0C\ket{0^n}}^2\leq 1/3$, with a computationally efficient prover $P^*$, such that the prover is accepted with probability at least $\omega^*-\eps$, for some $\eps > 0$, in the state test. Here $\omega^*=\tfrac{1}{6}(5+\cos^2(\frac{\pi}{8}))$ is the optimal winning probability of the Clifford test. Then the verifier accepts $P^*$ in the delegation game with probability at most $\frac{7}{9}+\gamma$, where $\gamma = O((\eps+\ea)^{1/2})$.
\end{lemma}
\begin{proof}
    By \cref{theorem:rsp-guarantee} (applied with qubit count $m$, since the state test is $\textprotocol{cliff}(X,Y,Z,F,G,m)$), success of at least $\omega^*-\eps$ in the state test implies the existence of a complex Hilbert space $\tilde\cH$ of finite dimension, an efficient isometry $V: \cH \to (\C^2)^m\ot\tilde\cH$ and a cryptographically small function $\ea$ such that
    \begin{align*}
        \overbrace{\E_{\tilde W}\sum_{v\in\zo^m} \ketbra{v,\tilde W}{v, \tilde W}_W\ot V\phi_v^{\Enc(\tilde W)}V^\dag}^{\Theta} \overset{c}{\approx}_{\gamma} \overbrace{\E_{v,\tilde W}\ketbra{v, \tilde W}{v, \tilde W}_W\ot\left(\tau_{\tilde W}^v\otimes \rho_{0} + \overline{\tau}_{\tilde W}^v\otimes \rho_{1}\right)}^{\Omega},\numberthis\label{eqn:soundness-rounding-state-test}
    \end{align*}
    where $\gamma = O((\eps+\ea)^{1/2})$ and $\phi_v^{\Enc(\tilde W)}$ is $P^*$'s state after the encrypted interaction (marginalized per decrypted outcome). 
    \par
    \medskip
    After the encrypted round the remaining interaction consists of $\poly(m)$ rounds of cleartext messages: the verifier samples the run type and the data $(N,T_0,T_1,d,\dots)$ as a function of the register $W$ and his own randomness, the prover replies, and the verifier applies a classical acceptance predicate to the transcript. The accepting contribution of $P^{*}$'s continuation, averaged over the verifier's randomness, can be written as a family of two-outcome POVMs $\{D_{v,\tilde W},\mathbb 1-D_{v,\tilde W}\}$ on $\cH$ with
    \[
    p_\text{accept} \;=\; \E_{\tilde W}\sum_{v\in\{0,1\}^m}
            \Tr\bigl[D_{v,\tilde W}\,\phi^{\Enc(\tilde W)}_v\bigr].
    \]
    The family depends on $(v,\tilde W)$ because the verifier's unencrypted messages are computed from $v=\Dec(\alpha)$, which the prover cannot compute himself; it is uniformly efficient in $(v,\tilde W)$, since the verifier is classical polynomial time, $P^*$'s continuation is efficient by assumption, and there are only $\poly(m)$ rounds. Let
    \[
    D:=\sum_{v,\tilde W}\ketbra{v,\tilde W}{v,\tilde W}_W\otimes D_{v,\tilde W},
    \qquad
    M:=(\mathbb 1_W\otimes V)\,D\,(\mathbb 1_W\otimes V^\dag).
    \]
    Since $V$ is an isometry we have $0\preceq M\preceq(\mathbb 1_W\otimes VV^\dagger)\preceq\mathbb 1$,
    so $\{M,\mathbb 1-M\}$ is a POVM, and it is uniformly efficient because $D$ and $V$ are.
    \par
    \medskip
    By $V^\dag V=\mathbb 1$ and cyclicity of the trace,
    $\Tr[V D_{v,\tilde W}V^\dag\, V\phi^{\Enc(\tilde W)}_v V^\dag]
    =\Tr[D_{v,\tilde W}\,\phi^{\Enc(\tilde W)}_v]$, hence
    \[
    \Tr[M\Theta] \;=\; p_\text{accept}\numberthis\label{eqn:soundness-rounding-povm-acceptance}.
    \]
    Since $M$ is a uniformly efficient POVM element, \cref{eqn:soundness-rounding-state-test} gives
    \[
    \bigl|\Tr[M\Theta]-\Tr[M\Omega]\bigr|\;\le\;\gamma\numberthis\label{eqn:soundness-rounding-povm-difference}.
    \]

    We claim that $\Tr[M\Omega]$ lower-bounds the acceptance probability of some prover holding the ideal state, i.e.\ the state on the right-hand side of \cref{eqn:soundness-rounding-state-test}. Recall that $M$ was derived from the particular interaction between the verifier and $P^*$, mapped into the dilated space of the isometry $V$. It remains to show that there exists a prover strategy, which acts on the ideal state (that lives in the dilated space), whose transcript distribution is identical to that of $P^*$ (on the parts of the ideal state that live inside the image of $V$).
    \par
    \medskip
    To this end define $U$ as the efficient unitary extension of $V$, such that 
    \[
    V\ket{\psi} = U(\ket{\psi}\ot\ket{0})
    \]
    for all $\ket{\psi}\in\cH$. Then we can define $P'$ as applying $U^\dag$, measuring the ancilla qubits in the computational basis, aborting if the outcome isn't zero and otherwise applying $P^*$'s continuation on $\cH$. The element of the verifier's interaction with $P'$ is then $A=M+A_\text{abort}$, where $A_\text{abort}$ is PSD, such that
    \[
    \Tr[M\Omega]\leq\Tr[A\Omega]\leq \frac{7}{9}.
    \]
    Here the last inequality follows from \Cref{lemma:soundness-against-malicious-bob}, since $P'$ holds the ideal state. Combining with \cref{eqn:soundness-rounding-povm-acceptance} and \cref{eqn:soundness-rounding-povm-difference} yields,
    \[
    p_\text{accept} \;\le\; \frac{7}{9}+\gamma,
    \]
    which completes the proof.
\end{proof}

\begin{definition}[$\textnormal{Q-CIRCUIT}$]\label{def:q-circuit}
    The input to the promise problem $\textnormal{Q-CIRCUIT}$ consists of a quantum circuit $C=C_T\cdots C_1$ acting on $n$ qubits, given in the universal gateset $\{\sigma_X,\sigma_Z,H,\mathrm{CNOT},T\}$ (as in \cref{protocol:verification}, every $H$ gate is replaced by $H(TTH)^3$). Let $\Pi_0$ denote the projector onto $\ket{0}$ of the first output qubit, and let
    \[
    p(C)\deq\norm{\Pi_0 C\ket{0^n}}^2
    \]
    be the probability of observing $0$ as a result of a computational-basis measurement of that qubit after evaluating $C$ on $\ket{0^n}$. Then define $\textnormal{Q-CIRCUIT}=\{\textnormal{Q-CIRCUIT}_{\mathrm{YES}},\,\textnormal{Q-CIRCUIT}_{\mathrm{NO}}\}$ with
    \[
    \textnormal{Q-CIRCUIT}_{\mathrm{YES}}\coloneqq \{C: p(C)\ge 2/3\}\,,\qquad
    \textnormal{Q-CIRCUIT}_{\mathrm{NO}}\coloneqq \{C: p(C)\le 1/3\}\,.
    \]
\end{definition}
    
\begin{theorem}\label{theorem:bqp-verification}
    Assuming LWE is hard for non-uniform quantum adversaries, in sense of \cref{def:crypto-small}, there exist constants $0<p<1$ and $\Delta>0$, $\delta_{SK}>0$ and thresholds $\lambda_0$, $m_0$ such that the following holds.
    For every instance $C$ of $\textnormal{Q-CIRCUIT}$ (\cref{def:q-circuit}) with $m=\Theta(|C|)\ge m_0$ and with security parameter $\lambda\ge\lambda_0$ instantiated from $m$ as in \cref{rem:lambda-vs-n}, the compiled\footnote{Since the test has multiple rounds of interaction with Bob, compiled here means that the first interaction with Alice is encrypted and all interactions with Bob happen in the clear.} protocol $\textprotocol{verify}(C,m,p)$ of \cref{protocol:verification} satisfies:
    \begin{enumerate}
        \item (Completeness:) If $C\in\textnormal{Q-CIRCUIT}_{\mathrm{YES}}$, then there is a strategy for the prover, consuming $O(\poly(\lambda, \log |C|)|C|)$ total resources, that is accepted with probability at least 
        \[
        p_c = p\left(\frac{1}{6}\left(5+\cos^2(\frac{\pi}{8})\right)\right) + (1-p)\frac{8}{9}-e^{-O(m)}-\ea-\delta_{\mathrm{SK}},
        \]
        where the cryptographically small term $\ea$ accounts for the correctness error of the QFHE evaluation, while $\delta_{\mathrm{SK}}$ accounts for gate-synthesis error.
        \item (Soundness:) If $C\in\textnormal{Q-CIRCUIT}_{\mathrm{NO}}$, then any efficient prover strategy is accepted with probability at most $p_s = p_c - \Delta$.
    \end{enumerate}
\end{theorem}
\begin{remark}[On the constants $\lambda_0,m_0,\Delta$]
The functions $\gamma$ and $\eta$ appearing in the proof below are, individually, only guaranteed to be cryptographically small for the specific prover under consideration; a priori this could make the threshold beyond which $\eta(\lambda)<\eps^*$ prover-dependent, which would make $\lambda_0$ (and hence $\Delta$) depend on the prover as well. This is not the case: by \cref{rem:non-uniformity}, for every $\lambda$ we may hardwire, as classical advice, the choice of prover (among the non-uniform class fixed by \cref{def:crypto-small} for the assumed hardness of LWE) that maximizes $\eta(\lambda)$, obtaining a single cryptographically small function that dominates $\eta(\lambda)$ for \emph{every} efficient prover simultaneously. It is this dominating function that is used to fix $\eps^*$, and hence $\lambda_0$ and $\Delta$, below; both are therefore genuine constants, valid uniformly over the entire class of efficient provers, not merely for a single fixed one.
\end{remark}
\begin{proof}
    \textbf{Completeness. } Completeness is straightforward to show; for more details about the honest prover strategy see \cref{table:honest-prover} and \cref{sec:protocols}. The Clifford subtest ($\textsc{cliff}$) has completeness $\frac{1}{6}(5+\cos^2(\frac{\pi}{8}))$, which is achieved by the prover holding $m+2$ EPR pairs and performing the honest measurements. By our consistency convention, the honest Alice measurements are actually the transposes of the requested operators. The delegation subtest corresponds to an execution of the prepare-and-send version of the Broadbent protocol, which has perfect completeness in the test rounds and completeness at least $\frac{2}{3}$ in the computation rounds (cf. \cite[Section 6]{Broadbent2018}), which yields overall completeness of at least $\frac{8}{9}$ since all three rounds are executed with equal probability. The correction term $e^{-O(m)}$ is due to the fact that a uniformly random $\tilde W\in\{X,Y,Z,F,G\}^m$ does not always contain every symbol sufficiently often, and the cryptographically small term $\ea$ accounts for the correctness error of the homomorphic evaluation performed by the honest prover in the encrypted round. Lastly, there is also an error term stemming from the synthesis of the honest Alice circuit into the gate set supported by QFHE. Let $G_{\mathrm{SK}}$ be the number of gates in that circuit that require approximation (in our construction, only controlled-$T$ gates). Synthesize each such gate to operator-norm error at most $\delta_{\mathrm{SK}}/(2G_{\mathrm{SK}})$. It is then possible to bound the diamond-norm distance between the ideal and synthesized Alice channels by $\delta_{\mathrm{SK}}$, which only changes the prover's acceptance probability by at most $\delta_{\mathrm{SK}}$. $\delta_{SK}$ is an independent and constant error term. Solovay--Kitaev gives a per-gate synthesis overhead of $\mathrm{polylog}(G_{\mathrm{SK}}/\delta_{\mathrm{SK}})$, which was already accounted for.
    \par
    \medskip
    We will now show that the basis choices mandated by the protocol correspond to the key assignments in the Broadbent protocol, taking the computation runs as example. On the qubits in $N$ we have $\tilde W_i=Z$, so the prover holds the $\sigma_Z$-eigenstate $\ket{v_i}=\sigma_X^{v_i}\ket{0}$; this is the one-time-padded all-zero input of \cite{Broadbent2018}, with $X$-key $a=\restr{v}{N}$ and trivial $Z$-key $b=0^n$. On the qubits in $T_0\cup T_1$ we have $\tilde W_i\in\{G,F\}$, uniformly. Writing $\ket{+_\theta}=(\ket{0}+e^{i\theta}\ket{1})/\sqrt2$, the eigenstates of $\sigma_G$ with outcomes $v_i=0,1$ are $\ket{+_{\pi/4}}$ and $\ket{+_{5\pi/4}}$, and those of $\sigma_F$ are $\ket{+_{3\pi/4}}$ and $\ket{+_{7\pi/4}}$, respectively. In our gate notation,
    \[
    \ket{+_\theta} = \sigma_X^{d_i} \sigma_Z^{e_i} \sigma_S^{y_i} \sigma_T\ket{+}
    \]
    up to global phase, with $\theta = \pi/4 + (y_i\oplus d_i)\pi/2+(e_i\oplus d_i)\pi$ for any $i\in[t]$. Substituting the keys $d\draw\zo^t$, $y=y'\oplus d$ and $e=\restr{v}{T_0\cup T_1}\oplus d$ of \cref{protocol:verification} recovers exactly the four eigenstates above, so the physical state is consistent with Broadbent's $T$-gadget auxiliaries.
    \par
    \medskip
    The $X$- and $Z$-test runs follow the same pattern; \cref{table:broadbent-key-correspondence} summarizes the correspondence for all three run types.
    \begin{table}[h]
    \centering
    \begin{tabular}{@{}lllll@{}}
    \toprule
    Run & Positions & $\tilde W$ symbol & Broadbent aux state & Keys\\
    \midrule
    Computation & $N$ & $Z$ & $\sigma_X^{v_i}\ket0$ & $a=\restr{v}{N},\ b=0^n$\\[0.2em]
     & $T_0\cup T_1$ & $G,F$ & $\sigma_X^{d_i}\sigma_Z^{e_i}\sigma_S^{y_i}\sigma_T\ket+$ & $d,\ y=y'\oplus d,\ e=\restr{v}{T_0\cup T_1}\oplus d$\\[0.5em]
    \hline\\[-0.5em]
    $X$-test & $N,T_0$ & $Z$ & $\sigma_X^{v_i}\ket0$ & $a=\restr{v}{N},\ b=0^n,\ d_0=\restr{v}{T_0}$\\[0.2em]
     & $T_1$ & $X,Y$ & $\sigma_S^{y_i}\sigma_Z^{(d_1)_i}\ket+$ & $d_1=\restr{v}{T_1},\ y_i=0\Leftrightarrow\tilde W_i=X$\\[0.5em]
    \hline\\[-0.5em]
    $Z$-test & $N$ & $X$ & $\sigma_Z^{v_i}\ket+$ & $b=\restr{v}{N},\ a=0^n$\\[0.2em]
     & $T_0$ & $X,Y$ & $\sigma_S^{y_i}\sigma_Z^{(d_0)_i}\ket+$ & $d_0=\restr{v}{T_0},\ y_i=0\Leftrightarrow\tilde W_i=X$\\[0.2em]
     & $T_1$ & $Z$ & $\sigma_X^{(d_1)_i}\ket0$ & $d_1=\restr{v}{T_1}$\\
    \bottomrule
    \end{tabular}
    \caption{Correspondence between the position classes sampled by \cref{protocol:verification} and the auxiliary states and keys of \cite{Broadbent2018}'s Interactive Proof System 1, for all three run types.}
    \label{table:broadbent-key-correspondence}
    \end{table}
    For the $X$-test, the $\sigma_X$/$\sigma_Y$-eigenstates on $T_1$ with outcome $(d_1)_i$ are exactly $\sigma_Z^{(d_1)_i}\ket+$ (for $\tilde W_i=X$, $y_i=0$) and $\sigma_S\sigma_Z^{(d_1)_i}\ket+$ (for $\tilde W_i=Y$, $y_i=1$), matching the definition of $y$ in \cref{protocol:verification}; the $Z$-test entries follow by the same computation with the roles of $(N,a,b)$ and of the two test-position classes exchanged. This shows that the physical state prepared on every position class, in every run type, is exactly the Broadbent auxiliary state with the keys assigned by \cref{protocol:verification}.
    \par
    \medskip
    The joint law of $(N,T_0,T_1)$, and hence of the keys, does not depend on the run type, and the unused positions are maximally mixed and independent of the keys, by the argument given in the proof of \cref{lemma:soundness-against-malicious-bob}. Finally, in every run type the interaction with Bob continues past the step named in \cref{protocol:verification}: steps A.4--A.5 (resp.\ B.4--B.5, C.4--C.5) of \cite{Broadbent2018}---the output measurement, decryption, and accept predicate---are executed as well, exactly as in the Broadbent protocol. Together with the Clifford-test completeness above, this shows that completeness of our protocol indeed follows from the completeness of the Clifford test and the Broadbent protocol.
    \par
    \medskip
    \textbf{Soundness. } It remains to show that any efficient prover will be accepted with probability at most $p_c - \Delta$, in an execution of $\textprotocol{verify}(C,m,p)$, where $C\in\textnormal{Q-CIRCUIT}_{\mathrm{NO}}$. We can make the following case distinction, where $\eps^*$ is a constant parameter we will set at the end of the proof:
    \begin{enumerate}[label={},leftmargin=14pt,labelsep=0pt]
        \item \emph{Good prover:} The prover is accepted in the state test with probability at least $\omega^*-\eps^*$.
        \item \emph{Bad prover:} The prover is accepted in the state test with probability at most $\omega^*-\eps^*$.
    \end{enumerate}
    
    In the case of a good prover, we invoke \cref{lemma:soundness-rounding} to conclude that $P^*$ succeeds with probability at most $\frac{7}{9}+\gamma(\eps^*+\ea)$. Since no efficient prover can win the state test with probability exceeding $\omega^*+\ea$ (cf.\ \cref{theorem:clifford-mixed-basis}), the upper bound on the winning probability of a good prover is $p_\text{good} \leq p(\omega^*+\ea) + (1-p)\left(\frac{7}{9}+\gamma(\eps^*+\ea)\right)$.
    \par
    \medskip
    In the case of a bad prover, we can't use the rigidity test, nor the soundness guarantee of the Broadbent protocol (because we can't round the state). The upper bound on the winning probability of a bad prover is $p_\text{bad} \leq p(\omega^*-\eps^*) + (1-p)$, since his winning probability in the state test is at most $\omega^*-\eps^*$.
    \par
    \medskip
    Since one of the two cases must hold, in general the winning probability of any prover on a no-instance is
    \[
    \max\left\lbrace p(\omega^*+\ea) + (1-p)\left(\frac{7}{9}+\gamma(\eps^*+\ea)\right), p(\omega^*-\eps^*) + (1-p)\right\rbrace.
    \]
    By choosing the protocol constant $p$ sufficiently close to $1$, specifically
    \[
    p = p^*(\eps^*)\deq \frac{\frac{2}{9}-\gamma(\eps^*)}{\frac{2}{9}-\gamma(\eps^*) + \eps^*}
    \]
    the first argument (corresponding to the good prover case) can be made to always dominate the maximum (the first argument dominates whenever $p(\eps^*+\ea)\geq(1-p)\bigl(\tfrac{2}{9}-\gamma(\eps^*+\ea)\bigr)$, and $p^*(\eps^*)$ satisfies this already for $\ea=0$). Then the highest possible winning probability on a no-instance is
    \[
    p_s = p^*(\eps^*)(\omega^*+\ea) + (1-p^*(\eps^*))\left(\frac{7}{9}+\gamma(\eps^*+\ea)\right).
    \]
    Recalling the completeness probability 
    \[
    p_c = p^*(\eps^*)\left(\frac{1}{6}\left(5+\cos^2(\frac{\pi}{8})\right)\right) + (1-p^*(\eps^*))\frac{8}{9}-e^{-O(m)}-\ea-\delta_{\mathrm{SK}},
    \]
    the completeness--soundness gap becomes
    \[
    \Delta = p_c - p_s = (1-p^*(\eps^*))\left(\frac{1}{9}-\gamma(\eps^*+\ea)\right) - O(\ea) - e^{-O(m)}-\delta_{\mathrm{SK}}.
    \]
    We are now in a position to set $\eps^*$: if we choose it such that $\gamma(2\eps^*) < 1/9$, then for all sufficiently large $\lambda$, $\eta(\lambda) < \eps^*$, so we have $\gamma(\eps^* + \eta(\lambda)) < 1/9$. This means there is some positive constant $K$ such that
    \[ \Delta \geq K - O(\eta(\lambda)) - e^{-O(m)} - \delta_{SK}. \]  We then choose the constant $\delta_{\mathrm{SK}}$  smaller than, for example, $K/2$. If $m$ and $\lambda$ are also sufficiently large, the terms $e^{-O(m)}$ and $O(\ea)$ do not change the sign of the gap, so $\Delta$ remains a positive constant. The more prover errors we wish to tolerate, the smaller the gap becomes and the closer to $1$ the probability of selecting a state test round ($p$) has to be. The only relevant regime is the one where $\gamma(\eps^*+\ea)< \frac{1}{9}$. In this work the analysis was performed without keeping exact count of the constants and the large constant for the fundamental anti-commutation test (see \cref{lemma:compiled-anti-commutation}) propagates through the analysis, requiring a very small $\eps^*$, and yielding a small but constant gap. The protocol is thus far from practical in its current state, although we believe that the constants can be reduced by a tighter analysis.
    \par
    \medskip
    Regarding resource requirements, an honest prover needs $O(\poly(\lambda)\poly(\log m)m)$ qubits (including EPR pairs) to pass the rigidity test, see \cref{subsec:qfhe-overhead} for more details regarding the overhead of the QFHE scheme. Similarly, the communication is of order $O(\poly(\lambda)m)$ bits. Completeness and soundness hold for a generic cryptographically small $\ea$ against efficient provers, in the sense of \cref{def:crypto-small}. Instantiating QFHE under polynomial hardness of LWE requires $\lambda=m^{\Omega(1)}$ (\cref{rem:lambda-vs-n}) and yields almost-linear overhead $O(m^{1+\eps})$ for every $\eps>0$. Instantiating under sub-exponential hardness of LWE, with $\lambda=(\log m)^{\Theta(1/\delta)}$ for some $\delta\in(0,1)$, yields resource requirements of $\widetilde{O}(m)$. When using our rigidity test in the nonlocal setting the honest prover resource requirements are only $O(m)$, since we manage to avoid the sampling overhead of the extended Pauli braiding test in \cite{Coladangelo2024} and do not consider blindness at this point.
\end{proof}

\subsection{Sequential repetition}
The gap $\Delta$ provided by \cref{theorem:bqp-verification} is a positive but small constant, so the completeness and soundness bounds sit near the Clifford test value $\omega^*$. The next lemma shows that the standard, sequential, threshold-based gap amplification works for our construction (allowing only classical advice), by performing a constant number of repetitions, with independent keys. The reduction only requires classical advice: leftover states are sampled by restarting the prover, rather than hardwired as quantum advice.
\begin{lemma}[Sequential repetition]\label{lemma:sequential-repetition}
    Let $p$, $\Delta$, $p_c$ and $p_s=p_c-\Delta$ be as in \cref{theorem:bqp-verification}, and assume $\lambda$ and $m$ are large enough for that theorem to apply. There exist constants $k\in\Nat$ and $t\in[k]$, independent of $\lambda$ and $|C|$, such that the following hold for the $k$-fold sequential repetition of the compiled Verification protocol $\textprotocol{verify}(C,m,p)$, obtained from \cref{protocol:verification}, in which each execution uses independently sampled keys and questions, and the verifier accepts if and only if at least $t$ executions accept.
    \begin{enumerate}
        \item (Completeness:) If $C\in\textnormal{Q-CIRCUIT}_{\mathrm{YES}}$, then there is a strategy for the prover, consuming $O(\poly(\lambda, \log |C|)|C|)$ total resources, that is accepted with probability at least $2/3$.
        \item (Soundness:) If $C\in\textnormal{Q-CIRCUIT}_{\mathrm{NO}}$, then any efficient prover strategy is accepted with probability at most $1/3$.
    \end{enumerate}
\end{lemma}
\begin{proof}
    For all sufficiently large $\lambda$, $p_c\geq\alpha$ is bounded from below and $p_s\leq\beta$ from above by constants separated by a positive gap $\Delta\coloneqq\alpha-\beta$.
    Let $\tau\deq (\alpha+\beta)/2=\alpha-\Delta/2$ denote the midpoint of the gap which is constant, and set $t\deq\lceil k\tau\rceil$ for a constant $k$ to be chosen below.
    \par
    \medskip
    \textbf{Completeness. } By the independence of the sequential runs, completeness follows from a concentration bound and choosing an appropriate number of repetitions. On a yes-instance the honest prover of \cref{theorem:bqp-verification} is accepted with probability at least $p_c\geq\alpha$ in a single execution. Repeating that strategy independently---preparing a fresh state at the start of each execution---yields $k$ independent Bernoulli random variables $Y_1,\dots,Y_k$ (indicating acceptance) with mean at least $p_c$ and Hoeffding's inequality gives for all $a>0$
    \[
        \Pr\Bigl[\sum_{i=1}^k Y_i\leq k\alpha - a\Bigr]\le\exp\bigl(-2a^2/k\bigr),
    \]
    where we used that
    \[
        \E\Bigl[\sum_{i=1}^k Y_i\Bigr]\geq k\alpha\,.
    \]
    The sequential verifier accepts iff $\sum_{i=1}^k Y_i\ge t$. Setting $a=k\alpha-t$ yields the following bound on the verifier's rejection probability
    \[
    \Pr\Bigl[\sum_{i=1}^k Y_i< t\Bigr]\le\exp\bigl(-2(k\alpha-t)^2/k\bigr).
    \]
    From $t\le k\tau+1$ and $\tau=\alpha-\Delta/2$ we have $k\alpha-t\ge k\Delta/2-1$. For $k\ge 8/\Delta$ this is at least $k\Delta/4$, hence
    \[
    \exp\bigl(-2(k\alpha-t)^2/k\bigr)\le\exp\bigl(-2(k\Delta/4)^2/k\bigr)=\exp(-k\Delta^2/8).
    \]
    Choosing $k\geq(8/\Delta^2)\ln 3>8/\Delta$ (since $\Delta<1$) makes the rejection probability at most $1/3$. The $k$ executions contribute only a constant factor to the resource bound of \cref{theorem:bqp-verification}.
    \par
    \medskip
    \textbf{Soundness. } Let $B$ be an efficient sequential prover, with only classical advice, and write $X=(X_1,\dots,X_k)\in\zo^k$ for the accept/reject bits of $k$ executions. For $i\in[k]$ and $u\in\zo^{i-1}$ write $p_i(u)\deq\Pr[X_{<i}=u]$ and $r_i(u)\deq\Pr[X_i=1\mid X_{<i}=u]$, taking $r_i(u)=0$ if $p_i(u)=0$. Fix a threshold constant $\gamma\deq 2^{-k}/6$. A prefix $u$ is \emph{common} if $p_i(u)\ge\gamma$ and \emph{rare} otherwise.
    \par
    \medskip
    Choose any $i\in[k]$. We first claim that $r_i(u)\le p_s+\ea$ for every \emph{common} prefix $u\in\zo^{i-1}$. Define a single-shot prover $A_{i,u}$ against $\textprotocol{verify}(C,m,p)$ as follows: the pair $(i,u)$ is classical advice. The adversary $A_{i,u}$ repeats the following procedure at most $\nu\deq\lceil\lambda/\gamma\rceil$ times: prepare the starting state of $B$ (which is efficient in our adversary model), internally simulate verifiers $1,\dots,i-1$ by sampling their keys and questions independently, perform the interaction and check whether the resulting verdict string equals $u$. If they match, $A_{i,u}$ relays $B$'s $i$-th execution to the external verifier; if no match is found, in $\nu$ iterations, $A_{i,u}$ aborts and is rejected. Each trial succeeds with probability $p_i(u)\ge\gamma$, so the probability that all $\nu$ trials fail is at most $(1-p_i(u))^\nu\leq(1-\gamma)^\nu\le e^{-\lambda}$, which is cryptographically small. The state of $B$ on success is distributed exactly as $B$'s leftover state conditioned on $X_{<i}=u$, hence
    \[
    \Pr[A_{i,u}\text{ is accepted}]=r_i(u)\bigl(1-(1-p_i(u))^\nu\bigr)\ge r_i(u)\bigl(1-e^{-\lambda}\bigr).
    \]
    The adversary $A_{i,u}$ is efficient in the sense of \cref{def:crypto-small}: $k$ is a constant, so $\nu=O(\lambda)$ and $A_{i,u}$ incurs only a linear overhead in $\lambda$ over $B$, in both the polynomial and the sub-exponential instantiations. \Cref{theorem:bqp-verification} therefore bounds its acceptance probability by $p_s$, which rearranges to $r_i(u)\le p_s+\ea$.
    \par
    \medskip
    Since we don't allow quantum auxiliary inputs in our model, the leftover state after $i$ executions, conditioned on a specific verdict, needs to be efficiently preparable. Based on the transcripts alone this is not the case, since it requires post-selecting on the prover's outgoing messages. However, splitting verdict prefixes into common and rare, the argument above shows that preparing the leftover state in the common case is efficient, by restarting $B$ from its classical advice until the verdict string matches.
    \par
    \medskip
    Let $R$ be the event that there exists an $i\in [k]$ such that the prefix $X_{<i}$ of $X$ is rare. There are at most $2^{i-1}$ prefixes of length $i-1$, each rare one having probability mass less than $\gamma$, so
    \[
    \Pr[R]\le\sum_{i=1}^k 2^{i-1}\gamma<2^k\gamma=\frac16.
    \]
    On the complementary event ($\neg R$) every prefix $X_{<i}$ is common, hence every successive conditional satisfies $r(X_{<i})\le p_s+\ea$. Coupling the verdict bits to independent Bernoulli random variables $Z_1,\dots,Z_k$ of mean $p_s+\ea$ in the usual way (draw $U_i$ uniformly from $[0,1]$ and set $X_i=\ind{U_i<r(X_{<i})}$, $Z_i=\ind{U_i<p_s+\ea}$) yields $X_i\le Z_i$ on $\neg R$, and therefore
    \[
    \Pr\Bigl[\sum_{i=1}^k X_i\ge t\Bigr]\le\Pr[R]+\Pr\Bigl[\sum_{i=1}^k Z_i\ge t\Bigr]\le\frac16+\Pr\Bigl[\sum_{i=1}^k Z_i\ge t\Bigr]\numberthis\label{eqn:sequential-repetition-soundness-bound}.
    \]
    The random variables $Z_1,\dots,Z_k$ are independent Bernoulli with mean $p_s+\ea$. Hoeffding's inequality gives, for all $a>0$,
    \[
        \Pr\Bigl[\sum_{i=1}^k Z_i \ge k(\tau-\Delta/4) + a\Bigr]\le\exp\bigl(-2a^2/k\bigr),
    \]
    where we used that for all sufficiently large $\lambda$ one has $p_s+\ea\le\tau-\Delta/4$, and thus
    \[
        \E\Bigl[\sum_{i=1}^k Z_i\Bigr]\le k(\tau-\Delta/4)\,,
    \]

    Setting $a=t-k(\tau-\Delta/4)$ (which is positive, since $t=\lceil k\tau\rceil\ge k\tau$) yields
    \[
        \Pr\Bigl[\sum_{i=1}^k Z_i \ge t\Bigr]\le\exp\bigl(-2\bigl(t-k(\tau-\Delta/4)\bigr)^2/k\bigr).
    \]
    From $t\ge k\tau$ we have $t-k(\tau-\Delta/4)\ge k\Delta/4$, hence
    \[
        \exp\bigl(-2\bigl(t-k(\tau-\Delta/4)\bigr)^2/k\bigr)\le\exp\bigl(-2(k\Delta/4)^2/k\bigr)=\exp(-k\Delta^2/8).
    \]
    Choosing $k\geq(8/\Delta^2)\ln 6> 8/\Delta$ (since $\Delta<1$) upper bounds the quantity by $1/6$, and the sequential verifier accepts $B$ with probability at most $1/3$, by the above and \cref{eqn:sequential-repetition-soundness-bound}.
\end{proof}

\begin{corollary}\label{cor:bqp-argument}
    Assuming LWE is hard for non-uniform quantum adversaries in the sense of \cref{def:crypto-small}, every language $L\in\mathsf{BQP}$ admits a single-prover, classical-verifier argument with completeness $2/3$, soundness $1/3$ and total resource requirements
    \[
    O(\poly(\lambda,\log |C_x|)|C_x|),
    \]
    where $C_x$ is a circuit deciding $L$ on input $x$ and $\lambda$ is the LWE security parameter. Instantiating the security parameter as in \cref{rem:lambda-vs-n} yields almost-linear overhead $O(|C_x|^{1+\eps})$ for every $\eps>0$ under polynomial hardness of LWE, and quasilinear overhead $\widetilde{O}(|C_x|)$ under sub-exponential hardness.
\end{corollary}
\begin{proof}
    The promise problem $\textnormal{Q-CIRCUIT}$ is $\mathsf{BQP}$-complete: for every language $L\in\mathsf{BQP}$ there exists a classical polynomial-time reduction which, on input $x$, outputs a circuit $C_x$ in the gateset of \cref{def:q-circuit}, of size $\poly(|x|)$, such that $C_x\in\textnormal{Q-CIRCUIT}_{\mathrm{YES}}$ if $x\in L$ and $C_x\in\textnormal{Q-CIRCUIT}_{\mathrm{NO}}$ if $x\notin L$.
    \par
    \medskip
    An argument system for all $L\in\BQP$ then immediately follows from \cref{lemma:sequential-repetition}: on input $x$, the verifier constructs $C_x$ and runs the $k$-fold sequential repetition of $\textprotocol{verify}(C_x,m,p)$, with $m=\Theta(|C_x|)$.
\end{proof}

\appendix
\section{Supplementary Material}\label{appendix:supplementary-material}

\subsection{Compiled (anti-)commutation tests}
\begin{lemma}\label{lemma:compiled-commutation}
    Let $P^*$ be any computationally efficient prover modeled as in Section \ref{section:modeling} that succeeds with probability $1-\varepsilon$ in the compiled commutation game $\textprotocol{com}(A,B)$, obtained from Protocol \ref{protocol:commutation-test}. Then for all $q\in\Qa$ there exists a cryptographically small function $\ea$ such that
    \[
    \|[A,B]\|_{\psi^{\Enc(q)}}^2 \leq O(\varepsilon)+\ea,
    \]
    with $A,B\in\Obs(\H)$, $P^*$'s observables that correspond to single-bit answer question labels.
\end{lemma}
\begin{proof}
    Let $W=(A,B)$, then $W_1=A$ and $W_2=B$. By a slight abuse of notation, these are simultaneously the labels for the questions sent by the verifier as well as the observables that the prover applies in the unencrypted part of the interaction. We know from the protocol specification that $\alpha$ is checked to be the encryption of a tuple, since the prover immediately loses if this is not the case, we can assume $\Dec(\alpha)\in\{0,1\}^2$. The winning probability is then given by
    \[
    p_\text{win}=\frac{1}{2}\left(1+\E_b \sum_\alpha (-1)^{\Dec(\alpha)_b}\Tr{W_b\pea{W}}\right)\geq 1-\varepsilon,
    \]
    where we inserted the assumed winning probability from the lemma statement. Rearranging, we obtain
    \[
    \E_b \overbrace{\sum_\alpha (-1)^{\Dec(\alpha)_b}\Tr{W_b\pea{W}}}^\gamma\geq 1-2\varepsilon.
    \]
    Since $\gamma\leq 1$ and $b$ is uniformly distributed over two options (see specification of protocol \ref{protocol:commutation-test}), the following must hold
    \[
    \gamma=\sum_\alpha (-1)^{\Dec(\alpha)_b}\Tr{W_b\pea{W}}\geq 1-4\varepsilon.
    \]
    We have
    \begin{align*}
        \sum_\alpha\|W_b-(-1)^{\Dec(\alpha)_b}\id\|_{\pea{W}}^2 = 2 - 2\sum_\alpha (-1)^{\Dec(\alpha)_b}\Tr{W_b\pea{W}}\leq 8\varepsilon,
    \end{align*}
    for any $b\in[2]$. Since the upper-bound holds for a convex combination of sums of non-negative terms, we know that
    \[
    \|W_b-(-1)^{\Dec(\alpha)_b}\id\|_{\pea{W}}^2 \leq \varepsilon_{\alpha},
    \]
    with $\sum_\alpha \varepsilon_{\alpha} = 8\varepsilon$. Since the identity matrix always commutes, we can use the fact that $W_1=A$ and $W_2=B$ approximately act as plus or minus the identity on Alice's post-measurement state to conclude that they must also approximately commute. For fixed $\alpha$, we have
    \begin{align*}
        W_1W_2&\approx_{\varepsilon_{\alpha}} W_1(-1)^{\Dec(\alpha)_2}\id = (-1)^{\Dec(\alpha)_2}W_1\\
        (-1)^{\Dec(\alpha)_2}W_1 &\approx_{\varepsilon_{\alpha}} (-1)^{\Dec(\alpha)_2}(-1)^{\Dec(\alpha)_1} = (-1)^{\Dec(\alpha)_1}(-1)^{\Dec(\alpha)_2}\\
        (-1)^{\Dec(\alpha)_1}(-1)^{\Dec(\alpha)_2}&\approx_{\varepsilon_{\alpha}} (-1)^{\Dec(\alpha)_1} W_2 = W_2(-1)^{\Dec(\alpha)_1}\\
        W_2(-1)^{\Dec(\alpha)_1}&\approx_{\varepsilon_{\alpha}} W_2W_1.
    \end{align*}
    Chaining these approximate equalities together, we obtain
    \[
    \|[W_1,W_2]\|_{\pea{W}}^2  = \|W_1W_2-W_2W_1\|_{\pea{W}}^2 \leq O(\varepsilon_{\alpha}).
    \]
    Reintroducing the sum over the Alice answer $\alpha$, this becomes
    \[
    \|[W_1,W_2]\|_{\psi^{\Enc(W)}}^2 = \sum_\alpha \|[W_1,W_2]\|_{\pea{W}}^2 \leq O(\varepsilon).
    \]
    Since the commutator is an efficient LCU, we can switch out the state by \cref{cor:efficient-operator-state-switching} (choosing a point distribution on $q$) and obtain
    \[
    \|[W_1,W_2]\|_{\psi^{\Enc(q)}}^2 \leq O(\varepsilon) + \ea,
    \]
    which completes the proof by the definition of $W_1$ and $W_2$ as $A$ and $B$, respectively.
\end{proof}

\begin{lemma}\label{lemma:compiled-anti-commutation}
    Let $P^*$ be any computationally efficient prover modeled as in Section \ref{section:modeling} that succeeds with probability $1-\varepsilon$ in the compiled anti-commutation game $\textprotocol{ac}(A,B)$, obtained from Protocol \ref{protocol:anti-commutation-test}. Then for all $q\in\Qa$ there exists a cryptographically small function $\ea$ such that
    \[
    \|\{A,B\}\|_{\psi^{\Enc(q)}}^2 \leq O(\varepsilon)+\ea,
    \]
    with $A,B\in\Obs(\H)$, $P^*$'s observables that correspond to single-bit answer question labels.
\end{lemma}
\begin{proof}
    In \cite[Theorem 6.7]{Cui2024} it is shown that $\exists\; \eta'(\lambda)$, s.t. $\forall\;q\in\mathcal{Q}_{MS}$ (the questions used in the magic-square test)
    \[
    \underset{\substack{\mathsf{sk}\leftarrow\mathsf{Gen}(1^\lambda)\\c\leftarrow \Enc_\mathsf{sk}(q)}}{\mathbb{E}}\sum_\alpha \|\{B_2,B_4\}\, \ket{\psi_{\alpha c}}\|^2\leq 17280\varepsilon + 52\eta'(\lambda),
    \]
    where $\eta'(\lambda)$ is a negligible function.
    In our setting $B_2=A$ (represented by question $A$), $B_4=B$ (represented by question $B$), observing $\ket{\psi_{\alpha c}} = \kca $. Since the anti-commutator is an efficient LCU, we can extend the statement to hold for any $q\in\Qa$ by \cref{cor:efficient-operator-state-switching}. Their result also holds in the setting of sub-exponential security of LWE (where $\nuQPT$/negligible are replaced by the second two classes of \cref{def:crypto-small}), as their protocol can be instantiated from the same QFHE scheme as ours and their analysis only performs the same same type of IND-CPA invocations as our work. The function $\eta(\lambda)\coloneqq 52\eta'(\lambda)$ is thus cryptographically small in general.
\end{proof}

\subsection{Isometry circuit calculation}\label{sec:circuit-calculation}
Using the expressions from the previous section, we constructed an explicit quantum circuit, where the choice of gates was informed by replicating the required phases and cross-register correlations (represented by the delta functions).
\par
\medskip
We will explicitly verify that the circuit in \cref{fig:circuit} implements the correct operations. The boxed operation on the first two qubits yields
\[
\ket{p_1,p_0,a,b,c}\mapsto \frac{1}{2}\sum_{\ell_0,\ell_1\in\zo} i^{(2p_1+p_0)(2\ell_1+\ell_0)}\ket{\ell_0,\ell_1,a,b,c}.
\]
This step already ensures that the factor shared between the classical and quantum irreps appears. We will first analyze the first half of the circuit, where every operation is controlled on the fact that $\ell_0=0$, after the first two gates, we have
\begin{align*}
    \ket{0,\ell_1,a,b,c}&\mapsto \sqrt{\frac{1}{2^{n}}}\sum_{\mu\in\zo^n}\prod_{j\in[n]}(-1)^{\ell_1 c_j+\mu_j c_j}\ket{0,\ell_1,a,b,\mu}
\end{align*}
The doubly controlled Hadamard gate then yields
\begin{align*}
    \ket{0,\ell_1,a,b,\mu}\mapsto\sqrt{\frac{2^{|\mu|}}{2^{n}}}\sum_{r\in\zo^n}\left(\prod_{j\in S(\mu\oplus 1^n)} (-1)^{a_j r_j}\right)\left(\prod_{j\in S(\mu)} \delta_{r_j,a_j}\right)\ket{0,\ell_1,r,b,\mu}.
\end{align*}
Combining the two we obtain
\begin{align*}
\ket{0,\ell_1,a,b,c}\mapsto\frac{1}{2^{n/2}}\sum_{\mu,r\in\zo^n}\sqrt{\frac{2^{|\mu|}}{2^{n}}}
&\left(\prod_{j\in S(\mu\oplus 1^n)} (-1)^{\ell_1c_j+a_j r_j}\right)\\
&\quad\left( \prod_{j\in S(\mu)} (-1)^{(\ell_1+1)c_j}\delta_{r_j,a_j}\right)\ket{0,\ell_1,r,b,\mu}.
\end{align*}
Applying the Hadamard on $b$ we get
\begin{align*}
\ket{0,\ell_1,a,b,c}\mapsto\frac{1}{2^{n}}\sum_{\mu,r,s\in\zo^n}\sqrt{\frac{2^{|\mu|}}{2^{n}}}
&\left(\prod_{j\in S(\mu\oplus 1^n)} (-1)^{\ell_1c_j+a_j r_j+b_js_j}\right)\\
&\quad\left( \prod_{j\in S(\mu)} (-1)^{(\ell_1+1)c_j+b_js_j}\delta_{r_j,a_j}\right)\ket{0,\ell_1,r,s,\mu}.
\end{align*}
The last step consists of a Toffoli and two swap gates, which yield
\begin{align*}
\ket{0,\ell_1,a,b,c}\mapsto\frac{1}{2^{n}}\sum_{\mu,r,s\in\zo^n}\sqrt{\frac{2^{|\mu|}}{2^{n}}}
&\left(\prod_{j\in S(\mu\oplus 1^n)} (-1)^{\ell_1c_j+a_j r_j+b_js_j}\right)\\
&\quad\left(\prod_{j\in S(\mu)} (-1)^{(\ell_1+1)c_j+b_js_j}\delta_{r_j,a_j}\right)\ket{0,\ell_1,\mu, r\oplus (s\cdot\mu), s}.
\end{align*}
Relabeling $r\oplus (s\cdot \mu)\mapsto r$, we obtain
\begin{align*}
\ket{0,\ell_1,a,b,c}\mapsto\frac{1}{2^{n}}\sum_{\mu,r,s\in\zo^n}\sqrt{\frac{2^{|\mu|}}{2^{n}}}
&\left(\prod_{j\in S(\mu\oplus 1^n)} (-1)^{\ell_1c_j+a_j r_j+b_js_j}\right)\\
&\quad\left(\prod_{j\in S(\mu)} (-1)^{(\ell_1+1)c_j+b_js_j}\delta_{r_j\oplus s_j,a_j}\right)\ket{0,\ell_1,\mu, r, s}.
\end{align*}
Putting everything together, the first part of the circuit implements the following operation:
\begin{align*}
\ket{p_1,p_0,a,b,c}\mapsto &\sum_{\ell_1\in\zo}\sum_{\mu,r,s\in\zo^n}\sqrt{\frac{2^{|\mu|}}{2^{3n+2}}}\left[\rho_{C,\mu}^{(2\ell_1)}(i^{2p_1+p_0}x(a)g(b)z(c))\right]_{r,s}\ket{0,\ell_1,\mu, r, s}\\
&+\frac{1}{2}\sum_{\ell_1\in\zo}i^{(2p_1+p_0)(2\ell_1+1)}\ket{1,\ell_1,a,b,c}.
\end{align*}
It remains to explicitly write out the operations that are controlled on $\ell_0=1$, after the Hadamard, $T$, $S$ and $CZ$ gates, we obtain
\begin{align*}
\ket{1,\ell_1,a,b,c}\mapsto \sum_{s\in\zo^n}\sqrt{\frac{1}{2^n}}\prod_{j\in[n]}(-1)^{c_js_j}
&\overbrace{e^{i\frac{\pi}{4}b_j-i\frac{\pi}{2}s_jb_j-i\frac{\pi}{2}\ell_1b_j+i\pi s_jb_j\ell_1}}^{=\omega_{2+\ell_1}(s_j)^{b_j}}\\
&\ket{1,\ell_1,a,b,s},
\end{align*}
where we recall that $\omega_\ell(s)= \mathbf{1}\{\ell = 1\}\omega(s)+\mathbf{1}\{\ell=3\}\overline{\omega}(s)$, with $\omega(s)=e^{i\frac{\pi}{4}(1-2s)}$. Next, we have two $\mathrm{CNOT}$ and a Hadamard gate
\begin{align*}
\ket{1,\ell_1,a,b,c}\mapsto \sum_{\mu,s\in\zo^n}\frac{1}{2^n}\prod_{j\in[n]}(-1)^{c_js_j+a_j\mu_j}\omega_{2+\ell_1}(s_j)^{b_j}
&\ket{1,\ell_1,\mu,b\oplus a\oplus s,s},
\end{align*}
the final operations are two $CZ$ gates, which yield
\begin{align*}
\ket{1,\ell_1,a,b,c}\mapsto \sum_{\mu,s\in\zo^n}\frac{1}{2^n}
&\prod_{j\in[n]}(-1)^{c_js_j+a_j\mu_j+(a_j+b_j+s_j)\mu_j+s_j\mu_j}\omega_{2+\ell_1}(s_j)^{b_j}\\
&\ket{1,\ell_1,\mu,b\oplus a\oplus s,s},
\end{align*}
simplifying the expression we obtain
\begin{align*}
\ket{1,\ell_1,a,b,c}\mapsto \sum_{\mu,r,s\in\zo^n}\frac{1}{2^n}
&\prod_{j\in[n]}(-1)^{c_js_j+b_j\mu_j}\omega_{2+\ell_1}(s_j)^{b_j}\delta_{r_j,s_j\oplus a_j\oplus b_j}\\
&\ket{1,\ell_1,\mu,r,s}.
\end{align*}
Combining this with the analysis of the first half of the circuit, we see that our circuit indeed implements the desired operation, which maps
\begin{align*}
\ket{p_1,p_0,a,b,c}\mapsto \sum_{\ell_0,\ell_1\in\zo}\sum_{\mu,r,s\in\zo^n}\sqrt{\frac{d_\mu}{2^{3n+2}}}
&\left[\rho_{\mu}^{(\ell_0+2\ell_1)}(i^{2p_1+p_0}x(a)g(b)z(c))\right]_{r,s}\\
&\ket{\ell_0,\ell_1,\mu, r, s}.
\end{align*}
Since our circuit only uses a constant number of gates per qubit, it implements $U_{\mathrm{QFT}}$ in $O(n)$, which is efficient in both senses of \cref{def:crypto-small}, under the corresponding relation between $n$ and $\lambda$.

\subsection{Parseval's identity}
\begin{lemma}[Parseval's identity]\label{lemma:parseval-identity}
    Let $g:\Z^n_2\to L(\H)$ and $\widehat{g}:\Z^n_2\to L(\H)$ be its Fourier transform, defined as
    \[
    \widehat{g}(x)=\E_a(-1)^{x\cdot a}g(a).
    \]
    Then
    \[
    \E_a g(a)^\dag g(a)=\sum_x \widehat{g}(x)^\dag\widehat{g}(x).
    \]
\end{lemma}
\begin{proof}
    \begin{align*}
        \sum_x \widehat{g}(x)^\dag\widehat{g}(x) &= \sum_x \Big(\E_a(-1)^{ax}g(a)^\dag\Big)\cdot\Big(\E_b(-1)^{bx}g(b)\Big) = \E_{a,b}\overbrace{\Big(\sum_x (-1)^{(a+b)x}\Big)}^{=0 \text{ if }a\neq b,\text{ else }2^n}g(a)^\dag g(b)\\
        &=\frac{1}{2^{2n}}\left(\sum_a \overbrace{\Big(\sum_x (-1)^{0\cdot x}\Big)}^{=2^n}g(a)^\dag g(a)+\sum_{\substack{a,b\\a\neq b}}\Big(\sum_x (-1)^{(a+b)x}\Big)g(a)^\dag g(b)\right)\\
        &=\E_a g(a)^\dag g(a)
    \end{align*}
\end{proof}

\begin{corollary}\label{cor:parseval}
    Let $g:\Z^n_2\to L(\H)$ and $\widehat{g}:\Z^n_2\to L(\H)$ be its Fourier transform, defined as
    \[
    \widehat{g}(x)=\E_a(-1)^{x\cdot a}g(a).
    \]
    Then for all $\rho\in\Pos(\H)$,
    \[
    \E_a \|g(a)\|^2_\rho=\sum_x \|\widehat{g}(x)\|^2_\rho.
    \]
\end{corollary}
\begin{proof}
    Apply \cref{lemma:parseval-identity} and the definition of the state-dependent norm, with the linearity of the trace.
\end{proof}

\subsection{Useful results}
\begin{lemma}\label{lemma:closeness-equivalence}
Let $\ket{\psi} \in \H$ with $\braket{\psi}{\psi} = \alpha$, $R\in U(\H)$, $S\in U(\H')$ and $V: \H \to \H'$ an isometry. Then
\[
\| (R-V^\dag SV)\ket{\psi} \|^2 \leq \varepsilon \quad \implies \quad \| (VR-SV) \ket{\psi} \|^2 \leq 4 \sqrt{\alpha\varepsilon}, 
\]
and
\[
\| (VR-SV) \ket{\psi} \|^2 \leq \varepsilon \quad \implies \quad \| (R-V^\dag SV)\ket{\psi} \|^2 \leq \varepsilon. 
\]
\end{lemma}
\begin{proof}
We start by showing the first implication. By the assumption of the lemma and the fact that $V$ is an isometry, we have:
\begin{equation}\label{eqn:closeness-equiv-VR-step}
\| (VR-VV^\dag SV) \ket{\psi} \|^2 = \| (R-V^\dag SV)\ket{\psi} \|^2 \leq \varepsilon. 
\end{equation}
Let $\alpha \coloneqq \braket{\psi}{\psi}$, be the squared norm of $\ket{\psi}$. 
The proof idea is to show that by the assumption of the first implication, the projection onto the image space of the isometry of the vector $ SV\ket{\psi}$ is close to the vector $VR\ket{\psi}$.
Specifically this can be shown using that $R$ is a unitary and using the triangle inequality 
\begin{align*}
\sqrt{\alpha} &= \| VR\ket{\psi} \| = \| \left(VV^\dag S V + VR - VV^\dag S V\right)\ket{\psi} \| \\
&\leq \| VV^\dag S V\ket{\psi} \| + \| \left(VR - VV^\dag S V\right) \ket{\psi} \| \\ 
\Leftrightarrow\quad &\| \Pi_V \underbrace{S V \ket{\psi}}_{\ket{\phi}} \| \geq \sqrt{\alpha} - \sqrt{\varepsilon}\quad\implies\quad \bra{\phi}\Pi_V\ket{\phi}\geq \left(\sqrt{\alpha} - \sqrt{\varepsilon}\right)^2,
\end{align*}
with $\Pi_V \coloneqq VV^\dag$, $\ket{\phi} \coloneqq SV\ket{\psi}$, and $\braket{\phi}{\phi} = \alpha$, since $S$ is unitary. 
\par 
\medskip
Since this lower-bound is close to the norm of $\ket{\phi}$ itself, we can conclude that the projection does not have a large effect and $\ket{\phi}$ primarily lies in the image space of the isometry, i.e. the unitary $S$ roughly preserves this space. Quantitatively this is expressed by the following upper-bound on the difference between $\ket{\phi}$ and $\Pi_V\ket{\phi}$:
\begin{equation}\label{eqn:closeness-equiv-projection-gap}
\| (VV^\dag SV - SV)\ket{\psi} \|^2 = \| (\Pi_V -\id)\ket{\phi} \|^2 = \alpha - \bra{\phi}\Pi_V\ket{\phi} \leq 2\sqrt{\alpha\varepsilon}-\varepsilon.
\end{equation}
By combining both observations in a transitive manner, we arrive at the following expression:
\begin{align*}
\| (VR-SV)\ket{\psi} \|^2 &= \| (VR - VV^\dag SV + VV^\dag SV - SV) \ket{\psi} \|^2\\
&\leq 2\varepsilon + 2\left(2\sqrt{\alpha\varepsilon}-\varepsilon\right) = 4\sqrt{\alpha\varepsilon}
\end{align*}
where the inequality follows from \eqref{eqn:closeness-equiv-VR-step} and \eqref{eqn:closeness-equiv-projection-gap} and the triangle inequality. This shows the first implication.
\par 
\medskip
For the reverse implication,
\begin{align*}
\| (R-V^\dag SV) \ket{\psi} \|^2 &= \| V^\dag \overbrace{(VR - SV) \ket{\psi}}^{\ket{\phi'}} \|^2 \\
&= \bra{\phi'}\Pi_V\ket{\phi'} \leq \braket{\phi'}{\phi'} \leq \varepsilon,
\end{align*}
with $\ket{\phi'}\coloneqq(VR - SV) \ket{\psi}$ and where the second-to-last inequality follows from the Cauchy--Schwarz inequality. The last inequality follows from the lemma assumption and concludes the proof of the second implication.
\end{proof}

\begin{lemma}\label{lemma:op-norm-adjoint-invariance}
    Let $A\in L(\H)$. Then $\|A^\dag A\|=\|A\|^2$.
\end{lemma}
\begin{proof}
    For the forward bound,
    \begin{align*}
        \|A\|^2 &= \Big(\sup_{\braket{v}{v}\leq 1}\|A\ket{v}\|\Big)^2 = \sup_{\braket{v}{v}\leq 1}\|A\ket{v}\|^2 = \sup_{\braket{v}{v}\leq 1}\bra{v}A^\dagger A\ket{v}\\
        &\leq \sup_{\braket{v}{v}\leq 1}\|\ket{v}\|\cdot\|A^\dagger A\ket{v}\|\leq \sup_{\braket{v}{v}\leq 1}\|A^\dagger A\ket{v}\| = \|A^\dagger A\|,
    \end{align*}
    where we used a well-known identity for the squared supremum over a set of non-negative reals and the Cauchy--Schwarz inequality. For the reverse bound,
    \begin{align*}
        \|A^\dag A\| = \sup_{\braket{u}{u},\braket{v}{v}\leq 1}|\bra{u}A^\dag A\ket{v}|\leq \sup_{\braket{v}{v}\leq 1}\| A\ket{v}\|^2=\|A\|^2,
    \end{align*}
    where we used the variational definition of the operator norm (in terms of optimizing over two states) and the Cauchy--Schwarz inequality.
\end{proof}

\bibliographystyle{halpha}
\bibliography{main}

\end{document}